\documentclass[11pt]{article}
\usepackage{amsfonts,amsmath,amssymb,amsthm,multirow}
\usepackage[pdftex]{graphicx}
\usepackage[hmargin=1in,vmargin=1in]{geometry}
\usepackage[longnamesfirst]{natbib}
\usepackage{bibunits}
\defaultbibliographystyle{ecta}
\defaultbibliography{BIBLIOGRAPHY}
\usepackage{enumerate,enumitem,caption}
\usepackage{booktabs}
\usepackage[para,online,flushleft]{threeparttable}
\usepackage[usenames,dvipsnames]{xcolor}
\usepackage{footmisc,fancyvrb} %
\usepackage[colorlinks,citecolor=blue!80!black, linkcolor=red!60!violet,urlcolor=blue!80!black]{hyperref}
\usepackage{tikz,graphicx,pgfplots,pgfkeys,subcaption}   %
\usepackage{etoolbox} %
\usepackage{eurosym}
\usepackage[toc,page,header]{appendix}

\usepackage{minitoc}

\VerbatimFootnotes
\usetikzlibrary{datavisualization.formats.functions}
\pgfplotsset{compat=1.12} %

\def\addlegendimage{\csname pgfplots@addlegendimage\endcsname}

\usepackage[pagewise,mathlines]{lineno}
\mathchardef\dash="2D
\theoremstyle{definition}
\newtheorem{theorem}{Theorem}[section]
\newtheorem{lemma}{Lemma}[section]
\newtheorem{definition}{Definition}[section]
\newtheorem{corollary}{Corollary}[section]
\newtheorem{proposition}{Proposition}[section]

\newtheorem{example}{Example}[section]

\newtheorem{assumption}{Assumption}

\newtheorem{remark}{Remark}[section]
\AtEndEnvironment{remark}{~\qed}%

\newcommand{\customfootnotetext}[2]{{%
  \renewcommand{\thefootnote}{#1}%
  \footnotetext[0]{#2}}}%
\usepackage{sepfootnotes}

\DeclareMathOperator*{\rank}{rank}

\DeclareMathOperator*{\argmin}{argmin}

\DeclareMathOperator*{\supp}{supp}
\DeclareMathOperator*{\sign}{sign}
\DeclareMathOperator*{\diag}{diag}

\DeclareMathOperator*{\MTE}{MTE}

\DeclareMathOperator{\cov}{cov}
\DeclareMathOperator{\var}{var}

\newcommand{\R}{\ensuremath{\mathbb{R}}}

\newcommand{\Exp}{\ensuremath{\mathbb{E}}} 
\newcommand{\Prob}{\ensuremath{\mathbb{P}}} 

\usepackage{bbm}
\newcommand{\indicator}{\ensuremath{\mathbbm{1}}}

\newcommand{\normal}{\ensuremath{\mathcal{N}}}

\makeatletter
\newcommand*{\indep}{%
  \mathbin{%
    \mathpalette{\@indep}{}%
  }%
}
\newcommand*{\nindep}{%
  \mathbin{
    \mathpalette{\@indep}{\not}
  }%
}
\newcommand*{\@indep}[2]{%
  \sbox0{$#1\perp\m@th$}
  \sbox2{$#1=$}
  \sbox4{$#1\vcenter{}$}
  \rlap{\copy0}
  \dimen@=\dimexpr\ht2-\ht4-.2pt\relax
  \kern\dimen@
  {#2}%
  \kern\dimen@
  \copy0 
} 
\makeatother

\author{Muyang Ren\thanks{Department of Economics, University of Tennessee, \href{mailto:mren7@utk.edu}{\texttt{mren7@utk.edu}}}}
\title{\textbf{Extrapolating LATE with Weak IVs}\thanks{I am very grateful for my advisor, Matt Masten, for his extensive guidance and encouragement. I am also grateful for my committee members, Federico Bugni, Adam Rosen, and Arnaud Maurel for their generous advice and feedback. I am also thankful to Jason Baron, Patrick Bayer, Xinyue Bei, Anna Bykhovskaya, Désiré Kédagni, Michael Pollmann, Daniel Yi Xu, and auidences of the Duke microeconometrics breakfast, Triangle Econometrics Conference, Hong Kong University, UC Santa Barbara, Rice University, Southern Methodist University, Aarhus University, University of Tennessee, and Queen's University for their useful comments and suggestions.}}
\date{\today}

\begin{document}
\begin{bibunit}   %

\maketitle

\begin{abstract}
     To evaluate the effectiveness of a counterfactual policy, it is often necessary to extrapolate treatment effects on compliers to broader populations. This extrapolation relies on exogenous variation in instruments, which is often weak in practice. This limited variation leads to invalid confidence intervals that are typically too short and cannot be accurately detected by classical methods. For instance, the $F$-test may falsely conclude that the instruments are strong. Consequently, I develop inference results that are valid even with limited variation in the instruments. These results lead to asymptotically valid confidence sets for various linear functionals of marginal treatment effects, including LATE, ATE, ATT, and policy-relevant treatment effects, regardless of identification strength. This is the first paper to provide weak instrument robust inference results for this class of parameters. Finally, I illustrate my results using data from \cite{agan/doleac/harvey:2023} to analyze counterfactual policies of changing prosecutors' leniency and their effects on reducing recidivism.
    \begin{description}
    \medskip
    \item \textit{Keywords}:
    Identification-robust Inference, Subvector Inference, Marginal Treatment Effects, Examiner Design
    \medskip
    \item \textit{JEL classification}: C12, C21, C26.
    \end{description}
\end{abstract}

\doparttoc %
\faketableofcontents

\part{} %

\newpage

\section{Introduction}
This paper provides the first formal treatment of weak identification analysis in the marginal treatment effect (MTE) model. Originally developed in the seminal works of \cite{bjorklund/moffitt:1987} and \cite{heckman/vytlacil:1999, heckman/vytlacil:2001, heckman/vytlacil:2005}, the MTE model has been widely adopted in various studies for extrapolating treatment effects, such as returns to schooling \citep{moffitt:2008, carneiro/heckman/vytlacil:2011}, analysis of recidivism effects \citep{bhuller/dahl/loken/mogstad:2020, agan/doleac/harvey:2023}, and evaluation of social insurance programs \citep{maestas/mullen/strand:2013, aizawa/mommaerts/corina/rennane:2023} (see Table 6 of \cite{mogstad/torgovitsky:2024} for a broad survey of MTE applications). Despite its widespread use in applied economic research, traditional confidence sets for the extrapolated causal effects within this model are often too short when the probability of receiving treatment (i.e., the propensity score) exhibits limited variation across the instrument’s support. Moreover, this variation can be very weak even if the usual $F$-test statistic is very large.

To achieve valid coverage of causal effects, this paper establishes the first set of inference results that are robust against weak IV variation in MTE models. For linear MTE models, I develop an asymptotically similar conditional Wald test that delivers uniformly valid confidence sets with exact coverage. For a more general class of polynomial MTE models, I propose a modified linear combination (MLC) test that produces uniformly valid confidence sets while achieving approximate asymptotic efficiency under strong identification. Additionally, I highlight limitations of the additive separability assumption, a functional form often used to address weak variation of propensity scores, by deriving explicit formulas for the bias of estimands when this specification is incorrectly specified.

\subsection*{Intuition for weak IVs in MTE models}
To demonstrate the consequences of weak IVs in MTE models and to illustrate why the $F$-statistic fails, I plot regression estimates of the outcome variable on propensity scores for treated samples in Figure \ref{fig:intuition_plot}. These estimators are computed from independent simulations based on a cubic MTE model with a discrete IV that takes four values. As will be shown in section \ref{sec:model}, the MTE can be directly recovered from the conditional regression, allowing us to study the weak IV problem by examining the finite-sample behavior of the estimates in Figure \ref{fig:intuition_plot}.

Figures \ref{fig:strong} through \ref{fig:weak} demonstrate that as propensity score variation diminishes, the regression estimates (shown in gray) become increasingly volatile and diverge from the true population quantity (shown in blue). This behavior implies that the MTE estimator loses consistency when propensity scores exhibit limited variation, which is a common consequence with instruments that weakly influence treatment selection. As a result, the traditional Wald confidence interval, which build on these compromised estimates, may fail to achieve its desired coverage probability. These estimation and inference problems arise not only when propensity scores cluster around one value but also when they approximate binary variation (Figure \ref{fig:partial}) under the cubic MTE design. In such cases, classical confidence intervals become unreliable, but the $F$-statistic can be misleadingly large because it detects the deviation from the null where all propensity scores are equal.

\bigskip

\begin{figure}[h!]
    \centering
    \caption{Estimators under Different Types of IV Strengths}
    \begin{subfigure}[t]{0.49\textwidth}
        \caption{Strong Variation}
        \label{fig:strong}
        \centering
        \includegraphics[width=\textwidth]{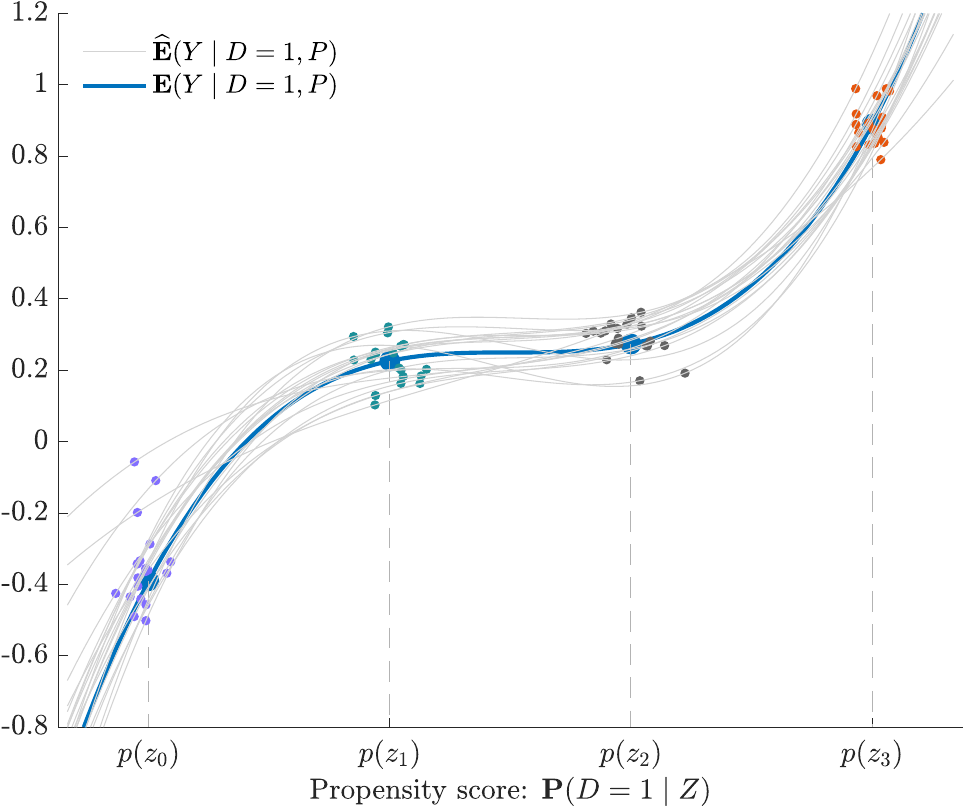}
    \end{subfigure}
    \hfill
    \begin{subfigure}[t]{0.49\textwidth}
        \caption{Intermediate Variation}
        \label{fig:semistrong}
        \centering
        \includegraphics[width=\textwidth]{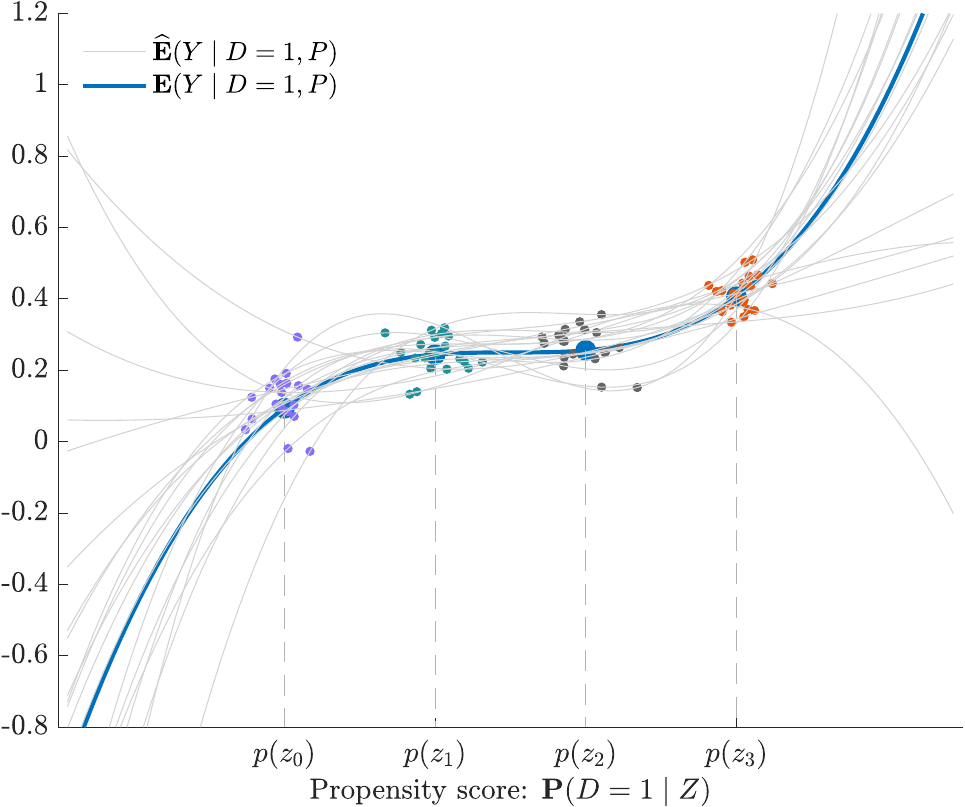}
    \end{subfigure}
    \vskip\baselineskip
    \begin{subfigure}[t]{0.49\textwidth}
        \caption{Limited Variation}
        \label{fig:weak}
        \centering
        \includegraphics[width=\textwidth]{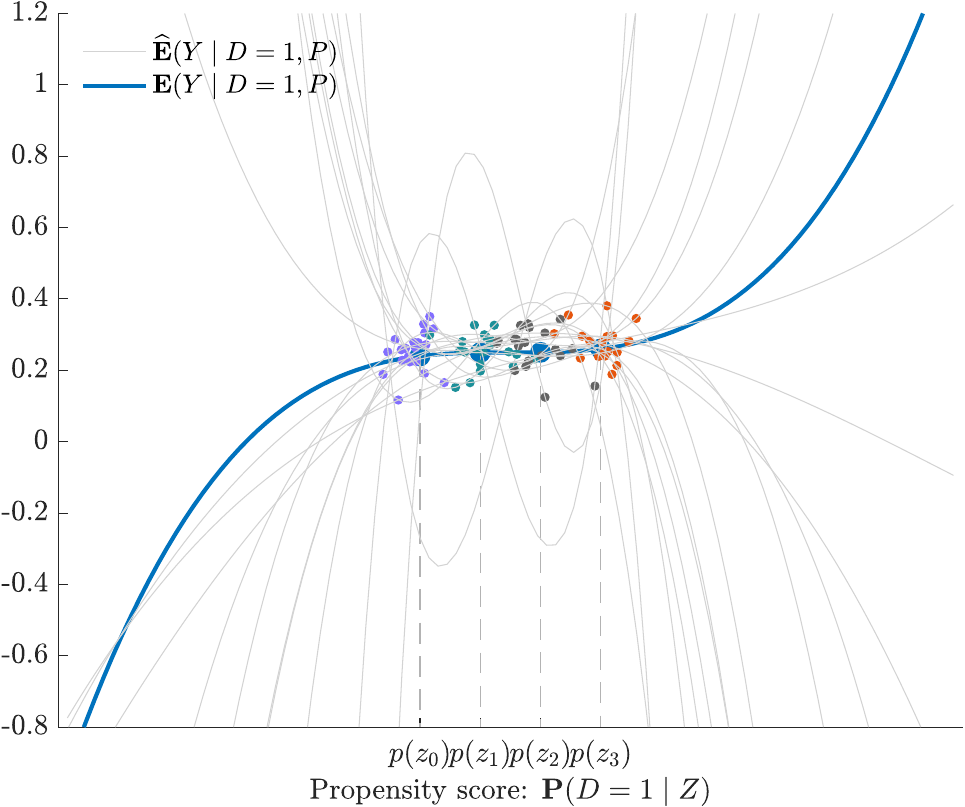}
    \end{subfigure}
    \hfill
    \begin{subfigure}[t]{0.49\textwidth}
        \caption{Binary Variation}
        \label{fig:partial}
        \centering
        \includegraphics[width=\textwidth]{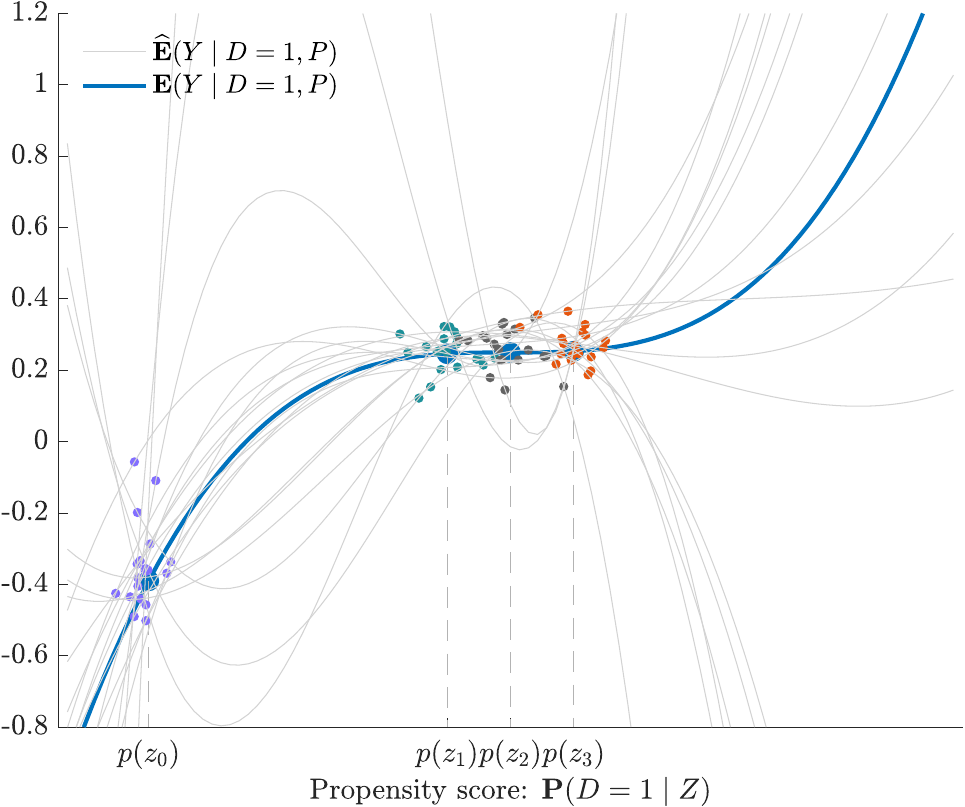}
    \end{subfigure}
    \label{fig:intuition_plot}

    \smallskip

    \begin{tablenotes}
    \footnotesize \raggedright
    Note: The figure shows the estimated expected outcomes for treated units conditional on propensity scores under different designs of propensity score variation, based on 20 independent simulations. All designs and simulations use a cubic specification for the MTE curve with a discrete IV that takes four values in its support. The sample size is 2,000.
    \end{tablenotes}
    
\end{figure}

\subsection*{Organization of this paper}

In section \ref{sec:model}, I describe the setup in \cite{brinch/mogstad/wiswall:2017} and \cite{kline/walters:2019} and present the MTE model with a discrete instrument. By imposing a parametric structure on the marginal treatment response (MTR) functions, the MTE is fully characterized by a finite-dimensional parameter, which can be point identified using discrete variation from the instrument. In finite samples, the MTE can be estimated by running separate regressions for both treated and control groups. However, these separate regression estimators, along with the corresponding Wald confidence sets, are highly vulnerable to limited variation in propensity scores. In light of this potential identification failure, this paper constructs uniformly valid confidence sets for causal parameters that are linear functionals of the MTE function, which cover a broad range of causal effects of interest including the MTE function itself, the average treatment effect (ATE), the average treatment effect on the treated (ATT), the local average treatment effect (LATE), and the policy-relevant treatment effect (PRTE).

In section \ref{sec:linearMTE_inference}, I develop a simple inference method for treatment effects in the linear MTE model, where the MTR functions vary linearly with the selection unobservable \citep{brinch/mogstad/wiswall:2017}.  This special structure offers a novel moment condition for constructing estimators of treatment effects of interest, bypassing the necessity for estimating the full model.  Building on this new moment condition, I construct a simple conditional Wald test and demonstrate its uniform validity under weak identification. This approach circumvents challenges posed by weakly identified nuisance parameters in the MTE, thus leading to confidence sets that achieve asymptotic similarity, a key advantage over existing approaches to subvector inference with weak identification as discussed in Online Appendix \ref{appendix:lit_subvector_inference}.

For a generic MTE model with a discrete instrument, it is infeasible to directly estimate causal effects of interest without relying on other primitive parameters in the MTE function. In such cases, these weakly identified primitive parameters hamper the ability to perform valid inference on linear functions of them. In section \ref{sec:polyMTE_inference}, I build on the improved projection approach developed in I. \citet{andrews.i:2018} and propose a MLC test that achieves uniform validity regardless of identification strength (see page 7 for a detailed comparison of my work with the literature). The inverted confidence set from this test has correct coverage under weak identification, is straightforward to compute, and is shown to be approximately as efficient as the Wald confidence set under strong identification.

In section \ref{sec:covariates}, I discuss the issue of incorporating covariates into the MTE model. First, I demonstrate that the commonly used additive separability condition can substantially bias causal effect estimands when it fails. Such bias does not vanish unless (1) unobserved treatment effect heterogeneity does not vary across covariates, or (2) individuals' treatment decisions do not depend on those covariates. If neither of these two assumptions hold, I show that causal effect estimands under additive separability differ from the true effects and potentially have the opposite sign. To deal with the bias induced by imposing additive separability, researchers may use the proposed methods in this paper to conduct inference by conditioning on covariates, while also achieving robustness against limited variation of propensity scores from conditioning. In Online Appendix \ref{sec:sidak_correction}, I consider a Bonferroni-type size correction for valid inference on aggregated effects in the absence of additive separability. For researchers interested in performing inference under additive separability, I also present an extension of the proposed inference methods to accommodate this framework in Online Appendix \ref{sec:inference_addsep}.

In section \ref{sec:simulation}, I examine the performance of MLC tests via a sequence of Monte-Carlo simulations in a quadratic MTE model with varying degrees of identification strength. The simulation results show that the MLC test is asymptotically size-correct whenever the model is strongly identified, weakly identified, or partially identified. In contrast, the classical Wald test reports over-rejections of the null value at 14\% frequency versus the 5\% significance level under partial identification and exhibits trivial power under weak identification. When the instrument strength is sufficiently strong, the proposed MLC test has power close to the asymptotically efficient Wald test. Theoretical justification is provided in Online Appendix \ref{appendix:power_analysis_MLC}.

In section \ref{sec:empirical}, I illustrate the proposed methods by using data from \cite{agan/doleac/harvey:2023}, who show that not prosecuting defendants with misdemeanor offenses reduces recidivism. While their analysis uses LATE estimates and ATE/ATT estimates under additive separability, I complement their findings by examining recidivism effects under different counterfactual prosecution policies. Specifically, I construct confidence sets for the reduction in recidivism under two scenarios: (1) a homogeneous marginal increase in nonprosecution rates across prosecutors, and (2) implementation of a minimum nonprosecution rate threshold. The results highlight that weak identification is a significant concern in higher-order polynomial models, and empirical conclusions differ substantially between robust and classical inference methods. Section \ref{sec:conclusion} concludes with a discussion of potential future extensions.

\subsection*{Related literature}

This paper contributes to three strands of the literature: marginal treatment effects, subvector inference in weakly identified models, and judge/examiner designs. For the rest of the introduction, I review the related literature. 

Although the MTE model has been used to extrapolate treatment effects in a variety of fields, there are relatively few studies on the theory of estimation and inference in MTE models. In a semiparametric MTE model with continuous propensity scores, \cite{heckman/urzua/vytlacil:2006a} use bootstrap methods to construct confidence bands for MTE functions. This approach is now common practice\footnote{The Stata implementation of MTE estimation by \cite{brave/walstrum:2014} and \cite{andresen:2018} uses bootstrap methods to produce confidence sets.}, but its theoretical validity has not yet been established. \cite{carneiro/lee:2009} analyze the pointwise limit distributions of MTE estimators using a separate regression approach. Building on this work, \cite{sasaki/ura:2023} further derive asymptotic theory for the PRTE using an orthogonalized score for double debiased estimation. Both papers assume a semiparametric MTE model with additive separability and strong IV/covariate variation to achieve $\sqrt{n}$-consistent estimation.  \cite{mogstad/santos/torgovitsky:2017, mogstad/santos/torgovitsky:2018} propose a bootstrap procedure for conducting inference on the PRTE when the MTE is partially identified. None of these papers study weak identification problems. Thus, compared to this existing literature, the inference approach I propose is the first to achieve robustness against weak instruments in MTE models.

To estimate the MTE over a large support of propensity scores, empirical researchers usually assume that the MTR functions---and their difference, the MTE---are additively separable in covariates and unobserved costs of treatment selection \cite[Assumption 2]{brinch/mogstad/wiswall:2017}. This assumption allows the MTE to be estimated on the unconditional support of propensity scores, by pooling variation in the propensity scores across different covariate values.  Despite its increasing popularity in empirical work, there is little theoretical justification for additive separability. In fact, this assumption can be strong enough to point identify the MTE without exogenous variation from instruments \citep{pan/wang/zhang/zhou:2024}.  In addition, if the MTE is misspecified as linearly additively separable, \cite{devereux:2022} provides numerical evidence indicating that omitting higher-order covariate terms can introduce significant bias in the estimated MTE slopes. I contribute to this knowledge by providing the first analytical bias formula for the causal parameters (e.g., ATE, conditional ATE, and MTE slope) when additive separability is misspecified in a widely used class of latent threshold crossing models \citep{kline/walters:2019}. To avoid this bias, researchers should extrapolate treatment effects conditional on covariates instead of relying on additive separability. Moreover, the robust inference procedures proposed in this paper can help address potential weak IV variation that may arise after conditioning on covariates.

The inference problem in this paper is also related to the literature on subvector inference in weakly identified models, where a subvector is a subset of the structural parameters. For inference on subvectors, \cite{dufour/mohamed:2005} suggest projecting the robust confidence sets of the full vector onto the subvector of interest. However, this procedure can be very conservative especially when the dimension of the full vector is much larger than that of the target subvector. To this end, there is a sequence of studies trying to reduce the conservativeness of projection inference. \cite{chaudhuri/zivot:2011} consider modifying the (Robust) Lagrangian Multiplier statistic \citep{kleibergen:2005} such that it is locally equivalent to asymptotic efficient subvector tests under strong identification and propose a Bonferroni method to improve the power of their tests at distant alternatives. Building on this idea, D. \cite{andrews.d:2017} improves the Bonferroni method such that the refined tests are asymptotically non-conservative and uniformly valid. However, D. \citet[page 2]{andrews.d:2017} acknowledges six key limitations, including computational challenges and the need for additional tuning parameters to categorize the identification strength, neither of which are required for my proposed MLC test. Moreover, his method does not directly address inference on a linear function of parameters, which is the primary focus of this paper. While his method could conceptually extend to inference on a linear function through model reparametrization, finding a universal reparametrization rule that works for all linear hypotheses of interest while maintaining tractable asymptotic analysis remains challenging.

For inference on a function of parameters in a weakly identified model, I. \cite{andrews.i:2018} generalizes the results in \cite{chaudhuri/zivot:2011} and proposes a two-step confidence set that achieves sequential validity with controlled coverage distortions under a set of high-level conditions imposed on a sequence of data generating processes (DGPs). The MLC test considered in this paper builds on the idea from I. \cite{andrews.i:2018} but differs in a few ways: (1) Most importantly, Andrews’ paper only provided \textit{sequential} validity results, concluding that ``conditions for \textit{uniform} asymptotic validity are an interesting open question’’ in Section V (page 347). I demonstrate that, with a minor modification to the test statistics, the robust test achieves uniform validity for inference on a scalar function. (2) While Andrews derived results for a general class of models using high-level conditions, my focus on the MTE model allows for more specific, primitive conditions for validity. (3) Instead of employing Andrews' two-step approach that alternates between non-robust Wald and robust confidence sets based on identification strength, I focus exclusively on his robust confidence sets. (4) I adapt his robust confidence sets, originally developed for GMM models, to the minimum distance framework arising from separate regressions in the MTE setting.  (5) By using the robust confidence set, this approach does not involve pretesting distortion in the two-step approach and maintains the desired asymptotic coverage level $1-\alpha$. (6) I prove that the local power difference from the asymptotically efficient Wald test can be arbitrarily small under strong identification (see Online Appendix \ref{appendix:power_analysis_MLC}). Applied to MTE models, this approach allows for valid and powerful inference on treatment effects even under limited variation of propensity scores. In Online Appendix \ref{appendix:lit_subvector_inference}, I also discuss other approaches to inference on functions (or subvectors) of parameters in weakly identified models and explain their inapplicability to the MTE model studied here.

Finally, the empirical analysis of this paper also connects to the judge/examiner design problems (see the survey by \cite{chyn/frandsen/leslie:2024}). By using the quasi-random assignment of examiners, researchers can identify the causal effects of judicial decisions for defendants at the margin of being treated, and then extrapolate these effects to the broader population under the assumption of pairwise monotonicity \citep{frandsen/lefgren/leslie:2023}. While the MTE model is commonly employed for such extrapolation, researchers often report LATE, ATE, ATUT, and ATT to inform potential policy decisions \citep{bhuller/dahl/loken/mogstad:2020, agan/doleac/harvey:2023, baron/gross:2023}. My paper provides methods for robust inference on all of these parameters, as well as the MTE function itself, and the policy counterfactual parameters studied in \cite{heckman/vytlacil:2001, heckman/vytlacil:2005}, and \cite{carneiro/heckman/vytlacil:2010}.

\clearpage 

\section{MTE Model with Discrete IVs}
\label{sec:model}
In this section, I describe the MTE model and the related identification result following \cite{brinch/mogstad/wiswall:2017}. Based on this result, the weakness of IVs can be characterized by limited variation of propensity scores. Then I introduce the relevant parameter space on which we achieve uniform validity of the proposed inference procedures.

\subsection{Setup}
Let $Y_1$ be the potential outcome of an individual who receive a binary treatment ($D = 1$) and $Y_0$ denote her potential outcome in the untreated state ($D = 0$). The observed outcome $Y$ is realized through
\begin{equation}
\label{eq:outcome}
    Y = (1-D)Y_0 + DY_1.  
\end{equation}
We further specify
\[
    Y_d = \mu_d + V_d, \quad d = 0,1,
\]
where $\mu_d \equiv \Exp[Y_d]$ is the mean of potential outcome. For simplicity, I leave out additional covariates and discuss them in section \ref{sec:covariates}. The treatment is determined by a weakly separable selection equation
\begin{equation}
\label{eq:selection}
    D = \indicator[U \leq \nu(Z)],
\end{equation}
where $\nu(\cdot)$ is an unknown function, $U$ is a continuous random variable representing the unobserved cost of selection into treatment, and $Z \in \supp(Z) = \{z_0, z_1, \ldots, z_K\}$ is the excluded discrete instrument. Researchers observe the outcome $Y$, the binary treatment $D$, and the excluded instrument $Z$ from data, while the unobservables are the potential outcomes $(Y_0, Y_1)$ and the variable $U$ in the selection equation. The MTE model allows individuals to be selected into treatment based on their information on potential outcomes, which leads to potential dependence between $(V_1, V_0)$ and $U$. 

The key identifying assumption from \citet[proposition 1.ii]{brinch/mogstad/wiswall:2017} is as follows:
\begin{assumption}[MTE model with discrete IVs] 
\label{asp:identification}  \hfill \vspace{-1em}
    \begin{enumerate}
    \item $Z \indep U$.
    \item $\Exp[Y_d\mid Z, U] = \Exp[Y_d\mid U]$ and $\Exp|Y_d| < \infty$ for $d \in \{0,1\}$.
    \item $U$ is continuously distributed.
    \item $0 < \Prob(D=1\mid Z=z) <1$ for all $z \in \supp(Z)$.
    \item Let $\{h_m(\cdot)\}_{m=1}^{M}$ be a set of known continuous functions defined on $(0,1)$. For $d \in \{0,1\}$, the MTR function is given by
    \[
        \Exp[Y_d\mid F_U(U) = u] = \mu_d + \sum_{m=1}^M \rho_{dm} h_m(u)\quad \text{for } u \in (0,1),
    \]
    where $F_U(\cdot)$ denotes the distribution function of $U$.
    \item Let $\lambda_{00}(\cdot) = \lambda_{10}(\cdot) \equiv 1$, and define
    \[
        \lambda_{1m}(p) \equiv \frac{1}{p}\int_{0}^p h_m(u) du \quad \text{and} \quad \lambda_{0m}(p) \equiv \frac{1}{1-p}\int_{p}^1 h_m(u) du \quad \text{for $m = 1,\ldots,M$}.
    \]
    for $p \in [0,1]$. 
    Then $\{\lambda_{1m}(\cdot)\}_{m=0}^M$ and $\{\lambda_{0m}(\cdot)\}_{m=0}^M$ are unisolvent\footnote{A set of $n$ functions $f_1, f_2, \ldots, f_n$ is unisolvent on domain $\Omega$ if the matrix $F \in \R^{n\times n}$ with entries $f_i(x_j)$ has nonzero determinant for any choice of $n$ distinct points $x_1, x_2, \ldots, x_n$ in $\Omega$.} on $(0,1)$.
    \item $\overline{K} \equiv |\{\Prob(D=1\mid Z=z): z=z_0, z_1, \ldots, z_K\}|\geq M+1$.
    \end{enumerate}
\end{assumption}

Assumptions \ref{asp:identification}.1 and \ref{asp:identification}.2 require the excluded instrument $Z$ to be exogenous to both the selection and outcome processes. Assumption \ref{asp:identification}.3 allows us to normalize the marginal distribution of $U$ to be uniformly distributed over $[0,1]$. That is, we can transform $U$ to a uniformly distributed variable $\widetilde{U} = F_{U}(U)$. Under the exogeneity of $Z$ in Assumption \ref{asp:identification}.1, the function $F_U(\nu(Z))$ can then be interpreted as propensity score:
\begin{align*}
	p(z) 	&\equiv \Prob(D=1\mid Z = z) \\
			&= \Prob(U\leq \nu(Z)\mid Z = z) \\
			&= \Prob(U\leq \nu(z)) \\ 
			&= F_U(\nu(z))
\end{align*}
where the third line uses the Assumption \ref{asp:identification}.1. Define $\widetilde{\nu}(z) \equiv F_U(\nu(Z))$, and then we can work with $(\widetilde{U}, \widetilde{\nu}(z))$ in place of $(U, \nu(z))$ without affecting the empirical content of the selection model. For simplicity, we drop out the tilde and assume $U$ is uniformly distributed throughout our analysis, and therefore $\nu(z) = p(z)$. Assumption \ref{asp:identification}.4 validates the overlap condition, so that we observe both treated and untreated individuals for each group defined by the value of the instrument. Assumption \ref{asp:identification}.5 imposes a parametric restriction on the MTR function $\Exp[Y_d\mid F_U(U)=u]$ so that the MTE can be extrapolated outside the discrete support of the instrument. Assumption \ref{asp:identification}.6 is a weak condition that rules out redundant specifications in $\{h_m(\cdot)\}_{m=0}^M$ that cause multicollinearity in $\{\lambda_{dm}(\cdot)\}_{m=0}^M$. In particular, the usual polynomial specification $h_m(u) = u^m - \frac{1}{m+1}$ satisfies this condition. Assumption \ref{asp:identification}.7 requires sufficient variation of the exogenous instrument to point identify the structural parameters.  While $\overline{K} \leq K+1$, with strict inequality when multiple instrument values yield identical propensity scores, point identification requires the number of distinct propensity scores, not instrument values, to exceed the order of the MTE model. 

Define $\theta_d \equiv (\mu_d, \rho_{d1}, \ldots, \rho_{dM})'$ for $d = 0,1$, and let $\theta \equiv (\theta_1', \theta_0')'$. Then we have the following identification result:

\clearpage

\begin{theorem}[Identification]
\label{thm:identification}
Suppose Assumption \ref{asp:identification} holds. Then $\theta$ is point identified.
\end{theorem}
\begin{proof}
	Based on the model \eqref{eq:outcome}--\eqref{eq:selection} and Assumption \ref{asp:identification}, we have
	\begin{align}
	\Exp[Y\mid D=1, Z=z] 
		&= \Exp[Y_1\mid U\leq p(z), Z=z] \notag \\
		&= \Exp[Y_1\mid U\leq p(z)] \notag \\
		&= \frac{1}{p(z)}\int_{0}^{p(z)} \Exp[Y_1\mid U=u] du \notag \\
		&= \mu_1 + \sum_{m=1}^M \frac{\rho_{1m}}{p(z)}\int_{0}^{p(z)} h_m(u) du. \notag \\
		&= (\lambda_{10}(p(z)), \lambda_{11}(p(z)), \ldots, \lambda_{1M}(p(z)))\theta_1. 	\label{eq:cond_exp_treated}
	\end{align}
	The first line holds by equations \eqref{eq:outcome}, \eqref{eq:selection}, and the normalization that $\nu(z) = p(z)$, which is jointly implied by Assumption \ref{asp:identification}.1 and \ref{asp:identification}.3. The second line holds by the exogeneity condition in Assumption \ref{asp:identification}.2. The third line follows by the overlap condition in Assumption \ref{asp:identification}.4 and the normalization that $U$ is uniformly distributed over $[0,1]$. The fourth line holds by parametric restriction in Assumption \ref{asp:identification}.5. Likewise, we have 
	\begin{equation}
	\label{eq:cond_exp_untreated}
	\Exp[Y\mid D=0, Z=z] = \mu_0 + \sum_{m=1}^M \frac{\rho_{0m}}{1-p(z)}\int_{p(z)}^{1} h_m(u) du = (\lambda_{00}(p(z)), \lambda_{01}(p(z)), \ldots, \lambda_{0M}(p(z)))\theta_0.
	\end{equation}
	Taking $z = z_0,z_1,\ldots,z_K$ in equations \eqref{eq:cond_exp_treated} and \eqref{eq:cond_exp_untreated} then yields two matrix equalities:
	\begin{equation}
	\label{eq:mateq}
		\beta_d = A_d\theta_d \quad \text{for } d = 0,1,
	\end{equation}
	where 
	\[
	A_d = 
	\begin{bmatrix}
		\lambda_{d0}(p(z_0)) & \lambda_{d1}(p(z_0)) & \ldots & \lambda_{dM}(p(z_0)) \\
		\lambda_{d0}(p(z_1)) & \lambda_{d1}(p(z_1)) & \ldots & \lambda_{dM}(p(z_1)) \\
		\vdots & \vdots & \ddots & \vdots \\
		\lambda_{d0}(p(z_K)) & \lambda_{d1}(p(z_K)) & \hdots & \lambda_{dM}(p(z_K))
	\end{bmatrix}
	\]
	and 
        \[
            \beta_d = (\Exp[Y\mid D=d, Z=z_0], \ldots, \Exp[Y\mid D=d, Z=z_K])'.
        \]
	Note that Assumption \ref{asp:identification}.6 and \ref{asp:identification}.7 guarantees that there exists a full-rank submatrix of $A_d$ with $\overline{K}$ rows and $M+1$ columns for each $d = 0,1$. Therefore, both $A_1$ and $A_0$ have full column rank. Then $\theta_1$ and $\theta_0$ are point identified by equation \eqref{eq:mateq}.
\end{proof}
\begin{remark}
    This identification result closely aligns with the proof of \citet[Proposition 1.ii]{brinch/mogstad/wiswall:2017}, who shows that Assumption \ref{asp:identification}.7 is a necessary condition for identifying $\theta$ in a MTE model satisfying Assumption \ref{asp:identification}.1--\ref{asp:identification}.5. Building on their result, I add Assumption \ref{asp:identification}.6 to establish sufficient conditions for point identification.
\end{remark}

Based on the parametric restriction in Assumption \ref{asp:identification}.5, treatment effects can often be written as linear functions of the primitive parameter $\theta$. For example, $\text{ATE} = \mu_1 - \mu_0$ and 
\begin{align*}
	\MTE(u) 
    &= \Exp[Y_1 - Y_0\mid U=u] \\
    &= c(u)'\theta_1 - c(u)'\theta_0,
\end{align*}
where $c(u) = (1, h_1(u),\ldots,h_M(u))$. 
Moreover, suppose we are interested in the parameters that inform potential policy decisions such as average treatment effects on the (un)treated groups or policy-relevant treatment effects. In that case, the weight $c$ is often unknown and depends on the underlying DGP (see Table 1 and 2 in \cite{mogstad/torgovitsky:2018}). If we write $c = (c_1', c_0')'$ with $c_{d} = (c_{d,0},\ldots,c_{d,M})'$ for $d = 0,1$ and denote $h_0(u) \equiv 1$, Table \ref{tab:treatment-effects-weight} summarizes the weights of various treatment effects under Assumption \ref{asp:identification}. In this paper, I consider inference on treatment effects of the form $c'\theta$ where $c_1 = -c_0$ since many causal effects of interest including those in Table \ref{tab:treatment-effects-weight} have symmetric weights: $c_{1,m} = -c_{0,m}$ for all $m = 1,\ldots,M$.
\begin{table}[ht]
\centering 
\caption{Weights for Treatment Effects}
\label{tab:treatment-effects-weight}
\begin{tabular}{ccl}
	\toprule
	Target parameter	& Expression 	&  Weights \\
	\midrule
	ATE &	 $\Exp[Y_1 - Y_0]$ & $c_{1,m} = \indicator[m=0]$ \\
	MTE & $\Exp[Y_1 - Y_0\mid U=u]$ & $c_{1,m} = h_m(u)$\\
	ATT & $\Exp[Y_1 - Y_0\mid D=1]$ & $c_{1,m} = \frac{\Exp(\int_0^{p(Z)}h_m(u)du)}{\Prob(D=1)}$ \\
	ATU & $\Exp[Y_1 - Y_0\mid D=0]$ & $c_{1,m} = \frac{\Exp(\int_{p(Z)}^1 h_m(u)du)}{\Prob(D=0)}$ \\
	LATE & $\Exp[Y_1 - Y_0\mid p(z_0) < U < p(z_k)]$ & $c_{1,m} = \frac{\int_{p(z_0)}^{p(z_k)} h_m(u)du}{p(z_k) - p(z_0)}$  \\
    Additive PRTE & PRTE with $p^\epsilon(z) = p(z) + \epsilon$  &   \multirow{3}{*}{$c_{1,m} =  \frac{\Exp[\int_{p(Z)}^{p^\epsilon(Z)} h_m(u) du]}{\Exp[p^\epsilon(Z) - p(Z)]}$}  \\
    Proportional PRTE & PRTE with $p^\epsilon(z) = (1+\epsilon) p(z)$ & \\    
    Quota & PRTE with $p^\epsilon(z) = \min\{p(z),\epsilon\}$ & \\
    \bottomrule
\end{tabular}
\end{table} 

Define $q(z_\ell) = \Prob(Z=z_\ell)$, i.e., the probability mass function of the discrete IV. For the commonly used polynomial MTE model in empirical studies where $h_m = u^{m} - \frac{1}{m+1}$, the weights of ATT, LATE, additive PRTE, and proportional PRTE become
{
\begin{align*}
	&c^{\text{ATT}}_{1,m} = \frac{1}{m+1} \left(\frac{\sum_{\ell=0}^K p(z_\ell)^{m+1} q(z_\ell)}{\sum_{\ell=0}^K p(z_\ell)q(z_\ell)} - 1\right)\\
	&c^{\text{LATE}}_{1,m} = \frac{1}{m+1} \left(\sum_{j=0}^m p(z_k)^j p(z_0)^{m-j} - 1\right)\\
        &c^{\text{A-PRTE}}_{1,m} = \frac{1}{m+1} \left(\sum_{\ell = 0}^K q(z_\ell) \sum_{j = 0}^m \left[p(z_\ell)^{j} (p(z_\ell) + \epsilon)^{m-j}\right] - 1 \right) \\
        &c^{\text{P-PRTE}}_{1,m} = \frac{1}{m+1} \left[\frac{(1+\epsilon)^{m+1} - 1}{\epsilon} \cdot \frac{\sum_{\ell = 0}^K p(z_\ell)^{m+1} q(z_\ell)}{\sum_{\ell = 0}^K p(z_\ell) q(z_\ell)} - 1\right]
\end{align*}
}
for $m = 1,\ldots,M$, and the first element $c_{1,0} = 1$ for all the above causal effects.

\subsection{Weak identification}
The identification strategy relies on the conditions in Assumption \ref{asp:identification}, particularly the existence of sufficiently many distinct propensity scores to point identify the structural parameter $\theta$ through Assumption \ref{asp:identification}.7. However, variation in propensity scores is often limited. This limited variation may not be enough to guarantee correct asymptotic approximation in the construction of classical confidence sets. In section \ref{sec:linearMTE_inference} and \ref{sec:polyMTE_inference}, I develop robust inference procedures for treatment effects of the form $c'\theta$ without requiring the potentially restrictive Assumptions \ref{asp:identification}.7. My results have two implications. Along a sequence of DGPs: (1) When Assumption \ref{asp:identification}.7 holds but is close to fail, the parameters are point identified, and the robust confidence set achieves correct coverage for the causal parameter of interest; (2) When Assumption \ref{asp:identification}.7 fails, the parameters are partially identified, and the proposed robust confidence set mantains correct coverage of the causal parameter $c'\theta$, where $\theta$ lies in the set defined by the linear system \eqref{eq:mateq}. Additionally, this uniform validity result does not rely on Assumption \ref{asp:identification}.6, though this assumption is typically satisfied under researchers' common specifications of the MTR functions.

\subsection{Parameter space restriction}
\label{sec:parameter_space}

This section introduces the parameter space for the joint distribution of $(Y,D,Z)$, on which the proposed robust inference procedures are uniformly valid (without covariates). First, consider the usual i.i.d. sampling distribution of the observables.
\begin{assumption}
\label{asp:iid}
	The random vectors $(Y_i, D_i, Z_i)$ for $i = 1,\ldots,n$ are i.i.d. with distribution $F$.
\end{assumption}
Next, I introduce a set of regularity conditions to be imposed on the joint distribution $F$. 
\begin{definition}[Parameter Space]
\label{def:parameter_space}
    For some $\delta, \zeta > 0$ and $\epsilon \in (0,1/2)$, define the parameter space $\mathcal{P}$ as the set of pairs $(\theta,F)$ satisfying the following properties:
\begin{enumerate}
	\item Equation \eqref{eq:mateq} is satisfied with $K \geq M$, where $\theta = (\theta_1', \theta_0')' \in \operatorname{int}(\Theta) \subseteq \R^{2(M+1)}$ for some compact set $\Theta$ with nonempty interior;
	\item $\sup_{d=0,1}\sup_{z\in \supp(Z)}\Exp_F[|Y|^{2+\delta}\mid D=d, Z=z] \leq \zeta$;
	\item $\epsilon \leq \inf_{z \in \supp(Z)}\Prob_F(D=1\mid Z=z) \leq \sup_{z \in \supp(Z)}\Prob_F(D=1\mid Z=z) \leq 1-\epsilon$;
	\item $\epsilon \leq \inf_{z\in\supp(Z)}\Prob_F(Z=z) \leq \sup_{z\in\supp(Z)} \Prob_{F}(Z=z) \leq 1-\epsilon$;
	\item $\epsilon \leq \inf_{d=0,1}\inf_{z\in\supp(Z)}\var_F(Y\mid D=d, Z=z)$.
\end{enumerate}
\end{definition}
The first condition assumes the correct model specification derived from  Assumptions \ref{asp:identification}.1-\ref{asp:identification}.5, and the order of the MTE model does not exceed the support of the instrument.\footnote{The condition $K\geq M$ is necessary but not sufficient for point identification, particularly when Assumption \ref{asp:identification}.7 does not hold.} The second condition is a mild restriction on the existence of moments of $Y$ conditional on $D$ and $Z$, which is essential for establishing the uniform version of the law of large numbers and central limit theorems. 
The third and fourth conditions, also known as the strong overlap conditions, require observing units for each value of the treatment and the instrument. Together with the final condition regarding sufficient variation in outcomes, these conditions rule out the singularity or near-singularity of the asymptotic variance of the moment equation \eqref{eq:mateq}. Since our parameter of interest has the form of $c'\theta$, let $\mathcal{P}_0 = \{(\lambda,F): \lambda = c'\theta, (\theta,F)\in\mathcal{P}\}$ denote the null parameter space for this linear function, in which the weight $c$ could possibly depend on underlying DGP $F$. It is important to emphasize that our parameter space does not require the relevance Assumptions \ref{asp:identification}.6 and \ref{asp:identification}.7, the latter of which is likely to fail. For a given  $\lambda \in \R$, our goal is to develop uniformly valid tests for assessing the hypothesis $H_0: c'\theta = \lambda$ that is robust to the limited variation of propensity scores.

\subsection{Notation and preliminaries}
\label{sec:notation}
The asymptotic behavior of test statistics to be proposed will be unified by the asymptotic distribution of estimators on marginal distribution of instrument $q(z_\ell) \equiv \Prob(Z=z_\ell)$, propensity score $p(z_\ell) = \Prob(D=1\mid Z=z_\ell)$, and the expected outcome conditional on the treatment and instrument $\beta_{d\ell} = \Exp[Y\mid D=d, Z=z_\ell]$ for $d=0,1$ and $\ell = 0,1,\ldots,K$. Under the parameter space restriction outlined in section \ref{sec:parameter_space}, we have the following sample-analog estimators for these quantities:
\begin{align*}
		&\hat{q}(z_\ell) = \frac{1}{n}\sum_{i=1}^n \indicator[Z_i =z_\ell] \\
		&\hat{p}(z_\ell) = \frac{1}{n}\sum_{i=1}^n \frac{\indicator[D_i = 1, Z_i = z_\ell]}{\hat{q}(z_\ell)} \\
		&\hat{\beta}_{d\ell} = \frac{1}{n}\sum_{i=1}^n\frac{Y_i\indicator[D_i = d, Z_i = z_\ell]}{\hat{q}(d,z_\ell)},
\end{align*}
where $\hat{q}(1,z_\ell) = \hat{p}(z_\ell)\hat{q}(z_\ell)$ and $\hat{q}(0,z_\ell) = (1-\hat{p}(z_\ell))\hat{q}(z_\ell)$. 

Collecting these estimators into vectors, we have $\hat{p} = (\hat{p}(z_\ell), \ldots, \hat{p}(z_K))'$, $\hat{q} = (\hat{q}(z_0),\ldots,\hat{q}(z_K))'$, $\hat{\beta}_d = (\hat{\beta}_{d0}, \ldots, \hat{\beta}_{dK})'$ for $d = 0,1$, and $\hat{\beta} = (\hat{\beta}_1', \hat{\beta}_0')'$. Lemma \ref{lem:convergence_p_beta}(a) provides the asymptotic distribution of the estimators $\hat{p}, \hat{q}$, and $\hat{\beta}$ under a drifting sequence in the parameter space $\mathcal{P}$. Moreover, Lemma \ref{lem:convergence_p_beta}(b) also shows that their asymptotic variances can be consistently estimated by 
\begin{align*}
	&\hat{\Sigma}_p = \diag\left\{\frac{\hat{p}(z_\ell)(1-\hat{p}(z_\ell))}{\hat{q}(z_\ell)}: \ell = 0,1,\ldots,K\right\} \\
	&\hat{\Sigma}_q = \{\hat{\Sigma}_q[i,j]\}_{i,j = 0,1,\ldots,K} \\
	&\hat{\Sigma}_{\beta_d} = \diag\left\{\frac{\hat{\sigma}_{d\ell}^2}{\hat{q}(d,z_\ell)}: \ell = 0,1,\ldots,K\right\} \\
        &\hat{\Sigma}_\beta = \diag\{\hat{\Sigma}_{\beta_1}, \hat{\Sigma}_{\beta_0}\},
\end{align*}
where 
\begin{align*}
    & \hat{\sigma}_{d\ell}^2 \equiv \frac{1}{n} \sum_{i=1}^n \frac{(Y_i - \hat{\beta}_{d\ell})^2\indicator[D_i=d, Z_i=z_\ell]}{\hat{q}(d,z_\ell)} \\
    & \hat{\Sigma}_q[i,j] 
        \equiv
        \begin{cases}
        \hat{p}(z_i) (1-\hat{p}(z_i))   & \text{ if } i = j\\
        -\hat{p}(z_i) \hat{p}(z_j)      & \text{ if } i \neq j.
        \end{cases}
\end{align*}

Throughout the paper, I use $\partial_x f \in \R^{k}$ to denote the gradient of a scalar function $f$ with respect to its argument $x\in \R^k$. If $f$ maps to a vector in $\R^l$, then $\partial_x f \in \R^{l\times k}$ denotes the Jacobian of $f$ with respect to $x\in\R^k$. Let $P_A = A(A'A)^{-1}A'$ denote the projection matrix onto the column spaces of matrix $A$ and define $M_A = I - P_A$ as the corresponding annihilator matrix. Let $\supp(X)$ denote the support of a random element $X$.

\section{Robust Inference in Linear MTE Models}
\label{sec:linearMTE_inference}
In this section, I start with the simplest functional specification and impose linearity on the MTR function. 
\begin{assumption}[Linear MTE model]
\label{asp:linear_MTE}
Assumption \ref{asp:identification}.5 holds with 
\[
	h_m(u) = u - \frac{1}{2}
\]
and $M=1$.
\end{assumption}
This assumption was first introduced by \cite{olsen:1980} to characterize sample selection bias and later generalized by \cite{brinch/mogstad/wiswall:2017} to model MTE functions. This linear MTE specification was recently adopted by \cite{kowalski:2023} to extrapolate treatment effects between two experimental studies. The parameters $\theta = (\mu_1, \rho_{11}, \mu_0, \rho_{01})$ can be identified using any pair of propensity scores $(p(z_0), p(z_k))$ for $k = 1,\ldots,K$ that differ from each other. Point identification fails if and only if all propensity scores ${p(z_0),\ldots,p(z_K)}$ are identical, i.e., when there is no variation in propensity scores.

The linear MTE specification allows us to derive closed-form estimands for $\theta$, in which case equation \eqref{eq:mateq} implies
\begin{align*}
	\begin{pmatrix}
	\mu_1 \\
	\rho_{11} 
	\end{pmatrix}
	&=
	\frac{1}{p(z_k) - p(z_0)}
	\begin{pmatrix}
	p(z_k) - 1 & -(p(z_0) - 1) \\
	-2 & 2
	\end{pmatrix}
	\begin{pmatrix}
	\Exp[Y\mid D=1, Z=z_0] \\
	\Exp[Y\mid D=1, Z=z_k]
	\end{pmatrix} 
\end{align*}
and
\begin{align*}
	\begin{pmatrix}
	\mu_0 \\
	\rho_{01} 
	\end{pmatrix}
	&=
	\frac{1}{p(z_k) - p(z_0)}
	\begin{pmatrix}
	p(z_k) & -p(z_0) \\
	-2 & 2
	\end{pmatrix}
	\begin{pmatrix}
	\Exp[Y\mid D=0, Z=z_0] \\
	\Exp[Y\mid D=0, Z=z_k]
	\end{pmatrix}.
\end{align*}
Recall $\beta_{dk} = \Exp[Y\mid D=d, Z=z_k]$. Then the ATE and the slope of MTE can be written as 
\[
	\begin{pmatrix}
	\mu_1 - \mu_0 \\
	\rho_{11} - \rho_{01} 
	\end{pmatrix}
 	=
	\frac{1}{p(z_k) - p(z_0)}
	\begin{pmatrix}
	p(z_k)\left[\beta_{10} - \beta_{00}\right] - p(z_0)\left[\beta_{1k} - \beta_{0k}\right] + \beta_{1k} - \beta_{10} \\
	2\left(\beta_{00} - \beta_{10} + \beta_{1k} - \beta_{0k}\right)
	\end{pmatrix}.
\]
Denote the numerator of the right-hand-side equation by
\[
	\Delta_\mu(z_0, z_k) \equiv p(z_k)\left[\beta_{10} - \beta_{00}\right] - p(z_0)\left[\beta_{1k} - \beta_{0k}\right] + \beta_{1k} - \beta_{10}
\]
and
\[
	\Delta_\rho(z_0, z_k) \equiv 2\left(\beta_{00} - \beta_{10} + \beta_{1k} - \beta_{0k}\right).
\]
Since these quantities can be directly identified from data, we can express the treatment effects parameters $\lambda = c'\theta$ with $c = (c_\mu, c_\rho, -c_\mu, -c_\rho)'$ as the solution of the following moment function:
\begin{equation}
\label{eq:moment_causal_effects}
	g_k(\lambda) = \left[p(z_k) - p(z_0)\right]\lambda - c_\mu\Delta_\mu(z_0,z_k) - c_\rho \Delta_\rho(z_0,z_k) = 0.
\end{equation}
If the instrument $Z$ is binary, then $c'\theta$ is just identified by this linear moment restriction, otherwise, we can construct a vector of moment functions:
\begin{equation}
\label{eq:moment_causal_effects_vector}
	g(\lambda) = (g_1(\lambda), \cdots, g_K(\lambda))' = 0_{K\times 1}
\end{equation}
to improve the efficiency for inference on $\lambda = c'\theta$. It is worth noting that other components of primitive parameter $\theta$ do not enter into \eqref{eq:moment_causal_effects_vector}. As a result, their nonstandard asymptotic behavior (under weak identification) is irrelevant in this context. To this end, this linear moment function $g(\lambda)$ will be used to construct asymptotically similar tests for treatment effects in linear MTE models. 

To fix ideas, I first consider inference for the parameter $c'\theta$ with a known weight $c$ and exploit the binary variation $\{z_0, z_k\}$ for $k = 1,\ldots,K$. This known-weight setting allows for inference on parameters such as the ATE and the MTE function. Online Appendix \ref{appendix:estimated_weights} discusses the extension to cases where $c$ needs to be estimated.

Equation \eqref{eq:moment_causal_effects} shows that $g_k(\lambda)$ suffices to estimate $\lambda$, the causal effects of interest. Let $\hat{g}_k(\lambda)$ be the sample analog of this moment function after plugging in the estimators $\{\hat{p}(z_\ell), \hat{\beta}_{0\ell}, \hat{\beta}_{1\ell}\}_{\ell = 0}^{K}$, i.e., 
\[
	\hat{g}_k(\lambda) = \left[\hat{p}(z_k) - \hat{p}(z_0)\right]\lambda - c_\mu\hat{\Delta}_\mu(z_0,z_k) - c_\rho \hat{\Delta}_\rho(z_0,z_k)
\]
where
\begin{align*}
	\hat{\Delta}_\mu(z_0,z_k) &= \hat{p}(z_k)[\hat{\beta}_{10} - \hat{\beta}_{00}] - \hat{p}(z_0)[\hat{\beta}_{1k} - \hat{\beta}_{0k}] + \hat{\beta}_{1k} - \hat{\beta}_{10} \\
	\hat{\Delta}_\rho(z_0, z_k) &= 2(\hat{\beta}_{00} - \hat{\beta}_{10} + \hat{\beta}_{1k} - \hat{\beta}_{0k}).
\end{align*}
Based on this moment condition, one might consider an Anderson-Rubin (AR) test statistic to assess the null hypothesis $H_0: c'\theta = \lambda$ as follows
\begin{equation}
\label{eq:linearMTE_AR_stat}
    \text{AR}_{n,k}(\lambda) = \left|\frac{\sqrt{n}\hat{g}_k(\lambda)}{\hat{s}_{k}(\lambda)}\right|^2
\end{equation}
where $\hat{s}_{k}^2(\lambda)$ consistently estimates the asymptotic variance of the sample moment $\hat{g}_k(\lambda)$. To construct this variance estimator, note that an asymptotic linear expansion of $\hat{g}_k(\lambda) - g_k(\lambda)$ gives
\[
	\hat{g}_k(\lambda) - g_k(\lambda) = 
	\begin{pmatrix}
	\partial_{p'} {g}_k(\lambda)[\hat{p} - p] \\
	+ ~\partial_{\beta_1'} {g}_k(\lambda)[\hat{\beta}_1 - \beta_1] \\
	+ ~\partial_{\beta_0'} {g}_k(\lambda)[\hat{\beta}_0 - \beta_0]
	\end{pmatrix}
	+ o_p(n^{-1/2}).
\]
with coefficients $\partial_{p} {g}_k(\lambda)$, $\partial_{\beta_1} {g}_k(\lambda)$, and $\partial_{\beta_0}{g}_k(\lambda)$ defined as the gradient of ${g}_k(\lambda)$ with respect to $p$, $\beta_1$, and $\beta_0$, respectively. It is natural to estimate these coefficients by plugging in sample analog estimators below:
\begin{equation}
\label{eq:fixed_weight_expansion}
\begin{aligned}
	& \partial_{p} \hat{g}_k(\lambda) = (-\lambda + (\hat{\beta}_{1k} - \hat{\beta}_{0k})c_\mu, 0, \ldots, 0, \lambda - (\hat{\beta}_{10} - \hat{\beta}_{00})c_\mu, 0, \ldots, 0)' \\
	& \partial_{\beta_1} \hat{g}_k(\lambda) = ((1-\hat{p}(z_k))c_\mu + 2c_\rho, 0, \ldots, 0, -(1-\hat{p}(z_0))c_\mu - 2c_\rho, 0, \ldots, 0)' \\
	& \partial_{\beta_0} \hat{g}_k(\lambda) = (\hat{p}(z_k)c_\mu - 2c_\rho, 0, \ldots, 0, -\hat{p}(z_0)c_\mu + 2c_\rho, 0, \ldots, 0)',
\end{aligned}
\end{equation}
where nonzero elements appear on the first and the $(k+1)$'th elements of vectors.

This leads to a consistent variance estimator
\[
    \hat{s}_{k}^2(\lambda) = \partial_{p'}\hat{g}_k(\lambda) \hat{\Sigma}_p \partial_{p}\hat{g}_k(\lambda) +  \partial_{\beta_1'}\hat{g}_k(\lambda) \hat{\Sigma}_{\beta_1} \partial_{\beta_1}\hat{g}_k(\lambda) +  \partial_{\beta_0'}\hat{g}_k(\lambda) \hat{\Sigma}_{\beta_0} \partial_{\beta_0}\hat{g}_k(\lambda).
\]
Let $\alpha \in (0,1)$ be the significant level. The proposition below establishes the uniform validity and asymptotic similarity of the AR test for $H_0:c'\theta=\lambda$ with $\lambda \in \R$. 

\begin{proposition}
\label{prop:AR_validity}
Let Assumption \ref{asp:iid} and \ref{asp:linear_MTE} hold, and suppose that the weight $c = (c_\mu, c_\rho, -c_\mu, -c_\rho)'$ is a nonzero fixed vector, then 
\[
	\liminf_{n\to\infty}\inf_{(\lambda,F)\in \mathcal{P}_0} \Prob_F\left(\text{AR}_{n,k}(\lambda) > q_{\chi_1^2}(1-\alpha)\right) = \limsup_{n\to\infty}\sup_{(\lambda,F)\in \mathcal{P}_0} \Prob_F\left(\text{AR}_{n,k}(\lambda) > q_{\chi_1^2}(1-\alpha)\right) = \alpha,
\] 
where $q_{\chi_1^2}(1-\alpha)$ denotes the $(1-\alpha)$-quantile of the $\chi_1^2$ distribution.
\end{proposition}

\begin{remark}
The above result shows that the proposed AR test is valid and (uniformly) asymptotically similar in the sense of \citet[eq.(2.6)]{andrews/cheng/guggenberger:2020}. Constructing an asymptotically valid subvector test is already challenging in models with weakly identified nuisance parameters\footnote{Since $c'\theta$ is the parameter of interest rather than $\theta$ itself, the vector of primitive parameters $\theta = (\mu_1, \rho_{11}, \mu_0, \rho_{01})$ becomes the weakly identified nuisance parameter under limited variation of propensity scores. Although we can show some functions of $\theta$ are strongly identified using the reparametrization technique proposed by \cite{han/mccloskey:2019}, two parameters $(\rho_{11}, \rho_{01})$ remain weakly identified in the transformed model (see Online Appendix \ref{appendix:control_function}).}, because the asymptotic distributions of many test statistics depend on those unknown nuisance parameters that cannot be consistently estimated \citep[page 1595]{andrews/mikusheva:2016b}.  However, I introduce a novel moment function that isolates the causal effects of interest while being independent of $\theta$, thereby making the simple AR test feasible in this context. Unlike other existing approaches to subvector inference with weak identification (discussed in Online Appendix \ref{appendix:lit_subvector_inference}), this method achieves asymptotic similarity, ensuring the inverted confidence set  maintains $1-\alpha$ coverage asymptotically under both strong and weak identification. %
\end{remark}

If $Z$ is a binary instrument taking values in $\{z_0, z_1\}$, there exists a unique AR test statistic $\text{AR}_{n,1}(\lambda)$ for robust inference on the causal parameter $c'\theta$. When the instrument is discrete with multiple values $Z\in\{z_0, \ldots, z_K\}$ where $K > 1$, we obtain multiple AR statistics $\{\text{AR}_{n,k}(\lambda)\}_{k=1}^K$. The informativeness of each statistic depends on the variation in propensity scores $p(z_0) - p(z_k)$. In this case, combining tests based on different propensity score pairs yields greater statistical power. Let 
\[
    \hat{\pi} = (\hat{p}(z_1) - \hat{p}(z_0), \ldots, \hat{p}(z_K) - \hat{p}(z_0))'
\]
and
\[
   \hat{\gamma} = (c_\mu \hat{\Delta}_\mu(z_0, z_1) + c_\rho \hat{\Delta}_\rho (z_0, z_1), \ldots, c_\mu \hat{\Delta}_\mu(z_0, z_K) + c_\rho \hat{\Delta}_\rho (z_0, z_K))'.
\]
Then our goal is to obtain a valid inference procedure based on a vector of linear moment functions: 
\[
    \hat{g}(\lambda)  = \hat{\pi}\lambda - \hat{\gamma} = (\hat{g}_1(\lambda), \ldots, \hat{g}_K(\lambda))' \in \R^K.
\]
Under strong identification where $\hat{\pi}$ converges to a nonzero limit, a Wald statistic based on efficient weighting matrix $\hat{S}(\lambda)^{-1}$ for testing $H_0: c'\theta = \lambda$ can be constructed below:
\[
    W_n(\lambda) = \frac{n\hat{g}(\lambda)'\hat{S}(\lambda)^{-1}\hat{\pi}\hat{\pi}'\hat{S}(\lambda)^{-1}\hat{g}(\lambda)}{\hat{\pi}'\hat{S}(\lambda)^{-1}\hat{\pi}},
\]
where 
\[
    \hat{S}(\lambda) = \partial_{p}\hat{g}(\lambda) \hat{\Sigma}_p \partial_{p'}\hat{g}(\lambda) + \partial_{\beta_1}\hat{g}(\lambda) \hat{\Sigma}_{\beta_1}\partial_{\beta_1'}\hat{g}(\lambda) +  
    \partial_{\beta_0}\hat{g}(\lambda) \hat{\Sigma}_{\beta_0}\partial_{\beta_0'}\hat{g}(\lambda)
\]
is a consistent estimator of the asymptotic covariance matrix of $\hat{\pi}\lambda - \hat{\gamma}$ under null hypothesis.

Under weak identification, the propensity score differences $\hat{\pi}$ may converge in probability to a zero vector, resulting in a nonstandard asymptotic distribution for $W_n(\lambda)$. To illusrate this problem, consider a sequence of DGPs where $\sqrt{n}\pi$ converges to a constant vector $\pi_0$. In this case, $\sqrt{n}\hat{\pi}$ converges in distribution to a multivariate normal distribution with mean $\pi_0$, which cannot be consistently estimated from the data. Consequently, the asymptotic distribution of $W_n(\lambda)$ contains the non-estimable nuisance parameter $\pi_0$, making it infeasible to consistently estimate the unconditional quantiles of $W_n(\lambda)$'s limiting distribution. However, similar to the arguments of \cite{moreira:2003} for linear IV models and \cite{andrews/mikusheva:2016b} for quasi-likelihood ratio tests in GMM models, the asymptotic distribution of $W_n(\lambda)$ becomes independent of $\pi_0$ when conditioned on a sufficient statistic, making it feasible to approximate the conditional distribution of $W_n(\lambda)$.

Define 
\[
    \hat{h}(\lambda) = \sqrt{n}\hat{\pi} - [\partial_{p} {\pi}]\hat{\Sigma}_p [\partial_{p'} \hat{g}(\lambda)] \hat{S}(\lambda)^{-1} \sqrt{n}\hat{g}(\lambda).
\]
The statistic $\hat{h}(\lambda)$ is asymptotically independent of $\sqrt{n}\hat{g}(\lambda)$ and contains sufficient information about $\pi_0$ such that the (asymptotic) distribution of $W_n(\lambda)$ conditional on $\hat{h}(\lambda)$ is free of $\pi_0$. To simulate this conditional distribution, I construct a simulated counterpart of $\sqrt{n}\hat{\pi}$, denoted as ${\pi}_s$, which is given by
\[
	{\pi}_s = \hat{h}(\lambda) + [\partial_{p} {\pi}]\hat{\Sigma}_p [\partial_{p'} \hat{g}(\lambda)] \hat{S}(\lambda)^{-1/2}\eta^*
\]
where $\eta^* \sim \normal(0_{K\times 1},I_{K\times K})$ is drawn independently of the data. We then approximate the distribution of $W_n(\lambda)$ by replacing $\sqrt{n}\hat{\pi}$ with ${\pi}_s$ and $\sqrt{n}\hat{g}(\lambda)$ with $\hat{S}(\lambda)^{1/2}\eta^*$, which yields
\[
	W_n^*(\lambda) = \frac{(\eta^*)'\hat{S}(\lambda)^{-1/2}\pi_s\pi_s'\hat{S}(\lambda)^{-1/2}(\eta^*)}{\pi_s'\hat{S}(\lambda)^{-1}\pi_s}.
\]
It follows that $W_n^*(\lambda)$ has the same asymptotic distribution as $W_n(\lambda)$ conditional on the realization of $\hat{h}(\lambda)$. Let $\hat{q}_{W^*}(1-\alpha)$ denote the $(1-\alpha)$-quantile of the distribution of $W_n^*(\lambda)$ conditional on the data. This critical value can be used to construct valid conditional Wald test:
\begin{theorem}
\label{thm:validity_cond_wald}
Let Assumption \ref{asp:iid} and \ref{asp:linear_MTE} hold, and suppose that the weight $c = (c_\mu, c_\rho, -c_\mu, -c_\rho)'$ is a nonzero fixed vector, then 
\[
	\liminf_{n\to\infty}\inf_{(\lambda,F)\in\mathcal{P}_0} \Prob_{F}\left(W_n(\lambda) > \hat{q}_{W^*}(1-\alpha)\right) = \limsup_{n\to\infty}\sup_{(\lambda,F)\in\mathcal{P}_0} \Prob_{F}\left(W_n(\lambda) > \hat{q}_{W^*}(1-\alpha)\right) = \alpha.
\]
\end{theorem}

\begin{remark}
    The statistic $W_n(\lambda)$ can be interpreted as a linear combination of AR statistics in equation \eqref{eq:linearMTE_AR_stat}, with weight 
    $\hat{S}^{-1}\hat{\pi}$. Under strong identification, this weighting vector can be shown to be optimal among all linear combinations of AR statistics, yielding the highest local asymptotic power. %
\end{remark}

\begin{remark}
In Lemma \ref{lem:pd_variance}, I show that the asymptotic variance of the moment function is uniformly positive definite as long as the weight $c = (c_\mu, c_\rho, -c_\mu, -c_\rho)$ is nonzero. If we incorporate additional moment conditions induced by the difference of propensity scores of the form $p(z_k) - p(z_j)$ for $k, j \neq 0$, then the asymptotic variance may become singular for some nonzero weight $c$. For example, if we set $c_\mu = 0$ and $c_\rho = 1$, it follows that 
\[
    [\hat{p}(z_k) - \hat{p}(z_j)]\lambda - c_\rho\hat{\Delta}_\rho(z_j,z_k) = \hat{g}_k(\lambda) - \hat{g}_j(\lambda), 
\]
implying that the moment condition constructed with the variation between $z_k$ and $z_j$ (on the left-hand side) is the difference of the moments constructed by using $(z_k, z_0)$ and $(z_j, z_0)$ (on the right-hand side).
\end{remark}

\begin{remark}
    An alternative statistic one could consider is the quasi-likelihood ratio statistic
    \[
        \text{QLR}_n(\lambda) = n\left[\hat{g}(\lambda)'\hat{S}(\lambda)^{-1}\hat{g}(\lambda) - \inf_{\lambda \in \R} \hat{g}(\lambda)'\hat{S}(\lambda)^{-1}\hat{g}(\lambda)\right]
    \]
    and its corresponding conditional approach discussed in \cite{andrews/mikusheva:2016b}. Here I focus on the conditional Wald approach because it is simpler to implement in practice and is first-order asymptotic equivalent to QLR test under strong identification.
\end{remark}

\section{Robust Inference in General MTE Models}
\label{sec:polyMTE_inference}
In this section, I relax the linearity Assumption \ref{asp:linear_MTE} and extend the analysis to allow for any functional forms specified in Assumption \ref{asp:identification}.5. This includes the linear MTE model and the polynomial specification $h_m(\cdot) = u^{m} - \frac{1}{m+1}$ that is commonly used in empirical studies. 

For this broader class of models, constructing moment conditions that only involve the target causal parameter is less obvious. Therefore, I develop an alternative \textit{improved projection} approach that builds on I. \cite{andrews.i:2018} to achieve valid inference. When compared to the conditional Wald test in section \ref{sec:linearMTE_inference}, this new approach applies to a wider range of MTE models but is more computationally intensive and is potentially more conservative. On the theoretical front, I argue that the high-level conditions imposed by I. \cite{andrews.i:2018} cannot be directly verified in the MTE setting. Therefore, I extend his sequential validity result (that based on high-level conditions) by establishing the uniform validity under primitive conditions on the parameter space outlined in Definition \ref{def:parameter_space}.

\subsection{Improved projection inference}
First, I describe how to adapt the improved projection test developed by I. \cite{andrews.i:2018} for conducting inference on $c'\theta$ using the following linear system of equations:
\begin{equation}
\label{eq:linsystem}
	A\theta = \beta	
\end{equation}
where 
	\begin{align*}
	A = 
	\begin{pmatrix}
	A_1 & 0_{(K+1)\times (M+1)} \\
	0_{(K+1)\times (M+1)}		& A_0
	\end{pmatrix}
	\qquad
	\theta =
	\begin{pmatrix}
	\theta_1 \\
	\theta_0
	\end{pmatrix}
	\qquad 
	\beta = 
	\begin{pmatrix}
	\beta_1 \\
	\beta_0
	\end{pmatrix},
	\end{align*}
as defined below equation \eqref{eq:mateq}. Note that $A$ is a matrix of transformed propensity scores and $\beta$ is a vector of conditional expectations, both of which can be consistently estimated by their sample analogs $\hat{A}$ and $\hat{\beta}$ under appropriate assumptions. Since the matrix $A$ captures the variation in propensity scores, it determines the strength of identification.

The conventional Wald test for conducting inference on $c'\theta$ uses the first-order asymptotic efficient estimator obtained by minimizing the following minimum-distance objective function:
\[
        \hat{\theta}^{\text{eff}} \in \argmin_{\theta\in \Theta} n(\hat{A}\theta - \hat{\beta})'{\hat{\Omega}(\theta)^{-1}}(\hat{A}\theta - \hat{\beta})
\]
where $\hat{\Omega}(\theta)$ is a consistent estimator for the asymptotic variance of $\sqrt{n}(\hat{A}\theta - \hat{\beta})$, formally defined below in equation \eqref{eq:asy_var_est}. %
Then the classical Wald statistic for assessing $H_0: c'\theta = \lambda$ is 
\begin{equation}
\label{eq:wald_stat}
	\text{Wald}_n(\lambda) = n(c'\hat{\theta}^{\text{eff}} - \lambda)' (c'(\hat{A}'\hat{\Omega}(\hat{\theta}^{\text{eff}})^{-1}\hat{A})^{-1}c)^{-1}(c'\hat{\theta}^{\text{eff}} - \lambda).
\end{equation}

When Assumption \ref{asp:identification}.7 is nearly violated, $\hat{A}$ may converge in probability to a matrix with deficient rank. For example, consider the following simple example on the linear MTE model with a binary IV.
\begin{example}[Linear MTE model with a binary IV]
    Suppose $K = M = 1$ and set $h_1(u) = u-\frac{1}{2}$. Along a sequence of DGPs $\{F_n\}$, let $p_{F_n}(z_0) = p \in (0,1)$ and $p_{F_n}(z_1) = p + \upsilon_{n} \in (0,1)$ with $\upsilon_n \to 0$. In such case, the probability limit of $\hat{p}(z_0)$ and $\hat{p}(z_1)$ are equal to $p$, and we have
    \[
        \hat{A}
        =
        \begin{pmatrix}
	\begin{matrix}
	1 & \frac{1}{2} (\hat{p}(z_0) - 1) \\
	1 & \frac{1}{2} (\hat{p}(z_1) - 1)
	\end{matrix} 
	&
	0_{2\times 2} \\
	0_{2\times 2} & 
	\begin{matrix}
	1 & \frac{1}{2}\hat{p}(z_0) \\
	1 & \frac{1}{2}\hat{p}(z_1)
	\end{matrix}
	\end{pmatrix}
        ~~
        \xrightarrow{p}
        ~~
        \begin{pmatrix}
	\begin{matrix}
	1 & \frac{1}{2} (p - 1) \\
	1 & \frac{1}{2} (p - 1)
	\end{matrix} 
	&
	0_{2\times 2} \\
	0_{2\times 2} & 
	\begin{matrix}
	1 & \frac{1}{2}p \\
	1 & \frac{1}{2}p
	\end{matrix}
	\end{pmatrix}
    \]
    Note that the probability limit of $\hat{A}$ becomes a singular matrix.
\end{example}
This singularity may result in multiple minimizers of the limiting objective function when defining $\hat{\theta}^{\text{eff}}$, suggesting that the efficient estimator $\hat{\theta}^{\text{eff}}$ may not exhibit the usual properties of consistency and asymptotic normality. To address this issue, one can derive inference results based on the moment condition:
\[
    {m}(\theta) \equiv {A}\theta - {\beta} = 0_{2(K+1)\times 1}
\]
instead of using the estimator $\hat{\theta}^{\text{eff}}$. Let $\hat{m}(\theta) \equiv \hat{A}\theta - \hat{\beta}$ denote the sample moment function. By substituting the true (or hypothesized) value $\theta$ for $\hat{\theta}^{\text{eff}}$ in $\hat{\Omega}(\hat{\theta}^{\text{eff}})$ and plugging the following first-order asymptotic expansion:
\[
	\sqrt{n}(\hat{\theta}^{\text{eff}} - \theta) = -(\hat{A}'\hat{\Omega}(\theta)^{-1}\hat{A})^{-1}\hat{A}'\hat{\Omega}(\theta)^{-1}\sqrt{n}(\hat{A}\theta - \hat{\beta}) + o_p(1)
\] 
into equation \eqref{eq:wald_stat}, this gives a locally equivalent Lagrangian Multiplier (LM) test statistic:
\[
	\text{LM}_n(\theta) \equiv n(\hat{A}\theta - \hat{\beta})'\hat{\Omega}(\theta)^{-1/2} P_{\hat{\Omega}(\theta)^{-1/2}\hat{A}(\hat{A}'\hat{\Omega}(\theta)^{-1}\hat{A})^{-1}c}\hat{\Omega}(\theta)^{-1/2}(\hat{A}\theta - \hat{\beta})
\]
which does not require $\hat{\theta}^{\text{eff}}$ to be consistent.

However, the singularity of $\hat{A}'\hat{\Omega}(\theta)^{-1}\hat{A}$ persists under the projection operator inside the LM statistic.  Following the insights of \cite{kleibergen:2005}, one can orthogonalize columns of $\hat{A} = (\hat{a}_1,\ldots, \hat{a}_{2(M+1)})$ with respect to the variation in the moment condition $\hat{m}(\theta)$ to obtain a new gradient estimator 
\[
	\hat{D}(\theta) = (\hat{d}_1(\theta), \hat{d}_2(\theta), \ldots, \hat{d}_{2(M+1)}(\theta))
\]
where 
\[
	\hat{d}_j(\theta) \equiv \hat{a}_j - \hat{\Gamma}_j(\theta)\hat{\Omega}(\theta)^{-1}(\hat{A}\theta - \hat{\beta}) \quad \text{for } j = 1,2,\ldots, 2(M+1).
\]
Here $\hat{\Gamma}_j(\theta)$ is a consistent estimator of the asymptotic covariance between $\sqrt{n}\hat{m}(\theta)$ and the $j$-th column in $\sqrt{n}(\hat{A} - A)$, formally defined below in equation \eqref{eq:asy_cov_est}. 

By this orthogonalization, $\sqrt{n}(\hat{D}(\theta) - A)$ is asymptotically independent of the moment condition $\sqrt{n}\hat{m}(\theta)$ in large samples. Replacing $\hat{A}$ with $\hat{D}(\theta)$ in the LM statistic then leads to the ``subvector'' version of Robust LM (RLM) statistics for inference on $c'\theta$:
\begin{equation}
\label{eq:KLM}
	\text{RLM}_n(\theta) 
        =  n(\hat{A}\theta - \hat{\beta})'\hat{\Omega}(\theta)^{-1/2} P_{\hat{\Omega}(\theta)^{-1/2}\hat{D}(\theta)(\hat{D}(\theta)'\hat{\Omega}(\theta)^{-1}\hat{D}(\theta))^{-1}c}\hat{\Omega}(\theta)^{-1/2}(\hat{A}\theta - \hat{\beta}).
\end{equation}
If $\sqrt{n}A$ converges to a fixed matrix $\mathcal{A}$ as in \citet[page 1108]{kleibergen:2005}, so that $A$ has a zero limit, then $\sqrt{n}\hat{D}(\theta)$ converges to a Gaussian matrix with mean $\mathcal{A}$ that is asymptotically independent of the moment vector. Consequently, the projection matrix in the RLM statistic becomes asymptotically independent of the moments on both sides, implying that $\text{RLM}_n(\theta)$ converges to the standard $\chi_1^2$ limiting distribution for a sequence of DGPs that induces a zero limit of $\hat{A}$ at a rate $n^{-1/2}$.

Compared to the classical RLM statistic for full vector inference in \cite{kleibergen:2005}, this subvector RLM statistic \eqref{eq:KLM} only attains power for deviations in the linear function $c'\theta$ rather than the full vector $\theta$.  This feature has two important implications. On the one hand, when identification is sufficiently strong, the subvector RLM statistic is locally equivalent to the efficient subvector Wald statistic \eqref{eq:wald_stat} when $\theta$ approaches its true value; On the other hand, this equivalence may fail since the subvector RLM statistic cannot distinguish alternative values of $\theta$ from the true value when they yield the same value of $c'\theta$. To overcome this  limitation, I. \cite{andrews.i:2018} introduces a linear combination (LC) statistic that combines the RLM and AR statistics:
\begin{equation}
\label{eq:lc_stat}
	\text{LC}_n(\theta) = \text{RLM}_n(\theta) + a \cdot \text{AR}_n(\theta)
\end{equation}
where
\[
	\text{AR}_n(\theta) = n(\hat{A}\theta - \hat{\beta})'\hat{\Omega}(\theta)^{-1}(\hat{A}\theta - \hat{\beta})
\]
and $a > 0$ is a tuning parameter on the weights attached to the AR statistic. By incorporating the AR term, the LC statistic diverges to infinity outside the $n^{-1/2}$ neighborhood of the true parameter $\theta$ under strong identification. This property guarantees power against any deviations from the true value. Moreover, when $a$ is sufficiently small and the model is strongly identified, the projection test based on the LC statistic becomes approximately equivalent to the subvector Wald test.

\subsection{Contributions and Modifications}
In a GMM model, I. \cite{andrews.i:2018} shows that $\text{LC}_n(\theta)$ converges to a mixture of two independent chi-squared distributions $(1+a)\chi^2_{1} + a\chi^2_{2K+1}$ in a sequence of DGPs $\{F_n\}_{n\geq 1}$ satisfying certain high-level conditions. As Andrews notes in his conclusion, the uniform validity of his result under more primitive conditions remains an open question. In this section, I make two key contributions: I show that his high-level conditions are not trivially satisfied in the MTE framework, and I establish the uniform validity of the LC test through a simple modification.

First, consider Assumption 4 of I. \citet{andrews.i:2018}. This assumption requires two convergence conditions by the existence of normalizing sequences: a sequence of full-rank matrices $\{\Lambda_{1,n}\}\subseteq \R^{2(M+1)\times 2(M+1)}$ and a sequence of nonzero constants $\{\Lambda_{2,n}\} \subseteq \R$. Under these sequences, the normalized matrix $\hat{D}(\theta)\Lambda_{1,n}$ must converge in distribution to a Gaussian matrix of full rank almost surely, and the normalized weight $\Lambda_{1,n}'c\Lambda_{2,n}$ must converge to a nonzero vector.

The convergence part of this assumption can be verified straightforwardly in the context of \cite{kleibergen:2005} as discussed above, who considers the sequence
\begin{equation}
\label{eq:A_shrink}
	\sqrt{n}A_{F_n} \to \mathcal{A} \in \R^{2(K+1)\times 2(M+1)}
\end{equation}
in which case $\Lambda_{1,n} = \sqrt{n}I$ and $\Lambda_{2,n} = \frac{1}{\sqrt{n}}$. More generally, \cite{stock/wright:2000} and \cite{chaudhuri/zivot:2011} consider a sequence
\begin{equation}
\label{eq:A_shrink_2}
	A_{F_n} = (0_{2(K+1)\times q}, A_{\text{full}}) + A_{\text{sing},F_n} \quad \text{and} \quad \sqrt{n}A_{\text{sing},F_n} \to \mathcal{A}
\end{equation}
in which case $A_{\text{full}}$ has full column rank, representing those parameters that are strongly identified. We can set $\Lambda_{1,n} = \diag\{\sqrt{n}I_{q},I_{2(M+1) - q}\}$ and $\Lambda_{2,n} = \frac{1}{\sqrt{n}}$ such that Assumption 4 still holds.

However, the sequences \eqref{eq:A_shrink} and \eqref{eq:A_shrink_2} are inappropriate for the MTE setup. To see this, consider again the example on a linear MTE model with a binary IV as discussed above:
\begin{example}[Linear MTE model with a binary IV]
Suppose $K=M=1$ and set $h_{1}(u) = u - \frac{1}{2}$. Along a sequence of DGPs $\{F_n\}$, let $p_{F_n}(z_0) = p_{F_n}(z_1) = p \in (0,1)$ for each $n\geq 1$. Since $\overline{K} = |\{p_{F_n}(z_0), p_{F_n}(z_1)\}| = 1 < M+1 = 2$, Assumption \ref{asp:identification}.7 fails in this example. Note that the matrix $A_{F_n}$ becomes
\[
	A_{F_n} = 
	\begin{pmatrix}
	\begin{matrix}
	1 & \frac{1}{2} (p - 1) \\
	1 & \frac{1}{2} (p - 1)
	\end{matrix} 
	&
	0_{2\times 2} \\
	0_{2\times 2} & 
	\begin{matrix}
	1 & \frac{1}{2}p \\
	1 & \frac{1}{2}p
	\end{matrix}
	\end{pmatrix}
\]
In this case, neither \eqref{eq:A_shrink} nor \eqref{eq:A_shrink_2} holds here.
\end{example}
More generally, a singular but nonzero limit of $A_{F_n}$ would not satisfy conditions \eqref{eq:A_shrink} or \eqref{eq:A_shrink_2}. Similar examples of weakly identified models that fail to meet these conditions are discussed in \cite{andrews/guggenberger:2017}, where a singular value decomposition (SVD) technique is introduced to establish uniform validity of the classical RLM test for the \textit{full} vector. To achieve the same goal of uniform validity, I generalize their SVD approach to address inference on parameter \textit{functionals} in the MTE framework. My approach proceeds with a SVD of $A_{F_n}$:
\[
	A_{F_n} = C_{F_n}
	\underbrace{
	\begin{bmatrix}
		\begin{pmatrix}
			\tau_{F_n,1} & 0 & \ldots & 0 \\
			0 & \tau_{F_n,2} & \ldots & 0 \\
			\vdots & \vdots & \ddots & \vdots \\
			0 & 0 & \ldots & \tau_{F_n, 2(M+1)}
	\end{pmatrix} \\
	0_{2(K-M)\times 2(M+1)}
	\end{bmatrix}}_{\Pi_{F_n}}
	B_{F_n}',
\]
where $C_{F_n} \in \R^{2(K+1)\times 2(K+1)}$ and $B_{F_n} \in \R^{2(M+1)\times 2(M+1)}$ are orthogonal matrices, and $\infty \geq \tau_{F_n,1} \geq \tau_{F_n,2} \geq  \ldots \geq \tau_{F_n,2(M+1)} \geq 0$ are singular values of $A_{F_n}$ in a descending order. I show that the following normalizing matrices establish convergence of $\hat{D}(\theta)\Lambda_{1,n}$ and $\Lambda_{1,n}'c\Lambda_{2,n}$:
\begin{align*}
	&\Lambda_{1,n} = B_{F_n} \diag\{(\tau_{F_n,1})^{-1},\ldots, (\tau_{F_n,q})^{-1}, \sqrt{n},\ldots, \sqrt{n}\} \\
	&\Lambda_{2,n} = \|\Lambda_{1,n}'c\|^{-1},
\end{align*}
where $q$ is the number of scaled singular values $\{\sqrt{n}\tau_{F_n,j}\}_{j=1}^{2(M+1)}$ that diverge to infinity.

Despite the convergence of $\hat{D}\Lambda_{1,n}$ under the constructed normalizing matrices, its asymptotic limit may be rank deficient. This issue prevents the use of the efficiently weighted projection matrix in the subvector RLM statistic. A similar issue was also noticed by \cite{andrews/guggenberger:2017}, who impose additional parameter space restrictions to avoid this problem.  To address this problem, I modify the matrix $\hat{D}$ by adding a small noise of size $n^{-1/2}$:
\begin{equation}
\label{eq:D_tilde}
	\widetilde{D}(\theta) = \hat{D}(\theta) + \kappa n^{-1/2}\xi,
\end{equation}
where $\kappa > 0$ is a tuning parameter and $\xi \in \R^{2(K+1)\times 2(M+1)}$ is a matrix of i.i.d. standard normal random variables independent of data. This perturbation ensures that $\widetilde{D}(\theta)\Lambda_{1,n}$ achieves full rank almost surely while having asymptotically negligible effects on the test statistic under strong identification.
This perturbation technique draws inspiration from the AR/LM test\footnote{In addition to the modification in \eqref{eq:D_tilde}, D. \citet[eq. (7.11)]{andrews.d:2017} introduces two tuning parameters, $(K_L^*, K_U^*)$, to categorize identification strength. However, this categorization is not required for our inference procedure, and there is no clear guidance for choosing these parameters.}\citep{andrews.d:2017}. While AR/LM test focuses specifically on inference for a subset of parameters, my paper addresses a different problem of inference on a linear functional of parameters.

Define the new matrix under projection as follows:
\[
    \hat{Q}(\theta) = \hat{\Omega}(\theta)^{-1/2}\widetilde{D}(\theta)(\widetilde{D}(\theta)'\hat{\Omega}(\theta)^{-1}\widetilde{D}(\theta))^{-1}c,
\]
and the corresponding modified RLM (MRLM) statistic becomes
\[
    \text{MRLM}_n(\theta) = n(\hat{A}\theta - \hat{\beta})' \hat{\Omega}(\theta)^{-1/2} P_{\hat{Q}(\theta)} \hat{\Omega}(\theta)^{-1/2}(\hat{A}\theta - \hat{\beta}).
\]
Then the modified LC (MLC) statistic is given by
\begin{align}
	\text{MLC}_n(\theta) 
		= \text{MRLM}_n(\theta) + a \cdot \text{AR}_n(\theta). \label{eq:modified_lc}
\end{align}
In section \ref{sec:unifom_size}, I establish the uniform validity of using \eqref{eq:modified_lc} for conducting inference on the linear function $c'\theta$ even if Assumption \ref{asp:identification}.7 might fail or be close to failing. 

\subsection{Implementation of the MLC test}
\label{sec:construct_MLC}
In this section, I describe the implementation of the projection test based on the MLC statistic as below:
\begin{enumerate}
\item Construct the key quantities for the test statistics
	\begin{enumerate}
	\item Construct the estimators of $\hat{p}$ and $\hat{\beta}$, as well as their asymptotic variances estimators $\hat{\Sigma}_p$ and $\hat{\Sigma}_\beta$, as in section \ref{sec:notation}. Plugging in the estimator $\hat{p}$ into the matrix $A$ then obtains the estimator $\hat{A}$.
	\item Define the asymptotic variance estimator of the moment function
	\begin{equation}
	\label{eq:asy_var_est}
	\hat{\Omega}(\theta) = H(\hat{p}, \theta)\hat{\Sigma}_p H(\hat{p}, \theta)' + \hat{\Sigma}_\beta
	\end{equation}
	where
	\[
	H(p,\theta) = 
	\begin{bmatrix}
		\diag\left\{\sum_{m=0}^M \theta_{1m} \lambda_{1m}'({p}(z_\ell)): \ell = 0,1,\ldots, K\right\} \\
		\diag\left\{\sum_{m=0}^M \theta_{0m} \lambda_{0m}'({p}(z_\ell)): \ell = 0,1,\ldots, K\right\}
	\end{bmatrix}.
	\]
	\item For each $m=0,1,\ldots,M$, let 
	\[
		L_m(\hat{p}) = \diag\{\lambda_{1m}'(\hat{p}(z_0)), \ldots, \lambda_{1m}'(\hat{p}(z_K))\}
	\] and  
	\[
		R_m(\hat{p}) = \diag\{\lambda_{0m}'(\hat{p}(z_0)), \ldots, \lambda_{0m}'(\hat{p}(z_K))\}.
	\] 
        For each $j=1,\ldots,2(M+1)$, define a column vector       $\hat{d}_j(\theta)$ and the covariance estimator $\hat{\Gamma}_j(\theta)$ as below
        \begin{align}
	& \hat{d}_j(\theta) \equiv \hat{a}_j - \hat{\Gamma}_j(\theta)\hat{\Omega}(\theta)^{-1}(\hat{A}\theta - \hat{\beta}) \notag \\
	& \hat{\Gamma}_{j}(\theta) = 
			M_{j} (\hat{p})
			\hat{\Sigma}_p 
			H(\hat{p},\theta)' \label{eq:asy_cov_est}
        \end{align}
        where
        \begin{equation}
            \label{eq:defn_M_mat}
            M_j(\hat{p}) = 
        	\begin{cases}
        		\begin{bmatrix}
        			L_{j-1}(\hat{p}) \\ 0_{(K+1)\times (K+1)} 
        		\end{bmatrix}
        		 & \text{if } j \leq M+1 \\
        		\begin{bmatrix}
        			0_{(K+1)\times (K+1)}  \\ R_{j-M-2}(\hat{p})
        		\end{bmatrix}
        		 & \text{if } j > M+1 \\
        	\end{cases}.
        \end{equation}
    Let $\hat{D}(\theta) = (\hat{d}_1(\theta),\ldots,\hat{d}_{2(M+1)}(\theta))$ and $\widetilde{D}(\theta) = \hat{D}(\theta) + \kappa n^{-1/2}\xi$, where $\xi\in \R^{2(K+1)\times 2(M+1)}$ is a matrix of standard normal random variables independent of data. $\kappa$ is a positive tuning parameter chosen as $10^{-6}$ following D. \cite{andrews.d:2017}. 
    \end{enumerate}
    
\item Given the estimators $\hat{A}$, $\hat{\beta}$, $\hat{\Omega}(\theta)$, and $\widetilde{D}(\theta)$ constructed from step 1, compute the modified LC statistics as in equation \eqref{eq:modified_lc} for $\theta \in \Theta$. For the empirical application and simulation studies, I set the AR weight coefficient as 0.05, which is close to the suggested weight $a(\gamma)$ introduced by I. \citet[page 343]{andrews.i:2018} by setting $\alpha = 5\%$, $\gamma = 10\%$, and $K = 15$, consistent with the application considered in section \ref{sec:empirical}.

\item Let $\alpha \in (0,1)$ be the significance level. The modified LC test rejects the null hypothesis $H_0: c'\theta = \lambda$ if the profiled test statistic over the linear manifold $c'\theta = \lambda$ is larger than their respective critical values, i.e.,
	\begin{align*}
		\hat{\phi}_{\text{MLC}}(\lambda) = \indicator\left[\inf_{c'\theta = \lambda} \text{MLC}_n(\theta) > q_{(1+a)\chi_1^2 + a\chi_{2K+1}^2}(1-\alpha)\right]
	\end{align*}
where $q_{(1+a)\chi_1^2 + a\chi_{2K+1}^2}(1-\alpha)$ denotes the $(1-\alpha)$-quantile of the mixture chi-square distribution $(1+a)\chi_1^2 + a\chi_{2K+1}^2$ with $\chi_1^2 \indep \chi_{2K+1}^2$. Additionally, define $\hat{\phi}_{\text{MLC}}(\lambda) = 1$ for $\lambda \not\in \{c'\theta: \theta \in \Theta\}$.
\end{enumerate}
A robust confidence set can be obtained by inverting the test, that is,
\[
	\mathcal{C}_{\text{MLC}} = \{\lambda\in\R: \hat{\phi}_{\text{MLC}}(\lambda) = 0\}.
\]
In practice, I report the convex hull of the confidence set so that it becomes an interval.

\subsection{Uniform validity}
\label{sec:unifom_size}
The following theorem establishes the uniform validity of the MLC test based on the profiling the MLC statistic.
\begin{theorem}
\label{thm:uniform_validity}
Let Assumption \ref{asp:iid} hold, and suppose that the weight $c$ is a nonzero fixed vector. Then we have
\[
	\limsup_{n\to\infty}\sup_{(\lambda,F)\in\mathcal{P}_0} \Exp_F[\hat{\phi}_{\text{MLC}}(\lambda)] \leq \alpha.
\]
\end{theorem}

By this theorem, the MLC test has an asymptotic size less than or equal to the nominal level $\alpha \in (0,1)$ for the parameter space $\mathcal{P}_0$. In other words, the MLC test is uniformly valid regardless of the model's identification status. In Online Appendix \ref{appendix:power_analysis_MLC}, I further demonstrate that the MLC test is consistent for distant (or fixed) alternatives and has comparable power to the asymptotic efficient Wald test under strong identification when the AR weight $a$ is sufficiently small. These findings highlight the usefulness of the MLC test for researchers seeking robust inference against weak identification.

\section{Incorporating Covariates}
\label{sec:covariates}

In this section, I consider a set of covariates $W \in \R^L$ included in the MTE model. I first establish an augmented set of assumptions under which the MTE function can be point identified conditional on $W=w$ for each $w \in \supp(W)$. Then I argue that a commonly used additive separability condition can introduce large bias in the estimand of target causal effects when unobserved heterogeneity interacts with the covariates. To this end, I recommend conducting the robust inference procedure by conditioning on covariates rather than imposing additive separability.

\subsection{Identifying assumptions with covariates}
Consider the following MTE model with an additional set of covariates $W$:
\begin{align*}
    & Y_d = \mu_d(W) + V_d \\
    & D = \indicator[U \leq \nu(W,Z)],
\end{align*}
where $\mu_d(W) \equiv \Exp[Y_d\mid W]$ denotes the conditional expectation of potential outcomes. Similar to Theorem \ref{thm:identification}, we can identify the MTE function $\Exp[Y_1-Y_0\mid F_{U|W}(U|W)=u, W=w]$ under the following assumptions.
\begin{assumption}[MTE model with covariates] 
\hfill \vspace{-1em}
\label{asp:identification_cov}
    \begin{enumerate}
    \item $Z \indep U\mid W$.
    \item $\Exp[Y_d\mid Z, W, U] = \Exp[Y_d\mid W, U]$ and $\Exp|Y_d| < \infty$ for $d \in \{0,1\}$.
    \item $U\mid W = w$ is continuously distributed for all $w \in \supp(W)$.
    \item $0 < \Prob(D=1\mid W=w, Z=z) <1$ for all $(w,z) \in \supp(W, Z)$.
    \item $\Exp[Y_d\mid F_{U\mid W}(U|W) = u, W=w] = \mu_d(w) + \sum_{m=1}^M \rho_{dm}(w) h_m(u)$ for some known continuous functions $\{h_m(\cdot)\}_{m=1}^{M}$ for each $d = 0,1$ and $u\in(0,1)$, where $F_{U\mid W}(\cdot|w)$ denotes the distribution function of $U$ conditional on $W=w$.
    \item $\{\lambda_{1m}(\cdot)\}_{m=0}^{M}$ and $\{\lambda_{0m}(\cdot)\}_{m=0}^M$ are unisolvent on $(0,1)$, where $\lambda_{00}(\cdot) = \lambda_{10}(\cdot) \equiv 1$, and 
    \[
    	\lambda_{1m}(p) \equiv \frac{1}{p}\int_{0}^p h_m(u) du \quad \text{and} \quad \lambda_{0m}(p) \equiv \frac{1}{1-p}\int_{p}^1 h_m(u) du \quad \text{for $m = 1,\ldots,M$}.
    \]
    \item $|\{\Prob(D=1\mid Z=z, W=w): z \in \supp(Z\mid W=w)\}|\geq M+1$ for all $w \in \supp(W)$.
    \end{enumerate}
\end{assumption}
Under the above assumption, we can normalize the distribution of $U$ conditional on the random covariate $W$ to be uniformly distributed over the unit interval (i.e., $F_{U\mid W}(u|W) = u$) and thus $\nu(w,z) = p(w,z)\equiv \Prob(D=1\mid W=w, Z=z)$. The structural parameter 
\[
    \theta(w) = (\mu_1(w), \{\rho_{1m}(w)\}_{m=1}^M, \mu_0(w),\{\rho_{0m}(w)\}_{m=1}^M)
\]
can be identified for all $w\in\supp(W)$ by the following separate regressions:
\begin{equation}
\label{eq:separate_reg}
\begin{aligned}
	\Exp[Y\mid D=1, Z=z, W=w]  
	& = \mu_1(w) + \sum_{m=1}^M \rho_{1m}(w) {\int_0^{p(w,z)} \frac{h_m(u)}{p(w,z)} du} \\
	& =  \mu_1(w) + \sum_{m=1}^M \rho_{1m}(w) \lambda_{1m}(p(w,z))\\
	\Exp[Y\mid D=0, Z=z, W=w] 
	&= \mu_0(w) + \sum_{m=1}^M \rho_{0m}(w) {\int_{p(w,z)}^1 \frac{h_m(u)}{1-p(w,z)} du} \\
	& = \mu_0(w) + \sum_{m=1}^M \rho_{0m}(w) \lambda_{0m}(p(w,z)). 
\end{aligned}
\end{equation}

For some $w\in \supp(W)$, the variation of $\{p(w,z): z\in\supp(Z\mid W=w)\}$ can be weak in practice, especially when conditioning on certain groups of individuals who have very high (or low) probability of being treated. 
To address this problem, researchers commonly impose the additive separability condition as follows.
\begin{assumption}[Additive separability]
    \label{asp:add_sep}
    The MTR function is linear and additively separable in covariates and selection unobservable. That is, for $d = 0,1,$
    \[
        \Exp[Y_d\mid W, U] = \mu_d + W'\tau_d + \Exp[V_d\mid U].
    \]
\end{assumption}
This assumption eliminates heterogeneous effects of covariates $W$ interacted with selection unobservable $U$ on potential outcomes, i.e., $\rho_{dm}(w)$ does not vary with $w$. Up to a linear term $W'\tau_d$, the covariate $W$ is mean independent of $Y_d$ conditional on $U$. This enables point identification of MTEs on the unconditional support of propensity scores $p(W,Z)$ by using variation from covariates in addition to exogenous variation from discrete instruments \citep[Section I.B]{carneiro/heckman/vytlacil:2011}. Next, I show that the failure of Assumption \ref{asp:add_sep} can lead to a biased estimator of treatment effects unless $W$ is uncorrelated with the treatment and propensity score. Given this result, I consider Assumption \ref{asp:add_sep} to be a strong assumption and implement inference in the empirical analysis by conditioning on covariates rather than relying on additive separability.

\subsection{Bias from additive separability}
\label{sec:bias_add_sep}
I show the population bias of treatment effect estimands for a specific class of DGP as follows:
\begin{equation}
\label{eq:DGP_for_bias}
    \Exp[Y_d \mid W = w, U = u] = \mu_d + w'\tau_d + \rho_d(w) h(u),
\end{equation}
where $\rho_d(w) = \rho_d + w'\eta_d$ for $d = 0,1$, and $h(\cdot)$ is strictly increasing and integrates to zero over $[0,1]$. This model can be viewed as an MTE model satisfying Assumption \ref{asp:identification_cov} with $M=1$ and with $\mu_d(w)$ linear in $w$. A related specification of this form has appeared in \citet[equation (12)]{kline/walters:2019} under additive separability, where $\rho_d(w)$ is treated as constant in $w$. So the model \eqref{eq:DGP_for_bias} can be considered as a direct generalization of their setting by allowing the effect of $W$ to vary with the selection unobservable $U$, therefore potentially violating the additive separability Assumption \ref{asp:add_sep}. Moreover, the bias derived under the specific model class \eqref{eq:DGP_for_bias} indicates that the worst-case bias becomes even larger when extending to a broader class of models where additive separability does not hold.

To obtain correct estimates on the model parameters, one should implement separate regressions as follows:
\begin{equation}
\label{eq:longreg}
\begin{aligned}
    &\Exp[Y\mid D=1, W=w, P=p] = \mu_1 + w'\tau_1 + \rho_1 \lambda_1(p) + [w\lambda_1(p)]'\eta_1 \\
    &\Exp[Y\mid D=0, W=w, P=p] = \mu_0 + w'\tau_0 + \rho_0 \lambda_0(p) + [w\lambda_0(p)]'\eta_0
\end{aligned}
\end{equation}
where $P = \Prob(D=1\mid Z, W)$ is the propensity score, and
\[
    \lambda_1(p) = \frac{1}{p} \int_0^p h(u) du
    \quad \text{and} \quad
    \lambda_0(p) = \frac{1}{1-p} \int_p^1 h(u) du.
\]

However, suppose a researcher mistakenly imposes additive separability. Thus $\rho_d(w)$ is considered as a constant in the model with $\eta_0 = \eta_1 = 0_{L\times 1}$. As a result, shorter regressions with omitted interaction terms $W\lambda_d(P)$ will be implemented: 
\begin{equation}
\label{eq:shortreg}
\begin{aligned}
    &Y = \tilde{\mu}_1 +  W'\tilde{\tau}_1 + \tilde{\rho}_1 \lambda_1(P) + Y^{\perp W, \lambda_1(P)\mid D=1}  &\text{conditional on } D = 1\\
    &Y = \tilde{\mu}_0 +  W'\tilde{\tau}_0 + \tilde{\rho}_0 \lambda_0(P) + Y^{\perp W, \lambda_0(P)\mid D=0} &\text{conditional on } D = 0
\end{aligned}
\end{equation}
where $Y^{\perp X\mid D=d}$ denotes the OLS residual of $Y$ from regressing $Y$ on $(1,X)$ conditional on subsamples $D=d$:
\[
    Y^{\perp X\mid D=d} \equiv Y - \tilde{X}'\Exp[\tilde{X}\tilde{X}'\mid D=d]^{-1}\Exp[\tilde{X}Y\mid D=d].
\]
where $\tilde{X} = (1,X')'$.

The next lemma compares the true parameters $(\rho_d, \tau_d)$, obtained as estimands from the correctly specified model in \eqref{eq:longreg}, with those from the misspecified regressions in \eqref{eq:shortreg}. 

\begin{lemma}
    \label{lem:bias_formula_general}
    Under the DGP \eqref{eq:DGP_for_bias}, the bias on $\rho_d$ equals 
    \begin{equation}
    \label{eq:bias_rho_general}
        \tilde{\rho}_d - \rho_d = \frac{\cov(W'\lambda_d(P), \lambda_d(P)^{\perp W\mid D=d} \mid D=d)}{\var(\lambda_d(P)^{\perp W\mid D=d}\mid D=d)} ~~\eta_d
    \end{equation}
    and the bias on $\tau_d$ equals 
    \begin{equation}
    \label{eq:bias_tau_general}
        \tilde{\tau}_d - \tau_d = \Exp[(W^{\perp \lambda_d(P)\mid D=d})(W^{\perp \lambda_1(P)\mid D=d})'\mid D=d]^{-1}\Exp[(W^{\perp \lambda_d(P)\mid D=d})(W'\lambda_d(P))\mid D=d]~\eta_d
    \end{equation}
\end{lemma}

The bias formulas are direct consequences of Frisch–Waugh–Lovell (FWL) theorem. It shows that the degree of treatment effects heterogeneity of the covariate $W$, measured by $\eta_d$, has a nontrivial impact on the bias of estimands. 

The bias formulas \eqref{eq:bias_rho_general} and \eqref{eq:bias_tau_general} simplify further under the following additional assumptions:
\begin{assumption}
    \label{asp:weak_exogeneous_covariates}
    Under the DGP \eqref{eq:DGP_for_bias}, the following conditions hold:
    \begin{enumerate}
        \item $W$ is uncorrelated with the control functions of propensity score conditional on treatment status: For $d = 0,1,$
        \begin{align*}
        & \cov(W, \lambda_d(P)\mid D=d) = 0_{L\times 1} \\  &\cov(WW',\lambda_d(P)\mid D=d) = 0_{L\times L} \\
        &\cov(W, \lambda_d(P)^2\mid D=d) = 0_{L\times 1}.
        \end{align*}
        \item $W$ is uncorrelated with treatment: $\Exp[W\mid D=1] = \Exp[W\mid D=0].$
    \end{enumerate}
\end{assumption}
Assumptions \ref{asp:weak_exogeneous_covariates}.1 posits that the covariate $W$ and the control function $\lambda_d(P)$ are uncorrelated up to their second moment, given the treatment or control groups. This assumption is notably strong, suggesting that the covariate $W$ cannot cause significant variation on the propensity scores within both groups. However, the bias of treatment effect estimands is shown to persist even under this stringent assumption. Only if researchers are willing to assume that $W$ is completely ``irrelavant''  under an additional Assumption \ref{asp:weak_exogeneous_covariates}.2, the misspecified model \eqref{eq:shortreg} would yield correct estimands of average treatment effects and the slope of MTEs as the ones obtained from the true model \eqref{eq:longreg}.

I focus on the bias for ATE and slope of MTE curve because the shape of the unconditional MTE, $\Exp[Y_1-Y_0\mid U=u]$, is uniquely determined by these quantities. I also consider the conditional ATE that captures  treatment effect heterogeneity induced by covariates. These parameters are defined below:
\begin{align*}
    & \text{ATE} = \Exp[Y_1 - Y_0] = \mu_1 - \mu_0 + \Exp[W'(\tau_1 - \tau_0)] \\
    & \text{CATE} = \Exp[Y_1 - Y_0\mid W=w] = \mu_1 - \mu_0 + w'(\tau_1 - \tau_0) \\
    & \text{Slope} = \Exp[\rho_1(W) - \rho_0(W)] = \rho_1 - \rho_0 + \Exp[W'(\eta_1 - \eta_0)].
\end{align*}
Under (misspecified) additive separability, the researcher may estimate these effects by 
\begin{align*}
    & \widetilde{\text{ATE}} = \tilde{\mu}_1 - \tilde{\mu}_0 + \Exp[W'(\tilde{\tau}_1 - \tilde{\tau}_0)] \\
    & \widetilde{\text{CATE}} = \tilde{\mu}_1 - \tilde{\mu}_0 + w'(\tilde{\tau}_1 - \tilde{\tau}_0) \\
    & \widetilde{\text{Slope}} = \tilde{\rho}_1 - \tilde{\rho}_0.
\end{align*}
The next theorem shows the difference between the causal parameters and their corresponding estimands under Assumption \ref{asp:weak_exogeneous_covariates}.
\clearpage
\begin{theorem}
Under the model \eqref{eq:DGP_for_bias}, 
    \label{thm:bias_formula_specific}
    \begin{enumerate}
        \item Suppose Assumption \ref{asp:weak_exogeneous_covariates}.1 holds, then 
        \begin{align*}
            \widetilde{\text{ATE}} - \text{ATE} 
            &= (\Exp[W] - \Exp[W\mid D=1])'\eta_1\times \Exp[\lambda_1(P)\mid D=1] \\
                &\quad - (\Exp[W] - \Exp[W\mid D=0])'\eta_0 \times \Exp[\lambda_0(P)\mid D=0] \\
            \widetilde{\text{CATE}} - \text{CATE}
            &= (w - \Exp[W\mid D=1])'\eta_1\times \Exp[\lambda_1(P)\mid D=1] \\
                &\quad - (w - \Exp[W\mid D=0])'\eta_0 \times \Exp[\lambda_0(P)\mid D=0]  \\
            \widetilde{\text{Slope}} - \text{Slope}
            &= \left(\Exp[W\mid D=1] - \Exp[W\mid D=0]\right)'(\Prob(D=0) ~ \eta_1 + \Prob(D=1) ~ \eta_0).
        \end{align*}
        \item Suppose Assumption \ref{asp:weak_exogeneous_covariates}.1 and \ref{asp:weak_exogeneous_covariates}.2 hold, then $\widetilde{\text{ATE}} = \text{ATE}$ and 
        $\widetilde{\text{Slope}} = \text{Slope}$.
    \end{enumerate}
\end{theorem}
\begin{remark}
    While the estimand of the ATE may remain unbiased under both conditions outlined in Assumption \ref{asp:weak_exogeneous_covariates}, it is important to note that the bias on CATE does not necessarily disappear under such assumption. Consequently, researchers should be cautious about interpreting CATE estimates in short regressions \eqref{eq:shortreg} even if they have justified the validity of Assumption \ref{asp:weak_exogeneous_covariates}.
\end{remark}
From this theorem, it becomes evident that the bias on ATE and the slope of MTE are driven by two main factors: the magnitude of $\eta_d$, which measures the heterogeneous effects of $W$, and the extent to which the covariate $W$ is unbalanced between the treatment and control groups. Therefore, the bias on those estimands can be quite significant if we omit heterogeneous effects of covariates that vary with unobserved heterogeneity $U$. In Online Appendix \ref{appendix:numerical_example}, I provide a numerical example illustrating that such bias can even alter the sign of the ATE estimand when additive separability is mistakenly imposed.

Due to potential bias from misspecified additive separability, Online Appendix \ref{sec:sidak_correction} shows how to apply the proposed methodology when conditioning on a set of discrete covariates using the Šidák–Bonferroni correction. For researchers who wish to impose additive separability in MTE models, Online Appendix \ref{sec:inference_addsep} provides guidelines on how to extend the robust inference procedure to that setting.

\section{Monte Carlo Simulation}
\label{sec:simulation}
In this section, I compare the finite-sample performance of the proposed MLC test with the classical Wald test using simulated data from a simple quadratic MTE model with a three-valued instrument. The power comparison between conditional Wald tests, MLC tests, and classical Wald tests in a linear MTE model can be found in Online Appendix \ref{appendix:add_simulation}.

Consider a DGP specified as follows: For each $d = 0,1,$
\begin{align*}
	& Y_d = \mu_d + V_d \\
	& D = \indicator[U \leq p(Z)] \\
	& V_d = \rho_{d1}\left(U - \frac{1}{2}\right) + \rho_{d2}\left(U^2 - \frac{1}{3}\right) + e_d,
\end{align*}
where $Z$ is uniformly distributed over three points $\{z_0,z_1,z_2\}$ and is independent of $(U, e_1, e_0)$. The error terms $(e_1, e_0)$ follow a joint normal distribution with zero mean and covariance matrix $\Sigma_e = 0.5 \cdot I_{2\times 2}$. For simplicity, I set $\mu_1 = \mu_0 = 0$ and impose strong endogeneity by specifying
\[
	\rho_{11} = \rho_{12} = -5 \quad \text{ and } \quad \rho_{01} = \rho_{02} = 5.
\] 
Regarding the specification of the propensity score $p(z)$ for $z \in \{z_0,z_1,z_2\}$, I fix $p(z_0) = 0.5$ and let $(p(z_1), p(z_2))$ vary across $(0.05,0.95)^2$ to generate a variety of degrees/directions of weak identification. By drawing 2,000 i.i.d. simulation samples, I implement the Wald test and the MLC test with AR weight $a = 0.05$ to assess the null hypothesis on testing ATE $H_0: \mu_1 - \mu_0 = 0$. The average null rejection rates based on 5\% significance level are displayed in Figure \ref{fig:3dsize}.

\begin{figure}[h!]

  \centering	
  
  \caption{3D Plot of Empirical Size for MLC and Wald Tests}
  
\includegraphics[width=0.95\textwidth]{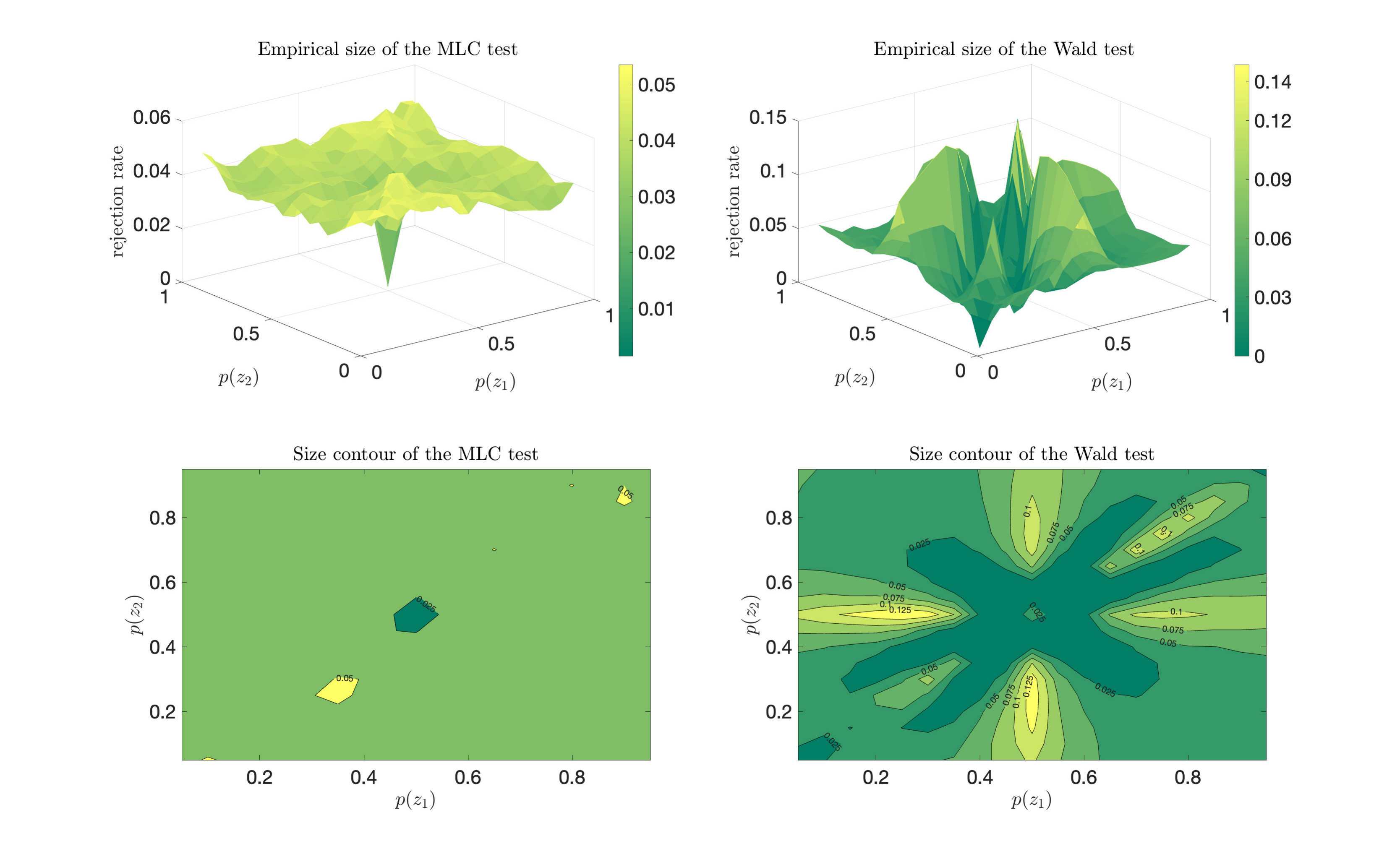}

\begin{tablenotes}
    \footnotesize \raggedright
    Note: The plots show empirical rejection rates for testing the null hypothesis $H_0: \text{ATE} = 0$ at the 5\% significance level. The upper panel presents three-dimensional surfaces of rejection rates, while the lower panel shows the corresponding contour plots. Each plot varies $p(z_1)$ and $p(z_2)$ across $(0.05,0.95)$ with $p(z_0)$ fixed at 0.5. The contour lines correspond to rejection rates of 2.5\%, 5\%, 7.5\%, 10\%, and 12.5\%, with regions between these lines representing intermediate rejection rates. Darker regions indicate lower rejection rates.
\end{tablenotes}
    
\label{fig:3dsize}
    
\end{figure}

The null rejection rates reveal distinct patterns in tests performance across propensity score variation. The conventional Wald test exhibits under-rejection when the three propensity scores are similar, indicating trivial power in this weakly identified scenario to be shown below. When the propensity scores cluster at two points (i.e., when either $p(z_1)$ or $p(z_2)$ equals $p(z_0) = 0.5$), the ATE parameter becomes partially identified, and the Wald test's rejection rates reach to 14\%, substantially exceeding the nominal 5\% significance level. This demonstrates the Wald test's invalidity under identification failure. In contrast, the proposed MLC test maintains proper size control across all propensity score combinations, demonstrating its robustness to weak identification.

\begin{figure}[htbp]
  \centering
  
  \caption{Power Curves of the MLC and Wald Test}
  
  \begin{subfigure}{0.95\textwidth}
  	\centering
  	\caption{$p(z) = [0.2,0.5,0.8]$}
    \includegraphics[width=0.5\linewidth]{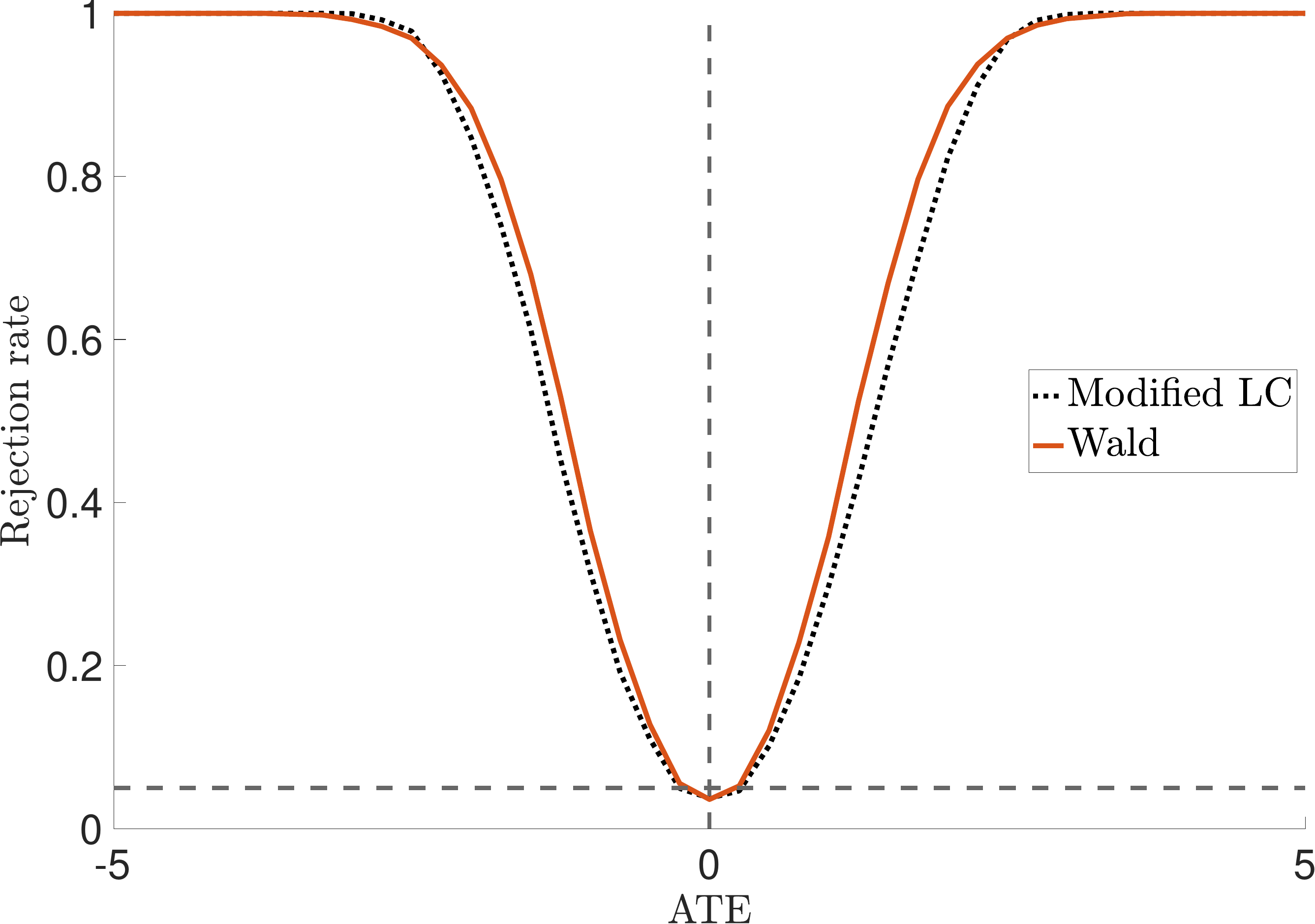}
    \label{fig:subplot1}
  \end{subfigure}
  
  \bigskip
  
  \begin{subfigure}{0.95\textwidth}
  	\centering	
  	\caption{$p(z) = [0.2,0.5,0.5]$}
    \includegraphics[width=0.5\linewidth]{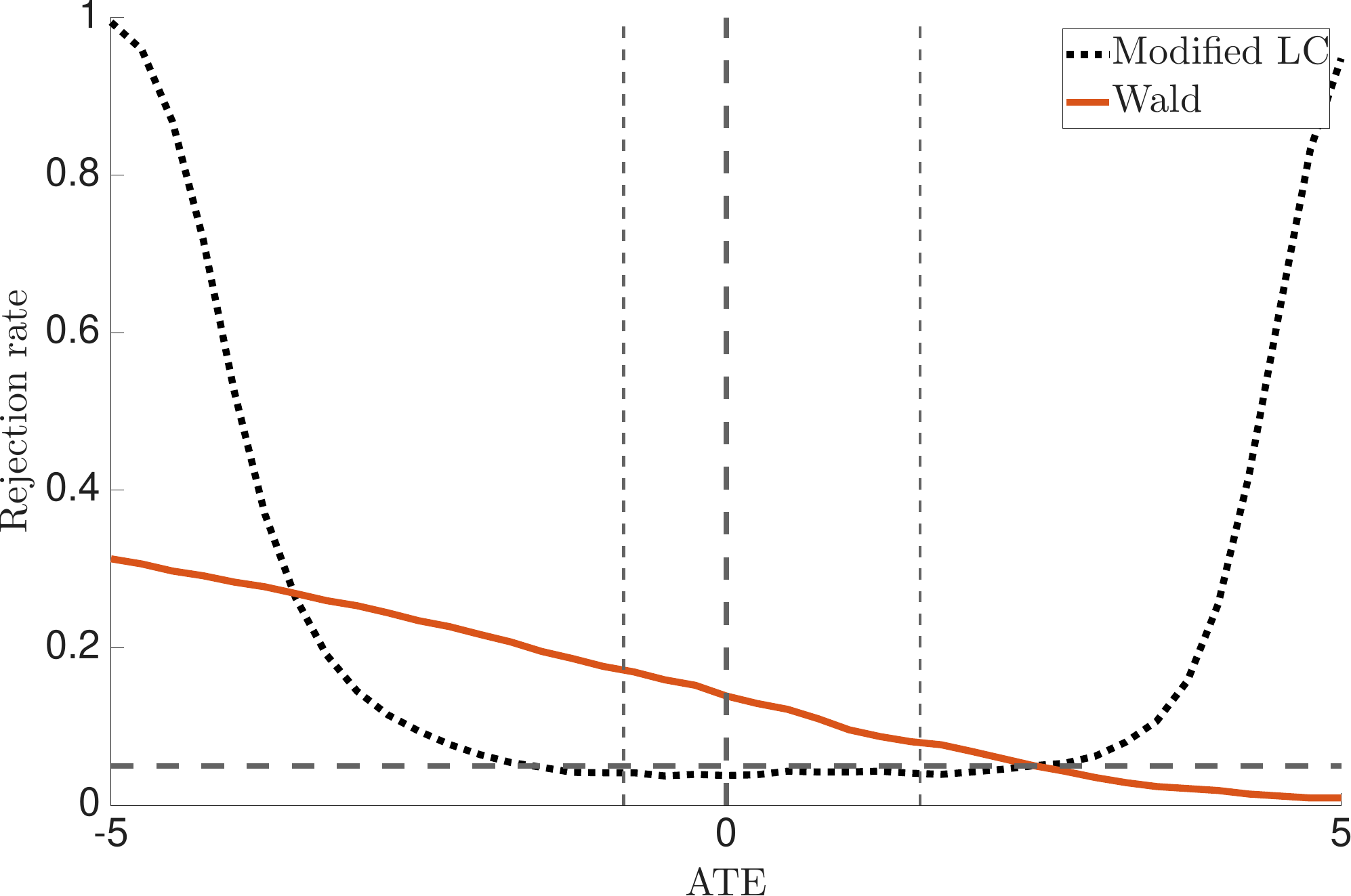}
    \label{fig:subplot2}
  \end{subfigure}
  
  \bigskip
  
  \begin{subfigure}{0.95\textwidth}
  	\centering
  	\caption{$p(z) = [0.4,0.5,0.6]$}
    \includegraphics[width=0.5\linewidth]{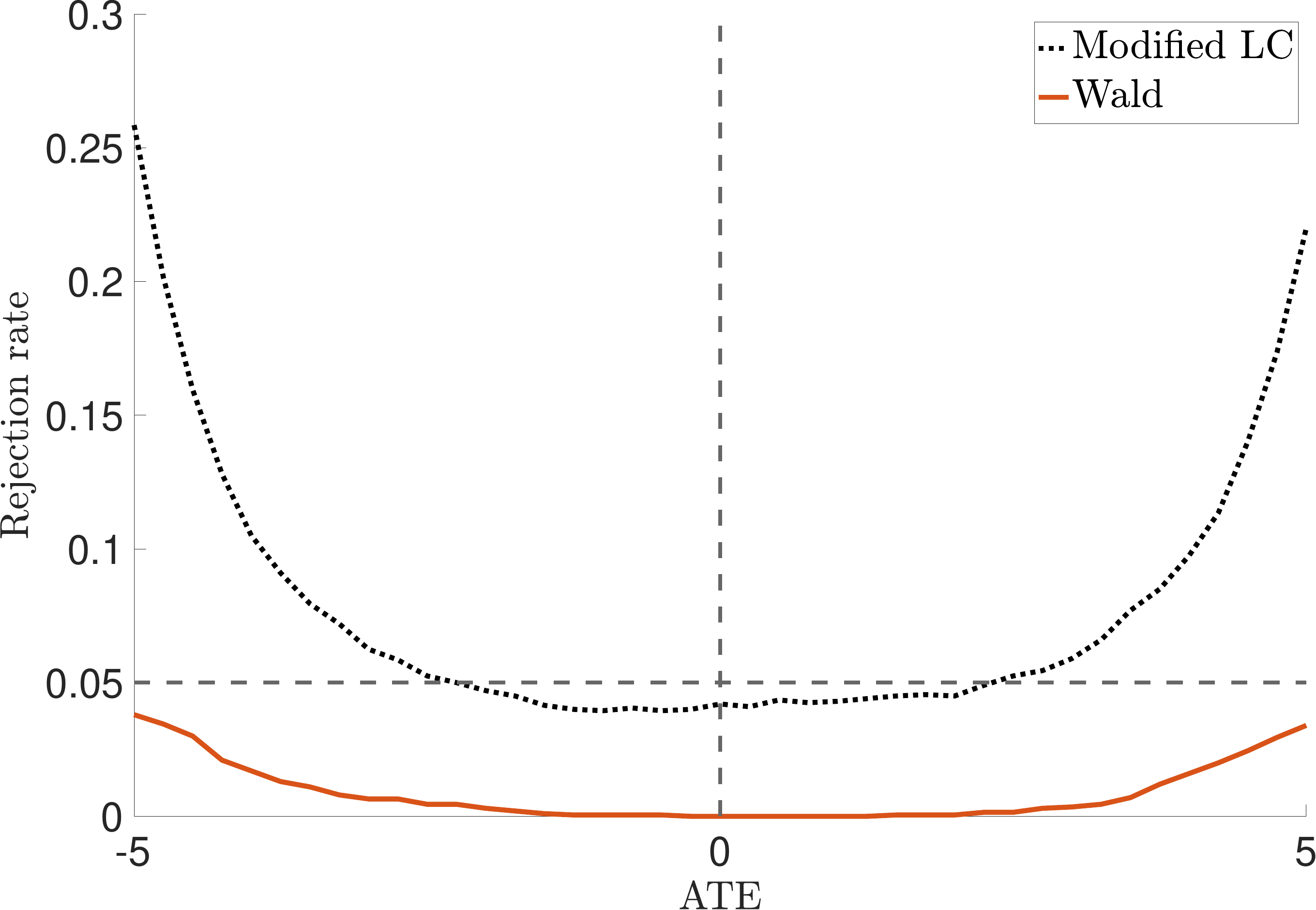}
    \label{fig:subplot3}
  \end{subfigure}
  
  \begin{tablenotes}
  \footnotesize \raggedright
      Note: {Testing ATE at values on $[-5,5]$ with the true effects fixed at zero. The significance level is set at 5\%. The sample size equals 2,000 and the average rejection rates are computed with 2,000 independent Monte-Carlo simulations. The two dashed lines around zero in subfigure (b) characterize the identified set of ATE by imposing a parameter space $[-10,10]$ for each element of $\theta$.}
  \end{tablenotes}
  
  \label{fig:power_curve_ATE}
  
\end{figure}
  
Figure \ref{fig:power_curve_ATE} compares the power of the MLC and Wald tests across different identification scenarios. In Panel (a), where propensity scores are well-separated (strong identification), the Wald test achieves asymptotic efficiency, and the MLC test delivers comparable rejection power. In Panel (b), where the instrument takes only two distinct values (partial identification), the Wald test exhibits not only size distortion at the null but also power loss at alternatives. In contrast, the MLC test maintains correct size and almost surely rejects distant alternatives outside the identified set. In Panel (c), where propensity scores are clustered together (weak identification), the Wald test's rejection rates fall consistently below the nominal level, indicating negligible power, while the MLC test retains nontrivial power at distant alternatives.

\section{Empirical Application to Misdemeanor Prosecution}
\label{sec:empirical}
In this section, I revisit the empirical analysis by \cite{agan/doleac/harvey:2023}, who examined the causal effects of misdemeanor prosecution on defendants' subsequent criminal activity. Based on the quasi-randomized assignment of nonviolent misdemeanor cases to assistant district attorneys (ADAs), their LATE and MTE estimates demonstrate that nonprosecution leads to a large reduction in the likelihood of defendants' future criminal involvement over the next two years. To study the policy effects of increasing nonprosecution, they analyzed the impact of imposing a presumption of misdemeanor nonprosecution by relying on an existing policy change issued by a new district attorney. In this section, I use their data to answer some policy-relevant questions concerning the exogenous change on ADA nonprosecution rates. My approach does not require the additional data or information related to an actual policy that is already implemented, but instead extrapolates causal effects to address this issue. 

Similar to their analysis, I use the MTE framework to extrapolate treatment effects outside compliers. However, my goal is to analyze causal effects of implementing several counterfactual policies on ADA leniency, which were not explored in their empirical studies. To adapt their data into the framework considered in this paper, I implement the proposed inference procedure conditional on each court and use ADAs' identity as the discrete instrument since ADAs are randomly assigned conditional on courts and time. To avoid bias introduced by many instruments (see the references in \cite{mikusheva/sun:2022} for further discussions), I combine ADAs into a total of 15 groups for each court, excluding the smaller courts (BRI, CHE, and CHA) due to insufficient ADA counts.

Let the treatment $D$ be the nonprosecution status of defendants (that takes value 1 if the defendant is not prosecuted) and the outcome $Y$ be the indicator of subsequent criminal complaints within two years postarraignment. To validate the strong overlap condition on propensity scores, I focus on ADAs that handle more than 50 cases and have nonprosecution rate at least 0.025 within each court. The summary statistics of ADAs' nonprosecution rates conditional on each court are provided in Table \ref{tab:sum_stat} below. From this table we can see that the range of propensity scores varies from 0.21 to 0.50. When employing a high-order polynomial MTE model such as the cubic specification in Figure IV of \cite{agan/doleac/harvey:2023} in this setup, the finite-sample performance of confidence sets may be too poor to guarantee valid coverage. Specifically, in Online Appendix \ref{appendix:iv_strength}, I show the evidence of weak instruments in cubic and quartic MTE models, while the classical $F$ test (with threshold 10) does not deliver the same conclusion. The reason is that $F$ statistic aims to detect deviations from the null where all propensity scores are equal to each other. However, this is not sufficient to strongly identify MTE models with flexible structure that require three or more propensity score to be well separated from each other (see Figure \ref{fig:intuition_plot}).

\begin{table}[htbp]
  \centering
  \caption{Summary Statistics of Nonprosecution Rates across Courts}
    \begin{tabular}{lccccc}
    \toprule
    \multirow{2}[2]{*}{Court} & \multirow{2}[2]{*}{Number of ADA} & \multicolumn{3}{c}{Nonprosecution Rate } & \multirow{2}[2]{*}{Sample Size} \\
          &       & Min   & Mean  & Max   &  \\
    \midrule
    SBO   & 16    & 0.04  & 0.24  & 0.54  & 3,921 \\
    EBOS  & 22    & 0.08  & 0.32  & 0.53  & 7,566 \\
    WROX  & 43    & 0.05  & 0.37  & 0.55  & 8,905 \\
    BMC   & 65    & 0.04  & 0.17  & 0.4   & 9,593 \\
    ROX   & 66    & 0.06  & 0.16  & 0.33  & 13,333 \\
    DOR   & 73    & 0.04  & 0.13  & 0.43  & 13,523 \\
    CHA   & 4     & 0.09  & 0.24  & 0.42  & 362 \\
    CHE   & 12    & 0.07  & 0.16  & 0.3   & 872 \\
    BRI   & 5     & 0.14  & 0.29  & 0.48  & 885 \\
          &       &       &       &       &  \\
    Total & 262   & 0.04  & 0.22  & 0.55  & 58,960 \\
    \bottomrule
    \end{tabular}%
    
    \begin{tablenotes}
        \footnotesize\raggedright
        Notes: Courts include South Boston (SBO), East Boston (EBOS), West Roxbury (WROX), Central (BMC), Roxbury (ROX), Dorchester (DOR), Charlestown (CHA), Chelsea (CHE), and Brighton (BRI).
    \end{tablenotes}
  \label{tab:sum_stat}%
\end{table}%

For all the counterfactual experiments, I consider an exogenous change to the ADA nonprosecution rates. Let $p^\epsilon(z)$ be the counterfactual propensity score of not prosecuting defendants, where $\epsilon$ is a nonnegative scalar (or vector) denoting the degree of deviation from status quo. I compare the induced outcome $Y^\epsilon$, defined as 
\[
		Y^{\epsilon} = Y_1 \indicator[U \leq p^{\epsilon}(Z)] + Y_0 \indicator[U > p^{\epsilon}(Z)]
\]
with the observed outcome $Y$ in the data.

Similar to the policy invariance assumption imposed by \cite{heckman/vytlacil:2005}, suppose that this counterfactual change on propensity score does not shift the distribution of potential outcomes and unobserved cost $U$ of being selected into treatment, therefore omitting any ``general equilibrium'' effects that may arise after the change of prosecution rates. I consider three different comparisons between the counterfactual outcome $Y^\epsilon$ and the observed outcome $Y$:
\begin{enumerate}
\item Non-normalized policy effects:
\begin{equation}
\label{eq:policy_effects}
	\alpha(\epsilon) \equiv \Exp[Y^\epsilon - Y].
\end{equation} 
This criterion follows from \cite{heckman/vytlacil:2001}. The term ``non-normalized'' distinguishes these effects from the policy-relevant treatment effects discussed below.

\item Policy-relevant treatment effects (PRTE)
\begin{equation}
\label{eq:norm_policy_effects}
	 \bar{\alpha}(\epsilon) \equiv \frac{ \Exp[Y^\epsilon - Y]}{\Exp[p^{\epsilon}(Z) - p(Z)]}.
\end{equation}
This criterion follows from \cite{heckman/vytlacil:2005}. This quantity can be interpreted as the treatment effect for units shifted into treatment via the counterfactual policy. Note that the additive/proportional PRTEs suggested by \cite{carneiro/heckman/vytlacil:2010,carneiro/heckman/vytlacil:2011} and \cite{mogstad/torgovitsky:2018} are special cases of $\bar{\alpha}(\epsilon)$ for particular choices of $p^{\epsilon}$, as described in equations \eqref{eq:PRTE-additive} and \eqref{eq:PRTE-prop} below. 

\item Marginal policy-relevant treatment effects (MPRTE)
\begin{equation}
\label{eq:marginal_policy_effects}
	\alpha_{+}(0) \equiv \lim_{\epsilon \searrow 0} \frac{ \Exp[Y^\epsilon - Y]}{\Exp[p^{\epsilon}(Z) - p(Z)]}.
\end{equation}
This criterion follows from \cite{carneiro/heckman/vytlacil:2010}. When $p^{\epsilon}(z) \geq p(z)$ for all $z\in\supp(Z)$, this parameter can be interpreted as the average treatment effects for marginal defendants at the edge of being prosecuted. The equivalence between average MTE and MPRTE $\alpha_+(0)$ has been established by \cite{carneiro/heckman/vytlacil:2010}. 
\end{enumerate}

Next, I discuss the counterfactual propensity scores of interests and present confidence sets for these policy effects.

\subsection{Marginal policy relevant treatment effects}
\label{sec:PRTE}
First, suppose policymakers increase the ADAs' nonprosecution rates up to a positive constant $\epsilon$ or to a proportion $\epsilon$, where $\epsilon > 0$.  Consider the additive or proportional marginal PRTE $\alpha_+(0)$ under the new policy:
\begin{align}
	\label{eq:PRTE-additive}
	\text{Additive change:} \quad & p^{\epsilon}(z) = p(z) + \epsilon, \\
	\label{eq:PRTE-prop}
	\text{Proportional change:} \quad  & p^{\epsilon}(z) = (1 + \epsilon) p(z).
\end{align}

\begin{table}[h!]
  \centering
	\caption{95\% (top) and 90\% (bottom) Confidence Sets for Additive MPRTE ${\alpha}_+(0)$}
  \begin{tabular}{lccc}
  \toprule 
		      &	MTE polynomial & Wald & MLC \\
		      \midrule
    Uncond.     & Linear & [-0.19, -0.11] & [-0.17, -0.13] \\
    Uncond.     & Quadratic & [-0.19, -0.11] & [-0.18, -0.13] \\
    Uncond.     & Cubic & [-0.17, -0.06] & [-0.15, -0.08] \\
    Uncond.     & Quartic & [-0.18, -0.05] & [-0.17, -0.06] \\
                &       &       &  \\
    Average     & Linear & [-0.24, -0.03] & [-0.23, -0.03] \\
    Average     & Quadratic & [-0.29, -0.01] & [-0.33, 0.09] \\
    Average     & Cubic & [-0.32, 0.03] & [-0.79, 0.59] \\
    Average     & Quartic & [-0.39, 0.01] & [-0.76, 0.61] \\
		      \midrule
    Uncond.     & Linear & [-0.18, -0.12] & [-0.16, -0.14] \\
    Uncond.     & Quadratic & [-0.19, -0.12] & [-0.17, -0.14] \\
    Uncond.     & Cubic & [-0.17, -0.07] & [-0.14, -0.09] \\
    Uncond.     & Quartic & [-0.17, -0.06] & [-0.13, -0.08] \\
                &       &       &  \\
    Average     & Linear & [-0.22, -0.04] & [-0.21, -0.06] \\
    Average     & Quadratic & [-0.27, -0.04] & [-0.28, 0.01] \\
    Average     & Cubic & [-0.29, -0.00] & [-0.60, 0.34] \\
    Average     & Quartic & [-0.35, -0.02] & [-0.63, 0.38] \\
		\bottomrule
	\end{tabular}
    
    \begin{tablenotes}
    \footnotesize \raggedright
      Note: {Results for the classical Wald and MLC confidence sets under polynomial specifications up to the fourth order. The counterfactual policy considered is a uniformly additive increase in nonprosecution rates by a marginal amount. The top panel reports the 95\% confidence sets, and the bottom panel reports the 90\% confidence sets. The “Uncond.” confidence sets are obtained without controlling for court identities, whereas the “Average’’ confidence sets are weighted averages of the court-specific confidence sets, with weights proportional to court size.}
    \end{tablenotes}
    \label{tab:Avg_A-MPRTE}
\end{table}

Table \ref{tab:Avg_A-MPRTE} collects the 95\% and 90\% confidence sets on marginal policy effects $\alpha_+(0)$ for the additive leniency increase in equation \eqref{eq:PRTE-additive} by letting $\epsilon$ approach zero. The ``unconditional'' confidence sets are obtained without controlling for court identities, whereas the ``average'' confidence sets are weighted average of court-specific confidence sets weighted by the size of courts. The differences between these two sets of results highlight the importance of adjusting for court identity as a confounding variable. Incorporating court identity generally increases the uncertainty of causal effect predictions. 

Comparing average Wald confidence sets with average MLC confidence sets, the results show that weak identification is a potential concern for models with polynomial orders higher than cubic. In such cases, robust confidence sets may not be informative about the sign of causal effects due to the inability to estimate a highly flexible model by using limited variation of propensity scores. Although the ranges of both confidence sets are similar under the quadratic specification, Wald confidence sets remain negatively significant, whereas MLC confidence sets lose such significance due to weak identification. In Online Appendix \ref{appendix:iv_strength} (Table \ref{tab:IVstrength_AMPRTE}), I evaluate the strength of identification for the marginal PRTEs based on the additive change in \eqref{eq:PRTE-additive}. The findings further confirm the presence of weak identification for MTE models with cubic or quartic orders.

In addition to these aggregated results, Table \ref{tab:A-MPRTE} in Online Appendix \ref{appendix:add_empirical_results} provides confidence sets conditional on each court, revealing significant heterogeneity in the causal impacts of increasing ADAs' leniency across courts. It is worth noting that increasing ADA leniency might be harmful to court ROX, indicating that defendants may be more likely to be deterred from increasing punishment after prosecution, rather than committing to further crimes. Similar results for proportional marginal PRTE based on \eqref{eq:PRTE-prop} can be found in Table \ref{tab:P-MPRTE} in Online Appendix \ref{appendix:add_empirical_results}.

\subsection{Quota}
Instead of mandating all ADAs to increase their nonprosecution rates simultaneously, it might be more convenient for policymakers to set up lower and upper bounds for ADA nonprosecution rates. For example, let $\epsilon = (\underline{\epsilon}, \bar{\epsilon}) \in [0,1]^2$ with $\underline{\epsilon} \leq \bar{\epsilon}$ denoting the lower and upper bounds of ADAs' nonprosecution rate, respectively. Then the counterfactual propensity score becomes
\[
	p^\epsilon(z) = \min(\max\{\underline{\epsilon}, p(z)\},  \bar{\epsilon})
\]
Since this counterfactual propensity score is non-smooth in $p(z)$ which may invalidate the inferential results, I employ a smooth approximation as follows:
\[
	p^\epsilon(z; \phi) = -\frac{1}{\phi}\log\left(\frac{1}{e^{\phi p(z)} + e^{\phi \underline{\epsilon}}} + \frac{1}{e^{\phi \bar{\epsilon}}}\right).
\]
One can show that $p^\epsilon(z; \phi) \to p^\epsilon(z)$ as $\phi \to \infty$. In practice, I set $\phi = 30$ to approximate the counterfactual propensity score $p^\epsilon(z)$. I analyze the (non-)normalized policy effects when policymakers implement a lower bound $\underline{\epsilon} \in [0.05, 0.3]$ while maintaining $\bar{\epsilon} = 1$ for each court. For notational consistency, I use $\epsilon$ to denote the value of this lower bound instead of $\underline{\epsilon}$.

\begin{figure}[h!]
\centering

\caption{95\% (top) and 90\% (bottom) Average Confidence Sets for Non-normalized Policy Effects $\alpha(\epsilon)$ when Setting a Lower Bound $\underline{\epsilon}$ for Nonprosecution Rate}
	\centering
	\includegraphics[width = \textwidth]{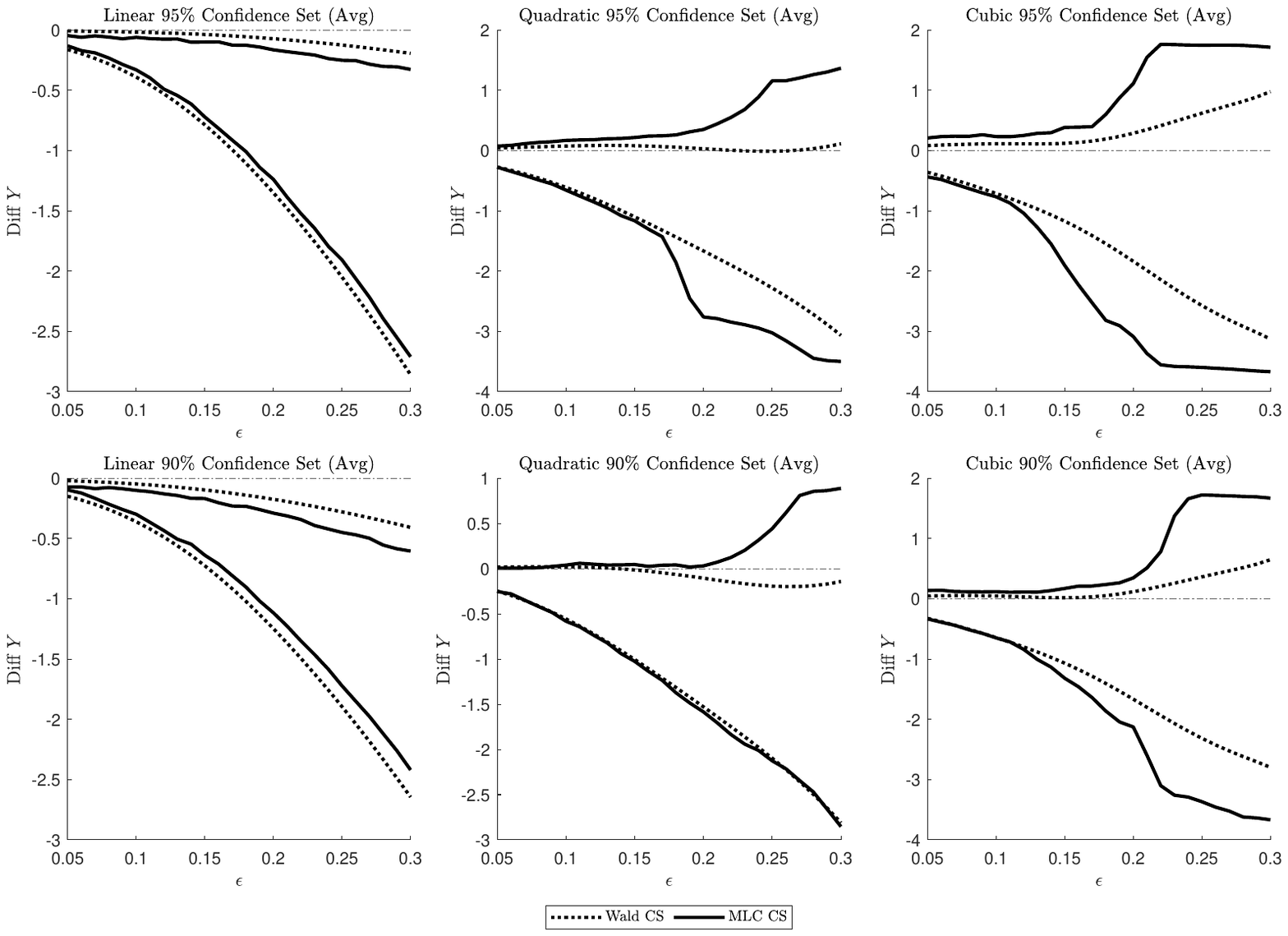}

    \begin{tablenotes}
    \footnotesize \raggedright
      Note: {Results for the classical Wald and MLC confidence sets under polynomial specifications up to the third order. The counterfactual policy under consideration is a universal lower bound on the nonprosecution rates shown on the $x$-axis. The $y$-axis depicts changes in recidivism rates (in percentage points). The top row presents the 95\% confidence sets, and the bottom row presents the 90\% confidence sets. All confidence sets are computed as weighted averages of the court-specific confidence sets, with weights proportional to court size.}
    \end{tablenotes}
    
\label{fig:quota_level_average}

\end{figure}

Figure \ref{fig:quota_level_average} displays the average confidence sets for the non-normalized policy effects $\alpha(\epsilon)$ of implementing a lower bound on nonprosecution rates. Under linear extrapolation, the results suggest that imposing a lower quota on nonprosecution is beneficial. However, this evidence becomes less conclusive with quadratic and cubic specifications as the lower bound ${\epsilon}$ increases, reflecting limitations of available exogenous variation at higher extrapolation levels. The Wald confidence sets appear ``spuriously precise'' by failing to account for this limited IV variation. Consequently, policy recommendations derived from the robust MLC approach may differ from those based on the classical Wald approach. For instance, based on the 90\% confidence sets from the quadratic MTE model, policymakers concerned with robustness to weak IVs may prefer setting the minimum nonprosecution rate below 20\%, since the reductions in recidivism become insignificant above this threshold under robust confidence sets. This sign reversal could potentially be driven by the fact that defendants with greater risk of recidivism are more likely to be released as the lower bound of nonprosecution rate increases.

\section{Conclusion}
\label{sec:conclusion}
In this paper, I propose two robust inference procedures for a class of causal effects identified by MTEs with discrete instruments. For the linear MTE model, I introduce a conditional Wald test, which is simple to implement and asymptotically similar regardless of identification strength. For a broader class of MTE models that nest polynomial specifications, I propose a modified linear combination test that achieves uniform validity against weak identification and has satisfactory power properties under strong identification. Finally, I use the proposed methods to investigate the counterfactual effects of manipulating ADAs' leniency in the study of misdemeanor prosecution by \cite{agan/doleac/harvey:2023}.

There are several avenues for future research. First, while this paper focuses on using discrete variation from instruments, extending the weak IV analysis to a semiparametric MTE model with continuous propensity scores would be empirically relevant. In addition, the MLC test requires the choice of tuning parameter on the weight of AR statistics, which trades off the power of test under different identification strengths (see more discussion in Online Appendix \ref{appendix:power_analysis_MLC}). It would be interesting to investigate the optimal choice of this tunning parameter following the regret analysis of \cite{andrews:2016}. Finally, extending robust inference techniques to accommodate many weak instruments would be valuable, especially considering the typically large number of judges involved in empirical studies \citep{jochmans:2023}.

\putbib
\end{bibunit}

\appendix

\begin{bibunit}
\section{Proofs for results in main text}
\label{sec:main_proofs}

\subsection{Notation}
Throughout the appendix, we employ the notation in Table \ref{tab:notation}, which was not necessarily introduced in the text.  In particular, the subscript $F$ emphasizes dependence on the distribution $F$, and the numerical subscripts (such as $d$, $\ell$, and $k$) following the subscript $F$ refer to the corresponding index of the parameter vector.

\begin{table}[htbp]
\centering
\caption{Important notation}
\label{tab:notation}
\begin{tabular}{p{3cm}p{0.7\textwidth}}
\toprule
$q_F(z_\ell)$ 						& $\Prob_F(Z=z_\ell)$ \\
$q_F(d, z_\ell)$ 					& $\Prob_F(D=d, Z=z_\ell)$ \\
$p_F(z_\ell)$ 						& $\Prob_F(D=1\mid Z=z_\ell)$ \\
$p_F$								& $(p_F(z_0), p_F(z_1), \ldots, p_F(z_K))'$ \\
$A_F$								& The matrix of propensity scores, defined in equations \eqref{eq:mateq} and \eqref{eq:linsystem} \\
$\beta_{F,d\ell}$					& $\Exp_F[Y\mid D=d, Z=z_\ell]$ \\
$\beta_{F,d}$						& $(\beta_{F,d0}, \ldots, \beta_{F,dK})'$ \\
$\beta_F$							& $(\beta_{F,1}', \beta_{F,0}')'$ \\
$\sigma_{F,d\ell}^2$				& $\var_F(Y\mid D=d, Z=z_\ell)$ \\
$\{\tau_{j,F}\}_{j=1}^{2(M+1)}$		& The singular values of $A_F$ (in a descending order) \\
$\hat{Z}_p, \hat{Z}_{\beta_d}, \hat{Z}_{q}$ & The normalized estimators defined in \eqref{eq:clt_p_beta} \\
$\mathcal{Z}_p, \mathcal{Z}_{\beta_d}, \mathcal{Z}_q$ & The asymptotic distribution of the normalized statistics \\
\bottomrule
\end{tabular}
\end{table}

In addition, let $\mathbb{S}_{++}^{k}$ denote the space of positive definite square matrices with $k$ rows (columns). Throughout the proof, we abbreviate ``weak law of large numbers'' as ``WLLN'' and ``central limit theorem'' as ``CLT'' for triangular arrays. Both theorems are stated in Lemmas \ref{lem:wlln_array} and \ref{lem:clt_array}.

\subsection{Proof of Proposition \ref{prop:AR_validity} and Theorem \ref{thm:validity_cond_wald}}
The proof of asymptotic similarity of the AR and conditional tests in Proposition \ref{prop:AR_validity} and Theorem \ref{thm:validity_cond_wald} uses the sub-sequencing techniques from \cite{andrews/cheng/guggenberger:2020}. Specifically, I verify their Assumption B\textsuperscript{*} in the following Proposition \ref{prop:subsequence_AR_validity}. 

\begin{proposition}
\label{prop:subsequence_AR_validity}
For any subsequence $\{p_n\}$ of $\{n\}$ and any sequence $\{(\lambda_{p_n}, F_{p_n}) \in \mathcal{P}_0\}$ for which 
\begin{enumerate}
	\item $\theta_{F_{p_n}} \to \theta_\infty \in \Theta$, where $\theta_{F_{p_n}}$ is a sequence of parameters such that $(\theta_{F_{p_n}}, F_{p_n}) \in \mathcal{P}$ and $\lambda_{{p_n}} = c'\theta_{F_{p_n}}$;
	\item $\beta_{F_{p_n}, d\ell} \to \beta_{\infty,d\ell}$ for $d = 0,1$ and $\ell = 0,1,\ldots,K$;
	\item $\sigma^2_{F_{p_n}, d\ell} \to \sigma^2_{\infty, d\ell}$ for $d = 0,1$ and $\ell = 0,1,\ldots,K$;
	\item $p_{F_{p_n}}(z_\ell) \to p_\infty(z_\ell)$ for all $\ell = 0,1,\ldots,K$;
	\item $q_{F_{p_n}}(z_\ell) \to q_{\infty}(z_\ell)$ for all $\ell = 0,1,\ldots,K$;
\end{enumerate}
The following result holds:
\begin{enumerate}
\item[(a)] the convergence of the type-I error rate of AR test:
\begin{equation*}
    \lim_{n\to\infty} \Prob_{F_{p_n}}\left(\text{AR}_{p_n,k}(\lambda_{p_n}) > q_{\chi_1^2}(1-\alpha)\right) = \alpha, \text{ for } k = 1,\ldots,K,
\end{equation*}
\item[(b)] 
Suppose further that 
\[
    \sqrt{p_n}\|\pi_{F_{p_n}}\| \to \psi_\infty\in [0,\infty], \text{~and~} \pi_{F_{p_n}}/\|\pi_{F_{p_n}}\| \to \iota_\infty,
\] %
where $\iota_\infty$ is a unit vector, then we have the convergence of the type-I error rate of the conditional test
\begin{equation*}
    \lim_{n\to\infty} \Prob_{F_{p_n}}\left(W_{p_n}(\lambda_{p_n}) > q_{W^*}(1-\alpha)\right) = \alpha.
\end{equation*}
\end{enumerate}
\end{proposition}

\bigskip

\begin{proof}[Proof of Proposition \ref{prop:AR_validity}]
This proposition is a consequence from Corollary 2.1(c) of \citet{andrews/cheng/guggenberger:2020} after replacing “CP (Coverage Probability)” with “RP (Null Rejection Probability)” as noted in Comment 4 below their Theorem 2.1.

Assumption B\textsuperscript{*} in their paper is verified by Proposition~\ref{prop:subsequence_AR_validity}(a) by taking their parameter $\lambda$ to be $(\lambda, F)$ in our notation, and by letting their function $h(\lambda)$ correspond to $(\theta, \beta, {\sigma^2_{d\ell}}_{d,\ell}, p, q)$, where $\theta$ is an element of $\Theta$ satisfying $c'\theta = \lambda$ and $(\theta, F) \in \mathcal{P}$. (By the definition of $\mathcal{P}_0$, such a $\theta$ exists.)
\end{proof}

\bigskip

\begin{proof}[Proof of Theorem \ref{thm:validity_cond_wald}]
This proposition is a consequence from Corollary 2.1(c) of \citet{andrews/cheng/guggenberger:2020} after replacing “CP (Coverage Probability)” with “RP (Null Rejection Probability)” as noted in Comment 4 below their Theorem 2.1.

Assumption B\textsuperscript{*} in their paper is verified by Proposition~\ref{prop:subsequence_AR_validity}(b) 
by taking their parameter $\lambda$ as $(\lambda, F)$ in our notation, 
and their function $h_n(\lambda)$ to correspond to $(\theta, \beta, \{\sigma^2_{d\ell}\}_{d,\ell}, p, q, \sqrt{n}\|\pi\|, \pi/\|\pi\|)$, 
where $\theta$ is an element of $\Theta$ satisfying $c'\theta = \lambda$ and $(\theta, F) \in \mathcal{P}$.
\end{proof}

\bigskip

\begin{proof}[Proof of Proposition \ref{prop:subsequence_AR_validity}]
For simplicity of notations, the proof is shown for the full sequence $\{n\}$. Then, we note that the same proof goes through with $p_n$ in place of $n$.

\medskip
\underline{Part (a)}: 
Note that $g_k(\lambda)$ is a continuously differentiable function of $\lambda$, $p$, $\beta_1$, and $\beta_0$. Applying Delta method and the convergence in Lemma \ref{lem:convergence_p_beta}(a) to the sample moment $\hat{g}_k(\lambda)$ yields
\begin{align*}
    \sqrt{n}\hat{g}_k(\lambda_{n})
    & = [\partial_{p'}{g}_k(\lambda_n)] \hat{Z}_p + [\partial_{\beta_1'}{g}_k(\lambda_n)] \hat{Z}_{\beta_1} + [\partial_{\beta_0'}{g}_k(\lambda_n)] \hat{Z}_{\beta_0} + o_p(1) \\
    & \xrightarrow{d} \normal(0, s_k^2)
\end{align*}
where 
\[
    s_k^2 = [\partial_{p'}{g}_k(\lambda_\infty)]\Sigma_{p,\infty}[\partial_{p}{g}_k(\lambda_\infty)] + [\partial_{\beta_1'}{g}_k(\lambda_\infty)]\Sigma_{\beta_1,\infty}[\partial_{p}{g}_k(\lambda_\infty)] + [\partial_{\beta_0'}{g}_k(\lambda)]\Sigma_{\beta_0,\infty}[\partial_{\beta_0}{g}_k(\lambda_\infty)],
\]
$\lambda_\infty$ is the limit of $\lambda_n = c'\theta_n$ as $n$ goes to infinity, and $g_k(\lambda_\infty)$ equals $g_k(\lambda)$ by sending $\lambda$, $p$, $\beta_1$, and $\beta_0$ to their limiting values.

Again by Lemma \ref{lem:convergence_p_beta}(b), $\Sigma_{p,\infty}$, $\Sigma_{\beta_1,\infty}$, and $\Sigma_{\beta_0,\infty}$ can be consistenly estimated by their sample analogs $\hat{\Sigma}_p$, $\hat{\Sigma}_{\beta_1}$, and $\hat{\Sigma}_{\beta_0}$ along the drifting sequence, respectively. The continuous mapping theorem then implies that 
\[
    \hat{s}_{k}^2(\lambda_{n}) 
    = \partial_{p'}\hat{g}_k(\lambda_{n}) \hat{\Sigma}_p \partial_{p}\hat{g}_k(\lambda_{n}) +  \partial_{\beta_1'}\hat{g}_k(\lambda_{n}) \hat{\Sigma}_{\beta_1} \partial_{\beta_1}\hat{g}_k(\lambda_{n}) +  \partial_{\beta_0'}\hat{g}_k(\lambda_{n}) \hat{\Sigma}_{\beta_0} \partial_{\beta_0}\hat{g}_k(\lambda_{n}) ~~\xrightarrow{p}~~ s_k^2.
\]
Note that Lemma \ref{lem:pd_variance} implies $s_k^2 \in (0,\infty)$. By Slutsky's theorem, we have
\[
    \frac{\sqrt{n}\hat{g}_k(\lambda_{n})}{\hat{s}_k(\lambda_{n})} ~~\xrightarrow{d}~~ \normal(0,1),
\]
which implies the desired result.

\medskip
\underline{Part (b)}: Following similar arguments in part (a), Lemma \ref{lem:convergence_p_beta} implies that
\begin{equation}
\label{eq:limit_baseline}
\begin{aligned}
     \sqrt{n}\hat{g}(\lambda_{n}) ~~ &\xrightarrow{d} ~~ \mathcal{Z}_{\mathfrak{g}} \equiv \partial_p g(\lambda_\infty) \mathcal{Z}_p  + \partial_{\beta_1} g(\lambda_\infty) \mathcal{Z}_{\beta_1} + \partial_{\beta_0} g(\lambda_\infty) \mathcal{Z}_{\beta_0} \sim \normal(0_{K\times 1}, S_\infty) \\
     \partial_{x} \hat{g}(\lambda_{n}) ~~ &\xrightarrow{p}~~ \partial_x g(\lambda_{\infty}) \quad \text{ for } x \in \{p, \beta_1, \beta_0\} \\
     \hat{\Sigma}_x ~~ &\xrightarrow{p}~~ \Sigma_{x,\infty} \quad \text{ for } x \in \{p, \beta_1, \beta_0\} \\
     \hat{S}(\lambda_n) ~~&\xrightarrow{p}~~ S_\infty
\end{aligned}
\end{equation}
where
\[
    S_\infty  \equiv \partial_p g(\lambda_\infty) \Sigma_{p,\infty} \partial_{p'} g(\lambda_\infty) + \partial_{\beta_1} g(\lambda_\infty) \Sigma_{\beta_1,\infty} \partial_{\beta_1'} g(\lambda_\infty) + \partial_{\beta_0} g(\lambda_\infty) \Sigma_{\beta_0,\infty} \partial_{\beta_0'} g(\lambda_\infty).
\]
and $S_\infty$ is bounded and positive definite by Lemma \ref{lem:pd_variance}. 

\medskip
\textbf{Case 1:} $\psi_\infty < \infty$.
\medskip

In this case, we have
\begin{equation}
\label{eq:limit_2}
\begin{aligned}
\sqrt{n} \hat{\pi} 
    &= \sqrt{n}\pi_{F_n} + [\partial_p \pi] \hat{Z}_p \\
    &\xrightarrow{d} \psi_\infty \iota_\infty + [\partial_p \pi] \mathcal{Z}_p\\
    &\equiv \mathcal{Z}_\pi
\end{aligned}
\end{equation}
and 
\begin{align*}
    \hat{h}(\lambda_n) 
    &= \sqrt{n}\hat{\pi} - [\partial_p \pi] \hat{\Sigma}_p [\partial_{p'}\hat{g}(\lambda_n)] \hat{S}(\lambda_n)^{-1} \sqrt{n}\hat{g}(\lambda_n) \\
    &\xrightarrow{d} \psi_\infty\iota_\infty + [\partial_p \pi] \mathcal{Z}_p - [\partial_p \pi]\Sigma_{p,\infty} [\partial_{p'}g(\lambda_\infty)]S_\infty^{-1} \mathcal{Z}_{\mathfrak{g}}\\
    &\equiv \mathcal{Z}_h
\end{align*}
From the independence $(\mathcal{Z}_{\beta_1}, \mathcal{Z}_{\beta_0}) \indep \mathcal{Z}_p$, it follows that $\mathcal{Z}_h \indep \mathcal{Z}_{\mathfrak{g}}$ since they are jointly normal and uncorrelated.

By continuous mapping theorem, we obtain the asymptotic distribution of the Wald statistic under the null:
\begin{align*}
    W_n(\lambda_n) 
    &= \frac{n\hat{g}(\lambda_n)'\hat{S}(\lambda_n)^{-1} \hat{\pi}\hat{\pi}'\hat{S}(\lambda_n)^{-1}\hat{g}(\lambda_n)}{\hat{\pi}'\hat{S}(\lambda_n)^{-1}\hat{\pi}} \\
    &\xrightarrow{d} \frac{\mathcal{Z}_{\mathfrak{g}}' S_\infty^{-1} \mathcal{Z}_\pi \mathcal{Z}_\pi' S_\infty^{-1} \mathcal{Z}_{\mathfrak{g}}}{\mathcal{Z}_\pi S_\infty^{-1} \mathcal{Z}_\pi}.
\end{align*}
On the other hand, we note that the simulated statistic $W_n^*(\lambda_n)$ can be written as a function of standard normal draw $\eta^*$ and a conditioning statistic 
\begin{equation}
\label{eq:Upsilon_hat}
    \hat{\Upsilon} \equiv (\hat{h}(\lambda_n), ~\hat{S}(\lambda_n), ~\hat{\Sigma}_p, ~\partial_p\hat{g}(\lambda_n)).
\end{equation}
Specifically, we write
\begin{equation}
\label{eq:defn_W_eta_C}
    W_n^*(\lambda_n) = W(\eta^*, \hat{\Upsilon}) \equiv \frac{(\eta^*)'\hat{S}(\lambda_n)^{-1/2} \pi_s(\eta^*,\hat{\Upsilon})\pi_s(\eta^*,\hat{\Upsilon})'\hat{S}(\lambda_n)^{-1/2}\eta^*}{\pi_s(\eta^*,\hat{\Upsilon})'\hat{S}(\lambda_n)^{-1}\pi_s(\eta^*,\hat{\Upsilon})} 
\end{equation}
where 
\[
    \pi_s(\eta^*, \hat{\Upsilon}) \equiv \hat{h}(\lambda_n) + [\partial_p \pi] \hat{\Sigma}_p [\partial_{p'}\hat{g}(\lambda_n)] \hat{S}(\lambda_n)^{-1/2} \eta^*.
\]
Therefore, $\hat{\Upsilon}$ is the source of sampling uncertainty on the simulated statistic $W_n^*(\lambda_n)$.

Let $G(\cdot\mid \hat{\Upsilon})$ denote the CDF of the simulated Wald statistic $W_n^*(\lambda_n)$ conditional on data, i.e.,
\begin{align*}
    G(x\mid \hat{\Upsilon}) 
    &\equiv \Prob(W_n^*(\lambda_n) \leq x \mid \{Y_i, D_i, Z_i\}_{i=1}^n) \\
    &= \Prob(W(\eta^*, \hat{\Upsilon}) \leq x\mid \hat{\Upsilon}).
\end{align*}
Define $q(1-\alpha,\Upsilon)$ as the $(1-\alpha)$-quantile of $G(\cdot \mid \Upsilon)$ for a quartet of elements $\Upsilon$ such as $\hat{\Upsilon}$ in \eqref{eq:Upsilon_hat}. By this definition, the critical value can be expressed as 
\[
    \hat{q}_{W^*}(1-\alpha) = q(1-\alpha,\hat{\Upsilon}).
\]
Next we show that $q(1-\alpha,\Upsilon)$
is a continuous function of $\Upsilon$ on the set $\mathcal{U}$, where
\begin{equation}
    \label{eq:defn_set_C}
    \begin{aligned}
    \mathcal{U} = 
    \{(u_1, u_2, u_3, u_4):~ 
    & u_1 + [\partial_p \pi]u_3u_4'u_2^{-1/2}\eta^* \neq 0_{K\times 1} \text{ a.s.},  \\
    & u_1\in\R^K, u_2 \in \mathbb{S}^{K}_{++}, u_3 \in \mathbb{S}^{K+1}_{++}, u_4 \in \R^{K\times (K+1)}\}.
    \end{aligned}
\end{equation}
Let $\{\Upsilon_n\}$ be a sequence in $\mathcal{U}$ such that $\Upsilon_n \to \Upsilon \in \mathcal{U}$ as $n\to \infty$. Given zero probability of discontinuity in the limit by the definition of $\mathcal{U}$, $\{W(\eta^*,\Upsilon_n)\}_{n \geq 1}$ converges almost surely to $W(\eta^*,\Upsilon)$. This implies the convergence in distribution for any continuity point $x$ of $G(\cdot\mid \Upsilon)$:
\[
    G(x\mid \Upsilon_n) = \Prob_{\eta^*}(W(\eta^*,\Upsilon_n) \leq x) ~~\to~~ \Prob_{\eta^*}(W(\eta^*,\Upsilon) \leq x) = G(x \mid \Upsilon).
\]
The distribution function $G(\cdot\mid \Upsilon)$ is increasing at its $(1-\alpha)$-quantile $q(1-\alpha,\Upsilon)$ because the random variable $W(\eta^*, \Upsilon)$ is continuously distributed. By \citet[Lemma 5]{andrews/guggenberger:2010}, it follows that $q(1-\alpha,\Upsilon_n) \to q(1-\alpha,\Upsilon)$. This establishes continuity of quantile function on the set $\mathcal{U}$.

The convergence results in \eqref{eq:limit_baseline} and \eqref{eq:limit_2} give  
\[
    \hat{\Upsilon} ~~\xrightarrow{d}~~ \Upsilon_\infty \equiv (\mathcal{Z}_h, ~S_\infty, ~\Sigma_{p,\infty}, ~\partial_p g(\lambda_\infty)),
\]
and by Lemma \ref{lem:NS_cond}, we have $\Prob(\Upsilon_\infty \in \mathcal{U}) = 1$. From the continuity of $q(1-\alpha,\Upsilon)$, it follows by continuous mapping theorem that
\[
    W_n(\lambda_n) - q(1-\alpha, \hat{\Upsilon}) ~~\xrightarrow{d}~~ \frac{\mathcal{Z}_{\mathfrak{g}}' S_\infty^{-1} \mathcal{Z}_\pi \mathcal{Z}_\pi' S_\infty^{-1} \mathcal{Z}_{\mathfrak{g}}}{\mathcal{Z}_\pi S_\infty^{-1} \mathcal{Z}_\pi} - q(1-\alpha, \Upsilon_\infty),
\]
which implies
\[
    \Prob_{F_n}\left(W_n(\lambda_n) > q(1-\alpha, \hat{\Upsilon})\right) ~~\to~~ \Prob\left(\frac{\mathcal{Z}_{\mathfrak{g}}' S_\infty^{-1} \mathcal{Z}_\pi \mathcal{Z}_\pi' S_\infty^{-1} \mathcal{Z}_{\mathfrak{g}}}{\mathcal{Z}_\pi S_\infty^{-1} \mathcal{Z}_\pi} > q(1-\alpha, \Upsilon_\infty)\right).
\]
Next, we show that the limit probability on the right-hand side equals $\alpha$. We first examine the conditional probability as follows.
\begin{align*}
    &\Prob\left(\frac{\mathcal{Z}_{\mathfrak{g}}' S_\infty^{-1} \mathcal{Z}_\pi \mathcal{Z}_\pi' S_\infty^{-1} \mathcal{Z}_{\mathfrak{g}}}{\mathcal{Z}_\pi S_\infty^{-1} \mathcal{Z}_\pi} > q(1-\alpha, \Upsilon_\infty) \mid \Upsilon_\infty\right) \\
    &= \Prob\left(\frac{(\eta^*)' S_\infty^{-1/2} \pi_s(\eta^*,\Upsilon_\infty) \pi_s(\eta^*,\Upsilon_\infty)' S_\infty^{-1/2} \eta^*}{\pi_s(\eta^*,\Upsilon_\infty) S_\infty^{-1} \pi_s(\eta^*,\Upsilon_\infty)} > q(1-\alpha, \Upsilon_\infty) \mid \Upsilon_\infty\right) \\
    &= \Prob\left(W(\eta^*,\Upsilon_\infty) > q(1-\alpha, \Upsilon_\infty) \mid \Upsilon_\infty\right) \\
    &= \alpha\quad \text{a.s.}
\end{align*}
The second line holds by the observation that $\eta^*$ and $S_\infty^{-1/2}\mathcal{Z}_{\mathfrak{g}}$ are both independent of $\Upsilon_\infty$, and that $\mathcal{Z}_\pi = \pi_s(S_\infty^{-1/2}\mathcal{Z}_{\mathfrak{g}},~\Upsilon_\infty)$, which has the same distribution as $\pi_s(\eta^*, \Upsilon_\infty)$. The third line holds by the definition of $W(\eta^*,\Upsilon_\infty)$. The last line holds by the definition of $q(1-\alpha,\Upsilon)$, and the fact that $W(\eta^*,\Upsilon)$ is continuously distributed for any $\Upsilon \in \mathcal{U}$. By taking the law of total probability, it follows that the unconditional rejection probability equals $\alpha$ as well. This completes the proof for the case $\psi_\infty < \infty$.

\medskip
\textbf{Case 2:} $\psi_\infty = \infty$.
\medskip

In this case, by the assumption $\sqrt{n}\|\pi_{F_n}\| \to \infty$, we have
\begin{equation}
\label{eq:limit_3}
\begin{aligned}
    \frac{\hat{\pi}}{\|\pi_{F_n}\|} 
    &= \frac{\pi_{F_n}}{\|\pi_{F_n}\|} + O_p(n^{-1/2}\|\pi_{F_n}\|^{-1}) \\
    &= \iota_\infty + o_p(1).
\end{aligned}
\end{equation}
and 
\begin{equation}
\label{eq:h_sqrt_limit}
\begin{aligned}
    \frac{\hat{h}(\lambda_n)}{\sqrt{n}\|\pi_{F_n}\|} 
    &= \frac{\hat{\pi}}{\|\pi_{F_n}\|} - [\partial_p \pi]\hat{\Sigma}_{p}[\partial_{p'} \hat{g}(\lambda_n)]\hat{S}^{-1}\sqrt{n}\hat{g}(\lambda_n) \|\sqrt{n}\pi_{F_n}\|^{-1} \\
    &= \iota_\infty + o_p(1).
\end{aligned}
\end{equation}

From the condition $\iota_\infty \neq 0_{K\times 1}$, \eqref{eq:limit_baseline}, and \eqref{eq:limit_3}, it follows that 
\begin{equation}
\label{eq:test_stat_converge}
    W_n(\lambda_n) = \frac{n\hat{g}(\lambda_n)'\hat{S}(\lambda_n)^{-1} (\hat{\pi}/\|\pi_{F_n}\|)(\hat{\pi}'/\|\pi_{F_n}\|)\hat{S}(\lambda_n)^{-1}\hat{g}(\lambda_n)}{(\hat{\pi}'/\|\pi_{F_n}\|)\hat{S}(\lambda_n)^{-1}(\hat{\pi}/\|\pi_{F_n}\|)} ~~\xrightarrow{d}~~ \chi_1^2.    
\end{equation}
Now we examine the stochastic behavior of critical value $\hat{q}_{W^*}(1-\alpha)$, defined as $(1-\alpha)$-quantile of $W_n^*(\lambda_n)$ conditional on data. First, define a normalized conditioning statistic in the construction of $W_n^*(\lambda_n)$:
\[
    \bar{\Upsilon} \equiv \left(\frac{\hat{h}(\lambda_n)}{\sqrt{n}\|\pi_{F_n}\|}, ~\hat{S}(\lambda_n), ~\hat{\Sigma}_p,  ~\frac{\partial_p\hat{g}(\lambda_n)}{\sqrt{n}\|\pi_{F_n}\|}\right).
\]
By the definition of $q(1-\alpha,\Upsilon)$ introduced in Case 1 and the fact that
\[
    W_n^*(\lambda_n) \equiv W_n(\eta^*, \hat{\Upsilon}) = W_n(\eta^*, \bar{\Upsilon}),
\]
we have 
\[
    \hat{q}_{W^*}(1-\alpha) = q(1-\alpha, \bar{\Upsilon}).
\]
It therefore suffices to obtain the limit of $\bar{\Upsilon}$ as the sample size diverges. By \eqref{eq:limit_baseline}, \eqref{eq:h_sqrt_limit}, and the condition that $\sqrt{n}\|\pi_{F_n}\| \to \infty$, we obtain
\[
    \bar{\Upsilon}~~\xrightarrow{p}~~\bar{\Upsilon}_{\infty} \equiv (\iota_\infty, S_\infty, \Sigma_{p,\infty}, 0_{K \times (K+1)}) ~\in~ \mathcal{U},
\]
where the inclusion in $\mathcal{U}$ follows from $\iota_\infty \neq 0_{K\times 1}$. 
This and the continuity of the quantile function $q(1-\alpha,\cdot)$ give $q(1-\alpha, \bar{\Upsilon}) \xrightarrow{p} q(1-\alpha, \bar{\Upsilon}_{\infty})$. This shows that the $(1-\alpha)$-conditional quantile of $W_n^*(\lambda_n)$ converges in probability to $q(1-\alpha, \bar{\Upsilon}_{\infty})$, which equals the $(1-\alpha)$-quantile of $\chi_1^2$ distribution by replacing $\hat{\Upsilon}$ with $\bar{\Upsilon}_\infty$ in \eqref{eq:defn_W_eta_C}. So we conclude that 
\begin{equation}
\label{eq:critical_value_converge}
        \hat{q}_{W^*}(1-\alpha) ~~\xrightarrow{p}~~ q_{\chi_1^2}(1-\alpha). 
\end{equation}
By \eqref{eq:test_stat_converge} and \eqref{eq:critical_value_converge}, we have $W_n(\lambda_n) - \hat{q}_{W^*}(1-\alpha) \xrightarrow{d} \chi_1^2 - q_{\chi_1^2}(1-\alpha)$. Then the desired conclusion follows by the definition of convergence in distribution.
\end{proof}

\subsection{Proof of Theorem \ref{thm:uniform_validity}}
The proof of uniform size control in Theorem \ref{thm:uniform_validity} uses the sub-sequencing techniques from \cite{andrews/cheng/guggenberger:2020}. Specifically, we verify part of their Assumption B in the following Proposition \ref{prop:subsequence_MLC_validity}.

\begin{proposition}
\label{prop:subsequence_MLC_validity}
For any subsequence $\{p_n\}$ of $\{n\}$ and any sequence $\{(\lambda_{p_n}, F_{p_n}) \in \mathcal{P}_0\}$ for which 
\begin{enumerate}
	\item $\theta_{F_{p_n}} \to \theta_\infty \in \Theta$, where $\theta_{F_{p_n}}$ is a sequence of parameters such that $(\theta_{F_{p_n}}, F_{p_n}) \in \mathcal{P}$ and $\lambda_{{p_n}} = c'\theta_{F_{p_n}}$;
	\item $\beta_{F_{p_n}, d\ell} \to \beta_{\infty,d\ell}$ for $d = 0,1$ and $\ell = 0,1,\ldots,K$;
	\item $\sigma^2_{F_{p_n}, d\ell} \to \sigma^2_{\infty, d\ell}$ for $d = 0,1$ and $\ell = 0,1,\ldots,K$;
	\item $p_{F_{p_n}}(z_\ell) \to p_\infty(z_\ell)$ for all $\ell = 0,1,\ldots,K$;
	\item $q_{F_{p_n}}(z_\ell) \to q_{\infty}(z_\ell)$ for all $\ell = 0,1,\ldots,K$;
	\item $B_{F_{p_n}} \to B_{\infty}$ and $C_{F_{p_n}} \to C_{\infty}$;
	\item $\sqrt{p_n}\tau_{j,F_{p_n}} \to t_j \in [0,\infty]$ for all $j = 1, \ldots, 2(M+1)$;
	\item $\iota_{n} \equiv \frac{\Psi_{F_{p_n}} B_{F_{p_n}}'c}{\|\Psi_{F_{p_n}} B_{F_{p_n}}'c\|}\to \iota_\infty \in \R^{2(M+1)}$, where $\Psi_{F_{p_n}} \in \R^{2(M+1)\times 2(M+1)}$ is a normalizing matrix defined in equation \eqref{eq:defn_S},\footnote{Note that the denominator $\|\Psi_{F_n}B_{F_n}'c\| \neq 0$ since $\Psi_{F_n}$ and $B_{F_n}$ are both full-rank matrices and the weight $c$ is nonzero. Hence $\iota_n$ is properly defined for each $n \geq 1$ and satisfies $\|\iota_n\| = 1$.}
\end{enumerate}
we have 
\[
	\limsup_{n\to\infty} \Prob_{F_{p_n}}\left(\inf_{c'\theta = \lambda_{p_n}} \text{MLC}_{p_n}(\theta) > q_{(1+a)\chi_1^2 + a\chi_{2K+1}^2}(1-\alpha)\right) \leq \alpha.
\]
\end{proposition}

\bigskip

\begin{proof}[Proof of Theorem \ref{thm:uniform_validity}]
This proposition is a consequence from Corollary 2.1(a) of \citet{andrews/cheng/guggenberger:2020} after replacing “CP (Coverage Probability)” with “RP (Null Rejection Probability)” as noted in Comment 4 below their Theorem 2.1.

Assumption B in their paper is verified by Proposition~\ref{prop:subsequence_MLC_validity}
by taking their parameter $\lambda$ as $(\lambda, F)$ in our notation, 
and their function $h_n(\lambda)$ to correspond to $(\theta, \beta, \{\sigma^2_{d\ell}\}_{d,\ell}, p, q,B,C,\{\sqrt{n}\tau_j\}_{j=1}^{2(M+1)}, \iota_n)$, 
where $\theta$ is an element of $\Theta$ satisfying $c'\theta = \lambda$ and $(\theta, F) \in \mathcal{P}$.
\end{proof}

\bigskip

\begin{proof}[Proof of Proposition \ref{prop:subsequence_MLC_validity}]
For simplicity of notations, the proof is shown for the full sequence $\{n\}$. Then, we note that the same proof goes through with $p_n$ in place of $n$.

By Lemma \ref{lem:convergence_p_beta}(a), we have 
\[
	\sqrt{n}
	\begin{pmatrix}
		\hat{p} - p_{F_n} \\
		\hat{\beta} - \beta_{F_n}
	\end{pmatrix}
	~~=~~
	\begin{bmatrix}
		\hat{Z}_p \\
		\hat{Z}_\beta
	\end{bmatrix}
	~~\xrightarrow{d}~~
	\normal\left(0_{3(K+1)\times 1},  \quad 
	\begin{bmatrix}
		\Sigma_{p,\infty} & 0_{(K+1)\times 2(K+1)} \\
		0_{2(K+1)\times (K+1)} & \Sigma_{\beta,\infty}
	\end{bmatrix}\right)
\]
This implies
	\begin{equation}
	\label{eq:A_mat_asymptotics}
		\sqrt{n}(\hat{A} - A_{F_n}) 
		= \begin{bmatrix}
			\begin{pmatrix}
				L_0(p_{F_n})\hat{Z}_p, & \ldots & L_M(p_{F_n})\hat{Z}_p  
			\end{pmatrix}
			& 0_{(K+1)\times(M+1)} \\
			0_{(K+1)\times (M+1)} & 
			\begin{pmatrix}
				R_0(p_{F_n})\hat{Z}_p, & \ldots & R_M(p_{F_n})\hat{Z}_p  
			\end{pmatrix}
		\end{bmatrix}
		+o_p(1)
	\end{equation}
	where
	\begin{align*}
		& L_m(p) = \diag\{\lambda_{1m}'(p(z_0)), \ldots,\lambda_{1m}'(p(z_K))\} \\
		& R_m(p) = \diag\{\lambda_{0m}'(p(z_0)), \ldots, \lambda_{0m}'(p(z_K))\}.
	\end{align*}

By continuous mapping theorem, we have 
\begin{align*}
	\sqrt{n}(\hat{A}\theta_{F_n} - \hat{\beta})
	&= \sqrt{n}(\hat{A} - A_{F_n})\theta_{F_n} - \sqrt{n}(\hat{\beta} - \beta_{F_n}) \\
	&= \begin{bmatrix}
		\begin{pmatrix}
			L_0({p}_{F_n})\hat{Z}_p, & \ldots & L_M({p}_{F_n})\hat{Z}_p  
		\end{pmatrix}
		& 0_{(K+1)\times(M+1)} \\
		0_{(K+1)\times (M+1)} & 
		\begin{pmatrix}
			R_0({p}_{F_n})\hat{Z}_p, & \ldots & R_M({p}_{F_n})\hat{Z}_p  
		\end{pmatrix}
	\end{bmatrix}
	\begin{bmatrix}
	\theta_{1,F_n} \\
	\theta_{0,F_n}
	\end{bmatrix} \\
	  & \quad - \hat{Z}_\beta + o_p(1) \\
	&\xrightarrow{d}
	\begin{bmatrix}
		\diag\left\{\sum_{m=0}^M \theta_{1m,\infty} \lambda_{1m}'({p}_{\infty}(z_\ell)): \ell = 0,1,\ldots, K\right\} \\
		\diag\left\{\sum_{m=0}^M \theta_{0m,\infty} \lambda_{0m}'({p}_{\infty}(z_\ell)): \ell = 0,1,\ldots, K\right\}
	\end{bmatrix} 
	\mathcal{Z}_p - \mathcal{Z}_\beta\\
	&= H(p_\infty,\theta_\infty)\mathcal{Z}_p - \mathcal{Z}_\beta.
\end{align*}
where we define
	\[
		H(p,\theta) 
		\equiv
		\begin{bmatrix}
			\diag\left\{\sum_{m=0}^M \theta_{1m} \lambda_{1m}'({p}(z_{\ell})): \ell = 0,1,\ldots, K\right\} \\
			\diag\left\{\sum_{m=0}^M \theta_{0m} \lambda_{0m}'({p}(z_{\ell})): \ell = 0,1,\ldots, K\right\}
		\end{bmatrix}.
	\]
Let 
\[
	\mathcal{Z}_{\mathfrak{m}} \equiv H(p_\infty,\theta_\infty)\mathcal{Z}_p - \mathcal{Z}_\beta.
\] 
Following Lemma \ref{lem:convergence_p_beta}, the condition $\theta_{F_n} \to \theta_{\infty}$, and note that $\{\lambda_{dm}'(\cdot)\}$ are continuouly differentiable by the continuity of $h_m(\cdot)$, continuous mapping theorem implies
\begin{align}
	\hat{\Omega}(\theta_{F_n}) 
	&~~=~~ H(\hat{p}, \theta_{F_n}) \hat{\Sigma}_p H(\hat{p}, \theta_{F_n})' + \hat{\Sigma}_\beta \notag \\
	&~~\xrightarrow{p}~~
	H(p_{\infty}, \theta_{\infty}) {\Sigma}_{p,\infty} H(p_{\infty}, \theta_{\infty})' + {\Sigma}_{\beta,\infty} \notag \\
	&~~=~~ \var(\mathcal{Z}_{\mathfrak{m}}) \label{eq:convergence_Omega} 
	\end{align}
and
\begin{align}
	\hat{\Gamma}_j(\theta_{F_n})
	&~~=~~ M_j(\hat{p})\hat{\Sigma}_p H(\hat{p}, \theta_{F_n})' \notag \\
	&~~\xrightarrow{p}~~
	M_j(p_{\infty}){\Sigma}_{p,\infty} H(p_{\infty}, \theta_{\infty})' \notag \\
	&~~=~~ \cov(M_j(p_\infty)\mathcal{Z}_p, \mathcal{Z}_{\mathfrak{m}}). \label{eq:convergence_Gamma}
\end{align}
where $M_j(p)$ is defined in \eqref{eq:defn_M_mat}. 
Note that $\var(\mathcal{Z}_{\mathfrak{m}})$ is positive definite by the positive definiteness of $\Sigma_{\beta,\infty}$, which is in turn implied by the parameter space restriction on $\mathcal{P}$. For each $j = 1,\ldots,2(M+1)$, then we have
\begin{align*}
	\sqrt{n}(\hat{d}_j(\theta_{F_n}) - a_{j,F_n}) 
	&= \sqrt{n}(\hat{a}_j - a_{j,F_n}) - \hat{\Gamma}_j(\theta_{F_n})\hat{\Omega}(\theta_{F_n})^{-1}\sqrt{n}(\hat{A}\theta_{F_n} - \hat{\beta}) \\
	&\xrightarrow{d} \underbrace{M_j(p_{\infty}) \mathcal{Z}_p - \cov(M_j(p_\infty)\mathcal{Z}_p, \mathcal{Z}_{\mathfrak{m}})\var(\mathcal{Z}_{\mathfrak{m}})^{-1}\mathcal{Z}_{\mathfrak{m}}}_{\equiv \mathcal{Z}_{d_j}},
\end{align*}
where the second line holds by combining results from \eqref{eq:A_mat_asymptotics}, \eqref{eq:convergence_Omega}, and \eqref{eq:convergence_Gamma}. Then we observe that
\[
	\cov(\mathcal{Z}_{d_j}, \mathcal{Z}_{\mathfrak{m}})
	= 0_{2(K+1)\times 2(K+1)}.
\]
This shows
\begin{equation}
\label{eq:convergence_moment}
	\sqrt{n}
	\begin{pmatrix}
		\hat{A}\theta_{F_n} - \hat{\beta} \\
		\operatorname*{vec}(\hat{D}(\theta_{F_n})) - \operatorname*{vec}(A_{F_n})
	\end{pmatrix}
	~~\xrightarrow{d}~~
	\begin{pmatrix}
		\mathcal{Z}_{\mathfrak{m}} \\
		\operatorname*{vec}(\mathcal{Z}_D)
  \end{pmatrix},
\end{equation}
where $\mathcal{Z}_D = (\mathcal{Z}_{d_1}, \ldots, \mathcal{Z}_{d_{2(M+1)}})$ is independent of $\mathcal{Z}_{\mathfrak{m}}$.

Based on \eqref{eq:convergence_moment}, applying Lemma \ref{lem:new_cond_stat} yields
\[
	n^{1/2}\widetilde{D}(\theta_{F_n}) B_{F_n} \Psi_{F_n} 
	~~\xrightarrow{d}~~
	\mathcal{D}_\xi
\]
where $\mathcal{D}_\xi$ is of full column rank with probability one and is also independent of $\mathcal{Z}_{\mathfrak{m}}$, and $\Psi_{F_n}$ is a diagonal matrix with positive diagonal elements 
\begin{equation}
\label{eq:defn_S}
		\Psi_{F_n} = 
		\begin{bmatrix}
		\begin{pmatrix}
		(\sqrt{n}\tau_{1,F_n})^{-1} & & \\
		& \ddots & \\
		& & (\sqrt{n}\tau_{q,F_n})^{-1} 
		\end{pmatrix} & 0_{q\times (2(M+1) - q)} \\
		0_{(2(M+1) - q)\times q} & I_{2(M+1) - q}
		\end{bmatrix}
\end{equation}
with $q \equiv \max\{j: t_j = \infty\}$.

	Define a normalizing constant $\gamma_n \equiv \|\sqrt{n} \Psi_{F_n} B_{F_n}'c\|^{-1}$. 
	 By continuous mapping theorem and given that $\mathcal{D}_\xi$ is of full rank with probability one, we have
	\begin{align}
		\hat{Q}(\theta_{F_n}) \gamma_n
		&~~=~~ 
		\hat{\Omega}(\theta_{F_n})^{-1/2}\left(\sqrt{n}\widetilde{D}(\theta_{F_n})B_{F_n} \Psi_{F_n} \right) \notag \\
		&~~\quad \left[\left(\sqrt{n}\widetilde{D}(\theta_{F_n})B_{F_n} \Psi_{F_n}\right)'\hat{\Omega}(\theta_{F_n})^{-1} \left(\sqrt{n}\widetilde{D}(\theta_{F_n})B_{F_n} \Psi_{F_n}\right)\right]^{-1}\underbrace{\frac{\sqrt{n} \Psi_{F_n} B_{F_n}' c}{\|\sqrt{n} \Psi_{F_n} B_{F_n}'c\|}}_{\iota_n} \notag \\
		&~~\xrightarrow{d}~~
		\mathcal{Q} ~~\equiv~~
			 \var(\mathcal{Z}_{\mathfrak{m}})^{-1/2}\mathcal{D}_\xi [\mathcal{D}_\xi'\var(\mathcal{Z}_{\mathfrak{m}})^{-1}\mathcal{D}_\xi ]^{-1} \iota_\infty.
		\label{eq:convergence_Q}
	\end{align}
Since $\iota_\infty \neq 0_{2(M+1)\times 1}$ and $\mathcal{D}_\xi$ has full rank with probability one, we conclude that $\mathcal{Q}$ is nonzero almost surely and  independent of $\mathcal{Z}_{\mathfrak{m}}$ by the independence between $\mathcal{D}_\xi$ and $\mathcal{Z}_{\mathfrak{m}}$.
	
	Then it follows that 
	\begin{align*}
	\text{MRLM}_n(\theta_{F_n}) 
	&~~=~~ \sqrt{n}(\hat{A}\theta_{F_n} - \hat{\beta})'\hat{\Omega}(\theta_{F_n})^{-1/2} P_{\hat{Q}(\theta_{F_n})}\hat{\Omega}(\theta_{F_n})^{-1/2}\sqrt{n}(\hat{A}\theta_{F_n} - \hat{\beta}) \\
	&~~=~~ \sqrt{n}(\hat{A}\theta_{F_n} - \hat{\beta})'\hat{\Omega}(\theta_{F_n})^{-1/2} P_{\hat{Q}(\theta_{F_n})\gamma_n}\hat{\Omega}(\theta_{F_n})^{-1/2}\sqrt{n}(\hat{A}\theta_{F_n} - \hat{\beta}) \\
	&~~\xrightarrow{d}~~ \mathcal{Z}_{\mathfrak{m}}'\var(\mathcal{Z}_{\mathfrak{m}})^{-1/2}P_{\mathcal{Q}}\var(\mathcal{Z}_{\mathfrak{m}})^{-1/2} \mathcal{Z}_{\mathfrak{m}} \\
	&~~\sim~~ \chi_1^2
	\end{align*}
	The third line holds by continuous mapping theorem based on \eqref{eq:convergence_Omega},  \eqref{eq:convergence_moment}, and \eqref{eq:convergence_Q}. The last line follows by that $\mathcal{Q}$ is nonzero almost surely and $\mathcal{Q} \indep \mathcal{Z}_{\mathfrak{m}}$, thus 
	\[
		\mathcal{Z}_{\mathfrak{m}}'\var(\mathcal{Z}_{\mathfrak{m}})^{-1/2}P_{\mathcal{Q}}\var(\mathcal{Z}_{\mathfrak{m}})^{-1/2} \mathcal{Z}_{\mathfrak{m}} ~~\sim~~ \chi_1^2
		\quad 
		\text{conditional on $\mathcal{Q}$}
	\]
	implies the unconditional distribution
	\[
		\mathcal{Z}_{\mathfrak{m}}'\var(\mathcal{Z}_{\mathfrak{m}})^{-1/2}P_{\mathcal{Q}}\var(\mathcal{Z}_{\mathfrak{m}})^{-1/2} \mathcal{Z}_{\mathfrak{m}} ~~\sim~~ \chi_1^2.
	\]

	We apply similar arguments to the difference between AR and MRLM statistics, yielding
	\[
	\begin{bmatrix}
	\text{AR}_n(\theta_{F_n}) - \text{MRLM}_n(\theta_{F_n}) \\
	\text{MRLM}_n(\theta_{F_n})
	\end{bmatrix}
 ~~\xrightarrow{d}~~ 
	\begin{bmatrix}
	\mathcal{Z}_{\mathfrak{m}}'\var(\mathcal{Z}_{\mathfrak{m}})^{-1/2}M_{\mathcal{Q}} \var(\mathcal{Z}_{\mathfrak{m}})^{-1/2}\mathcal{Z}_{\mathfrak{m}}  \\
	\mathcal{Z}_{\mathfrak{m}}'\var(\mathcal{Z}_{\mathfrak{m}})^{-1/2}P_{\mathcal{Q}} \var(\mathcal{Z}_{\mathfrak{m}})^{-1/2}\mathcal{Z}_{\mathfrak{m}} 
	\end{bmatrix}
	~~\sim~~
	\begin{bmatrix}
	\chi_{2K + 1}^2 \\
	\chi_{1}^2
	\end{bmatrix},
	\]
	where $M_\mathcal{Q} \equiv I - P_\mathcal{Q}$ denotes the annihilator operator. Note that $\chi_{2K + 1}^2$ is independent of $\chi_1^2$ because $(P_{\mathcal{Q}}\var(\mathcal{Z}_{\mathfrak{m}})^{-1/2}\mathcal{Z}_{\mathfrak{m}}, \quad M_{\mathcal{Q}}\var(\mathcal{Z}_{\mathfrak{m}})^{-1/2}\mathcal{Z}_{\mathfrak{m}})$ are uncorrelated:
	\begin{align*}
	\cov(P_{\mathcal{Q}}\var(\mathcal{Z}_{\mathfrak{m}})^{-1/2}\mathcal{Z}_{\mathfrak{m}}, M_{\mathcal{Q}}\var(\mathcal{Z}_{\mathfrak{m}})^{-1/2}\mathcal{Z}_{\mathfrak{m}})
	&= \Exp\left(\cov(P_{\mathcal{Q}}\var(\mathcal{Z}_{\mathfrak{m}})^{-1/2}\mathcal{Z}_{\mathfrak{m}}, M_{\mathcal{Q}}\var(\mathcal{Z}_{\mathfrak{m}})^{-1/2}\mathcal{Z}_{\mathfrak{m}} \mid \mathcal{Q})\right) \\
	&= \Exp(P_{\mathcal{Q}}M_{\mathcal{Q}}) \\
	&= 0_{2(K+1)\times 2(K+1)}.
	\end{align*}
	The first line holds by the law of total covariance and the fact that $\Exp[\mathcal{Z}_{\mathfrak{m}}\mid \mathcal{Q}] = \Exp[\mathcal{Z}_{\mathfrak{m}}] = 0$. The second line holds by the independence between $\mathcal{Q}$ and $\mathcal{Z}_{\mathfrak{m}}$. The third line holds by the definition $M_{\mathcal{Q}} = I - P_{\mathcal{Q}}$. So we have
	\[
		\text{MLC}_n(\theta_{F_n}) = \text{MRLM}_n(\theta_{F_n}) + a\cdot\text{AR}_n(\theta_{F_n}) ~~\xrightarrow{d}~~ (1+a)\chi_1^2 + a\chi_{2K+1}^2 
	\]
	by continuous mapping theorem.

	As a result, we show that
	\begin{align*}
		& \limsup_{n\to\infty} \Prob_{F_{n}}\left(\inf_{c'\theta = \lambda_{n}} \text{MLC}_n(\theta) > q_{(1+a)\chi_1^2 + a\chi_{2K+1}^2 }(1-\alpha)\right) \\
		& \leq 
		\lim_{n\to\infty} \Prob_{F_{n}}\left( \text{MLC}_n(\theta_{F_n}) > q_{(1+a)\chi_1^2 + a\chi_{2K+1}^2 }(1-\alpha)\right) \\
		& = \alpha.
	\end{align*}
	Then the proof is complete.
\end{proof}

\clearpage

\subsection{Proof of bias formula in section \ref{sec:bias_add_sep}}

\begin{proof}[Proof of Lemma \ref{lem:bias_formula_general}]
Applying FWL theorem to partial out the linear effects of $W$ in the misspecified regression \eqref{eq:shortreg}, we have
\begin{align*}
    &Y^{\perp W\mid D=1} = \tilde{\rho}_1 \lambda_1(P)^{\perp W\mid D=1} + Y^{\perp W, \lambda_1(P)\mid D=1}  &\text{conditional on } D = 1\\
    &Y^{\perp W\mid D=0} = \tilde{\rho}_0 \lambda_0(P)^{\perp W\mid D=0} + Y^{\perp W, \lambda_0(P)\mid D=0} &\text{conditional on } D = 0
\end{align*}
Conditional on $D=d$, simple OLS regression gives 
\begin{align*}
    \tilde{\rho}_d
    &= \frac{\cov(Y^{\perp W\mid D=d}, \lambda_1(P)^{\perp W\mid D=d}\mid D=d)}{\var(\lambda_d(P)^{\perp W\mid D=d}\mid D=d)} \\
    &= \frac{\cov((\rho_d \lambda_d(P) + W'\eta_d \lambda_d(P))^{\perp W\mid D=d}, \lambda_d(P)^{\perp W\mid D=d}\mid D=d)}{\var(\lambda_d(P)^{\perp W\mid D=d}\mid D=d)} \\
    &= \rho_d + \frac{\cov((W'\lambda_d(P))^{\perp W\mid D=d}, \lambda_d(P)^{\perp W\mid D=d} \mid D=d)}{\var(\lambda_d(P)^{\perp W\mid D=d}\mid D=d)} ~~\eta_d \\
    &= \rho_d + 
    \frac{\cov(W'\lambda_d(P), \lambda_d(P)^{\perp W\mid D=d} \mid D=d)}{\var(\lambda_d(P)^{\perp W\mid D=d}\mid D=d)} ~~\eta_d 
\end{align*}
The second line holds by noting that correct regression equation \eqref{eq:longreg} gives
\[
    Y = \mu_d + W'\tau_d + \rho_d \lambda_d(P) + W'\eta_d \lambda_d(P) + \epsilon_d \quad \text{conditional on } D = d.
\]
where $\epsilon_d \equiv Y - \Exp[Y\mid D=d, W, P]$. This implies 
\[
    Y^{\perp W\mid D=d} = \left(\rho_d \lambda_d(P) + W'\eta_d \lambda_d(P)\right)^{\perp W\mid D=d}
\]
since $\epsilon_d$ is uncorrelated with $W$ conditional on $D=d$. So we  establish the bias formula for $\tilde{\rho}_d - \rho_d$.

Alternatively, we can also apply the FWL theorem to partial out the linear effects of $\lambda_d(W)$ in the misspecified regression, giving 
\begin{align*}
    &Y^{\perp \lambda_1(P)\mid D=1} = \tilde{\tau}_1'W^{\perp \lambda_1(P)\mid D=1} + Y^{\perp W, \lambda_1(P)\mid D=1}  &\text{conditional on } D = 1\\
    &Y^{\perp \lambda_0(P)\mid D=0} = \tilde{\tau}_0'W^{\perp \lambda_0(P)\mid D=0} + Y^{\perp W, \lambda_0(P)\mid D=0} &\text{conditional on } D = 0
\end{align*}
In this way, $\tilde{\tau}_d$ is the OLS coefficient on $W$ in a regression of $Y^{\perp \lambda_d(P)\mid D=d}$ on $W^{\perp \lambda_d(P)\mid D=d}$: 
\[
    \tilde{\tau}_d = \Exp[(W^{\perp \lambda_d(P)\mid D=d})(W^{\perp \lambda_d(P)\mid D=d})'\mid D=d]^{-1}\Exp[(W^{\perp \lambda_d(P)\mid D=d})(Y^{\perp \lambda_d(P)\mid D=d})\mid D=d].
\]
The correct regression \eqref{eq:longreg} implies that  
$Y^{\perp \lambda_d(P)\mid D=d} = (W'\tau_d + W'\eta_d\lambda_d(P))^{\perp \lambda_d(P)\mid D=d}$. Plugging it into the OLS estimand $\tilde{\tau}_d$ then gives
\begin{align*}
    \tilde{\tau}_d - \tau_d 
    &= \Exp[(W^{\perp \lambda_d(P)\mid D=d})(W^{\perp \lambda_d(P)\mid D=d})'\mid D=d]^{-1}\Exp[(W^{\perp \lambda_1(P)\mid D=d})(W'\lambda_d(P))^{\perp \lambda_d(P)\mid D=d}\mid D=d]~\eta_d \\
    &= \Exp[(W^{\perp \lambda_d(P)\mid D=d})(W^{\perp \lambda_d(P)\mid D=d})'\mid D=d]^{-1}\Exp[(W^{\perp \lambda_1(P)\mid D=d})(W'\lambda_d(P))\mid D=d]~\eta_d 
\end{align*}
which establishes the bias formula for $\tilde{\tau}_d - \tau_d$.
\end{proof}

\begin{proof}[Proof of Theorem \ref{thm:bias_formula_specific}]
We begin by noting that $\cov(W, \lambda_d(P)\mid D=d) = 0_{L\times 1}$ implies 
\begin{equation}
\label{eq:reg_resid_1}
    W^{\perp \lambda_d(P)\mid D=d} = W - \Exp[W\mid D=d] \quad \text{and} \quad \lambda_d(P)^{\perp W\mid D=d} = \lambda_d(P) - \Exp[\lambda_d(P)\mid D=d].
\end{equation}
First, we simplify the bias formula in Lemma \ref{lem:bias_formula_general}, which would lead us to compute the bias on causal parameter of interests. The general bias formula on $\tau_d$ can be simplified by the following derivations:
\begin{align}
    \tilde{\tau}_d
    &= \tau_d + \Exp[(W^{\perp \lambda_d(P)\mid D=d})(W^{\perp \lambda_d(P)\mid D=d})'\mid D=d]^{-1}\Exp[(W^{\perp \lambda_d(P)\mid D=d})(W'\lambda_d(P))\mid D=d]~\eta_d  \notag \\
    &= \tau_d + \var(W\mid D=d)^{-1}\Exp[(W-\Exp[W\mid D=d])W'\lambda_d(P)\mid D=d] ~\eta_d \notag \\
    &= \tau_d + \var(W\mid D=d)^{-1}\left(\Exp[WW'\lambda_d(P)\mid D=d] - \Exp[W\mid D=d]\Exp[W'\lambda_d(P)\mid D=d]\right) \eta_d \notag \\
    &= \tau_d + \var(W\mid D=d)^{-1}\var(W\mid D=d) \Exp[\lambda_d(P)\mid D=d] ~\eta_d \notag \\
    &= \tau_d + \Exp[\lambda_d(P)\mid D=d] \eta_d. \label{eq:bias_tau}
\end{align}
The first equality is given by the Lemma \ref{lem:bias_formula_general}. The second equality holds by \eqref{eq:reg_resid_1}. The fourth equality holds by assumption $\cov(WW',\lambda_d(P)\mid D=d) = 0_{L\times L}$ and $\cov(W,\lambda_d(P)\mid D=d) = 0_{L\times 1}$. 

On the other hand, the general bias formula on $\rho_d$ can be simplified by the following derivations:
\begin{align}
    \tilde{\rho}_d
    &= \rho_d + \frac{\cov(W'\lambda_d(P), \lambda_d(P)^{\perp W\mid D=d} \mid D=d)}{\var(\lambda_d(P)^{\perp W\mid D=d}\mid D=d)} ~~\eta_d \notag \\
    &= \rho_d + \frac{\cov(W'\lambda_d(P), \lambda_d(P) - \Exp[\lambda_d(P)\mid D=d]\mid D=d)}{\var(\lambda_d(P)\mid D=d)}~~\eta_d \notag \\
    &= \rho_d + \frac{\Exp[W'\lambda_1(P)^2\mid D=d] - \Exp[W'\mid D=d]\Exp[\lambda_d(P)\mid D=d]^2}{\var(\lambda_d(P)\mid D=d)}~~\eta_d \notag \\
    &= \rho_d + \Exp[W'\eta_d\mid D=d]. \label{eq:bias_rho}
\end{align}
The first equality is given by the Lemma \ref{lem:bias_formula_general}. The second equality holds by \eqref{eq:reg_resid_1}. The third equality holds by assumption that $\cov(W, \lambda_d(P)\mid D=d) = 0_{L\times 1}$. The final equality holds by the assumption that $\cov(W, \lambda_d(P)^2) = 0_{L\times 1}$.

Having simplified the bias formulae in Lemma \ref{lem:bias_formula_general}, now we can derive the bias formula for $\mu_d$ as follows:
\begin{align}
    \tilde{\mu}_d
    &= \Exp[Y - W'\tilde{\tau}_d - \tilde{\rho}_d\lambda_d(P)\mid D=d] \notag \\
    &= \Exp[\mu_d + W'\tau_d + \rho_d\lambda_d(P) + [W\lambda_d(P)]'\eta_d - W'\tilde{\tau}_d - \tilde{\rho}_d\lambda_d(P)\mid D=d] \notag \\
    &=\Exp[\mu_d + [W\lambda_d(P)]'\eta_d - W'\eta_d \Exp[\lambda_d(P)\mid D=d] - \Exp[W'\eta_d\mid D=d]\lambda_d(P)\mid D=d] \notag \\
    &= \mu_d - \Exp[W'\eta_d\mid D=d]\Exp[\lambda_d(P)\mid D=d] \label{eq:bias_mu}
\end{align}
The first equality follows by misspecified regression \eqref{eq:shortreg}. The second equality follows by the correctly specified model \eqref{eq:longreg}. The third equality follows from the bias formula in \eqref{eq:bias_tau} and \eqref{eq:bias_rho}. The last equality holds by the assumption $\cov(W,\lambda_d(P)\mid D=d) = 0_{L\times 1}$.

Given the bias formula in \eqref{eq:bias_tau}, \eqref{eq:bias_rho}, and  \eqref{eq:bias_mu}, now we compute the bias on estimating the slope of MTE curve:
\begin{align*}
    \widetilde{\text{Slope}} - \text{Slope}
    &= [\tilde{\rho}_1 - \tilde{\rho}_0] - \Exp[\rho_1(W) -\rho_0(W)]\\
    &= \Exp[W\mid D=1]'\eta_1 - \Exp[W\mid D=0]'\eta_0 - \Exp[W]'(\eta_1 - \eta_0) \\
    &= \left(\Exp[W\mid D=1] - \Exp[W\mid D=0]\right)'(\Prob(D=0) ~ \eta_1 + \Prob(D=1) ~ \eta_0).
\end{align*}
The second line holds by \eqref{eq:bias_rho}. This establishes the bias formula for the slope of MTE curve.

Next note that \eqref{eq:bias_tau} and \eqref{eq:bias_mu} imply
\begin{align*}
    \widetilde{\text{CATE}} - \text{CATE}
    &= [\tilde{\mu}_1 - \tilde{\mu}_0] + w'(\tilde{\tau}_1 - \tilde{\tau}_0) - [\mu_1 - \mu_0] - w'(\tau_1 - \tau_0) \\
    &= (w - \Exp[W\mid D=1])'\eta_1\times \Exp[\lambda_1(P)\mid D=1] \\
    & \quad - (w - \Exp[W\mid D=0])'\eta_0 \times \Exp[\lambda_0(P)\mid D=0].
\end{align*}
and 
\begin{align*}
    \widetilde{\text{ATE}} - \text{ATE}
    &= [\tilde{\mu}_1 - \tilde{\mu}_0] + \Exp[W]'(\tilde{\tau}_1 - \tilde{\tau}_0) - [\mu_1 - \mu_0] - \Exp[W]'(\tau_1 - \tau_0) \\
    &= (\Exp[W] - \Exp[W\mid D=1])'\eta_1\times \Exp[\lambda_1(P)\mid D=1] \\
    & \quad - (\Exp[W] - \Exp[W\mid D=0])'\eta_0 \times \Exp[\lambda_0(P)\mid D=0].
\end{align*}

So we have established the bias formula for ATE, CATE, and slope of MTE curve under Assumption \ref{asp:weak_exogeneous_covariates}.1. From this, note that the bias on ATE and slope would vanish if we additionally impose Assumption \ref{asp:weak_exogeneous_covariates}.2, i.e., $\Exp[W] = \Exp[W\mid D=1] = \Exp[W\mid D=0]$, so the proof is complete.
\end{proof}

\subsection{Lemmas for main results}
\label{sec:main_lemmas}

In this section, I provide the statement of supporting lemmas used in  Appendix \ref{sec:main_proofs}. The proofs of these lemmas can be found in Online Appendix \ref{appendix:lemma_proofs}.

\begin{lemma}
\label{lem:convergence_p_beta}
	For a sequence of distributions $\{(\theta_{F_n}, F_n)\}_{n=1}^\infty \subset \mathcal{P}$, suppose 
	\begin{enumerate}
		\item $q_{F_n}(z_\ell) \to q_{\infty}(z_\ell)$ for each $\ell = 0,1,\ldots,K$,
		\item $p_{F_n} \to p_\infty$,
		\item $\beta_{F_n} = (\beta_{F_n,1}', \beta_{F_n,0}')' \to \beta_{\infty}$,
		\item $\sigma^2_{F_n,d\ell} \to \sigma_{\infty,d\ell}^2$ for $d = 0,1$ and $\ell = 0,1,\ldots,K$.
	\end{enumerate}
	Then we have 
	\begin{enumerate}
	\item[(a)] The following convergence holds
	\begin{equation}
	\label{eq:clt_p_beta}
            \begin{pmatrix}
            \hat{Z}_p \\
            \hat{Z}_{\beta_1} \\
            \hat{Z}_{\beta_0} \\
            \hat{Z}_{q}
            \end{pmatrix}
            ~~\equiv~~
		\sqrt{n}
		\begin{pmatrix}
			\hat{p} - p_{F_n} \\
			\hat{\beta}_1 - \beta_{F_n, 1} \\
			\hat{\beta}_0 - \beta_{F_n, 0} \\
			\hat{q} - q_{F_n}
		\end{pmatrix}
		~~\xrightarrow{d}~~
		\begin{bmatrix}
			\mathcal{Z}_p \\
			\mathcal{Z}_{\beta_1} \\
			\mathcal{Z}_{\beta_0} \\
			\mathcal{Z}_q
		\end{bmatrix},
	\end{equation}
	where 
	\begin{align*}
            (\mathcal{Z}_p', \mathcal{Z}_{\beta_1}', \mathcal{Z}_{\beta_0}', \mathcal{Z}_q')' 
            ~~&\sim~~
		\normal\left(0_{3(K+1)\times 1},  \quad 
		\diag\{\Sigma_{p,\infty}, \Sigma_{\beta_1,\infty}, \Sigma_{\beta_0,\infty},\Sigma_{q,\infty}\}\right), \\
		\Sigma_{p,\infty} 
            &\equiv \diag\left\{\frac{p_\infty(z_\ell) (1-p_\infty(z_\ell))}{q_\infty(z_\ell)}: \ell = 0,1,\ldots,K\right\} \\
		\Sigma_{\beta_{1},\infty} 
            &\equiv
			\diag\left\{\frac{\sigma_{\infty,1\ell}^2}{q_\infty(z_\ell)p_\infty(z_\ell)}:\ell = 0,1,\ldots,K\right\} \\
		\Sigma_{\beta_0, \infty} 
            &\equiv 
			\diag\left\{\frac{\sigma_{\infty,0\ell}^2}{q_\infty(z_\ell)(1-p_\infty(z_\ell))}: \ell = 0,1,\ldots,K\right\} \\
		\Sigma_{q,\infty} 
            &\equiv \{\Sigma_{q,\infty}[i,j]\}_{i,j = 0,1,\ldots,K},
	\end{align*}
        with
        \[
            \Sigma_{q,\infty}[i,j] = 
            \begin{cases}
            p_\infty(z_i) (1-p_\infty(z_i)) \\
            -p_\infty(z_i)p_\infty(z_j).
            \end{cases}
        \]
	\item[(b)] We have consistent estimators for the asymptotic variance:
	\begin{align*}
		&\hat{\Sigma}_{p}  ~~\xrightarrow{p}~~ \Sigma_{p,\infty} \\
		&\hat{\Sigma}_\beta ~~\xrightarrow{p}~~ \Sigma_{\beta,\infty} \equiv \diag\{\Sigma_{\beta_1,\infty}, \Sigma_{\beta_0, \infty}\} \\
		&\hat{\Sigma}_q ~~\xrightarrow{p}~~ \Sigma_{q,\infty}.
	\end{align*}
        \item[(c)] The convergence results in part (a) and (b) also hold for any subsequence $\{p_n\}$ provided that the conditions given in the lemma hold along $\{p_n\}$.
	\end{enumerate}
\end{lemma}

\begin{lemma}
\label{lem:pd_variance}
Suppose $c_\mu \neq 0$ or $c_\rho \neq 0$ are fixed weights and let $\lambda = c'\theta$. For any $(\theta,F)\in\mathcal{P}$, the asymptotic variance matrix 
\[
    S \equiv [\partial_p g(\lambda)]\Sigma_p[\partial_{p'} g(\lambda)] + [\partial_{\beta_1} g(\lambda)]\Sigma_{\beta_1}[\partial_{\beta_1'} g(\lambda)] + [\partial_{\beta_0} g(\lambda)]\Sigma_{\beta_0}[\partial_{\beta_0'} g(\lambda)]
\]
is uniformly bounded and positive definite. That is, 
\begin{equation}
\label{eq:bdd_above_zero}
    \inf_{(\theta,F)\in \mathcal{P}} \lambda_{\min}(S) > 0,
\end{equation}
and 
\begin{equation}
\label{eq:bdd_below_inf}
    \sup_{(\theta,F)\in \mathcal{P}} \lambda_{\max}(S) < \infty,
\end{equation}
where $\lambda_{\min}(S)$ and $\lambda_{\max}(S)$ denote the smallest and largest eigenvalues (or singular value) of $S$, respectively.
\end{lemma}
\begin{remark}
This lemma shows that any choice of nonzero linear weights $c_\mu$ and $c_\rho$ would not entail an asymptotic singular covariance matrix for AR test, which is essential for establishing uniform validity.
\end{remark}
\begin{remark}
Indeed, we can relax the restriction that $\lambda = c'\theta$ and show this result for all $\lambda$ that belongs to a compact subset on the real line, which will be useful to establish the power property of the test. Since the above statement is sufficient to prove the uniform validity of the conditional Wald test, I do not pursue this extension here.
\end{remark}
\begin{lemma}
\label{lem:NS_cond}
    Suppose the conditions in Proposition \ref{prop:subsequence_AR_validity} hold for a sequence $(\lambda_n, F_n)$ for which $\sqrt{n}\|\pi_{F_{n}}\| \to s_\infty < \infty$. Then $\Upsilon_\infty = (\mathcal{Z}_h, S_\infty, \Sigma_{p,\infty}, \partial_p g(\lambda_\infty))$ belongs to $\mathcal{U}$ defined in \eqref{eq:defn_set_C} almost surely. That is, 
    \[
        \Prob_{\eta^*}\left(\mathcal{Z}_h + [\partial_p \pi] \Sigma_{p,\infty} \partial_{p'} g(\lambda_\infty) S_\infty^{-1/2}\eta^* \neq 0_{K\times 1} \mid \mathcal{Z}_h \right) = 1 \quad \text{a.s.}
    \]
\end{lemma}

\begin{lemma}
\label{lem:cond_stat}[\citet[Lemma 10.3]{andrews/guggenberger:2017}]
Assume the following conditions hold under a sequence of DGPs $\{(\theta_{F_n}, F_n)\}_{n\geq 1}$:
\begin{enumerate}
\item The scaled singular values converge
\[
	\sqrt{n}\tau_{j,F_n} \to 
	\begin{cases}
		\infty & \text{if } j \leq q \\
		t_j \in [0,\infty) & \text{if } j > q
	\end{cases}
\]
\item $C_{F_n}$ and $B_{F_n}$ converge to their limits $C_{\infty}$ and $B_{\infty}$, respectively.
\item The moment condition and the normalized $\hat{D}(\theta_{F_n})$ converge jointly:
\begin{align*}
	\sqrt{n}
	\begin{pmatrix}
		\hat{A}\theta_{F_n} - \hat{\beta} \\
		\operatorname*{vec}(\hat{D}(\theta_{F_n})) - \operatorname*{vec}(A_{F_n})
	\end{pmatrix}
	~~\xrightarrow{d}~~
	\begin{pmatrix}
		\mathcal{Z}_{\mathfrak{m}} \\
		\operatorname{vec}(\mathcal{Z}_D)
	\end{pmatrix}
\end{align*}
where $\mathcal{Z}_D = (\mathcal{Z}_{d_1}, \ldots, \mathcal{Z}_{d_{2(M+1)}})$ is independent of $\mathcal{Z}_{\mathfrak{m}}$.
\end{enumerate}
Then we have
\[
	n^{1/2}\hat{D}(\theta_{F_n}) B_{F_n} \Psi_{F_n} ~~\xrightarrow{d}~~ \mathcal{D} = O_p(1)
\]
where
\[
	\Psi_{F_n} \equiv \diag\{(\sqrt{n}\tau_{1,F_n})^{-1},\ldots, (\sqrt{n}\tau_{q,F_n})^{-1}, 1,\ldots, 1\} \in \R^{2(M+1)\times 2(M+1)}.
\]
Moreover, $\mathcal{D}$ is independent of $\mathcal{Z}_{\mathfrak{m}}$.
\end{lemma}

\begin{lemma}
	\label{lem:new_cond_stat}
	Suppose the assumptions in Lemma \ref{lem:cond_stat} hold. Let $\xi\in\R^{2(K+1)\times 2(M+1)}$ be a matrix of i.i.d. standard normal random variables and $\kappa > 0$. Then 
	\[
		n^{1/2} (\hat{D}(\theta_{F_n}) + \kappa n^{-1/2}\xi) B_{F_n} \Psi_{F_n} 
		~~\xrightarrow{d}~~
		\mathcal{D}_\xi	
	\]
	where $\mathcal{D}_\xi$ has full rank with probability one and is independent of $\mathcal{Z}_{\mathfrak{m}}$.
\end{lemma}

\begin{lemma}[The $c_r$ inequality]
\label{lem:cr_ineq}
	Let $X$ and $Y$ be random variables with $\Exp|X|^r < \infty$ and $\Exp|Y|^r < \infty$ for some $r > 0$. Then 
	\[
		\Exp|X+Y|^r \leq c_r\left(\Exp|X|^r + \Exp|Y|^r\right)	
	\]  
	where $c_r = 1$ if $r \leq 1$ and $c_r = 2^{r-1}$ if $r > 1$.
\end{lemma}

\begin{lemma}[WLLN for $L^{1+\delta}$ triangular array]
\label{lem:wlln_array}
Let $\{X_{ni}\}$ be a row-wise i.i.d. triangular array of random variables. Suppose
\begin{equation}
\label{eq:wlln}
	\sup_{n}\Exp |X_{ni}|^{1+\delta} < \infty
\end{equation}
for some $\delta > 0$. 
Then 
\[
	\frac{1}{n}\sum_{i=1}^n X_{ni} - \Exp[X_{ni}] ~~\xrightarrow{p}~~ 0.
\]
\end{lemma}

\begin{lemma}[Multivariate Lyapunov CLT]
\label{lem:clt_array}
	Let $\{X_{ni}\}$ be a row-wise i.i.d. triangular array of random vectors in $\R^k$ with expectations $\Exp[X_{ni}] = 0$ and covariance matrices $\Sigma_{n} = \Exp[X_{ni}X_{ni}']$.  Suppose 
\begin{equation}
\label{eq:lyapunov_var}
	\Sigma_n \to \Sigma \geq 0
\end{equation}
and for some $\delta > 0$
\begin{equation}
\label{eq:lyapunov_moment}
	\sup_{n} \Exp\|X_{ni}\|^{2+\delta} < \infty.
\end{equation}
Then as $n\to \infty$,
\[
	\frac{1}{\sqrt{n}}\sum_{i=1}^n X_i ~~\xrightarrow{d}~~ \normal(0,\Sigma).
\]
\end{lemma}

\putbib
\end{bibunit}

\begin{bibunit} %
\begin{titlepage}
  \begin{center}
  \vskip 60pt
  \Large \textbf{Online appendices for 
``Extrapolating LATE with Weak IVs''} \par
  \vskip 1em
  \large Muyang Ren\textsuperscript{$\dagger$}\customfootnotetext{$\dagger$}{Department of Economics, University of Tennessee, \href{mailto:mren7@utk.edu}{\texttt{mren7@utk.edu}}} 
  
  \vskip 1em

  \today

  \vskip 1em

    \begin{abstract}
        This document contains a set of results that accompany the main text. Appendix \ref{appendix:lemma_proofs} provides proofs of the lemmas stated in Appendix \ref{sec:main_lemmas}. Appendix \ref{appendix:power_analysis_MLC} establishes consistency and conducts local power analysis for the MLC test introduced in section \ref{sec:polyMTE_inference}. Appendix \ref{appendix:estimated_weights} shows how to extend the proposed methods to settings with estimated weights. Appendix \ref{appendix:inference_covariates} discusses robust inference with additional covariates. Appendix \ref{appendix:add_empirical_results} presents the court-specific confidence sets omitted from the section \ref{sec:empirical}. Appendix \ref{appendix:iv_strength} discusses issues related to testing IV strength in the empirical application. Appendix \ref{appendix:numerical_example} provides a numerical example illustrating the bias of the ATE estimand when additive separability is misspecified. Appendix \ref{appendix:litreview} provides further discussion of the weak IV literature on subvector inference and inference with covariates. Finally, Appendix \ref{appendix:control_function} shows that an alternative two-stage regression approach can lead to a (nearly) singular variance of moment conditions under weak identification.
    \end{abstract}
    
  \end{center}
    
  \part{}
  \vspace{-1.5cm}
  \setcounter{parttocdepth}{1}  %
  \parttoc %
\end{titlepage}

\section{Proofs of lemmas in Appendix \ref{sec:main_lemmas}}
\label{appendix:lemma_proofs}

\begin{proof}[Proof of Lemma \ref{lem:convergence_p_beta}]
    \underline{Part (a)}:
	We first show the influence function representation of $\sqrt{n}(\hat{\beta} - \beta_{F_n})$. Recall that
	\begin{align*}
		q_{F}(d,z_\ell) 
		& = \Prob_F(D=d, Z=z_\ell) \\
		& = \begin{cases}
			q_{F}(z_\ell)p_F(z_\ell) & \text{ if } d = 1\\
			q_F(z_\ell)(1-p_F(z_\ell)) & \text{ if } d = 0.
		\end{cases} 
	\end{align*} 
	Denote $q_\infty(d,z_\ell) \equiv \lim_{n\to\infty} q_{F_n}(d,z_\ell) \in [\epsilon^2, (1-\epsilon)^2]$ under the parameter space restriction.
	Note that for each $d = 0,1$ and $\ell = 0,1,\ldots,K$,
	{\allowdisplaybreaks
	\begin{align*}
		\sqrt{n}(\hat{\beta}_{d\ell} - \beta_{F_n,d\ell})
		&= \sqrt{n}\left(\frac{\sum_{i=1}^n Y_i\indicator[D_i = d, Z_i = z_\ell]}{\sum_{i=1}^n \indicator[D_i = d, Z_i = z_\ell]} - \beta_{F_n,d\ell} \right) \\
		&= \sqrt{n}\left(\frac{q_{F_n}(d,z_\ell)}{\hat{q}(d,z_\ell)}\cdot \frac{1}{n}\sum_{i=1}^n \frac{Y_i \indicator[D_i = d, Z_i = z_\ell]}{q_{F_n}(z_\ell)} - \beta_{F_n,d\ell} \right) \\
		&= \frac{1}{\sqrt{n}}\sum_{i=1}^n \left(\frac{Y_i\indicator[D_i=d, Z_i = z_\ell]}{q_{F_n}(d, z_\ell)} - \beta_{F_n,d\ell} \right) \\
		&\quad + \frac{1}{n}\sum_{i=1}^n \frac{Y_i\indicator[D_i=d, Z_i=z_\ell]}{q_{F_n}(d,z_\ell)}\cdot \sqrt{n}\left(\frac{q_{F_n}(d,z_\ell)}{\hat{q}(d,z_\ell)} - 1\right) \\
		&= \frac{1}{\sqrt{n}}\sum_{i=1}^n \left(\frac{Y_i\indicator[D_i=d, Z_i = z_\ell]}{q_{F_n}(d, z_\ell)} - \beta_{F_n,d\ell} \right) \\
		&\quad - \frac{1}{\sqrt{n}}\sum_{i=1}^n \left(\indicator[D_i = d, Z_i = z_\ell] - q_{F_n}(d,z_\ell)\right) \cdot \left(\frac{\beta_{F_n,d\ell}}{q_{F_n}(d, z_\ell)} + o_p(1)\right) \\
		&= \frac{1}{\sqrt{n}} \sum_{i=1}^n \frac{Y_i\indicator[D_i = d, Z_i = z_\ell] - \indicator[D_i=d, Z_i = z_\ell]\beta_{F_n, d\ell}}{q_{F_n}(d,z_\ell)} + o_p(1) \\
		&= \frac{1}{\sqrt{n}} \sum_{i=1}^n  (Y_i - \Exp_{F_n}(Y\mid D=d, Z=z_\ell)) \frac{\indicator[D_i = d, Z_i = z_\ell]}{\Prob_{F_n}(D_i = d, Z_i=z_\ell)} + o_p(1).
	\end{align*}}
	The fourth equality holds by applying the WLLN (Lemma \ref{lem:wlln_array}) to $\{Y_i \indicator[D_i = d, Z_i = z_\ell]/q_{F_n}(d,z_\ell)\}_{i=1}^n$ and to $\{\indicator[D_i=d, Z_i = z_\ell]\}_{i=1}^n$, which yields
	\begin{align*}
		& \frac{1}{n}\sum_{i=1}^n \frac{Y_i\indicator[D_i = d, Z_i=z_\ell]}{q_{F_n}(d,z_\ell)} = \beta_{F_n, d\ell} + o_p(1) \\
		& \hat{q}(d,z_\ell) = q_{F_n}(d,z_\ell) + o_p(1).
	\end{align*}
	Note that $L^{1+\delta}$-integrability \eqref{eq:wlln} required by WLLN holds here since $\{\indicator[D_i=d, Z_i = z_\ell]\}_{i=1}^n$ are uniformly bounded and for $\{Y_i \indicator[D_i = d, Z_i = z_\ell]/q_{F_n}(d,z_\ell)\}_{i=1}^n$, we have 
	\begin{align*}
		\Exp_{F_n}\left|\frac{Y_i\indicator[D_i = d, Z_i = z_\ell]}{q_{F_n}(d,z_\ell)}\right|^{1+\delta}
		&= \frac{\Exp_{F_n}(|Y_i|^{1+\delta} \mid D_i = d, Z_i = z_\ell)}{q_{F_n}(d,z_\ell)^{\delta}}	\\
		&\leq \frac{\Exp_{F_n}(|Y_i|^{2+\delta} \mid D_i = d, Z_i = z_\ell)^{\frac{1+\delta}{2+\delta}}}{q_{F_n}(d,z_\ell)^{\delta}} \quad \text{by H\"{o}lder's inequality}\\
		&<
		\frac{\zeta^{\frac{1+\delta}{2+\delta}}}{\epsilon^{2\delta}} < \infty.
	\end{align*}
	The fifth equality holds by applying Lyapunov CLT to $\{n^{-1/2}(\indicator[D_i = d, Z_i = z_\ell] - \Exp_{F_n}(D_i = d, Z_i = z_\ell))\}_{i=1}^n$:
	\[
		\frac{1}{\sqrt{n}}\sum_{i=1}^n (\indicator[D_i = d, Z_i = z_\ell] - q_{F_n}(d,z_\ell)) = O_p(1)
	\]
	and by noting that $O_p(1)o_p(1) = o_p(1)$.

	Likewise, we can derive the influence function represetntation for $\sqrt{n}(\hat{p} - p_{F_n})$. For each $\ell = 0,1,\ldots,K$, we have 
	\[
		\sqrt{n}(\hat{p}(z_\ell) - p_{F_n}(z_\ell)) = \frac{1}{\sqrt{n}} \sum_{i=1}^n (D_i - \Prob_{F_n}(D=1\mid Z=z_\ell))\frac{\indicator[Z_i = z_\ell]}{\Prob_{F_n}(Z = z_\ell)} + o_p(1).
	\]

	Therefore, 
	\begin{align*}
		\sqrt{n}
		\begin{pmatrix}
			\{\hat{p}(z_\ell) - p_{F_n}(z_\ell)\}_{\ell} \\
			\{\hat{\beta}_{1\ell} - \beta_{F_n,1\ell}\}_{\ell} \\
			\{\hat{\beta}_{0\ell} - \beta_{F_n,0\ell}\}_{\ell} \\
			\{\hat{q}(z_\ell) - q_{F_n}(z_\ell)\}_{\ell}
		\end{pmatrix}	
		&= 
		\frac{1}{\sqrt{n}} \sum_{i=1}^n 
		\begin{pmatrix}
			\left\{(D_i - \Prob_{F_n}(D=1\mid Z=z_\ell))\frac{\indicator[Z_i = z_\ell]}{\Prob_{F_n}(Z = z_\ell)}\right\}_{\ell} \\
			\left\{(Y_i - \Exp_{F_n}(Y\mid D=1, Z=z_\ell)) \frac{\indicator[D_i = 1, Z_i = z_\ell]}{\Prob_{F_n}(D_i = 1, Z_i=z_\ell)}\right\}_{\ell} \\
			\left\{(Y_i - \Exp_{F_n}(Y\mid D=0, Z=z_\ell)) \frac{\indicator[D_i = 0, Z_i = z_\ell]}{\Prob_{F_n}(D_i = 0, Z_i=z_\ell)}\right\}_{\ell} \\
			\left\{(\indicator[Z_i = z_\ell] - \Prob_{F_n}(Z=z_\ell)\right\}_{\ell}
		\end{pmatrix} + o_p(1) \\
		&\equiv
		\frac{1}{\sqrt{n}} \sum_{i=1}^n \phi_n(Y_i, D_i, Z_i) + o_p(1).
	\end{align*}
As $n$ goes to infinity, 
	\[
		\var_{F_n}(\phi_n(Y_i,D_i,Z_i))
		=
		\diag\{\Sigma_{p,F_n}, \Sigma_{\beta_1,F_n}, \Sigma_{\beta_0,F_n}, \Sigma_{q,F_n}\}	
		\to
		\diag\{\Sigma_{p,\infty}, \Sigma_{\beta_1,\infty}, \Sigma_{\beta_0,\infty}, \Sigma_{q,\infty}\}
	    \geq 
            0
        \]
	where 
	\begin{align*}
		& \Sigma_{p,F_n} = \diag\left\{\frac{p_{F_n}(z_\ell)(1-p_{F_n}(z_\ell))}{q_{F_n}(z_\ell)}: \ell = 0,1,\ldots,K\right\} \\
		& \Sigma_{\beta_1, F_n} = \diag\left\{\frac{\sigma^2_{F_n, 1\ell}}{q_{F_n}(z_\ell)p_{F_n}(z_\ell)}: \ell = 0,1,\ldots, K\right\} \\
		& \Sigma_{\beta_0, F_n} = \diag\left\{\frac{\sigma^2_{F_n, 0\ell}}{q_{F_n}(z_\ell)(1- p_{F_n}(z_\ell))}: \ell = 0,1,\ldots, K\right\} \\
		& \Sigma_{q,F_n} = \{\Sigma_{q,F_n}[i,j]\}_{i,j = 0,1,\ldots,K}
	\end{align*}
        with
        \[
            \Sigma_{q,F_n}[i,j] = 
            \begin{cases}
                p_{F_n}(z_i) (1 - p_{F_n}(z_i)) & \text{ if } i = j \\
                -p_{F_n}(z_i)p_{F_n}(z_j)       & \text{ if } i \neq j.
            \end{cases}
        \]
	Here the positive semi-definiteness of $\diag\{\Sigma_{p,\infty}, \Sigma_{\beta_1,\infty}, \Sigma_{\beta_0,\infty}, \Sigma_{q,\infty}\}$ is implied by the positive semi-definiteness of each covariance matrix in the diagonal position.

	To apply the Lyapunov CLT (Lemma \ref{lem:clt_array}) to $\phi_{n}(Y_i, D_i, Z_i)$, it suffices to verify the Lyapunov condition \eqref{eq:lyapunov_moment}. We note that
	 $\left\{(D_i - p_{F_n}(z_\ell))\frac{\indicator[Z_i = z_\ell]}{q_{F_n}(z_\ell)}\right\}_{\ell}$ and 
	 $\left\{\indicator[Z_i = z_\ell] - q_{F_n}(z_\ell)\right\}_{\ell}$
	are uniformly bounded, and
	\begin{align*}
		&\Exp_{F_n}\left|(Y_i - \Exp_{F_n}(Y\mid D=d, Z=z_\ell)) \frac{\indicator[D_i = d, Z_i = z_\ell]}{\Prob_{F_n}(D = d, Z=z_\ell)}\right|^{2+\delta} \\
		&= \frac{\Exp_{F_n}[|Y_i - \Exp_{F_n}(Y_i\mid D=d, Z=z_\ell)|^{2+\delta}\mid D=d, Z=z_\ell]}{\Prob_{F_n}(D = d, Z=z_\ell)^{1+\delta}} \\
		&< \frac{2^{2+\delta}\zeta}{\epsilon^{2(1+\delta)}} < \infty,
	\end{align*}
where the last line holds by the $L^{2+\delta}$-integrability of the centralized moment:
\begin{equation}
\label{eq:bdd_central_moment}
\begin{aligned}
	& \Exp[|Y - \Exp[Y\mid D=d, Z=z]|^{2+\delta}\mid D=d, Z=z] \\
	& \leq 2^{1+\delta}\left(\Exp[|Y|^{2+\delta}\mid D=d, Z=z] + |\Exp[Y\mid D=d, Z=z]|^{2+\delta}\right) \\
	& \leq 2^{2+\delta}\Exp[|Y|^{2+\delta}\mid D=d, Z=z] \\
	& \leq 2^{2+\delta} \zeta.
\end{aligned}
\end{equation}
Note that the second line of \eqref{eq:bdd_central_moment} follows by $c_r$ inequality given in Lemma \ref{lem:cr_ineq}, and the last line holds by the parameter space restriction on the existence of $(2+\delta)$'th moment of outcomes.

Hence the desired convergence \eqref{eq:clt_p_beta} follows by Lyapunov CLT.

    \medskip
    \underline{Part (b)}: 
	Next, we establish the consistency of the variance estimators. First note that 
	\[
		\hat{\Sigma}_{p} - \Sigma_{p,F_n}
		= 
		\diag\left\{\frac{\hat{p}(z_\ell) (1-\hat{p}(z_\ell))}{\hat{q}(d,z_\ell)} - \frac{p_{F_n}(z_\ell)(1-p_{F_n}(z_\ell))}{q_{F_n}(d,z_\ell)}\right\}
		= o_p(1)
	\]
	since $\hat{p} - p_{F_n} = o_p(1)$ by the convergence \eqref{eq:clt_p_beta} and for $d = 0,1$ and $\ell = 0,1,\ldots,K$,
	\begin{equation}
	\label{eq:aux_var_1}
		\hat{q}(d, z_\ell) - q_{F_n}(d, z_\ell) = o_p(1),
	\end{equation}
	which is implied by WLLN for bounded arrays.
 Therefore, $\hat{\Sigma}_{p}$ is consistent to $\Sigma_{p,\infty}$ since $\Sigma_{p,F_n} \to \Sigma_{p,\infty}$ as $n \to \infty$. Applying the similar arguments to $\hat{\Sigma}_q$ also yields the consistency of $\hat{\Sigma}_q$ to $\Sigma_{q,\infty}$.

	Now we examine the variance estimators $\hat{\sigma}_{d\ell}^2$:
	\begin{align*}
		\hat{\sigma}_{d\ell}^2 
		& = \frac{1}{n} \sum_{i=1}^n \frac{(Y_i - \hat{\beta}_{d\ell})^2\indicator[D_i=d, Z_i=z_\ell]}{\hat{q}(d,z_\ell)} \\
		& = \frac{1}{\hat{q}(d,z_\ell)} \left[\frac{1}{n}\sum_{i=1}^n \left((Y_i - \beta_{F_n, d\ell})^2 + 2(Y_i - \beta_{F_n,d\ell})(\hat{\beta}_{d\ell} - \beta_{F_n, d\ell}) + (\hat{\beta}_{d\ell} - \beta_{F_n, d\ell})^2 \right)\indicator[D_i =d, Z_i = z_\ell]\right] \\
		& = \frac{q_{F_n}(d,z_\ell)}{\hat{q}(d,z_\ell)} \left[ \frac{1}{n}\sum_{i=1}^n (Y_i - \beta_{F_n, d\ell})^2 \frac{\indicator[D_i = d, Z_i = z_\ell]}{q_{F_n}(d,z_\ell)}  \right] \\
		&\quad  + 2(\hat{\beta}_{d\ell} - \beta_{F_n, d\ell})\frac{q_{F_n}(d,z_\ell)}{\hat{q}(d,z_\ell)}\left[\frac{1}{n}\sum_{i=1}^n (Y_i - \beta_{F_n, d\ell})\frac{\indicator[D_i = d, Z = z_\ell]}{q_{F_n}(d,z_\ell)}\right] \\
		& \quad + (\hat{\beta}_{d\ell} - \beta_{F_n,d\ell})^2
	\end{align*}
	Note that equation \eqref{eq:bdd_central_moment} implies that $L^{1+\delta}$-integrability \eqref{eq:wlln} holds for 
	\[
		X_{ni} \in \left\{(Y_i - \beta_{F_n,d\ell})^2 \frac{\indicator[D_i=d, Z_i = z_\ell]}{q_{F_n}(d,z_\ell)}, ~~ (Y_i - \beta_{F_n, d\ell})\frac{\indicator[D_i=d, Z_i = z_\ell]}{q_{F_n}(d,z_\ell)}\right\}
	\]
for appropriately chosen $\delta$.

	By WLLN, we have
	\begin{equation}
	\label{eq:aux_var_2}
	\begin{aligned}
		& \frac{1}{n}\sum_{i=1}^n (Y_i - \beta_{F_n, d\ell})^2 \frac{\indicator[D_i = d, Z_i = z_\ell]}{q_{F_n}(d,z_\ell)}   - \sigma^2_{F_n, d\ell}  = o_p(1) \\
		& \frac{1}{n}\sum_{i=1}^n (Y_i - \beta_{F_n, d\ell})\frac{\indicator[D_i = d, Z = z_\ell]}{q_{F_n}(d,z_\ell)} = o_p(1) 
	\end{aligned}
	\end{equation}
	As implied by the convergence \eqref{eq:clt_p_beta}, we have $\hat{\beta}_{d\ell} - \beta_{F_n, d\ell} = o_p(1)$. This combined with equations \eqref{eq:aux_var_1} and \eqref{eq:aux_var_2} then yield 
	\[
		\hat{\sigma}_{d\ell}^2 - \sigma^2_{F_n, d\ell} ~~\xrightarrow{p}~~ 0.
	\]
	Hence we have 
	\[
		\hat{\Sigma}_{\beta} - \Sigma_{\beta,F_n} = \diag\left\{\frac{\hat{\sigma}_{d\ell}^2}{\hat{q}(d,z_\ell)} - \frac{\sigma_{F_n, d\ell}^2}{q_{F_n}(d,z_\ell)}: d=1,0; \ell = 0,1,\ldots,K\right\} = o_p(1).
	\]
	This establish the consistency of $\hat{\Sigma}_{\beta}$ since $\Sigma_{\beta,F_n} \to \Sigma_{\beta,\infty}$ as $n \to \infty$. Finally, note that part (c) can be proved by replacing $n$ with the subsequence $p_n$ throughout the arguments given above. Then the proof is complete.
\end{proof}

\bigskip

\begin{proof}[Proof of Lemma \ref{lem:pd_variance}]
First, we show the inequality \eqref{eq:bdd_above_zero}.
Note that $S$ is a sum of three positive semi-definite matrices. 
Define
\[
    \bar{S} \equiv [\partial_{\beta_1} g(\lambda)] \Sigma_{\beta_1} [\partial_{\beta_1'} g(\lambda)]
    + [\partial_{\beta_0} g(\lambda)] \Sigma_{\beta_0} [\partial_{\beta_0'} g(\lambda)].
\]
Given the decomposition $S = \bar{S} + \partial_p\hat{g}(\lambda)\hat{\Sigma}_p\partial_{p'}\hat{g}(\lambda)$, the eigenvalue inequality\footnote{This inequality holds by noting that $\lambda_{\min}(P+Q) = \min_{\|x\| = 1}x'(P+Q)x \geq \min_{\|x\| = 1}x'Px + \min_{\|x\| = 1}x'Qx = \lambda_{\min}(P) + \lambda_{\min}(Q)$.} 
\[
    \lambda_{\min}(P+Q) \ge \lambda_{\min}(P) + \lambda_{\min}(Q), \quad \text{ for positive semi-definite } P,Q 
\]
implies $\lambda_{\min}(S) \ge \lambda_{\min}(\bar{S})$, since $\lambda_{\min}(\partial_p\hat{g}(\lambda)\hat{\Sigma}_p\partial_{p'}\hat{g}(\lambda)) \ge 0$. Hence, to bound $\lambda_{\min}(S)$ away from zero, it suffices to show that the minimal eigenvalue of $\bar{S}$ is bounded away from zero.

Next, note that $\Sigma_{\beta_d}$ is a diagonal matrix with positive diagonal elements $\Sigma_{\beta_d}[\ell,\ell] = \frac{\sigma^2_{d\ell}}{q(d,z_\ell)}$ for $d = 0,1$ and $\ell = 0,1,\ldots,K$. Let $\nu_d$ denote the first column of $[\partial_{\beta_d}g(\lambda)]$, whose $k$-th row ($1\leq k \leq K$) of its sample analog version is defined in equation \eqref{eq:fixed_weight_expansion}, then it follows that 
\[
    [\partial_{\beta_1}g(\lambda)] \Sigma_{\beta_1} [\partial_{\beta_1'}g(\lambda)]
    =
    \Sigma_{\beta_1}[0,0]\cdot \nu_1'\nu_1 + 
    \diag
    \begin{pmatrix}
        \Sigma_{\beta_1}[1,1]([1 - p(z_0)]c_\mu + 2c_\rho)^2 \\
        \vdots \\
        \Sigma_{\beta_1}[K,K]([1 - p(z_0)]c_\mu + 2c_\rho)^2
    \end{pmatrix}
\]
and 
\[
    [\partial_{\beta_0}g(\lambda)] \Sigma_{\beta_0} [\partial_{\beta_0'}g(\lambda)]
    =
    \Sigma_{\beta_0}[0,0]\cdot \nu_0'\nu_0 + 
    \diag
    \begin{pmatrix}
        \Sigma_{\beta_0}[1,1](- p(z_0)c_\mu + 2c_\rho)^2 \\
        \vdots \\
        \Sigma_{\beta_0}[K,K](- p(z_0)c_\mu + 2c_\rho)^2
    \end{pmatrix}.
\]
Hence we have
\begin{align*}
    \bar{S} 
    &= \left(\Sigma_{\beta_1}[0,0] \cdot \nu_1'\nu_1 + \Sigma_{\beta_0}[0,0] \cdot \nu_0'\nu_0\right) \\
    &\quad + ~ 
    \diag
    \begin{pmatrix}
        \Sigma_{\beta_1}[1,1]([1 - p(z_0)]c_\mu + 2c_\rho)^2 + \Sigma_{\beta_0}[1,1](- p(z_0)c_\mu + 2c_\rho)^2 \\
        \vdots \\
        \Sigma_{\beta_1}[K,K]([1 - p(z_0)]c_\mu + 2c_\rho)^2 +
        \Sigma_{\beta_0}[K,K](- p(z_0)c_\mu + 2c_\rho)^2
    \end{pmatrix}
\end{align*}
We now show that each diagonal element in the second term has a positive lower bound that does not depend on the distribution of data. For each $\ell = 0,1,\ldots,K$, we have $\Sigma_{\beta_d}[\ell,\ell] \geq \epsilon/(1-\epsilon)^2 \equiv \underline{\delta}(\epsilon) > 0$ implied by the parameter space restriction in section \ref{sec:parameter_space}. This implies the $k$-th diagonal element is bounded below by
\begin{align*}
    & ~ \Sigma_{\beta_1}[k,k]([1 - p(z_0)]c_\mu + 2c_\rho)^2 + \Sigma_{\beta_0}[k,k](- p(z_0)c_\mu + 2c_\rho)^2 \\
    & > \underline{\delta}(\epsilon) \underbrace{\left(([1-p(z_0)]c_\mu + 2c_\rho)^2 + (-p(z_0)c_\mu + 2c_\rho)^2\right)}_{\text{LB}} 
\end{align*}
On the one hand, 
\begin{equation}
\label{eq:lower_bd_1}
    \text{LB} \geq \frac{1}{2} c_\mu^2 
\end{equation}
by Cauchy-Schwarz inequality: $a^2 + b^2 \geq \frac{1}{2}(a-b)^2$. On the other hand,
\begin{equation}
\label{eq:lower_bd_2}
\begin{aligned}
    \text{LB} 
    &= 8c_\rho^2+ \left[p(z_0)^2 + (1-p(z_0))^2\right]c_\mu^2 + 4c_\rho c_\mu[1-2p(z_0)] \\
    &\geq 8c_\rho^2+ \frac{1}{2} c_\mu^2 + 4c_\rho c_\mu[1-2p(z_0)] \\
    &= \left(2\sqrt{2}(1-2p(z_0))c_\rho + \frac{1}{\sqrt{2}}c_\mu\right)^2 + 8\left[1-(1-2p(z_0))^2\right]c_\rho^2 \\
    &\geq 32 p(z_0)\left(1 - p(z_0)\right)c_\rho^2 \\
    &\geq 32\epsilon(1-\epsilon) c_\rho^2,
\end{aligned}
\end{equation}
where the second line holds by Cauchy-Schwarz inequality. 
Combining inequalities \eqref{eq:lower_bd_1} and \eqref{eq:lower_bd_2} then yields
\[
    \text{LB} \geq \max\left\{\frac{1}{2}c_\mu^2, ~~32\epsilon(1-\epsilon)c_\rho^2\right\} > 0
\]
since $c_\mu \neq 0$ or $c_\rho \neq 0$. This implies that
\begin{align*}
\lambda_{\min}(\bar{S}) 
    &\geq 
    \min
    \begin{pmatrix}
        \Sigma_{\beta_1}[1,1]([1 - p(z_0)]c_\mu + 2c_\rho)^2 + \Sigma_{\beta_0}[1,1](- p(z_0)c_\mu + 2c_\rho)^2 \\
        \vdots \\
        \Sigma_{\beta_1}[K,K]([1 - p(z_0)]c_\mu + 2c_\rho)^2 +
        \Sigma_{\beta_0}[K,K](- p(z_0)c_\mu + 2c_\rho)^2
    \end{pmatrix} \\
    &>
    \underline{\delta}(\epsilon) \max\left\{\frac{1}{2}c_\mu^2, ~~32\epsilon(1-\epsilon)c_\rho^2\right\},    
\end{align*}
where the first line uses eigenvalue inequality. 
This establishes the inequality \eqref{eq:bdd_above_zero}, as desired. 

Next, we show the inequality \eqref{eq:bdd_below_inf}. Following the eigenvalue inequality
\[
    \lambda_{\max}(P+Q) \leq \lambda_{\max}(P) + \lambda_{\max}(Q)
\]
given that $A$ and $B$ are both positive semi-definite, it suffices to show that
\[
    \lambda_{\max} \left([\partial_p g(\lambda)]\Sigma_p[\partial_{p'} g(\lambda)]\right), ~~
    \lambda_{\max} \left([\partial_{\beta_1} g(\lambda)]\Sigma_{\beta_1}[\partial_{\beta_1} g(\lambda)]\right), 
    ~~ \text{and} ~~
    \lambda_{\max} \left([\partial_{\beta_0} g(\lambda)]\Sigma_{\beta_0}[\partial_{\beta_0} g(\lambda)]\right)
\]
are bounded above by some positive constant that does not depend on the distribution of data. For simplicity, we establish this statement for $\lambda_{\max} \left([\partial_p g(\lambda)]\Sigma_p[\partial_{p'} g(\lambda)]\right)$, and the remaining cases follow by similar arguments. Let $\nu_p$ denote the first column of $[\partial_p g(\lambda)]$. For each $\ell = 0,1,\ldots,K$, we have $\Sigma_{p}[\ell,\ell] = \frac{p(z_\ell)(1-p(z_\ell))}{q(z_\ell)} \leq \frac{1}{4\epsilon} \equiv \bar{\delta}(\epsilon) < \infty$. Then it follows that
\begin{align*}
    \lambda_{\max}([\partial_p g(\lambda)] \Sigma_p [\partial_p' g(\lambda)])
    ~~&\leq~~
    \Sigma_p[0,0] \nu_p'\nu_p + 
    \left(\max_{1\leq k \leq K}\Sigma_p[k,k]\right)|\lambda - (\beta_{10} - \beta_{00}) c_\mu|^2 \\
    ~~&<~~
    \bar{\delta}(\epsilon)\sum_{\ell = 0}^K |-\lambda + (\beta_{1\ell} - \beta_{0\ell})c_\mu|^2.
\end{align*}
The parameter space $\mathcal{P}$ requires $\lambda$ and $\beta$ to be constrained by a compact space that depends only through $\Theta$ and $\zeta > 0$ in Definition \ref{def:parameter_space}, respectively. Thus the right-hand side can be bounded above by a positive constant that is independent of data distribution, which completes the proof.
\end{proof}

\bigskip

\begin{proof}[Proof of Lemma \ref{lem:NS_cond}]
Let $\mathbb{H} = [\partial_p \pi] \Sigma_{p,\infty} \partial_{p'}g(\lambda_\infty) \in \R^{K\times K}$. We divide the proof into two cases: 

\medskip
\textbf{Case 1:} $\rank(\mathbb{H}) > 0$. 
\medskip

We show a stronger conclusion: $h + \mathbb{H}S_{\infty}^{-1/2} \eta^* \neq 0_{K\times 1}$ a.s. for all $h\in\R^K$ and prove it by contradiction. Suppose there exists $h_0 \in \R^K$ such that $\Prob(h_0 + \mathbb{H}S_{\infty}^{-1/2}\eta^* = 0_{K\times 1}) > 0$. By the condition that $\rank(\mathbb{H}) > 0$, there exists a vector $a \in \R^K$ such that $\mathbb{H}'a \neq 0_{K\times 1}$. Therefore, the variance of the normal random variable $a'(h_0 + \mathbb{H} S_{\infty}^{-1/2} \eta^*)$ is nonzero:
\[
    a'\mathbb{H} S_{\infty}^{-1} \mathbb{H}' a \neq 0.
\]
This shows $a'(h_0 + \mathbb{H} S_{\infty}^{-1/2} \eta^*)$ is a nondegenerate normal random variable, which implies 
\[
    \Prob(a'(h_0 + \mathbb{H} S_{\infty}^{-1/2} \eta^*) = 0) = 0.
\]
However, this contradicts with the claim that $\Prob(h_0 + \mathbb{H} S_{\infty}^{-1/2} \eta^* = 0_{K\times 1}) > 0$ since 
\[
    \Prob(h_0 + \mathbb{H} S_{\infty}^{-1/2} \eta^* = 0_{K\times 1}) ~~\leq~~ \Prob(a'(h_0 + \mathbb{H} S_{\infty}^{-1/2} \eta^*) = 0) ~~=~~0.
\]
So the desired conclusion is established.

\medskip
\textbf{Case 2:} $\rank(\mathbb{H}) = 0$.
\medskip

Note that $\rank(\mathbb{H}) = 0$ gives $\mathbb{H} = 0_{K\times K}$, which implies
\begin{align*}
    \mathcal{Z}_h + \mathbb{H} S_\infty^{-1/2}\eta^*   
    &= \mathcal{Z}_h \\
    &= s_\infty\iota_\infty + [\partial_p \pi] \mathcal{Z}_p - \mathbb{H} S_\infty^{-1} \mathcal{Z}_{\mathfrak{g}} \\
    &= s_\infty\iota_\infty + [\partial_p \pi] \mathcal{Z}_p.
\end{align*}
Note that the variance of the random vector $s_\infty\iota_\infty + [\partial_p \pi] \mathcal{Z}_p$ equals a full-rank matrix $[\partial_p \pi] \Sigma_{p,\infty} [\partial_{p'} \pi]$ since $\partial_p \pi$ and $\Sigma_{p,\infty}$ have full rank. This implies $a'( s_\infty\iota_\infty + [\partial_p \pi] \mathcal{Z}_p)$ is a non-degenerate normal random variable for all $a\neq 0_{K\times 1}$. Following similar arguments in Case 1, we conclude that $s_\infty\iota_\infty + [\partial_p \pi] \mathcal{Z}_p \neq 0_{K\times 1}$ almost surely. Thus the desired conclusion has been established.
\end{proof}

\bigskip

\begin{proof}[Proof of Lemma \ref{lem:cond_stat}]
This lemma is a simplified version of Lemma 10.3(d) in \citet{andrews/guggenberger:2017}. For completeness of exposition, I provide a proof below that closely follows their arguments. For notational simplicity, denote $d_\theta \equiv 2(M+1)$ and $d_m = 2(K+1)$, which represent the numbers of columns and rows of matrix $\hat{D}(\theta)$, respectively. Let $\Psi_q$ denote the upper $q$-block of $\Psi_{F_n}$ defined in equation \eqref{eq:defn_S}. Then
	\begin{equation}
	 	\Psi_q = \diag\{(\sqrt{n}\tau_{1,F_n})^{-1},\ldots, (\sqrt{n}\tau_{q,F_n})^{-1}\} 
		\quad\text{and}\quad
		\Psi_{F_n} = 
		\begin{bmatrix}
		\Psi_q & 0_{q\times (d_\theta - q)} \\
		0_{(d_\theta - q)\times q} & I_{d_\theta - q}
		\end{bmatrix}.
	\end{equation}
	Let $B_{F_n} \equiv (B_q ~\vdots~ B_{d_\theta - q})$ and $C_{F_n} \equiv (C_q ~\vdots~ C_{d_m - q})$.
	Then 
	\[
	\sqrt{n}\hat{D}(\theta_{F_n}) B_{F_n} \Psi_{F_n}
	 = (\sqrt{n} \hat{D}(\theta_{F_n}) B_q \Psi_q, \sqrt{n} \hat{D}(\theta_{F_n}) B_{d_\theta-q}).
	\]
	Now we analyze the asymptotic behavior of each part. For the first part, note that 
	\begin{align*}
		\sqrt{n}\hat{D}(\theta_{F_n}) B_q \Psi_q 
		& = \sqrt{n}(\hat{D}(\theta_{F_n}) - A_{F_n})B_q \Psi_q +  A_{F_n} B_q (\sqrt{n}\Psi_q) \\
		& = o_p(1) + C_{F_n}\Pi_{F_n} (B_q, B_{d_\theta - q})' B_q (\sqrt{n}\Psi_q) \\
		& = o_p(1) + C_{F_n}\Pi_{F_n} (I_q, 0_{q \times (d_\theta - q)})' (\sqrt{n}\Psi_q) \\
		& = o_p(1) + C_{F_n}(I_q, 0_{q\times (d_m - q)})' \\
		& \xrightarrow{p} C_{q,\infty}
	\end{align*}
	where $C_{q,\infty}$ denotes the first $q$-columns of matrix $C_\infty$. The second line holds by $\sqrt{n}(\hat{D}(\theta_{F_n}) - A_{F_n}) = O_p(1)$, $\Psi_q = o(1)$, and recalling that $A_{F_n} = C_{F_n} \Pi_{F_n} B_{F_n}'$ via singular value decomposition. The third line holds by the construction that $B_{F_n}$ is an orthogonal matrix. 

	For the second part, note that 
	\begin{align*}
		\sqrt{n} A_{F_n} B_{d_\theta - q} 
		&= 
		\sqrt{n} C_{F_n} \Pi_{F_n} (B_q, B_{d_\theta - q})' B_{d_\theta - q} \\
		&=
		C_{F_n} (\sqrt{n}\Pi_{F_n}) (0_{(d_\theta - q)\times q}, I_{d_\theta - q})' \\
		&\to
		C_\infty
		\begin{bmatrix}
			0_{q\times(d_\theta - q)} \\
			\diag\{t_{q+1}, \ldots, t_{d_\theta}\} \\
			0_{(d_m-d_\theta)\times (d_\theta - q) }
		\end{bmatrix}.
	\end{align*}
	On the other hand, by continuous mapping theorem,
	\[
		\sqrt{n} (\hat{D}(\theta_{F_n}) - A_{F_n}) B_{d_\theta - q} ~~\xrightarrow{d}~~ \mathcal{Z}_D
		B_{d_\theta - q, \infty}
	\]
	where $B_{d_\theta-q,\infty}$ denotes the last $(d_\theta - q)$-columns of $B_\infty$.
	Therefore, we have
	\[
		\sqrt{n}\hat{D}(\theta_{F_n})B_{d_\theta - q} ~~\xrightarrow{d}~~ C_\infty 
		\begin{bmatrix}
			0_{q\times(d_\theta - q)} \\
			\diag\{t_{q+1}, \ldots, t_{d_\theta}\} \\
			0_{(d_m-d_\theta)\times (d_\theta - q) }
		\end{bmatrix} 
		+
		\mathcal{Z}_D
		B_{d_\theta - q, \infty}
	\]
	Combine the above results, we have 
	\[
		\sqrt{n}\hat{D}(\theta_{F_n})B_{F_n}\Psi_{F_n} ~~\xrightarrow{d}~~ \mathcal{D} \equiv \left(C_{q,\infty}, \quad C_\infty 
		\begin{bmatrix}
			0_{q\times(d_\theta - q)} \\
			\diag\{t_{q+1}, \ldots, t_{d_\theta}\} \\
			0_{(d_m-d_\theta)\times (d_\theta - q) }
		\end{bmatrix} 
		+
		\mathcal{Z}_D
		B_{d_\theta - q, \infty}\right) 
	\]
	whose stochastic behavior only depends on $\mathcal{Z}_D$, thus is independent of $\mathcal{Z}_{\mathfrak{m}}$.
\end{proof}

\bigskip

\begin{proof}[Proof of Lemma \ref{lem:new_cond_stat}]
	For notational simplicity, we still maintain the notation $d_\theta = 2(M+1)$ and $d_m = 2(K+1)$ from Lemma \ref{lem:cond_stat}, which shows  
	\[
		n^{1/2} \hat{D}(\theta_{F_n}) B_{F_n} \Psi_{F_n} ~~\xrightarrow{d}~~ \mathcal{D},	
	\]
	where $\mathcal{D}$ is independent of $\mathcal{Z}_{\mathfrak{m}}$.

	When $q = d_\theta$, we have
	\begin{align*}
		n^{1/2} \kappa n^{-1/2} \xi B_{F_n} \Psi_{F_n} 
		&~=~ \kappa \xi B_{F_n} \Psi_{F_n} \\
		&~\xrightarrow{p}~ \kappa \xi B_{\infty} 0_{d_\theta \times d_\theta} \\
		&~=~ 0_{d_m \times d_\theta}.
	\end{align*}
	Then 
	\[
		n^{1/2} (\hat{D}(\theta_{F_n}) + \kappa n^{-1/2}\xi) B_{F_n} \Psi_{F_n} 
		~~\xrightarrow{p}~~
		C_\infty	
	\]
	which is the same limit as the $n^{1/2} \hat{D}(\theta_{F_n}) B_{F_n} \Psi_{F_n}$. Thus the introduction of this small perturbation $\xi$ would not harm the asymptotic behavior of the test under strong identification.

	When $q < d_\theta$, we have 
	\begin{align*}
		n^{1/2} \kappa n^{1/2} \xi B_{F_n} \Psi_{F_n}
		&~=~ \kappa \xi B_{F_n} \Psi_{F_n} \\
		&~\xrightarrow{p}~ \kappa \xi B_{\infty} \diag(0_{q\times q}, I_{d_\theta - q}) \\
		&~=~ (0_{q\times q}, \quad \kappa \xi B_{d_\theta - q,\infty}).
	\end{align*}

	Combining with the results from Lemma \ref{lem:cond_stat}, we have
	\[
		n^{1/2} (\hat{D}(\theta_{F_n}) + \kappa n^{-1/2}\xi) B_{F_n} \Psi_{F_n}
		~~\xrightarrow{d}~~
        \mathcal{D}_\xi \equiv
		\left(C_{q,\infty}, \quad C_\infty 
		\begin{bmatrix}
			0_{q\times(d_\theta - q)} \\
			\diag\{t_{q+1}, \ldots, t_{d_\theta}\} \\
			0_{(d_m-d_\theta)\times (d_\theta - q) }
		\end{bmatrix} 
		+
		 (\mathcal{Z}_D + \kappa\xi)
		B_{d_\theta - q, \infty}\right).
	\]
	Now we show that $\mathcal{D}_\xi$ has full rank of probability one. Define 
	\[
		C_{d_\theta - q, \xi} \equiv C_\infty 
		\begin{bmatrix}
			0_{q\times(d_\theta - q)} \\
			\diag\{t_{q+1}, \ldots, t_{d_\theta}\} \\
			0_{(d_m-d_\theta)\times (d_\theta - q) }
		\end{bmatrix} 
		+
		(\mathcal{Z}_D + \kappa\xi)
		B_{d_\theta - q, \infty}.
	\]
	Since $\mathcal{D}_\xi$ has full rank is equivalent to $C_{\infty}'\mathcal{D}_\xi$ has full rank, which equals
	\begin{align*}
		C_{\infty}'\mathcal{D}_\xi
		&=
		\begin{bmatrix}
			C_{q,\infty}' \\
			C_{d_m - q, \infty}'
		\end{bmatrix}
		(C_{q,\infty}, C_{d_\theta - q, \xi}) \\
		&=
		\begin{bmatrix}
			C_{q,\infty}' C_{q, \infty} & C_{q,\infty}'C_{d_\theta - q, \xi} \\
			C_{d_m-q,\infty}' C_{q,\infty} & C_{d_m-q,\infty}'C_{d_\theta - q, \xi}
		\end{bmatrix} \\
		&=
		\begin{bmatrix}
			I_q & C_{q,\infty}'C_{d_\theta - q, \xi} \\
			0_{(d_m-q)\times q} & C_{d_m-q,\infty}'C_{d_\theta - q, \xi}
		\end{bmatrix}.
	\end{align*}
	Hence, $\mathcal{D}_\xi$ has full rank is equivalent to $C_{d_m-q,\infty}'C_{d_\theta - q, \xi}$ having full rank. By Lemma 16.1 from \cite{andrews/guggenberger:2017}, it suffices to show that the variance of $\operatorname{vec}(C_{d_m-q,\infty}'C_{d_\theta - q, \xi})$ is positive definite. 
	
	Note that 
	\begin{align*}
		\var(\operatorname*{vec}(C_{d_m-q,\infty}'C_{d_\theta - q, \xi}))
		&= \var(\operatorname*{vec}(C_{d_m-q,\infty}'(\mathcal{Z}_D + \kappa\xi)B_{d_\theta - q, \infty})) \\
		&= \var\left((B_{d_\theta - q, \infty}' \otimes C_{d_m-q,\infty}')\operatorname*{vec}(\mathcal{Z}_D + \kappa\xi)\right) \\
		&= (B_{d_\theta - q, \infty}' \otimes C_{d_m-q,\infty}') \var\left(\operatorname*{vec}(\mathcal{Z}_D + \kappa\xi)\right) (B_{d_\theta - q, \infty}' \otimes C_{d_m-q,\infty}')'.
	\end{align*}
	By the property that $\rank(A\otimes B) = \rank(A)\rank(B)$, we have 
	\begin{align*}
		\rank(B_{d_\theta - q, \infty} \otimes C_{d_m-q,\infty}) 
		&= \rank(B_{d_\theta - q, \infty}) \rank(C_{d_m-q,\infty})  \\
		&= (d_\theta - q)(d_m - q),
	\end{align*}
	implying that $B_{d_\theta - q, \infty} \otimes C_{d_m-q,\infty}$ is a full-rank matrix. 
	On the other hand, 
	\begin{align*}
		\var\left(\operatorname*{vec}(\mathcal{Z}_D + \kappa \xi)\right)
		&= \var\left(\operatorname*{vec}(\mathcal{Z}_D)\right) + \kappa^2 \var\left(\operatorname*{vec}(\xi)\right)
	\end{align*}
	with the equality holding by the independence between $\xi$ and data. Note that $\var(\operatorname*{vec}(\xi)) = I_{d_md_\theta}$ by the assumption that $\xi$ is a matrix of i.i.d. standard normal variables.  Therefore, $\var\left(\operatorname*{vec}(\mathcal{Z}_D + \kappa\xi)\right)$ is the sum of two positive semi-definite matrices, with one of them being positive definite, so we conclude that it is positive definite.

	For every $x \neq 0_{(d_m-q)(d_\theta - q)\times 1}$, we have $(B_{d_\theta - q, \infty} \otimes C_{d_m-q,\infty})x \neq 0_{(d_\theta-q)(d_m-q)\times 1}$ by the full rank of $B_{d_\theta - q, \infty} \otimes C_{d_m-q,\infty}$, and thus 
	\[
		x'(B_{d_\theta - q, \infty}' \otimes C_{d_m-q,\infty}') \var(\operatorname*{vec}(\mathcal{Z}_D + \kappa\xi)) (B_{d_\theta - q, \infty}' \otimes C_{d_m-q,\infty}')'x > 0
	\]
	by the positive definiteness of $\var(\operatorname*{vec}(\mathcal{Z}_D + \kappa\xi))$. From this, we conclude that the covariance matrix of $\operatorname{vec}(C_{d_m-q,\infty}'C_{d_\theta - q, \xi})$ is positive definite. The proof is therefore complete.
\end{proof}

\begin{proof}[Proof of Lemma \ref{lem:cr_ineq}]
The proof of $c_r$ inequality can be found in various reference, for example, see \citet[Proposition 3.8]{white:1999}.
\end{proof}

\begin{proof}[Proof of Lemma \ref{lem:wlln_array}]
	For each $\epsilon > 0$, note that 
\begin{align*}
	\Prob\left(\left|\frac{1}{n} \sum_{i=1}^n X_{ni} - \Exp[X_{ni}]\right| > \epsilon \right)
	&\leq \frac{\Exp|n^{-1}\sum_{i=1}^n (X_{ni} - \Exp X_{ni})|^{1+\delta}}{\epsilon^{1+\delta}} \\
	&\leq \frac{\sum_{i=1}^n \Exp|X_{ni} - EX_{ni}|^{1+\delta}}{n^{1+\delta}\epsilon^{1+\delta}} \\
	&\leq \frac{2^{\delta} \left(\Exp|X_{ni}|^{1+\delta} + |\Exp X_{ni}|^{1+\delta}\right)}{n^{\delta} \epsilon^{1+\delta}} \\
	&\leq \frac{2^{1+\delta} \sup_{n}\Exp |X_{n,i}|^{1+\delta}}{n^{\delta}\epsilon^{1+\delta}} \\
	&\to 0
\end{align*}
as $n\to \infty$. The first line holds by Markov's inequality. The second line holds by triangular inequality. The third line holds by $c_r$ inequality. The fourth line holds by absolute inequality. Hence the desired conclusion holds.
\end{proof}

\begin{proof}[Proof of Lemma \ref{lem:clt_array}]
This central limit theorem for a triangular array of random vectors is a special case of heterogeneous CLT in \citet[Theorem 9.4]{hansen:2022} by restricting observations in each row are identically distributed.
\end{proof}

\newpage

\section{Power Analysis of MLC Tests}
\label{appendix:power_analysis_MLC}
\subsection{Main results}
In this section, I analyze the power of MLC test under strong identification. To define the sequences of DGPs that are strongly identified, I impose the following additional restrictions on the space $\mathcal{P}$. 

\begin{definition}[Parameter space under strong identification]
\label{def:strong_id_para_space}
For some $\delta, \zeta > 0$ and $\epsilon \in (0,1/2)$, define the set $\mathcal{P}_s$ of pairs $(\theta, F)$ that satisfy the following conditions:
\begin{enumerate}
    \item $(\theta, F) \in \mathcal{P}$. That is, $(\theta,F)$ satisfies the conditions imposed in Definition \ref{def:parameter_space} with the given $\delta, \zeta > 0$, and $\epsilon \in (0,1/2)$.
    \item Equation \eqref{eq:mateq} holds with functions $\{h_m(\cdot)\}_{m=1}^M$ satisfying Assumption \ref{asp:identification}.6.
    \item There exists a set of index $\mathcal{S} \subseteq \{0,1,\ldots,K\}$ with $|\mathcal{S}| = M+1$ such that 
    \[
        \min_{j,k\in\mathcal{S}, j \neq k} |p(z_j) - p(z_k)| \geq \epsilon.
    \]
\end{enumerate}
\end{definition}

Specifically, the second condition brings back the unisolvent property on the specified functions, and the third condition assumes that there is a set of isolated propensity scores sufficient to point identify primitive parameters in the MTE model. 

\begin{lemma}
\label{lem:full_rank_A}
With $\mathcal{P}_s$ defined in Definition \ref{def:strong_id_para_space}, We have 
\[
    \inf_{(\theta,F)\in \mathcal{P}_s} \lambda_{\min}(A_F) > 0.
\]
where $A_F$ is defined below equation \eqref{eq:linsystem} for some distribution $F$.
\end{lemma}

The parameter space $\mathcal{P}_s$ guarantees that the smallest singular value of $A_F$, denoted as $\tau_{F,2(M+1)}$, is uniformly bounded away from zero. Since $A_F$ has full rank for each $(\theta,F)\in\mathcal{P}_s$, $\theta = (A_F'A_{F})^{-1}A_{F}'\beta_F$ is uniquely determined by the data distribution $F$, so does $\lambda = c'\theta$. Therefore, $\mathcal{P}_s$ can be regarded as a collection of distribution $F$, and we denote $\theta_F \equiv (A_F'A_{F})^{-1}A_{F}'\beta_F$ and $\lambda_F \equiv c'\theta_F$ for each $F \in \mathcal{P}_s$. The next result establishes the consistency and local power property of the MLC test under strong identification:
\begin{proposition}
\label{prop:local_power_MLC}
Suppose $\{F_n: n \geq 1\} \subseteq \mathcal{P}_s$, where for some $F_0 \in \mathcal{P}_s$, assume 
\begin{enumerate}
    \item $p_{F_n}(z_\ell) \to p_{F_0}(z_\ell)$ for all $\ell = 0,1,\ldots,K$,
    \item $\beta_{F_n} \to \beta_{F_0}$,
    \item $\sigma^2_{F_n,d\ell} \to \sigma^2_{F_0,d\ell}$ for all $d = 0,1$ and $\ell = 0,1,\ldots,K$,
    \item $q_{F_n}(z_\ell) \to q_{F_0}(z_\ell)$ for all $\ell = 0,1,\ldots,K$.
\end{enumerate}
Consider a sequence of null values $\lambda_n^* = \lambda_{F_n} + bn^{-r}$ with $b \in \R$ and $b \neq 0$, then the following conclusion holds
\begin{enumerate}
    \item[(a)] If $r \in [0,1/2)$, then we have 
    \begin{equation}
    \label{eq:MLC_consistency}
        \lim_{n\to\infty} \Prob_{F_n}(\hat{\phi}_{\text{MLC}}(\lambda_n^*) = 1) = 1.
    \end{equation}
    \item[(b)] If $r = 1/2$, then we have
    \[
        \lim_{a\searrow 0}\limsup_{n\to\infty}\Exp_{F_n}\left[\hat{\phi}_{\text{Wald}}(\lambda_n^*) - \hat{\phi}_{\text{MLC}}(\lambda_n^*)\right] = 0.
    \]
\end{enumerate}
\end{proposition}

It is worth noting that the proposed MLC test is approximately as powerful as the asymptotic efficient Wald test under strong identification if we set the weight assigned to AR statistic sufficiently small. \citet[Theorem 3]{andrews.i:2018} establishes a related result by showing that a $(1-\alpha-\gamma)$ LC confidence set is contained by a $(1-\alpha)$ Wald confidence set with probability approaching one for all $\gamma > 0$ based on a specific choice of weight function $a(\gamma)$ such that the critical values of the two tests are equal:  $q_{(1+a(\gamma))\chi_1^2 + a(\gamma)\chi_{2K+1}^2}(1-\alpha-\gamma) = q_{\chi_1^2}(1-\alpha)$. Nevertheless, it is less clear from his result that the MLC test would have similar local power as the Wald test at the \textit{same significance level, regardless of how we choose the AR weight}. This result is now formally established in Proposition \ref{prop:local_power_MLC}(b).

It is not recommended to set $a = 0$ for the MLC test. On the one hand, the AR statistic helps direct the optimizer of the profiled MLC statistic to converge within a small neighborhood of the true parameter $\theta$ under strong identification. Within this neighborhood, the MRLM statistic is first-order equivalent to the Wald statistic. However, the MRLM statistic alone cannot detect deviations from the true parameter except in the direction of the target parameter $c'\theta$. On the other hand, assigning a larger weight to the AR statistic increases power under weak identification, as the MRLM statistic might be small for distant alternatives due to the near-singularity of $\hat{A}$ \citep{kleibergen:2005}. While \cite{andrews:2016} discusses the optimal choice of $a$ for full vector inference problems, the optimal choice of $a$ for  inference of parameter functions remains an open question for future research.

\subsection{Lemmas for local power analysis}
The following lemmas are used to establish Proposition \ref{prop:local_power_MLC}.

\begin{lemma}
\label{lem:wald_local_power}
    Under the same assumptions in Proposition \ref{prop:local_power_MLC}, the local power of Wald statistic equals:
    \[
        \lim_{n\to\infty} \Exp_{F_n}[\hat{\phi}_{\text{Wald}}(\lambda_n^*)] = \Prob((\mathcal{Z} + \nu)'P_\nu (\mathcal{Z} + \nu) > q_{\chi_1^2}(1-\alpha)).
    \]
    where $\mathcal{Z} \sim \normal(0_{2(K+1)\times 1}, I_{2(K+1)})$ and $\nu$ is defined in equation \eqref{eq:defn_nu_vec}.
\end{lemma}

For each population distribution $F$, let 
\begin{align*}
    & \Omega_{F}(\theta) \equiv H(p_F, \theta) \Sigma_{p,F} H(p_F,\theta) \\ %
    & \Gamma_{j,F}(\theta) \equiv M_j(p_F)\Sigma_{p,F} H(p_F,\theta)', %
\end{align*}
where $H(p,\theta)$ and $M_j(p)$ are defined in section \ref{sec:construct_MLC}.
\begin{lemma}
\label{lem:consistent_cov_strong_id}
    Suppose $\{F_n\}_{n\geq 1} \subseteq \mathcal{P}_s$ is a sequence of distributions satisfying the conditions in Proposition \ref{prop:local_power_MLC}. The following results hold:
    \begin{align}
    & 0 < \inf_{\theta \in\Theta} \lambda_{\min}(\Omega_{F_0}(\theta)) \leq 
    \sup_{\theta \in\Theta}\lambda_{\max}(\Omega_{F_0}(\theta)) < \infty \label{eq:lem_locpower_regular1} \\
    & \sup_{\theta \in\Theta}\|\hat{\Omega}(\theta) - \Omega_{F_0}(\theta)\| = o_p(1) 
    \label{eq:lem_locpower_regular2} \\
    & \sup_{\theta \in\Theta}\|\hat{\Gamma}_j(\theta) - \Gamma_{j,F_0}(\theta)\| = o_p(1).
    \label{eq:lem_locpower_regular3}
    \end{align}
\end{lemma}

\subsection{Proofs in Appendix \ref{appendix:power_analysis_MLC}}

\begin{proof}[Proof of Lemma \ref{lem:full_rank_A}]
    Since the matrix $A_F$ is the block diagonal of $A_{1F}$ and $A_{0F}$ defined in \eqref{eq:mateq}, it suffices to show that the minimum singular value of $A_{dF}$ is bounded away from zero uniformly over $\mathcal{P}_s$ for both $d = 0,1$.  Fix an arbitrary $d = 0,1$, the proof is divided into three steps.

    \underline{Part (a)}: Establish the positive lower bound of the determinant of an alternant matrix over a compact subset. \par
    Suppose $\{h_m(\cdot)\}_{m=1}^M$ satisfies the Assumption \ref{asp:identification}.6. Define the function 
    \[
        f(p_0, \ldots, p_{M}) \equiv \lambda_{\min}(\mathbb{A}(p_0, \ldots, p_M)) \quad \text{for each }(p_0,\ldots, p_M) \in (0,1)^{M+1},
    \]
     where $\lambda_{\min}$ denotes the smallest singular value of an alternant matrix $\mathbb{A}(p_0, \ldots, p_M) \in \R^{(M+1) \times (M+1)}$ whose $(\ell,m)$'th element is defined as
    \[
        \mathbb{A}_{\ell m}(p_0, \ldots, p_M) \equiv \lambda_{dm}(p_\ell).
    \]
    By Assumption \ref{asp:identification}.6, $\{\lambda_{dm}(\cdot)\}_{m=0}^M$ are continuous functions on $(0,1)$. This implies that $f(p_0, \ldots, p_M)$ is also continuous on $(0,1)^{M+1}$. Consider a compact subset
    \begin{equation}
    \label{eq:compact_subset}
        \mathcal{E} \equiv \left\{(p_0,\ldots,p_M): p_\ell \in [\epsilon, 1-\epsilon] \text{ for all } \ell = 0,\ldots,M,~~ \min_{j \neq k}|p_j - p_k| \geq  \epsilon\right\}.
    \end{equation}
    Then $f$ is strictly positive and continuous on $\mathcal{E}$ by the unisolvent property\footnote{The unisolvent property implies that $\mathbb{A}(p_0,\ldots,p_M)$ has nonzero determinant, whose absolute value equals the product of all singular values of $\mathbb{A}(p_0,\ldots,p_M)$. From this, it implies that the smallest singular value should be positive.} imposed in Assumption \ref{asp:identification}.6. Therefore, there is a natural positive lower bound, denoted by $\epsilon_d$, such that
    \[
        \inf_{(p_0,\ldots, p_M) \in \mathcal{E}} f(p_0, \ldots, p_M) > \epsilon_d > 0.
    \]

    \underline{Part (b)}: Show that the smallest singular value of a submatrix of $A_{dF}$ is uniformly positive. \par
    By the third condition imposed in Definition \ref{def:strong_id_para_space} and the third condition imposed in Definition \ref{def:parameter_space}, for each $(\theta,F)\in\mathcal{P}$, there exists a subset of propensity scores $p_{\mathcal{S}}\equiv \{p_F(z_\ell): \ell \in \mathcal{S} \}$ such that $p_{\mathcal{S}} \in \mathcal{E}$ defined in \eqref{eq:compact_subset}. Then it follows from part (a) that 
    \[
        f(p_\mathcal{S}) =  \lambda_{\min}(\mathbb{A}(p_{\mathcal{S}}))> \epsilon_d > 0.
    \]

    \underline{Part (c)}: Conclude the proof. \par
    Note that $\mathbb{A}(p_\mathcal{S})$ is a submatrix of $A_{dF}$ for row index belonging to the set $\mathcal{S}$. Therefore, 
    \begin{align*}
        \lambda_{\min}(A_{dF}) 
        &= \min_{\|x\| = 1} \|(A_{dF}) x\| \\
        &\geq \min_{\|x\| = 1} \|\mathbb{A}(p_{\mathcal{S}}) x\| \\
        &= \lambda_{\min}(\mathbb{A}(p_{\mathcal{S}})) \\
        &> \epsilon_d > 0
    \end{align*}
    for all $(\theta,F)\in\mathcal{P}_s$. The first and third line holds by the min-max principle for singular values, and the second line follows by the fact that $\mathbb{A}(p_{\mathcal{S}})$ is a submatrix of $A_{dF}$. Hence the desired conclusion has been established.
\end{proof}

\begin{proof}[Proof of Proposition \ref{prop:local_power_MLC}]
    \underline{Part (a)}. First, we show that 
        \begin{equation}
        \label{eq:diverge_distant_alt}
            \inf_{c'\theta = \lambda_n^*} \|\sqrt{n}(\hat{A}\theta - \hat{\beta})\| \to \infty \quad \text{almost surely}.
        \end{equation}
    To prove this, fix an arbitrary $\theta_n^*$ that satisfies $c'\theta_n^* = \lambda_n^*$.\footnote{If such $\theta_n^*$ does not exist, then the test always rejects the null hypothesis by its construction in section \ref{sec:construct_MLC}. So the desired conclusion trivially holds.} Consider the following derivation
    \begin{align}
        \|\sqrt{n}(\hat{A}\theta_n^* - \hat{\beta})\|
        &= \|\sqrt{n}(\hat{A}\theta_{F_n} - \hat{\beta}) + \sqrt{n}\hat{A}(\theta_{n}^* - \theta_{F_n})\| \notag \\
        &\geq \sqrt{n}\|\hat{A}(\theta_n^* - \theta_{F_n})\| - \|\sqrt{n}(\hat{A}\theta_{F_n} - \hat{\beta})\| \notag \\
        &\geq \lambda_{\min}(\hat{A}) \sqrt{n}\|\theta_n^* - \theta_{F_n}\| - \|\sqrt{n}(\hat{A}\theta_{F_n} - \hat{\beta})\| \label{eq:bound_on_diff_theta} \\
        &\geq \lambda_{\min}(\hat{A}) \frac{|c'\theta_{F_n} - \lambda_{n}^*|}{\sqrt{c'c}} - \|\sqrt{n}(\hat{A}\theta_{F_n} - \hat{\beta})\| \notag \\
        &= \lambda_{\min}(A_{F_0}) \frac{|b|n^{-r+1/2}}{\sqrt{c'c}} + O_p(1). \notag
    \end{align}
    The second line follows by triangle inequality. The third line holds by taking $x = \theta_n^* - \theta_{F_n}$ in the following equality: 
    \begin{equation}
    \label{eq:norm_ineq}
        \|\hat{A}x\| \geq \|x\| \inf_{\|x\| \neq 0}\frac{\|\hat{A}x\|}{\|x\|} = \|x\| \lambda_{\min}(\hat{A}),
    \end{equation}
    where the equality in \eqref{eq:norm_ineq} holds by Courant–Fischer–Weyl min-max principle for singular values.
    The fourth line holds by the fact that the distance between $\theta_{F_n}$ and $\theta_{n}^*$ is bounded below by the distance between $\theta_{F_n}$ and the hyperplane $\{\theta: c'\theta = \lambda_n^*\}$, where the latter equals $|c'\theta_{F_n} - \lambda_n^*|/\sqrt{c'c}$. The last line holds by $\lambda_n^* = \lambda_{F_n} + bn^{-r}$, $\hat{A} \xrightarrow{p} A_{F_0}$, and $\sqrt{n}(\hat{A}\theta_{F_n} - \hat{\beta}) = O_p(1)$ as implied by Lemma \ref{lem:convergence_p_beta}. Finally, based the assumption $r \in [0,1/2)$ and $\lambda_{\min}(A_{F_0}) > 0$ from  Lemma \ref{lem:full_rank_A}, the desired intermediate step \eqref{eq:diverge_distant_alt} is established.
    
    Second, with a slight abuse of notation, let $\theta_n^*$ denote the minimizer in the profiled MLC test statistic:
    \begin{equation}
    \label{eq:minimize_MLC}
    \begin{aligned}
        \theta_n^* 
        &\in \argmin_{c'\theta = \lambda_n^*} \text{MLC}_n(\theta) \\
        & = \argmin_{c'\theta = \lambda_n^*} n(\hat{A}\theta - \hat{\beta})' \hat{\Omega}(\theta)^{-1/2} P_{\hat{Q}(\theta)} \hat{\Omega}(\theta)^{-1/2} (\hat{A}\theta - \hat{\beta}) + a\cdot n(\hat{A}\theta - \hat{\beta})'\hat{\Omega}(\theta)^{-1}(\hat{A}\theta - \hat{\beta}).
    \end{aligned}
    \end{equation}
    The first term of the objective function is always non-negative, we show that the second term diverges to infinity. Note that
    \begin{align}
    n(\hat{A}\theta_{n}^* - \hat{\beta})'\hat{\Omega}(\theta_{n}^*)^{-1}(\hat{A}\theta_{n}^* - \hat{\beta}) 
    &= \|\sqrt{n}\hat{\Omega}(\theta_n^*)^{-1/2}(\hat{A}\theta_n^* - \hat{\beta})\|^2  \notag \\
    &\geq \lambda_{\max}(\hat{\Omega}(\theta_n^*))^{-1} \|\sqrt{n}(\hat{A}\theta_n^* - \hat{\beta})\|^2  \notag \\ 
    &\geq \left[\lambda_{\max}(\Omega_{F_0}(\theta_n^*))^{-1} + o_p(1)\right] \|\sqrt{n}(\hat{A}\theta_n^* - \hat{\beta})\|^2 \label{eq:expand_AR_term}
    \end{align}
    The second inequality holds by applying the same arguments in \eqref{eq:norm_ineq}, and the third line holds by equation \eqref{eq:lem_locpower_regular2} from Lemma \ref{lem:consistent_cov_strong_id}. Again following this Lemma, inequality \eqref{eq:lem_locpower_regular1} implies $\lambda_{\max}(\Omega_{F_0}(\theta_n^*))^{-1} > 0$. Combining \eqref{eq:diverge_distant_alt} and \eqref{eq:expand_AR_term} implies that 
    \[
        \inf_{c'\theta = \lambda_n^*} \text{MLC}_n(\theta) = \text{MLC}_n(\theta_n^*) \to \infty \quad \text{almost surely}.
    \]
     As a result, the desired conclusion \eqref{eq:MLC_consistency} holds.

    \underline{Part (b)}. We divide the proof into several steps:
    
    \textbf{Step 1}: There exists a $\bar{\theta}_n$ that satisfies $c'\bar{\theta}_n = \lambda_n^*$ such that $\text{MLC}_n(\bar{\theta}_n)$ converges to a mixture of noncentral $\chi^2$ distributions.

    Consider 
    \[
        \bar{\theta}_n = \theta_{F_n} + \frac{(A_{F_0}'\Omega_{F_0}(\theta_{F_0})^{-1} A_{F_0})^{-1} c}{c'(A_{F_0}'\Omega_{F_0}(\theta_{F_0})^{-1} A_{F_0})^{-1} c} \cdot \frac{b}{\sqrt{n}},
    \]
    Note that $\bar{\theta}_n \in \Theta$ for sufficiently large $n$ since $\theta \in \operatorname{int}(\Theta)$, and 
    \[
        \|\Omega_{F_0}(\theta_{F_0})^{-1/2}A_{F_0}(\bar{\theta}_n - \theta_{F_n})\| = \frac{|b|n^{-1/2}}{\sqrt{c'(A_{F_0}'\Omega_{F_0}(\theta_{F_0})^{-1} A_{F_0})^{-1} c}} = O(n^{-1/2}).
    \]
    So we have $\|\bar{\theta}_n - \theta_{F_0}\| = o(1)$ as $\|\theta_{F_n} - \theta_{F_0}\| = o(1)$ holds by the given conditions. Then the $j$'th column of $\widetilde{D}(\bar{\theta}_n)$, denoted by $\widetilde{D}_j(\bar{\theta}_n)$, equals
    \begin{align*}
        \widetilde{D}_j(\bar{\theta}_n) 
        &= \hat{a}_j - \hat{\Gamma}_j(\bar{\theta}_n) \hat{\Omega} (\bar{\theta}_n)^{-1}(\hat{A}\bar{\theta}_n - \hat{\beta}) + n^{-1/2}\kappa\xi_j\\
        &= \hat{a}_j - {\Gamma}_{j,F_0}(\theta_{F_0}) {\Omega}_{F_0}(\theta_{F_0})^{-1}(\hat{A}\theta_{F_n} - \hat{\beta}) + o_p(1) \\
        &= \hat{a}_j + o_p(1)
    \end{align*}
     for each $j = 1,\ldots,2(M+1)$, where the second line follows by \eqref{eq:lem_locpower_regular3} from Lemma \ref{lem:consistent_cov_strong_id} and the fact that $\|\theta_{F_n} - \theta_{F_0}\| = o(1)$. This implies $\|\widetilde{D}(\bar{\theta}_n) - \hat{A}\| = o_p(1)$. Since $\hat{\Omega}(\bar{\theta}_n)^{-1} = \Omega_{F_0}(\theta_{F_0})^{-1} + o_p(1) = O_p(1)$ given by \eqref{eq:lem_locpower_regular1} and \eqref{eq:lem_locpower_regular2} from Lemma \ref{lem:consistent_cov_strong_id}, we have
     \begin{align*}
        \hat{Q}(\bar{\theta}_n) 
        &\equiv \hat{\Omega}(\bar{\theta}_n)^{-1/2}\widetilde{D}(\bar{\theta}_n)(\widetilde{D}(\bar{\theta}_n)'\hat{\Omega}(\bar{\theta}_n)^{-1}\widetilde{D}(\bar{\theta}_n))^{-1}c \\
        &= \hat{\Omega}(\bar{\theta}_n)^{-1/2}\hat{A}(\hat{A}'\hat{\Omega}(\bar{\theta}_n)^{-1}\hat{A})^{-1}c + o_p(1).
     \end{align*} %
    Plugging this into $\text{MRLM}_n(\bar{\theta}_n)$ yields
    \begin{equation}
    \label{eq:MRLM_local_asymtotic}
    \begin{aligned}
        \text{MRLM}_n(\bar{\theta}_n) 
        &= \frac{\left[\sqrt{n}(\hat{A}\bar{\theta}_n - \hat{\beta})'\hat{\Omega}(\bar{\theta}_n)^{-1}\widetilde{D}(\bar{\theta}_n)\left(\widetilde{D}(\bar{\theta}_n)'\hat{\Omega}(\bar{\theta}_n)^{-1}\widetilde{D}(\bar{\theta}_n)\right)^{-1}c\right]^2}{c'\left(\widetilde{D}(\bar{\theta}_n)'\hat{\Omega}(\bar{\theta}_n)^{-1}\widetilde{D}(\bar{\theta}_n)\right)^{-1}c} \\
        &= \frac{\left[\sqrt{n}(\hat{A}\theta_{F_n} - \hat{\beta})'\hat{\Omega}(\bar{\theta}_n)^{-1}\hat{A}\left(\hat{A}'\hat{\Omega}(\bar{\theta}_n)^{-1}\hat{A}\right)^{-1}c + \sqrt{n}(\bar{\theta}_n - \theta_{F_n})'c + o_p(1)\right]^2}{c'\left(\hat{A}'\hat{\Omega}(\bar{\theta}_n)^{-1}\hat{A}\right)^{-1}c + o_p(1)} \\
        &\xrightarrow{d} \left(\frac{\mathcal{Z}'\nu}{\|\nu\|} + \|\nu\|\right)^2  = (\mathcal{Z} + \nu)' P_\nu (\mathcal{Z} + \nu).
    \end{aligned}
    \end{equation}
    where
    \begin{equation}
    \label{eq:defn_nu_vec}
        \nu 
        \equiv {\Omega}_{F_0}({\theta}_{F_0})^{-1/2}{A}_{F_0}\sqrt{n}(\bar{\theta}_n - \theta_{F_n})
        = \frac{\Omega_{F_0}(\theta_{F_0})^{-1/2}{A}_{F_0}(A_{F_0}'\Omega_{F_0}(\theta_{F_0})^{-1} A_{F_0})^{-1} c}{c'(A_{F_0}'\Omega_{F_0}(\theta_{F_0})^{-1} A_{F_0})^{-1} c} \cdot b
    \end{equation}
    and $\mathcal{Z}$ is defined as the limit law of $\sqrt{n}\Omega_{F_n}(\theta_{F_n})^{-1/2}(\hat{A}\theta_{F_n} - \hat{\beta})$, following a normal distribution $\normal(0_{2(K+1)\times 1}, I_{2(K+1)})$.
    
    For the $\text{AR}_n(\bar{\theta}_n)$, observe that 
    \begin{align*}
        \text{AR}_n(\bar{\theta}_n)
        &= \|\sqrt{n}(\hat{A}\bar{\theta}_n - \hat{\beta})'\hat{\Omega}(\bar{\theta}_n)^{-1/2}\|^2 \\
        &= \|\sqrt{n}(\hat{A}\theta_{F_n} - \hat{\beta})'\hat{\Omega}(\theta_{F_n})^{-1/2} + \sqrt{n}(\bar{\theta}_n - \theta_{F_n})'\hat{A}'\hat{\Omega}(\theta_{F_n})^{-1/2} + o_p(1)\|^2 \\
        &\xrightarrow{d} \|\mathcal{Z} + \nu\|^2.
    \end{align*}

    Combining the results for MRLM and AR statistics, it follows that 
    \[
        \text{MLC}_n(\bar{\theta}_n) 
        ~~\xrightarrow{d}~~ 
        (1+a)\cdot (\mathcal{Z} + \nu)'P_\nu (\mathcal{Z} + \nu) + a\cdot(\mathcal{Z} + \nu)'M_{\nu}(\mathcal{Z} + \nu).
    \]
    
    \medskip
    \textbf{Step 2}: Show that the minimizer of the profiled MLC statistic is consistent to $\theta_{F_0}$.
    \medskip
    
    Let $\theta_n^*$ denote the minimizer of the problem \eqref{eq:minimize_MLC}. Then we have
    \[
        a\cdot \text{AR}_n(\theta_n^*) \leq \text{MLC}_n(\theta_n^*) \leq \text{MLC}_n(\bar{\theta}_n) = O_p(1),
    \]
    where the first inequality follows by the fact that MRLM statistic is always non-negative, the second inequality follows by the definition of $\theta_n^*$ and that $\bar{\theta}_n$ from step 1 is feasible under the constraint, and the last equality follows by the conclusion of step 1.

    This implies $\|\sqrt{n}(\hat{A}\theta_n^* - \hat{\beta})\| = O_p(1)$. Using inequality \eqref{eq:bound_on_diff_theta}, 
    \[  
        \|\theta_n^* - \theta_{F_n}\| \leq 
        \frac{\sqrt{n}(\|\hat{A}\theta_n^* - \hat{\beta}\| + \|\hat{A}\theta_{F_n} - \hat{\beta}\|)}{\sqrt{n}\lambda_{\min}(\hat{A})},
    \]
    we conclude $\|\theta_n^* - \theta_{F_n}\| = O_p(n^{-1/2})$. Given this result, the same arguments in Step 1 continue to hold after replacing $\theta_n^*$ with $\bar{\theta}_n$ in \eqref{eq:MRLM_local_asymtotic} while maintaining the same definition of $\nu$:
    \[
        \nu = \frac{\Omega_{F_0}(\theta_{F_0})^{-1/2}{A}_{F_0}(A_{F_0}'\Omega_{F_0}(\theta_{F_0})^{-1} A_{F_0})^{-1} c}{c'(A_{F_0}'\Omega_{F_0}(\theta_{F_0})^{-1} A_{F_0})^{-1} c} \cdot b.
    \] 
    Hence,
    \[
        \text{MRLM}_n(\theta_n^*) ~~\xrightarrow{d}~~ (\mathcal{Z}+ \nu)'P_\nu(\mathcal{Z} + \nu).
    \]
    
    \medskip
    \textbf{Step 3}: Conclude the proof \par
    \medskip
    To save notations, let $q(a)$ denote the $(1-\alpha)$-quantile of the mixture chi-square distributions $(1+a)\chi_1^2 + a\chi_{2K+1}^2$. 
    Note that the following inequality 
    \[
        \text{MRLM}_n(\theta_n^*) \leq \text{MLC}_n(\theta_n^*) \leq \text{MLC}_n(\bar{\theta}_n)
    \]
    implies the inequality on rejection probabilities:
    \[
        \Prob_{F_n}(\text{MRLM}_n(\theta_n^*) > q(a)) 
        ~~\leq~~ \Exp_{F_n}[\hat{\phi}_{\text{MLC}}(\lambda_n^*)] 
        ~~\leq~~ \Prob_{F_n}(\text{MLC}_n(\bar{\theta}_n) > q(a))
    \]
    Taking the limit on both sides gives
    \[
        L(a) \leq \liminf_{n\to\infty}\Exp_{F_n}[\hat{\phi}_{\text{MLC}}(\lambda_n^*)] \leq \limsup_{n\to\infty}\Exp_{F_n}[\hat{\phi}_{\text{MLC}}(\lambda_n^*)] \leq U(a)
    \]
    where 
    \[
        L(a) = \Prob((\mathcal{Z} + \nu)'P_\nu(\mathcal{Z} + \nu) > q(a))
    \]
    and
    \[
        U(a) = \Prob((1+a)\cdot (\mathcal{Z} + \nu)'P_\nu (\mathcal{Z} + \nu) + a\cdot(\mathcal{Z} + \nu)'M_{\nu}(\mathcal{Z} + \nu) > q(a)).
    \]
    Note that both $L(a)$ and $U(a)$ are continuous function of $a$, taking $a$ to zero from above then yields:
    \[
        L(0) \leq \liminf_{a\searrow 0} \liminf_{n\to\infty} \Exp_{F_n}[\hat{\phi}_{\text{MLC}}(\lambda_n^*)] \leq  \limsup_{a\searrow 0} \liminf_{n\to\infty} \Exp_{F_n}[\hat{\phi}_{\text{MLC}}(\lambda_n^*)] \leq U(0).
    \]
    Since $L(0) = U(0) = \lim_{n\to\infty} \Exp_{F_n}[\hat{\phi}_{\text{Wald}}(\lambda_n^*)]$ by Lemma \ref{lem:wald_local_power}, the desired conclusion has been established.
\end{proof}

\begin{proof}[Proof of Lemma \ref{lem:wald_local_power}]
    It suffices to show that 
    \[
        \text{Wald}_n(\lambda_n^*) ~~\xrightarrow{d}~~ (\mathcal{Z} + \nu)'P_\nu (\mathcal{Z} + \nu)
    \]
    under $\{F_n\}_{n\geq 1} \subseteq \mathcal{P}_s$. 
    Under strong identification such as the parameter space restriction imposed in $\mathcal{P}_s$, the efficient estimator $\hat{\theta}^{\text{eff}}$ (the continuously updated GMM or two-step GMM) is consistent and asymptotic normal with the asymptotic variance $\left[A_{F_0}'\Omega_{F_0}(\theta_{F_0})^{-1} A_{F_0}\right]^{-1}$ \citep[Theorem 5.2]{newey/mcfadden:1994}. Then it follows that
    \begin{align*}
        \text{Wald}_n(\lambda_n^*) 
        &= 
        n(c'\hat{\theta}^{\text{eff}} - \lambda_{n}^*)' \left[c'(\hat{A}'\hat{\Omega}(\hat{\theta}^{\text{eff}})^{-1}\hat{A})^{-1}c\right]^{-1}(c'\hat{\theta}^{\text{eff}} - \lambda_{n}^*) \\
        &=
        \left[\sqrt{n}(c'\hat{\theta}^{\text{eff}} - \lambda_{F_n}) - b\right]' \left[c'(\hat{A}'\hat{\Omega}(\hat{\theta}^{\text{eff}})^{-1}\hat{A})^{-1}c\right]^{-1}\left[\sqrt{n}(c'\hat{\theta}^{\text{eff}} - \lambda_{F_n}) - b\right] \\
        &= 
        \left(\frac{\sqrt{n}(c'\hat{\theta}^{\text{eff}} - \lambda_{F_n})}{\sqrt{c'(A_{F_0}'\Omega_{F_0}(\theta_{F_0})^{-1}A_{F_0})^{-1}c + o_p(1)}} - \|\nu\|\sign(b) + o_p(1)\right)^2
    \end{align*}
    For an efficient estimator, we have the following asymptotic representation:
    \[
        \sqrt{n}(\hat{\theta}^{\text{eff}} - \theta_{F_n}) = -(A_{F_n}'\Omega_{F_n}(\theta_{F_n})^{-1} A_{F_n})^{-1} A_{F_n}'\Omega_{F_n}(\theta_{F_n})^{-1} \sqrt{n}(\hat{A}\theta_{F_n} - \hat{\beta}) + o_p(1).
    \]
    which implies
    \[
        \frac{\sqrt{n}(c'\hat{\theta}^{\text{eff}} - \lambda_{F_n})}{\sqrt{c'(A_{F_0}'\Omega_{F_0}(\theta_{F_0})^{-1}A_{F_0})^{-1}c}} ~~\xrightarrow{d}~~ -\frac{\mathcal{Z}'\nu}{\|\nu\|} \cdot \sign(b).
    \]
    Plugging this into the Wald test statistic yields the desired result.
    \[
        \text{Wald}_n(\lambda_n^*)  ~~\xrightarrow{d}~~ \left(-\frac{\mathcal{Z}'\nu}{\|\nu\|} - \|\nu\|\right)^2 = (\mathcal{Z} + \nu) P_\nu (\mathcal{Z} + \nu).
    \]
\end{proof}

\begin{proof}[Proof of Lemma \ref{lem:consistent_cov_strong_id}]
For simplicity of notation, I leave the dependence on $F_0$ implicit in this proof. 
First, we show that the smallest and largest singular values of $\Omega(\theta)$ are finite and bounded away from zero uniformly across $\theta\in \Theta$ under distribution $F_0$. The first inequality can be established by noting that 
\begin{align*}
    \lambda_{\min}(\Omega(\theta)) 
    &= \lambda_{\min}(H(p,\theta)'\Sigma_{p}H(p,\theta) + \Sigma_{\beta}) \\ 
    &\geq \lambda_{\min}(\Sigma_\beta) \\
    &= \min_{d, \ell} \frac{\sigma_{d\ell}^2}{q(d,z_\ell)} \\
    &> 0,
\end{align*}
where the first inequality holds by eigenvalue inequality since $H(p,\theta)'\Sigma_p H(p,\theta)$ is positive semi-definite, and the second inequality follows by the parameter space restriction on $\mathcal{P}$. Therefore, the lower bound inequality is established. Regarding the upper bound. Note that 
\begin{align*}
    \lambda_{\max}(\Omega(\theta))
    &= \lambda_{\max}(H(p,\theta)'\Sigma_{p}H(p,\theta) + \Sigma_{\beta}) \\
    &\leq \lambda_{\max}(H(p,\theta)'\Sigma_{p}H(p,\theta)) + \lambda_{\max}(\Sigma_{\beta}) \\
    &= \lambda_{\max}(\Sigma_p^{1/2}H(p,\theta))^2 + \lambda_{\max}(\Sigma_{\beta}) \\
    &\leq \lambda_{\max}(\Sigma_p) \lambda_{\max}(H(p,\theta))^2 + \lambda_{\max}(\Sigma_\beta),
\end{align*}
where the first inequality holds by triangle inquality for spectral norm $\lambda_{\max}(\cdot)$, and the second inequality holds by the Cauchy-Schwarz inequality. Since each element of $H(p,\theta)$ is continuous in $\theta \in \Theta$, and $\Theta$ is a compact set, $\lambda_{\max}(H(p,\theta))$ is bounded uniformly over $\theta \in \Theta$. Note that 
\begin{align*}
     \lambda_{\max}(\Sigma_p) = \max_{\ell} \frac{p(z_\ell)(1-p(z_\ell))}{q(z_\ell)} < \infty 
    \quad \text{and} \quad 
     \lambda_{\max}(\Sigma_\beta) = \max_{d, \ell} \frac{\sigma_{d\ell}^2}{q(d,z_\ell)} < \infty
\end{align*}
as implied by the conditions imposed on $\mathcal{P}$. So the inequality \eqref{eq:lem_locpower_regular1} is established.

Next we establish the uniform consistent estimation on $\hat{\Omega}(\theta)$ in \eqref{eq:lem_locpower_regular2}. First note that the difference in spectral norm can be expressed as
\begin{align*}
    \|\hat{\Omega}(\theta) - \Omega(\theta)\|
    &= \|H(\hat{p},\theta)' \hat{\Sigma}_p H(\hat{p},\theta) - H(p,\theta)' \Sigma_{p} H(p, \theta) + \hat{\Sigma}_\beta - \Sigma_{\beta}\| \\
    &\leq \|H(\hat{p},\theta)' \hat{\Sigma}_p H(\hat{p},\theta) - H(\hat{p},\theta)' \Sigma_{p} H(\hat{p}, \theta)\| \\
    &\quad + \|H(\hat{p},\theta)' \Sigma_{p} H(\hat{p}, \theta) - H(p,\theta)' \Sigma_{p} H(p, \theta)\| \\
    &\quad + \|\hat{\Sigma}_\beta - \Sigma_{\beta}\|  \\
    &\leq \left(\|\hat{\Sigma}_p^{1/2}H(\hat{p},\theta)\| + \|\Sigma_{p}^{1/2}H(\hat{p},\theta)\|\right)\|H(\hat{p},\theta)\|\cdot \|\hat{\Sigma}_p^{1/2} - {\Sigma}_{p}^{1/2}\| \\
    &\quad + \left(\|\Sigma_{p}^{1/2} H(\hat{p},\theta)\| + \|\Sigma_{p}^{1/2} H(p,\theta)\|\right)\|\Sigma_{p}^{1/2}\| \cdot \|H(\hat{p},\theta) - H(p,\theta)\| \\
    &\quad + \|\hat{\Sigma}_\beta - \Sigma_{\beta}\|,
\end{align*}
where the first inequality follows by triangle inequality, and the second inequality holds by noting that 
\begin{align*}
    \|X'X - Y'Y\| 
    &= \|X'X - X'Y + X'Y - Y'Y\| \\
    &\leq \|X'(X-Y)\| + \|(X-Y)'Y\| \\
    &\leq (\|X\| + \|Y\|)\|X-Y\|,
\end{align*}
for any matrices $X$ and $Y$ that have the same number of columns. 

Next establish several auxiliary bounds that will be used to control these terms:
\begin{enumerate}
    \item $\sup_{\theta \in \Theta} \|H(p,\theta)\| < \infty$. This is implied by the fact that $H(p,\theta)$ is continuous in $\theta$, and $\Theta$ is compact and thus $H(p,\theta)$ is uniformly bounded above. 
    
    \item $\sup_{\theta \in \Theta}\|H(\hat{p},\theta)\| = O_p(1)$. This is implied by the following inequality:
    \begin{align*}
    &\Prob_{F_n}\left(\sup_{\theta\in\Theta}H(\hat{p},\theta) \leq \sup\{H(p,\theta): \epsilon/2 \leq p(z_\ell) \leq 1-\epsilon/2, ~ \forall \ell = 0,1,\ldots,K; ~ \theta\in\Theta\}\right) \\
    &\geq \Prob_{F_n}\left(\|\hat{p} - p_{F_n}\| < \epsilon/2\right) \\
    &\to 1,
    \end{align*}
    where the second line holds since $\epsilon\leq p_{F_n}(z_\ell) \leq 1-\epsilon$, and the last line holds by the consistency of $\hat{p}$ from Lemma \ref{lem:convergence_p_beta}. 
    
    \item $\|\hat{\Sigma}_p\| = O_p(1)$. This is implied by $\hat{\Sigma}_p \xrightarrow{p} \Sigma_{p,F_0}$ from Lemma \ref{lem:convergence_p_beta}, and $\Sigma_{p,F_0}$ being positive definite.
    
    \item $\sup_{\theta \in \Theta}\|H(\hat{p},\theta) - H(p,\theta)\| = o_p(1)$, which holds by the following inequality:
    \begin{align*}
        \|H(\hat{p},\theta) - H(p,\theta)\| \leq \sum_{m=0}^M \sum_{d=0,1}|\theta_{dm}|\times|\lambda'_{dm}(\hat{p}(z_\ell)) - \lambda'_{dm}(p(z_\ell))|,
    \end{align*}
    and given that $\lambda'_{dm}(\cdot)$ being continuous and $\hat{p}~~ \xrightarrow{p}~~p$. 
\end{enumerate}

Combining these results gives
\[
    \sup_{\theta\in\Theta}\|\hat{\Omega}(\theta) - \Omega(\theta)\| = O_p(1) \cdot \|\hat{\Sigma}_p^{1/2} - {\Sigma}_{p}^{1/2}\| + O_p(1) \cdot o_p(1) + O_p(1)\cdot \|\hat{\Sigma}_\beta - \Sigma_{\beta}\|.
\]
Then the desired result follows by the consistency of variance estimators from Lemma \ref{lem:convergence_p_beta}.

The proof of \eqref{eq:lem_locpower_regular3} follows similar arguments as the proof of \eqref{eq:lem_locpower_regular2} and thus omitted.
\end{proof}

\newpage

\section{Inference with Estimated Weights}
\label{appendix:estimated_weights}

Sometimes the weight $c$ needs to be estimated if researchers' interests focus on causal parameters such as ATT and LATE. In this case, the weight usually depends on the joint distribution of $(D,Z)$ as in Table \ref{tab:treatment-effects-weight}, therefore I make the following assumption:
\begin{assumption}
\label{asp:estimated_weight}
The weight $c = c(p, q)$ is a function of propensity scores $p = (p(z_0),\ldots,p(z_K))'$ and marginal distribution of instrument $q=(q(z_0), \ldots, q(z_K))'$, and this function is continuously differentiable.
\end{assumption}

Under this assumption, I develop inferential theory for causal parameter that depends on estimated weights. In section \ref{sec:cond_wald_est_weight}, I generalize the conditional Wald test proposed in section \ref{sec:linearMTE_inference} for linear MTE models and extend the theory for MLC test proposed in section \ref{sec:polyMTE_inference} for a general class of MTE models that include the polynomial structure.

\subsection{Linear MTE models}
\label{sec:cond_wald_est_weight}
I plug a consistent and asymptotic normal estimator $c(\hat{p}, \hat{q})$ into the construction of the moment function. This gives a modified sample moment function:
\[
    \hat{g}_k^\dagger(\lambda) = \left[\hat{p}(z_k) - \hat{p}(z_0)\right]\lambda - c_\mu(\hat{p},\hat{q})\hat{\Delta}_\mu(z_0,z_k) - c_\rho(\hat{p},\hat{q}) \hat{\Delta}_\rho(z_0,z_k)
\]
and 
\[
    \hat{g}^{\dagger}(\lambda) = (\hat{g}_1^\dagger(\lambda), \ldots, \hat{g}_K^\dagger(\lambda))'\in\R^K.
\]
Since the weight $c$ is estimated, its sampling uncertainty directly impacts  the asymptotic distribution of the moment function. Let 
\[
    {\Delta}_\mu = ({\Delta}_\mu(z_0,z_1), \ldots, {\Delta}_\mu(z_0,z_K))' \quad \text{and} \quad {\Delta}_\rho = ({\Delta}_\rho(z_0,z_1), \ldots, {\Delta}_\rho(z_0,z_K))'.
\]
Define $\hat{\Delta}_\mu$ and $\hat{\Delta}_\rho$ as their corresponding estimators, where each entry is obtained by replacing ${\Delta}_\mu(z_0,z_k)$ and ${\Delta}_\rho(z_0,z_k)$ with their respective estimators $\hat{\Delta}_\mu(z_0,z_k)$ and $\hat{\Delta}_\rho(z_0,z_k)$.

Applying the Delta method, we have the following first-order asymptotic expansion:
\[
    \hat{g}^\dagger(\lambda) - g(\lambda)
    = \begin{pmatrix}
	\partial_{p} {g}(\lambda)[\hat{p} - p] \\
	+ ~\partial_{\beta_1} {g}(\lambda)[\hat{\beta}_1 - \beta_1] \\
	+ ~\partial_{\beta_0} {g}(\lambda)[\hat{\beta}_0 - \beta_0] \\
        - ~\Delta_\mu(\partial_{p'}c_\mu[\hat{p} - p] + \partial_{q'}c_\mu[\hat{q} - q]) \\
        - ~\Delta_\rho(\partial_{p'}c_\rho[\hat{p} - p] + \partial_{q'}c_\rho[\hat{q} - q])
      \end{pmatrix}
	+ o_p(n^{-1/2}),
\]
where $\partial_p c$ and $\partial_q c$ denote the partial derivative vectors of $c$ with respect to vectors $p$ and $q$, respectively, for $c \in \{c_\mu(p,q), c_\rho(p,q)\}$.

The asymptotic variance can be consistently estimated by
\begin{align*}
    \hat{S}^\dagger(\lambda) 
    & = \left(\partial_{p}\hat{g}^\dagger(\lambda) - \hat{\Delta}_\mu [\partial_{p'} \hat{c}_\mu] - \hat{\Delta}_\rho [\partial_{p'} \hat{c}_\rho]\right)\hat{\Sigma}_p\left(\partial_{p}\hat{g}^\dagger(\lambda) - \hat{\Delta}_\mu [\partial_{p'} \hat{c}_\mu] - \hat{\Delta}_\rho [\partial_{p'} \hat{c}_\rho]\right)' \\
    & \quad ~+ \left(\hat{\Delta}_\mu [\partial_{q'}\hat{c}_\mu] + \hat{\Delta}_\rho [\partial_{q'}\hat{c}_\rho]\right)\hat{\Sigma}_q \left(\hat{\Delta}_\mu [\partial_{q'}\hat{c}_\mu] + \hat{\Delta}_\rho [\partial_{q'}\hat{c}_\rho]\right)' \\
    & \quad ~+ \partial_{\beta_0}\hat{g}^\dagger(\lambda) \hat{\Sigma}_{\beta_0}\partial_{\beta_0'}\hat{g}^\dagger(\lambda) + \partial_{\beta_1}\hat{g}^\dagger(\lambda) \hat{\Sigma}_{\beta_1}\partial_{\beta_1'}\hat{g}^\dagger(\lambda),
\end{align*}
where $\partial_x \hat{g}^\dagger(\lambda)$ is obtained from $\partial_x g(\lambda)$ by replacing $c$ with $\hat{c}$ for $x \in \{p,\beta_0,\beta_1\}$, and $\partial_x \hat{c}_\mu$ (or $\partial_x \hat{c}_\rho$) denotes the sample analog of $\partial_x c_\mu$ (or $\partial_x c_\rho$) evaluated at $(\hat{p}, \hat{q})$ for $x \in \{p,q\}$.
 
To achieve asymptotic independence from $\sqrt{n}\hat{g}^\dagger(\lambda)$, I modify the sufficient statistic $\hat{h}(\lambda)$ to
\[
    \hat{h}^\dagger(\lambda) = \sqrt{n}\hat{\pi} - [\partial_{p} {\pi}]\hat{\Sigma}_p \left(\partial_{p'} \hat{g}^\dagger(\lambda) - \hat{\Delta}_\mu [\partial_{p'} \hat{c}_\mu] - \hat{\Delta}_\rho [\partial_{p'} \hat{c}_\rho]\right) \hat{S}^\dagger(\lambda)^{-1}\sqrt{n}\hat{g}^\dagger(\lambda).
\]
The simulation counterpart of $\hat{\pi}_s$ can then be constructed as
\[
    {\pi}_s^\dagger = \hat{h}^\dagger(\lambda) + [\partial_{p} {\pi}]\hat{\Sigma}_p \left(\partial_{p'} \hat{g}^\dagger(\lambda) - \hat{\Delta}_\mu [\partial_{p'} \hat{c}_\mu] - \hat{\Delta}_\rho [\partial_{p'} \hat{c}_\rho]\right) \hat{S}^\dagger(\lambda)^{-1/2}\eta^*,
\]
where $\eta^* \sim \normal(0_{K\times 1},I_{K\times K})$ are simulated draws independent of data.
Under these modifications, the asymptotic distribution of the plug-in Wald statistic
\[
    W_n^\dagger(\lambda) = \frac{n\hat{g}^\dagger(\lambda)\hat{S}^\dagger(\lambda)^{-1}\hat{\pi}\hat{\pi}'\hat{S}^\dagger(\lambda)^{-1}\hat{g}^\dagger(\lambda)}{\hat{\pi}'\hat{S}^\dagger(\lambda)^{-1}\hat{\pi}}
\]
has the same conditional distribution as
\[
    \frac{[\eta^*]' \hat{S}^\dagger(\lambda)^{-1/2} [\pi_s^\dagger] [\pi_s^\dagger]' \hat{S}^\dagger(\lambda)^{-1/2} [\eta^*]}{[\pi_s^\dagger]'\hat{S}^\dagger(\lambda)^{-1}[\pi_s^\dagger]}
\]
when conditioning on $\hat{h}^\dagger(\lambda)$ as sample size diverges. This leads to the following corollary regarding the uniform validity of the modified testing procedure:
\begin{corollary}
\label{corollary:linear_MTE_estimated_weight}
    Consider the weight $c$ that depends on the data distribution $F$. Let Assumptions \ref{asp:iid}, \ref{asp:linear_MTE} and \ref{asp:estimated_weight} hold, and suppose that $\inf_{(\theta, F)\in\mathcal{P}} \|c(p_F,q_F)\| > 0$. Replacing $(\hat{g}(\lambda), \hat{S}(\lambda), \pi_s)$ with their counterparts $(\hat{g}^\dagger(\lambda), \hat{S}^\dagger(\lambda), \pi^\dagger_s)$ in Theorem \ref{thm:validity_cond_wald} yields the same conclusion.
\end{corollary}

\begin{remark}
    We can also modify the AR test described in Proposition \ref{prop:AR_validity} by replacing $\hat{g}_k(\lambda)$ with $\hat{g}_k^\dagger(\lambda)$ and replacing $\hat{s}_k(\lambda)$ with the $(k,k)$-th diagonal element of $\hat{S}^\dagger(\lambda)$ to make it asymptotically valid with the estimated weight.
\end{remark}

\subsection{Polynomial MTE models}

In this section, I maintain Assumption \ref{asp:estimated_weight} and consider modified linear combination test for the treatment effects parameters $\lambda = c(p,q)'\theta$. Naively plugging in the estimator $\hat{c} = c(\hat{p}, \hat{q})$ into the test statistic can introduce estimation errors in the constraint $\hat{c}'\theta = \lambda$, potentially leading to asymptotic bias in the limiting distribution of the test statistic. To address this, I modify the test statistic to account for the effects of these estimation errors.

One possible solution involves reparameterizing the model so that the target parameter $\hat{c}'\theta$ becomes an element of the reparameterized model parameter, e.g., see D. \citet{andrews.d:2017}. However, given the broad class of causal parameters of interest, finding a universal reparameterization rule that works for all these parameters is challenging, and such reparameterization would greatly complicates the asymptotic analysis.

As an alternative, I impose a normalization on the weights of the treatment effect parameters. This assumption is broadly applicable to various causal effects, including ATE, ATT, LATE, and the normalized counterfactual policy effects $\bar{\alpha}(\epsilon)$ considered in this paper. 

\begin{assumption}
\label{asp:weight_for_TE}
    The first and $(K+1)$-th element of the weight function $c(p,q)$ are 1 and $-1$, respectively, and weight is symmetric: $c = (c_1', c_0')'$ with $c_1 = -c_0$.
\end{assumption}

Under Assumption \ref{asp:weight_for_TE}, I consider a sequence of vectors $\tilde{\theta}_n = \theta_{F_n} + \delta_n$ such that $\hat{c}'\tilde{\theta}_n = \lambda_n = c(p_{F_n}, q_{F_n})' \theta_{F_n}$. Hence $\tilde{\theta}_n$ falls into the estimated constraint set. There are two choices of $\delta_n$ that satisfies this requirement:
\begin{align*}
    & \delta_{n,1} = (-[\hat{c} - c(p_{F_n},q_{F_n})]'\theta_{F_n}, 0_{1\times (2K+1)})' \\
    & \delta_{n,0} = (0_{1\times (K+1)}, [\hat{c} - c(p_{F_n},q_{F_n})]'\theta_{F_n}, 0_{1\times K})'.
\end{align*}
For any $r\in[0,1]$, define a sequence 
\begin{equation}
\label{eq:alt_theta_seq}
    \theta_{n,r} \equiv r(\theta_{F_n} + \delta_{n,1}) + (1-r)(\theta_{F_n} + \delta_{n,0}).
\end{equation} 
This construction generates a continuum of sequences $\{\theta_{n,r}\}_{n\geq 1}$ that satisfy the estimated constraint. By examining the asymptotic behavior of the test statistic under the sequence $\{\theta_{n,r}\}_{n \geq 1}$, we can establish an asymptotically valid test.

\begin{proposition}
\label{prop:MLC_validity_est_weight}
Let Assumption \ref{asp:iid}, \ref{asp:estimated_weight}, and \ref{asp:weight_for_TE} hold. Define $\xi_r \equiv (-r1_{1\times (K+1)}, (1-r)1_{1\times (K+1)})'$ with $r\in [0,1]$, where $1_{1\times (K+1)}$ denotes a $1\times (K+1)$ vector of ones. 
In the MLC testing procedures outlined in section \ref{sec:construct_MLC}, consider replacing $\hat{\Omega}(\theta)$ by 
\[
    \hat{\Omega}(\theta; r) \equiv \left(H(\hat{p},\theta) + \xi_r\theta' [\partial_p \hat{c}]\right) \hat{\Sigma}_p \left(H(\hat{p},\theta) + \xi_r\theta' [\partial_p \hat{c}]\right)' + (\xi_r\theta'[\partial_q \hat{c}])\hat{\Sigma}_q (\xi_r\theta'[\partial_q \hat{c}])' + \hat{\Sigma}_\beta,
\]
replacing $\hat{\Gamma}_{j}(\theta)$ by
\[
    \hat{\Gamma}_{j}(\theta; r) \equiv \hat{M}_j(\hat{p}) \hat{\Sigma}_p \left(H(\hat{p},\theta) + \xi_r'\theta [\partial_p \hat{c}]\right)',
\]
and replacing $c$ with $\hat{c}$ in the constraint set. With these modifications, the uniform validity of the MLC test is still maintained.
\end{proposition}

\subsection{Proofs in Appendix \ref{appendix:estimated_weights}}

\begin{proof}[Proof of Corollary \ref{corollary:linear_MTE_estimated_weight}]
Note that Delta method gives a different asymptotic expansion for the moment along a converging sequence:
\begin{align*}
    \sqrt{n}\hat{g}(\lambda_{n}) 
    ~~\xrightarrow{d}~~&  \left(\partial_pg(\lambda_\infty) - \Delta_\mu[\partial_{p'} c_\mu] - \Delta_\rho [\partial_{p'}c_\rho]\right) \mathcal{Z}_p  \\
    & \quad - \left(\Delta_\mu [\partial_{q'}c_\mu] + \Delta_{\rho}[\partial_{q'}c_\rho]\right)\mathcal{Z}_q \\
    & \quad + \partial_{\beta_1} g(\lambda_\infty) 
    \mathcal{Z}_{\beta_1} + \partial_{\beta_0} g(\lambda_\infty) \mathcal{Z}_{\beta_0} \\
    ~~\sim~~& \normal(0, S^\dagger(\lambda_\infty)).
\end{align*}
The asymptotic covariance matrix $S^\dagger(\lambda_\infty)$ is uniformly positive definite by noting that Lemma \ref{lem:pd_variance} has already established that $\partial_{\beta_1} g(\lambda_\infty) \mathcal{Z}_{\beta_1} + \partial_{\beta_0} g(\lambda_\infty)\mathcal{Z}_{\beta_0}$ is non-degenerate uniformly as long as $c$ is nonzero. Since $\mathcal{Z}_p$, $\mathcal{Z}_q$, and $\mathcal{Z}_\beta$ are jointly independent, $\hat{S}^\dagger(\lambda)$ consistently estimates the asymptotic variance $S^\dagger(\lambda_\infty)$. The rest of the proof follows by the same arguments in Proposition \ref{prop:subsequence_AR_validity} and Theorem \ref{thm:validity_cond_wald}.
\end{proof}

\begin{proof}[Proof of Proposition \ref{prop:MLC_validity_est_weight}]
Under the drifting sequence satisfying the conditions outlined in Proposition \ref{prop:subsequence_MLC_validity}, our goal is to establish
\[
    \limsup_{n\to\infty} \Prob\left(\inf_{\hat{c}'\theta = \lambda_n} \text{MLC}_n(\theta) > q_{(1+a)\chi_1^2 + a\chi_{2K+1}^2}(1-\alpha)\right) \leq \alpha,
\]
where the two chi-square distributions are independent. Then the rest of the proof follows the same arguments in the proof of Theorem \ref{thm:uniform_validity}.

Consider $\theta_{n,r}$ for $r\in [0,1]$ defined in \eqref{eq:alt_theta_seq}. Then we note that $\lim_{n\to\infty}\Prob(\tilde{\theta}_n \in \Theta) = 1$ since $\theta_{F_n} \in \operatorname{int}(\Theta)$ and the difference between $\theta_{n,r}$ and $\theta_{F_n}$ converges in probability to a zero vector. Next, it can be seen that
\begin{align*}
    \sqrt{n}(\hat{A}{\theta}_{n,r} - \hat{\beta})
    &= \sqrt{n}(\hat{A}\theta_{F_n} - \hat{\beta}) + r\cdot \sqrt{n}\hat{A}\delta_{n,1} + (1-r) \cdot \sqrt{n}\hat{A}\delta_{n,0} \\
    &\xrightarrow{d} \left(H(p_\infty,\theta_\infty) + \xi_r \theta_{F_n}' [\partial_{p} c]\right) \mathcal{Z}_p + (\xi_r  \theta_{F_n}' [\partial_{q} c]) \mathcal{Z}_q - \mathcal{Z}_\beta,
\end{align*} 
where $\partial_p c$ and  $\partial_q c$ denote the gradients of $c$ with respect to $p$ and $q$, respectively, by setting $p = p_\infty$ and $q = q_\infty$. 
Note that $\hat{\Omega}(\theta_{n,r};r)$ consistently estimate the asymptotic variance of the above moment condition. Moreover, 
\[
    \hat{\Gamma}_j(\theta_{n,r};r)
    ~~\xrightarrow{p}~~
    M_j(p_\infty) \Sigma_{p,\infty} \left(H(p_\infty, \theta_\infty) + \xi_r'\theta_{\infty}[\partial_p c]'\right),
\]
which equals the asymptotic covariance between $\sqrt{n}(\hat{A}{\theta}_{n,r} - \hat{\beta})$ and $\sqrt{n}(\hat{a}_j - a_{j,F_n})$, where $\hat{a}_{j,F_n}$ and $a_{j,F_n}$ denotes the $j$-th row of $\hat{A}$ and $A_{F_n}$, respectively. 

Following the same arguments in Proposition \ref{prop:subsequence_MLC_validity} but replacing $\theta_{F_n}$ with $\theta_{n,r}$ in the argument of test statistics, the consistency of the (co)variance estimators implies
\begin{equation}
\label{eq:convergence_MLC_altseq}
    \text{MLC}_n(\theta_{n,r}) ~~\xrightarrow{d}~~ (1+a)\chi_1^2 + a\chi_{2K+1}^2.
\end{equation}
Let $\mathcal{B}$ denote the event that $\inf_{\hat{c}'\theta = \lambda_n} \text{MLC}_n(\theta) > q_{(1+a)\chi_1^2 + a\chi_{2K+1}^2} (1-\alpha)$. 
Then we have 
\begin{align*}
    \Prob_{F_n}(\mathcal{B})
    & = 
    \Prob_{F_n}(\mathcal{B}, \theta_{n,r}\in\Theta) + \Prob_{F_n}(\mathcal{B}, \theta_{n,r}\not\in\Theta) \\
    & \leq
    \Prob_{F_n}(\text{MLC}(\theta_{n,r}) > q_{(1+a)\chi_1^2 + a\chi_{2K+1}^2}, \theta_{n,r}\in\Theta) + \Prob_{F_n}(\theta_{n,r}\not\in \Theta) \\
    & \leq 
    \Prob_{F_n}(\text{MLC}(\theta_{n,r}) > q_{(1+a)\chi_1^2 + a\chi_{2K+1}^2}) + \Prob_{F_n}(\theta_{n,r}\not\in \Theta),
\end{align*}
where the second line holds by the fact that $\theta_{n,r} \in \Theta$ and satisfies the constraint set under the first term of probability.  Taking the limit as $n \to \infty$ and noting $\Prob_{F_n}(\theta_{n,r}\not\in\Theta) \to 0$ along with \eqref{eq:convergence_MLC_altseq}, the desired result holds.
\end{proof}

\newpage

\section{Inference with Covariates}
\label{appendix:inference_covariates}
In this appendix, I describe two approaches for implementing the proposed inference procedure with covariates. Section \ref{sec:sidak_correction} presents a method that conducts the robust inference procedure conditional on a set of discrete covariates and then aggregates the causal parameters using the Šidák and Bonferroni's correction. This approach does not rely on the additive separability assumption imposed in Assumption \ref{asp:add_sep}. Section \ref{sec:inference_addsep} then explains how the proposed procedure can be adapted to settings in which researchers impose additive separability in the MTE model.

\subsection{\v Sidák-Bonferroni's correction}
\label{sec:sidak_correction}
In this section, I implement the proposed identification-robust inference procedures conditional on a set of discrete covariates $W$. Assumption \ref{asp:iid} is extended to incorporate additional covariates:
\begin{assumption}
\label{asp:iid_covariate}
The random vectors $(Y_i,D_i,Z_i,W_i)$ for $i = 1,\ldots,n$ are i.i.d. with distribution $F$, where $W_i$ has finite support.
\end{assumption}

Next I introduce a set of regularity conditions to be imposed on the distribution $F$. Let $\theta = \{\theta(w): w\in\supp(W)\}$. For some $\delta, \zeta > 0$ and $\epsilon \in (0,1/2)$, define the parameter space $\mathcal{P}$ as the set of pairs $(\theta,F)$ satisfying the following properties:
\begin{enumerate}
	\item Equation \eqref{eq:separate_reg} is satisfied with $K \geq M$, where $\theta(w) \in \operatorname{int}(\Theta) \subseteq \R^{2(M+1)}$ for some compact set $\Theta$, for all $w\in \supp(W)$,
	\item $\sup_{d=0,1}\sup_{(z,w) \in \supp(Z,W)}\Exp_F[|Y|^{2+\delta}\mid D=d, Z=z, W=w] \leq \zeta$,
	\item $\epsilon \leq \inf_{(z,w) \in \supp(Z,W)}\Prob_F(D=1\mid W=w,Z=z) \leq \sup_{(z,w) \in \supp(Z,W)}\Prob_F(D=1\mid W=w, Z=z) \leq 1-\epsilon$,
	\item $\epsilon \leq \inf_{(z,w) \in \supp(Z,W)}\Prob_F(Z=z \mid W=w) \leq \sup_{(z,w) \in \supp(Z,W)} \Prob_{F}(Z=z \mid W=w) \leq 1-\epsilon$,
	\item $\epsilon \leq \inf_{w\in\supp(W)} \Prob_F(W=w) \leq \sup_{w\in\supp(W)} \Prob_F(W=w) \leq 1-\epsilon$,
 	\item $\epsilon \leq \inf_{d=0,1}\inf_{(z,w) \in \supp(Z,W)}\var_F(Y\mid D=d, Z=z, W=w)$.
\end{enumerate}
In particular, conditions 3, 4, and 5 together imply that 
\[
	\supp(D,Z,W) = \{0,1\}\times \supp(Z) \times \supp(W)
\]
and strong overlap holds.

For a fixed $w \in \supp(W)$, let $\mathcal{P}(w)$ be the projection of $\mathcal{P}$ onto space of $\left(\theta(w),F_{Y,D,Z\mid W = w}\right)$. That is, $\left(\theta(w),F_{Y,D,Z\mid W = w}\right)$ belongs to $\mathcal{P}(w)$ if there exists a DGP $(\theta,F)\in\mathcal{P}$ that generates $\left(\theta(w),F_{Y,D,Z\mid W = w}\right)$. Note that elements in $\mathcal{P}(w)$ still satisfy conditions 1-6 but with a fixed $w\in\supp(W)$. Define $\mathcal{P}_0(w)$ as the space of the conditional causal effects $\lambda(w) = c(w)'\theta(w)$\footnote{The weight $c(w)$ depends on covariates $w$ if it contains unknown propensity scores, e.g., ATT, LATE, and PRTE.}:
\[
	\mathcal{P}_0(w) = \left\{(\lambda(w), F_{Y,D,Z\mid W=w}): \lambda(w) = c(w)'\theta(w), (\theta(w), F_{Y,D,Z\mid W=w})\in \mathcal{P}(w)\right\}.
\]
Note that the weight $c(w)$ can depend on the underlying distribution $F_{Y,D,Z\mid W = w}$, but I omit this dependence for the simplicity of notations.

The space for the aggregated effects $\lambda = \sum_{w\in\supp(W)}q(w)\lambda(w)$ is denoted by
\[
	\mathcal{P}_0 = \left\{(\lambda, F): \lambda = \sum_{w\in\supp(W)}q_F(w)c(w)'\theta(w), \quad (\theta, F) \in \mathcal{P} \right\}
\]
where $q_F(w)$ represents the marginal or conditional distribution of covariates, for example: 
\begin{enumerate}
	\item The marginal probability of covariates $\Prob(W=w)$, for ATE.
	\item The weighted marginal probability: $\Prob(W=w \mid D=d)$, for ATT and ATU.
	\item The weighted marginal probability: $\Prob(W=w \mid p(W,z_0) < U < p(W,z_k))$, for LATE.
\end{enumerate}

The methods described in sections \ref{sec:linearMTE_inference} and \ref{sec:polyMTE_inference} have led to a confidence set ${\mathcal{C}}(w)$ for the conditional effect $\lambda(w)$ with uniformly valid coverage requirement as below:
\begin{definition}
\label{def:valid_conditional_cs}
For $\omega \in \supp(W)$, a confidence set ${\mathcal{C}}(w)$ for the conditional causal effects $\lambda(w) = c(w)'\theta(w)$ is said to be uniformly valid with asymptotic level $1-\alpha$ if it depends on the samples $\{(Y_i,D_i,Z_i): W_i =w\}$ and satisfies
\[
	\liminf_{n\to\infty} \inf_{(\lambda(w),F_{Y,D,Z\mid W=w})\in\mathcal{P}_0(w)}\Prob_{F}\left(\lambda(w) \in {\mathcal{C}}(w)\right) \geq 1-\alpha.
\]
\end{definition}

The following lemma then shows that a valid confidence set for the unconditional effect $\lambda$ can be obtained by combining the confidence sets ${\mathcal{C}}(w)$ across the support of covariate $W$ using an adjusted critical value.
\begin{lemma}
\label{lem:inference_covariates}
Let Assumption \ref{asp:iid_covariate} hold. Suppose for each $w \in \supp(W)$, there exists a uniformly valid confidence set ${\mathcal{C}}(w)$ for $\lambda(w)$ with asymptotic level $(1-\alpha)^{1/|\supp(W)|}$. Define a confidence set ${\mathcal{C}}$ by taking a weighted average:
\[
	{\mathcal{C}} = \left\{\lambda = \sum_{w\in\supp(W)}q_F(w)\lambda(w): \lambda(w) \in {\mathcal{C}}(w)\right\}.
\]
Then we have
\[
	\liminf_{n\to\infty} \inf_{(\lambda, F)\in \mathcal{P}_0} \Prob_F\left(\lambda \in {\mathcal{C}} \right) \geq 1-\alpha.
\]
\end{lemma}

Consider a unknown $q_F(w)$ that needs to be estimated. Denote ${\mathcal{C}}_q$ the valid confidence set of $q_F = \{q_F(w): w\in\supp(W)\}$ with asymptotic level $1-\alpha_1$ ($0<\alpha_1<\alpha$), which satisfies
\[
	\liminf_{n\to\infty} \inf_{(\lambda, F) \in \mathcal{P}_0} \Prob_F\left(q_F \in {\mathcal{C}}_q\right) \geq 1-\alpha_1.
\]
One can construct such confidence set by inverting the Wald test based on the sample analog estimator of $q_F$. This leads to the following confidence set for $\lambda$ when $q(w)$ is unknown but needs to be estimated.

\begin{lemma}
\label{lem:inference_covariates_estimated_mass}
Let Assumption \ref{asp:iid_covariate} hold. Suppose for each $w \in \supp(W)$, there exists a uniformly valid confidence set ${\mathcal{C}}(w)$ for $\lambda(w)$ with asymptotic level $(1-\alpha_2)^{1/|\supp(W)|}$, and there exists a uniformly valid confidence set for $q_F$ with asymptotic level $1-\alpha_1$, where $\alpha_1 = \alpha - \alpha_2$. Define the confidence set $\mathcal{C}$ for the unconditional effects as below:
\[
	{\mathcal{C}} = \left\{\lambda = \sum_{w\in\supp(W)}q(w)\lambda(w): \lambda(w) \in {\mathcal{C}}(w), q \in {\mathcal{C}}_q\right\}.
\]
Then we have 
\[
	\liminf_{n\to\infty} \inf_{(\lambda, F)\in \mathcal{P}_0} \Prob_F\left(\lambda \in {\mathcal{C}} \right) \geq 1-\alpha.
\]
\end{lemma}

In practice, we usually can achieve precise estimation of $q(z)$ if the sample size conditional on covariates is considerably large. In this case, the standard errors from estimating $q(z)$ are often small, and therefore the corresponding impact can usually be ignored.

\subsection{Proofs in section \ref{sec:sidak_correction}}

\begin{proof}[Proof of Lemma \ref{lem:inference_covariates}]
For any fixed $(\lambda,F)\in\mathcal{P}_0$, we have 
	\begin{align*}
	\Prob_F\left(\lambda \in {\mathcal{C}}\right)
	&\geq
	\Prob_F\left(\bigcap_{w\in\supp(W)} \left\{\lambda(w) \in {\mathcal{C}}(w)\right\}\right) \\
	&=
	\prod_{w\in\supp(W)}\Prob_F\left(\lambda(w) \in {\mathcal{C}}(w)\right) 
	\end{align*}
The second line uses the fact that $\left\{\lambda(w) \in {\mathcal{C}}(w)\right\}$ is independent of $\left\{\lambda(w')\in\mathcal{C}(w')\right\}$ if $w\neq w'$ since confidence sets are estimated from independent samples following Definition \ref{def:valid_conditional_cs}. Taking infimum on both sides obtains
\begin{align*}
	\inf_{(\lambda, F)\in\mathcal{P}_0} \Prob_F\left(\lambda \in {\mathcal{C}}\right) 
	&= \inf_{(\lambda, F)\in\mathcal{P}_0} \prod_{w\in\supp(W)} \Prob_F(\lambda(w) \in {\mathcal{C}}(w)) \\
	&\geq \prod_{w\in\supp(W)} \inf_{(\lambda(w), F_{Y,D,Z\mid W=w}) \in \mathcal{P}_0(w)}\Prob_F(\lambda(w) \in {\mathcal{C}}(w))
\end{align*}
The second line holds since every element $(\lambda,F)\in\mathcal{P}_0$ corresponds to a valid element of $(\lambda(w), F_{Y,D,Z\mid W=w})$ contained in $\mathcal{P}_0(w)$ for each $w \in \supp(W)$. Let sample size $n$ go to infinity and then it follows that
\[
		\liminf_{n\to\infty}\inf_{(\lambda, F)\in\mathcal{P}_0} \Prob_F\left(\lambda \in {\mathcal{C}}\right) 
		\geq
		\prod_{w\in\supp(W)} \liminf_{n\to\infty} \inf_{(\lambda(w), F_{Y,D,Z\mid W=w}) \in \mathcal{P}_0(w)}\Prob_F(\lambda(w) \in {\mathcal{C}}(w)) \geq 1-\alpha.
\]
So the desired result is proved.
\end{proof}

\bigskip

\begin{proof}[Proof of Lemma \ref{lem:inference_covariates_estimated_mass}]
	For a fixed pair $(\lambda,F)\in\mathcal{P}_0$, we have 
\begin{align*}
	\Prob_F(\lambda\in {\mathcal{C}}) 
	&= \Prob_F\left(\sum_{w\in\supp(W)}q(w)\lambda(w) \in {\mathcal{C}}\right) \\
	&\geq \Prob_F\left(q\in\hat{\mathcal{C}}_q, \lambda(w) \in {\mathcal{C}}(w) \text{ for each } w\in\supp(W) \right) \\
	&\geq \Prob_F\left(q\in\hat{\mathcal{C}}_q\right) + \prod_{w\in\supp(W)} \Prob_F\left(\lambda(w)\in {\mathcal{C}}(w)\right) - 1
\end{align*}
Following the same arguments in the proof of Lemma \ref{lem:inference_covariates}, we have 
\begin{align*}
	\liminf_{n\to\infty}\inf_{(\lambda,F)\in\mathcal{P}_0} \Prob_F(\lambda \in {\mathcal{C}}) 
	& \geq \liminf_{n\to\infty}\inf_{(\lambda,F)\in\mathcal{P}_0}\Prob_F\left(q\in{\mathcal{C}}_q\right) \\
	& \quad + \prod_{w\in\supp(W)} \liminf_{n\to\infty}\inf_{(\lambda(w),F_{Y,D,Z\mid W=w})\in\mathcal{P}_0(w)} \Prob_F\left(\lambda(w)\in{\mathcal{C}}(w)\right) - 1 \\
	& = 1-\alpha_1-\alpha_2 \\
	& = 1-\alpha.
\end{align*}
So the desired conclusion is proved.
\end{proof}

\subsection{Discussion of inference under additive separability}
\label{sec:inference_addsep}
In this section, I briefly describe how to implement the proposed robust inference procedures if researchers still want to maintain additive separability.  Following the existing literature discussed in Appendix \ref{appendix:cov_weakIV},  additional functional restrictions should be imposed on both stages, such as the additive separability (Assumption \ref{asp:add_sep}) and a distributional assumption on selection unobservable $U$. To fix ideas, consider a model with covariates $W \in \R^{L}$ and a discrete instrument $Z$ below:
\begin{align*}
	& Y_d = \mu_d + W'\tau_d + V_d \\
	& D = \indicator[U \leq W'\pi + \delta(Z)] \\
        & U \indep (W,Z)\quad \text{and} \quad U\sim F_U(\cdot)
\end{align*}
where $F_U$ denotes the CDF for a known distribution (e.g., logit or normal specifications). Additionally, assume the additive separable structure holds:
\[
	\Exp[V_d\mid W=w, F_U(U)=u] = \Exp[V_d \mid F_U(U)=u]
\]
for which I maintain the parametric specification of control functions:
\[
	\Exp[V_d \mid F_U(U)=u] = \sum_{m=1}^M \rho_{dm} h_m(u).
\]
The separate regression approach gives the following separate regressions for identification of the structural parameters:
\begin{equation}
\label{eq:separate_regression}
\begin{aligned}
	\Exp[Y\mid D=1, W=w, Z=z] 
        & = \mu_1 +  w'\tau_1 + \sum_{m=1}^M \rho_{1m}\lambda_{1m}(w'\pi + \delta(z))  \\
        &\equiv \mu_1 + w'\tau_1 + \lambda_1(w'\pi+\delta(z))'\rho_1\\
	\Exp[Y\mid D=0, W=w, Z=z] 
        & = \mu_0 + w'\tau_0 + \sum_{m=1}^M \rho_{0m}\lambda_{0m}(w'\pi+\delta(z))  \\
        & \equiv \mu_0 + w'\tau_0 + \lambda_0(w'\pi+\delta(z))'\rho_0
\end{aligned}
\end{equation}
where $\rho_d = (\rho_{d1}, \ldots, \rho_{dM})'$ and $\lambda_d(t) = (\lambda_{d1}(t), \ldots, \lambda_{dM}(t))'$ for $t\in \R$, in which we  redefine the control functions as 
\begin{align*}
	&\lambda_{1m}(t) \equiv \frac{1}{F_U(t)}\int_0^{F_U(t)} h_m(u)du \\
	&\lambda_{0m}(t) \equiv \frac{1}{1-F_U(t)}\int_{F_U(t)}^1 h_m(u)du.
\end{align*}

With this additively separable structure, the instrument $Z$ is not necessarily required for point identification of the structural parameters $\mu, \tau$, and $\rho$ as shown by \citet{pan/wang/zhang/zhou:2024}. Consider the following example:

\begin{example}
In a Normal MTE model where $(U,V_d) \sim \normal(0,\Sigma_d)$, and
\[
	\Sigma_d = 
	\begin{pmatrix}
		1 & \rho_d \\
		\rho_d & \sigma_d^2
	\end{pmatrix},
\]
we have
\[
	\Exp[V_d\mid F_U(U) = u] = \rho_d \Phi^{-1}(u),
\]
where $\Phi^{-1}$ denotes the inverse of standard normal CDF. 
So this implies $M = 1$ with $h_1(u) = \Phi^{-1}(u)$ in the equation \eqref{eq:separate_regression}, 
where $\lambda_d(\cdot)$ denotes the inverse Mills' ratio:
\[
	\lambda_1(u) = - \frac{\phi(u)}{\Phi(u)} 
	\quad\text{and}\quad
	\lambda_0(u) = \frac{\phi(u)}{1-\Phi(u)}.
\]
In this model, the structural parameters are not identified when both $\pi = 0$ and $\delta(z) \equiv \delta$. If we only have $\delta(z) \equiv \delta$ for all $z\in \supp(Z)$ (i.e., there is no exogenous variation from IV at all), we can still point identify all parameters given that $\rank(1, W, \lambda_d(W'\pi + \delta)) = 2+L$ almost surely. This point identification is a joint consequence of the additively separable structure and nonlinearity of the transformation $\lambda_d(\cdot)$. It also illustrates how the variation in covariates $W$ can help identify the endogenous coefficients $(\rho_1, \rho_0)$ without instruments, even if $W$ are not excluded from outcome equation. 
\end{example}
From this example, weak identification under additive separability depends jointly on the coefficients of covariates and instruments in the first stage. Even if instruments are weak and cannot generate much variation, the structural coefficients might still be strongly identified if covariates strongly influence individuals' treatment decisions (that $\pi$ is sufficiently distant from zero).

Recall that the parameters of interest are linear functionals of the MTE function:
\begin{align*}
	\MTE(w,u) = \mu_1 - \mu_0 + w'(\tau_1 - \tau_0) + \sum_{m=1}^M (\rho_{1m} - \rho_{0m}) h_m(u).
\end{align*}
For example, consider the PRTE:
\begin{align*}
	\frac{\Exp [Y^\epsilon - Y]}{\Exp [D^\epsilon - D]}
	&= \left.\Exp_{Z,W}\left[\int_{p(W,Z)}^{p^\epsilon(W,Z)} \MTE(W,u) du\right]\right/\Exp[D^\epsilon - D] \\
	&= \mu_1 - \mu_0 + \Exp[W]'(\tau_1 - \tau_0) + \sum_{m=1}^M\frac{\rho_{1m} - \rho_{0m}}{\Exp[D^\epsilon - D]}\times \Exp\left[\int_{p(W,Z)}^{p^\epsilon(W,Z)}h_m(u) du\right]
\end{align*}
where we let 
\[
	\Exp[D^\epsilon - D] = \Exp[p^\epsilon(W,Z) - p(W,Z)] \quad
	\text{and} \quad p(w,z) = F_U(w'\pi + \delta(z)),
\]
where $p^\epsilon(w,z)$ is the counterfactual propensity score discussed in section \ref{sec:empirical} of the main text.

To proceed with inference on the target parameter, I suggest the following step:
\begin{enumerate}
	\item Estimate the first stage and obtain the estimated coefficients $\hat{\pi}, \hat{\delta}(\cdot)$.
	\item Fix the parameters of $\rho_d = \{\rho_{dm}\}_{m=1}^M$ and estimate $\tau_d$ in separate regressions \eqref{eq:separate_regression} by regressing $Y_i - \lambda_d(W_i'\hat{\pi} + \hat{\delta}(Z_i))'\rho_d$ on $(1,W_i')$. Under the assumption that $W$ has full rank (or $\lambda_{\min}(\var(W)) > \epsilon$ uniformly over the class of valid DGPs), we can obtain consistent estimators of $\mu_d$ and $\tau_d$ under a fixed value of $\rho_d$. Denote the estimators as $\hat{\mu}_d(\rho_d)$ and $\hat{\tau}_d(\rho_d)$, respectively.
	\item Following the previous step, we can construct the moment condition of $(\rho_0, \rho_1)$ as below:
	\begin{align*}
		& \hat{m}_1(\rho_1) = \sum_{i=1}^n f_1(W_i, Z_i) D_i\left(Y_i - \hat{\mu}_1(\rho_1) - W_i'\hat{\tau}_1(\rho_1) - \lambda_{1}(W_i'\hat{\pi} + \hat{\delta}(Z_i))'\rho_1\right) = 0_{p\times 1} \\
		& \hat{m}_0(\rho_0) = \sum_{i=1}^n f_0(W_i, Z_i) (1-D_i)\left(Y_i - \hat{\mu}_0(\rho_0) - W_i'\hat{\tau}_0(\rho_0) - \lambda_0(W_i'\hat{\pi} + \hat{\delta}(Z_i))'\rho_0\right) = 0_{p\times 1},
	\end{align*}
	where $f_d(w,z) \in \R^p$ is a measurable function of $(w,z)$ with $p \geq M$ for $d = 0,1$. The functions $\{f_d(w,z)\}_{d=0,1}$ are chosen to avoid the issue discussed in Appendix \ref{appendix:control_function}, such that the asymptotic variance of the moment functions is non-singular for all $\rho_d \in \R^M$ for $d = 0,1$.
	\item With these moment conditions, we can conduct inference on the target parameter\footnote{In the target parameter, $\mu_d$ and $\tau_d$ are replaced with their corresponding estimator $\hat{\mu}_d(\rho_d)$ and $\hat{\tau}_d(\rho_d)$, respectively. Other consistently estimable quantities are replaced with their corresponding estimators.} by applying the improved projection method from section \ref{sec:polyMTE_inference}.
\end{enumerate}

\newpage

\section{Power Comparison in Linear MTE Models}

\label{appendix:add_simulation}

In this appendix, I compare the power of the proposed two approaches, namely, the conditional Wald test proposed in section \ref{sec:linearMTE_inference} and modified linear combination test proposed in section \ref{sec:polyMTE_inference}, and contrast them with the conventional Wald test in a linear MTE model. The linear MTE model is generated by the DGP below: For each $d = 0,1,$
\begin{align*}
    & Y_d = \mu_d + V_d \\
    & D = \indicator[U \leq p(Z)] \\
    & V_d = \rho_{d}\left(U-\frac{1}{2}\right) + e_d,
\end{align*}
where $U$ is uniformly distributed over a unit interval $[0,1]$, $Z$ has the same distribution as the one used in section \ref{sec:simulation}, i.e., uniformly distributed over $\{z_0,z_1,z_2\}$ and independent of $(U,e_1,e_0)$, and $(e_1, e_0)$ follows the joint normal distribution with zero mean and covariance matrix $\Sigma_e = 0.5\cdot I_{2\times 2}$. Since instrument has ternary support, the linear MTE model is over-identified and is weakly identified if all propensity scores converge to a single point. I consider the following specification of the parameters in the Monte-Carlo simulation:
\begin{itemize}
    \item Mean potential outcomes: $\mu_1 = \mu_0 = 0$,
    \item Slope of MTR functions: $\rho_{0} = \rho_{1} = 5$
    \item Propensity scores: 
    \begin{itemize}
        \item[(i)] Strong identification: $p^s(z) = [0.2,0.5,0.8]$
        \item[(ii)] Weak identification: $p^w(z) = [0.4,0.5,0.6]$
    \end{itemize}
\end{itemize}
The data is generated with $n = 500$ units to evaluate the power of the conditional Wald test, modified linear combination test, and the classical Wald test. I repeat the Monte-Carlo experiments 2,000 times to compute the average rejection rates for testing ATE $H_0: \mu_1 - \mu_0 = \delta_\mu$ for $\delta_\mu \in [-5,5]$.

The power curves are plotted in Figure \ref{fig:power_curve_lmte} below under two different levels of identification strength. If the variation on the propensity score is strong enough, all considered approaches are valid under the true value of ATE and exhibit similar power when testing against fixed alternatives in finite samples. This result reveals that the proposed MLC test does not sacrifice too much power under strong identification when the weight assigned to AR statistic is small $(a=0.05)$, which coincides with our local power analysis. Under weak identification, both MLC test and conditional Wald test control the size under the true value at zero, whereas the non-robust Wald test over-rejects the null. Regarding the power of tests, the MLC test may have deficient power at distant alternatives as it exhibits a non-monotonic power curve similar to the RLM test proposed by \citet{kleibergen:2005} for full vector inference. On the other hand, the conditional Wald test has power against both positive and negative values of ATE, which is therefore recommended for practical implementation in linear MTE models.

\begin{figure}[htbp]
    \centering
    
    \caption{Power Curves of the Conditional Wald, MLC, and Wald Tests}
    
    \begin{subfigure}{\textwidth}
        \centering
        \caption{Strong identification: $p(z) = [0.2,0.5,0.8]$}
        \includegraphics[width=0.6\linewidth]{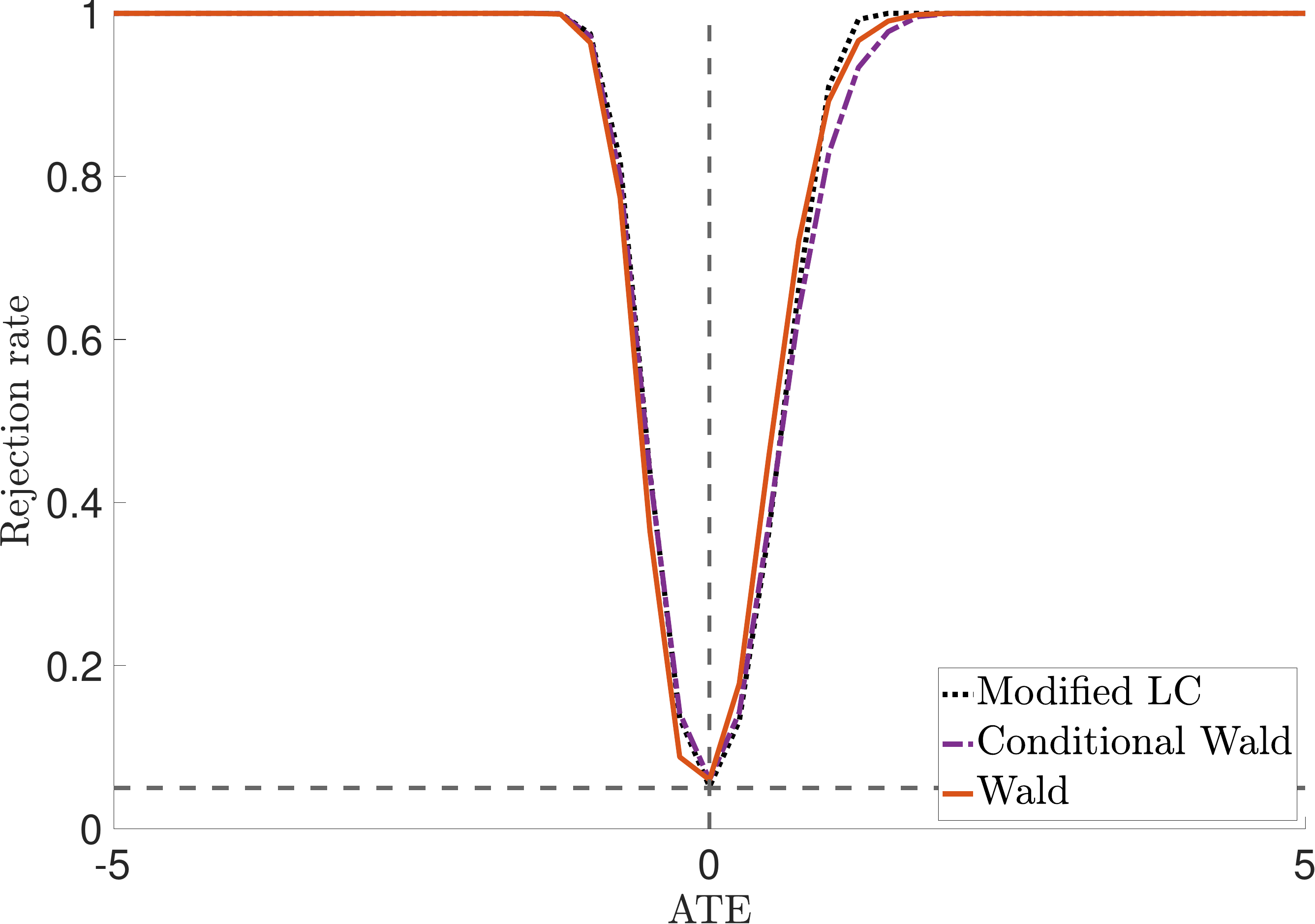}
        \label{fig:subplot_lmte_strong}
    \end{subfigure}
    
    \bigskip
    
    \begin{subfigure}{\textwidth}
        \centering
        \caption{Weak identification: $p(z) = [0.4,0.5,0.6]$}
        \includegraphics[width=0.6\linewidth]{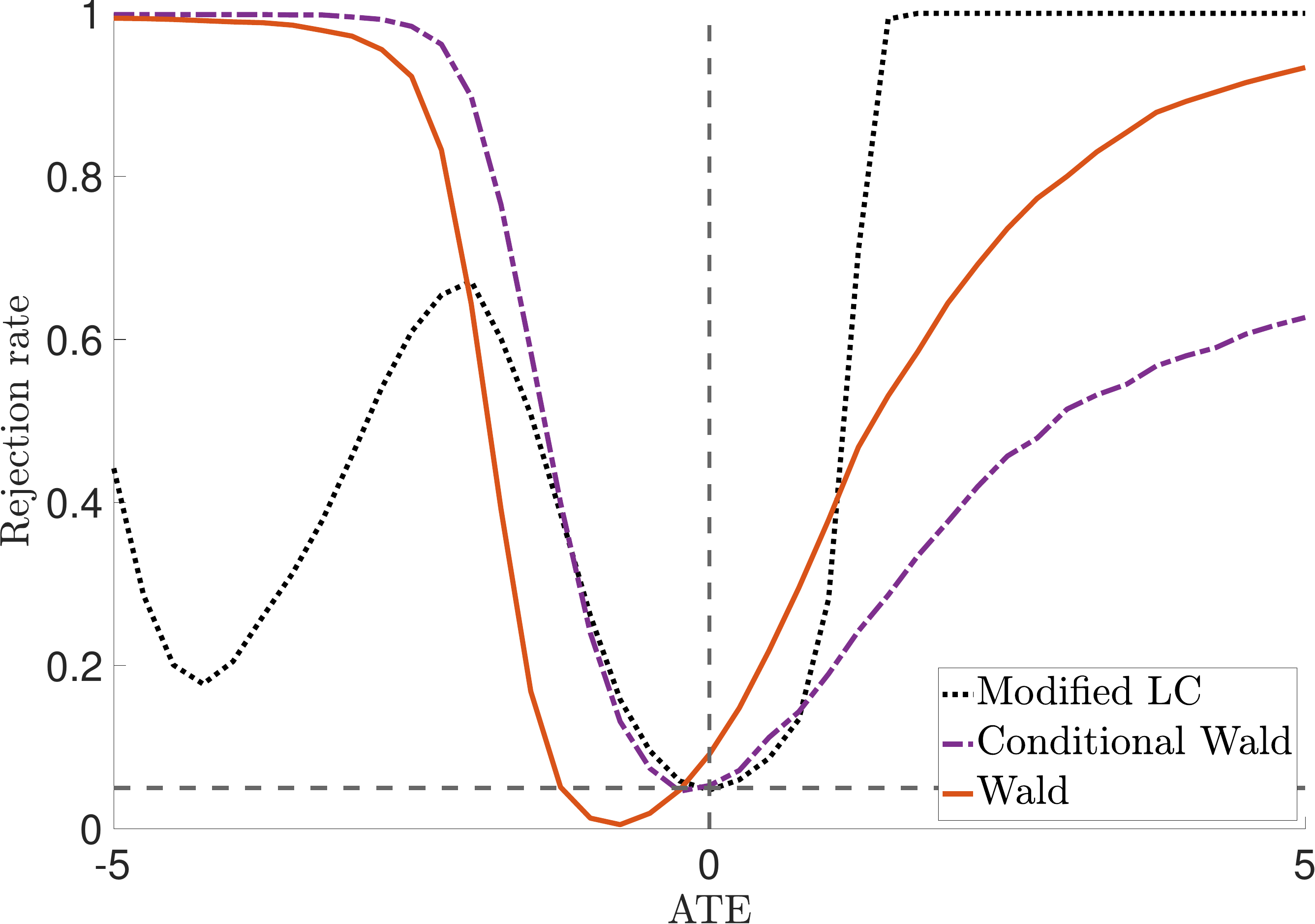}
        \label{fig:subplot_lmte_weak}
    \end{subfigure}

    \smallskip
    
    \begin{tablenotes}
        \footnotesize \raggedright
        {Note: Testing ATE at values on $[-5,5]$ with the true effects fixed at zero. The significance level is 5\%. The sample size equals 500. The average rejection rates are computed with 2,000 independent Monte-Carlo simulations.}
    \end{tablenotes}
    
    \label{fig:power_curve_lmte}
\end{figure}

\newpage

\section{Additional Empirical Results}
\label{appendix:add_empirical_results}

\begin{table}[htbp]
\scriptsize
  \centering
  \caption{95\% (left) and 90\% (right) Court-specific Confidence Sets for Additive MPRTE ${\alpha}_+(0)$}
    \begin{tabular}{llcc}
    \toprule
    Court & MTE Poly. & Wald  & MLC \\
    SBO   & Linear & [-0.17, -0.02] & [-0.16, -0.03] \\
    SBO   & Quadratic & [-0.17, 0.18] & [-0.32, 0.28] \\
    SBO   & Cubic & [-0.40, 0.45] & [-1.00, 0.96] \\
    SBO   & Quartic & [-0.64, 0.46] & [-1.00, 1.00] \\
          &       &       &  \\
    EBOS  & Linear & [-0.24, -0.10] & [-0.24, -0.10] \\
    EBOS  & Quadratic & [-0.22, -0.06] & [-0.25, -0.04] \\
    EBOS  & Cubic & [-0.34, -0.04] & [-0.95, 0.63] \\
    EBOS  & Quartic & [-0.36, 0.04] & [-0.96, 0.57] \\
          &       &       &  \\
    WROX  & Linear & [-0.32, -0.14] & [-0.32, -0.14] \\
    WROX  & Quadratic & [-0.33, -0.12] & [-0.36, -0.05] \\
    WROX  & Cubic & [-0.33, -0.13] & [-1.00, 0.68] \\
    WROX  & Quartic & [-0.33, -0.12] & [-1.00, 0.63] \\
          &       &       &  \\
    BMC   & Linear & [-0.24, 0.00] & [-0.22, -0.01] \\
    BMC   & Quadratic & [-0.28, -0.02] & [-0.26, 0.34] \\
    BMC   & Cubic & [-0.28, 0.00] & [-0.65, 0.79] \\
    BMC   & Quartic & [-0.37, -0.07] & [-0.62, 0.77] \\
          &       &       &  \\
    ROX   & Linear & [-0.03, 0.23] & [0.01, 0.22] \\
    ROX   & Quadratic & [-0.02, 0.28] & [-0.03, 0.30] \\
    ROX   & Cubic & [-0.02, 0.34] & [-0.37, 0.70] \\
    ROX   & Quartic & [-0.10, 0.31] & [-0.28, 0.77] \\
          &       &       &  \\
    DOR   & Linear & [-0.42, -0.18] & [-0.42, -0.20] \\
    DOR   & Quadratic & [-0.60, -0.26] & [-0.72, -0.20] \\
    DOR   & Cubic & [-0.60, -0.26] & [-1.00, 0.15] \\
    DOR   & Quartic & [-0.67, -0.29] & [-0.98, 0.25] \\
    \bottomrule
    \end{tabular}%
		\quad
    \begin{tabular}{llcc}
    \toprule
    Court & MTE Poly. & Wald  & MLC \\
    SBO   & Linear & [-0.16, -0.04] & [-0.15, -0.05] \\
    SBO   & Quadratic & [-0.14, 0.15] & [-0.26, 0.24] \\
    SBO   & Cubic & [-0.33, 0.38] & [-0.91, 0.78] \\
    SBO   & Quartic & [-0.55, 0.37] & [-0.98, 1.00] \\
          &       &       &  \\
    EBOS  & Linear & [-0.23, -0.11] & [-0.23, -0.11] \\
    EBOS  & Quadratic & [-0.21, -0.08] & [-0.22, -0.07] \\
    EBOS  & Cubic & [-0.31, -0.06] & [-0.77, 0.54] \\
    EBOS  & Quartic & [-0.32, 0.01] & [-0.84, 0.42] \\
          &       &       &  \\
    WROX  & Linear & [-0.31, -0.15] & [-0.31, -0.15] \\
    WROX  & Quadratic & [-0.31, -0.14] & [-0.33, -0.10] \\
    WROX  & Cubic & [-0.31, -0.14] & [-0.84, 0.23] \\
    WROX  & Quartic & [-0.31, -0.13] & [-0.97, 0.51] \\
          &       &       &  \\
    BMC   & Linear & [-0.22, -0.02] & [-0.19, -0.04] \\
    BMC   & Quadratic & [-0.26, -0.04] & [-0.23, 0.15] \\
    BMC   & Cubic & [-0.26, -0.02] & [-0.45, 0.45] \\
    BMC   & Quartic & [-0.35, -0.09] & [-0.49, 0.46] \\
          &       &       &  \\
    ROX   & Linear & [-0.01, 0.21] & [0.04, 0.19] \\
    ROX   & Quadratic & [0.00, 0.26] & [0.02, 0.26] \\
    ROX   & Cubic & [0.01, 0.31] & [-0.20, 0.51] \\
    ROX   & Quartic & [-0.06, 0.28] & [-0.06, 0.39] \\
          &       &       &  \\
    DOR   & Linear & [-0.40, -0.20] & [-0.40, -0.23] \\
    DOR   & Quadratic & [-0.57, -0.29] & [-0.60, -0.27] \\
    DOR   & Cubic & [-0.57, -0.28] & [-0.75, -0.09] \\
    DOR   & Quartic & [-0.64, -0.32] & [-0.85, 0.01] \\
    \bottomrule
    \end{tabular}%
  \label{tab:A-MPRTE}%

    \begin{tablenotes}
    \footnotesize \raggedright
      Note: {Results for classical Wald and MLC confidence sets under polynomial specifications up to the fourth order for courts of South Boston (SBO), East Boston (EBOS), West Roxbury (WROX), Central (BMC), Roxbury (ROX), and Dorchester (DOR). The counterfactual policy considered is a uniformly additive increase in nonprosecution rates by a marginal amount. The left panel reports the 95\% confidence sets, and the right panel reports the 90\% confidence sets.}
    \end{tablenotes}
    
\end{table}%

\begin{table}[htbp]
	\scriptsize
  \centering
  \caption{95\% (left) and 90\% (right) Court-specific Confidence Sets for Proportional MPRTE ${\alpha}_+(0)$}
    \begin{tabular}{llcc}
		\toprule
    Court & MTE Poly. & Wald  & MLC \\
    SBO   & Linear  & [-0.12, 0.04] & [-0.11, 0.03] \\
    SBO   & Quadratic  & [-0.29, 0.20] & [-0.52, 0.36] \\
    SBO   & Cubic  & [-0.45, 0.52] & [-1.00, 1.00] \\
    SBO   & Quartic  & [-0.46, 1.03] & [-1.00, 1.00] \\
          &       &       &  \\
    EBOS  & Linear  & [-0.26, -0.12] & [-0.27, -0.12] \\
    EBOS  & Quadratic  & [-0.26, -0.04] & [-0.27, 0.01] \\
    EBOS  & Cubic  & [-0.37, -0.01] & [-1.00, 1.00] \\
    EBOS  & Quartic  & [-0.48, 0.16] & [-1.00, 0.95] \\
          &       &       &  \\
    WROX  & Linear  & [-0.33, -0.14] & [-0.33, -0.14] \\
    WROX  & Quadratic  & [-0.34, -0.10] & [-0.38, -0.03] \\
    WROX  & Cubic  & [-0.36, -0.12] & [-1.00, 0.88] \\
    WROX  & Quartic  & [-0.37, -0.12] & [-1.00, 0.90] \\
          &       &       &  \\
    BMC   & Linear  & [-0.23, 0.01] & [-0.20, 0.00] \\
    BMC   & Quadratic  & [-0.22, 0.02] & [-0.19, 0.86] \\
    BMC   & Cubic  & [-0.20, 0.08] & [-0.74, 0.53] \\
    BMC   & Quartic  & [-0.47, 0.01] & [-0.72, 0.74] \\
          &       &       &  \\
    ROX   & Linear  & [-0.04, 0.21] & [0.00, 0.20] \\
    ROX   & Quadratic  & [-0.04, 0.21] & [-0.10, 0.23] \\
    ROX   & Cubic  & [-0.03, 0.27] & [-0.49, 0.86] \\
    ROX   & Quartic  & [-0.12, 0.42] & [-0.34, 0.97] \\
          &       &       &  \\
    DOR   & Linear  & [-0.42, -0.18] & [-0.42, -0.21] \\
    DOR   & Quadratic  & [-0.48, -0.23] & [-1.00, -0.10] \\
    DOR   & Cubic  & [-0.50, -0.23] & [-1.00, 0.47] \\
    DOR   & Quartic  & [-0.84, -0.28] & [-1.00, 0.47] \\
          &       &       &  \\
    Uncond.  & Linear  & [-0.19, -0.12] & [-0.17, -0.14] \\
    Uncond.  & Quadratic  & [-0.18, -0.10] & [-0.16, -0.11] \\
    Uncond.  & Cubic  & [-0.16, -0.04] & [-0.13, -0.05] \\
    Uncond.  & Quartic  & [-0.16, 0.04] & [-0.08, 0.02] \\
          &       &       &  \\
    Average & Linear  & [-0.24, -0.03] & [-0.23, -0.04] \\
    Average & Quadratic  & [-0.27, -0.01] & [-0.42, 0.20] \\
    Average & Cubic  & [-0.30, 0.04] & [-0.84, 0.74] \\
    Average & Quartic  & [-0.46, 0.11] & [-0.80, 0.80] \\
    \bottomrule
    \end{tabular}%
		\quad
    \begin{tabular}{llcc}
		\toprule
    Court & MTE Poly. & Wald  & MLC \\
    SBO   & Linear  & [-0.11, 0.03] & [-0.09, 0.02] \\
    SBO   & Quadratic  & [-0.25, 0.16] & [-0.43, 0.26] \\
    SBO   & Cubic  & [-0.37, 0.44] & [-1.00, 1.00] \\
    SBO   & Quartic  & [-0.34, 0.91] & [-1.00, 1.00] \\
          &       &       &  \\
    EBOS  & Linear  & [-0.25, -0.13] & [-0.25, -0.13] \\
    EBOS  & Quadratic  & [-0.24, -0.06] & [-0.25, -0.04] \\
    EBOS  & Cubic  & [-0.34, -0.04] & [-1.00, 0.96] \\
    EBOS  & Quartic  & [-0.43, 0.11] & [-1.00, 0.81] \\
          &       &       &  \\
    WROX  & Linear  & [-0.31, -0.16] & [-0.31, -0.16] \\
    WROX  & Quadratic  & [-0.32, -0.12] & [-0.34, -0.09] \\
    WROX  & Cubic  & [-0.34, -0.14] & [-1.00, 0.34] \\
    WROX  & Quartic  & [-0.35, -0.14] & [-1.00, 0.80] \\
          &       &       &  \\
    BMC   & Linear  & [-0.21, -0.01] & [-0.17, -0.03] \\
    BMC   & Quadratic  & [-0.20, 0.00] & [-0.16, 0.59] \\
    BMC   & Cubic  & [-0.17, 0.06] & [-0.52, 0.40] \\
    BMC   & Quartic  & [-0.43, -0.03] & [-0.37, 0.16] \\
          &       &       &  \\
    ROX   & Linear  & [-0.02, 0.19] & [0.02, 0.17] \\
    ROX   & Quadratic  & [-0.02, 0.19] & [0.00, 0.18] \\
    ROX   & Cubic  & [-0.01, 0.24] & [-0.37, 0.48] \\
    ROX   & Quartic  & [-0.07, 0.38] & [-0.18, 0.54] \\
          &       &       &  \\
    DOR   & Linear  & [-0.40, -0.20] & [-0.39, -0.23] \\
    DOR   & Quadratic  & [-0.46, -0.25] & [-0.49, -0.18] \\
    DOR   & Cubic  & [-0.48, -0.25] & [-0.82, 0.22] \\
    DOR   & Quartic  & [-0.79, -0.32] & [-0.81, 0.20] \\
          &       &       &  \\
    Uncond.  & Linear  & [-0.19, -0.13] & [-0.16, -0.15] \\
    Uncond.  & Quadratic  & [-0.17, -0.10] & [-0.15, -0.12] \\
    Uncond.  & Cubic  & [-0.15, -0.05] & [-0.11, -0.07] \\
    Uncond.  & Quartic  & [-0.15, 0.02] & [-0.07, -0.01] \\
          &       &       &  \\
    Average & Linear  & [-0.22, -0.05] & [-0.20, -0.06] \\
    Average & Quadratic  & [-0.25, -0.03] & [-0.26, 0.10] \\
    Average & Cubic  & [-0.27, 0.01] & [-0.73, 0.48] \\
    Average & Quartic  & [-0.41, 0.06] & [-0.66, 0.50] \\
    \bottomrule
    \end{tabular}%
  \label{tab:P-MPRTE}%

    \begin{tablenotes}
    \footnotesize \raggedright
      Note: {Results for classical Wald and MLC confidence sets under polynomial specifications up to the fourth order for courts of South Boston (SBO), East Boston (EBOS), West Roxbury (WROX), Central (BMC), Roxbury (ROX), Dorchester (DOR). The ``Uncond.'' confidence sets are obtained without controlling for court identities, whereas the ``Average'' confidence sets are weighted averages of the court-specific confidence sets with weights proportional to court size. The counterfactual policy considered is a uniformly proportional increase in nonprosecution rates by a marginal amount. The left panel reports the 95\% confidence sets, and the right panel reports the 90\% confidence sets.}
    \end{tablenotes}
    
\end{table}

\newpage

\section{Testing IV Strength}
\label{appendix:iv_strength}

Several studies have developed tests for IV strength \citep{stock/yogo:2005, montielolea/pflueger:2013, lewis/mertens:2022}. These tests typically focus on the bias of IV estimators relative to OLS estimators and the size distortion of conventional Wald and t-tests under their null hypotheses. This problem is well-studied in linear IV regressions with homoskedastic error structure \citep[Section 4]{andrews/stock/sun:2019}. However, little is known about nonlinear models with a focus on subvector inference. In this section, I evaluate the strength of the identification in the empirical application along multiple dimensions.

\subsection{Pre-testing weak identification by size distortion}
Following I. \cite{andrews.i:2018}, we regard the target parameter as being weakly identified if the Wald confidence set fails to include a locally equivalent robust confidence set with a strictly smaller coverage level.

Let $\gamma \in [0,1-\alpha)$. We choose the weight $a = a(\gamma)$ in the MLC test such that $a(\gamma)$ solves
\[
	q_{(1+a(\gamma))\chi_1^2 + a(\gamma)\chi_{2K+1}^2}(1 - \alpha - \gamma) = q_{\chi_1^2}(1-\alpha).
\]
With the choice of weight $a(\gamma)$ and the critical value $q_{\chi_1^2}(1-\alpha)$, the following set $\mathcal{C}_P(\gamma)$ achieves asymptotic coverage rate $1-\alpha-\gamma$ under weak identification:
\[
	\mathcal{C}_{P}(\gamma) = \left\{\lambda\in\R: \inf_{c'\theta = \lambda}\text{MRLM}_n(\theta) + a(\gamma)\cdot \text{AR}_n(\theta) < q_{\chi_1^2}(1-\alpha)\right\}.
\]
As shown by \citet[Theorem 3]{andrews.i:2018}, this set is contained by a non-robust $(1-\alpha)$ Wald confidence set $\mathcal{C}_{W}$ with probability approaching one under \textit{strong identification}. In other words, the failure of this containment relationship suggests evidence of weak identification. This leads to the following test of weak identification:
\[
	\phi_{ICS}(\gamma) = \indicator[\mathcal{C}_P(\gamma) \not\subseteq \mathcal{C}_W].
\]
The tuning parameter $\gamma$ can be regarded as the maximum amount of size distortion on Wald confidence set that researchers can tolerate. If $\phi_{ICS}(\gamma) = 1$, then researchers are willing to use the robust confidence set with a smaller coverage level ($1-\alpha-\gamma$) since it is always valid while the classical $(1-\alpha)$ Wald confidence set is unreliable if the test rejects. However, this type of test may not have sufficient power to detect weak identification, which implies the model might still be weakly identified even if we did not reject the test.

For the example of additive MPRTE in Table \ref{tab:A-MPRTE}, I report the Andrews' pretesting result in Table \ref{tab:IVstrength_AMPRTE} by specifying $\gamma = 10\%$ and $\alpha = 5\%$ (i.e., the researchers can tolerate at most 10\% size distortion). The result indicates that cubic and quartic MTE models lead to potential concerns of weak identification conditional on most courts, while linear MTE models are strongly identified. Quadratic MTE models are strongly identified for several courts, but not all of them. This result reveals the fact that the identification strength depends on both the variation of instruments and the flexibility of the model specified by researchers, while this information is not reflected by the conventional rule-of-thumb applied to $F$-statistics, since the latter rule only works under homoskedastic linear IV models with a scalar endogenous coefficient \citep{andrews/stock/sun:2019}.

\begin{table}[htbp]
  \centering
  \caption{Testing Weak Identification of Additive MPRTE at $\gamma = 10\%$ and $\alpha = 5\%$}
  \begin{threeparttable}
    \begin{tabular}{llccccc}
    \toprule
    Court & MTE Poly. & Valid Test & Wald (95\%) & MLC (85\%) & F-stat & RoT \\
    \midrule
    SBO   & Linear MTE & Possibly Strong & [-0.17, -0.02] & [-0.14, -0.06] & \multirow{4}[1]{*}{77.90} & \multirow{4}[1]{*}{Strong} \\
    SBO   & Quadratic MTE & Weak  & [-0.17, 0.18] & [-0.18, 0.20] &       &  \\
    SBO   & Cubic MTE & Weak  & [-0.40, 0.45] & [-0.67, 0.55] &       &  \\
    SBO   & Quartic MTE & Weak  & [-0.64, 0.46] & [-0.76, 0.75] &       &  \\
          &       &       &       &       &       &  \\
    EBOS  & Linear MTE & Possibly Strong & [-0.24, -0.10] & [-0.23, -0.12] & \multirow{4}[0]{*}{50.01} & \multirow{4}[0]{*}{Strong} \\
    EBOS  & Quadratic MTE & Possibly Strong & [-0.22, -0.06] & [-0.20, -0.08] &       &  \\
    EBOS  & Cubic MTE & Weak  & [-0.34, -0.04] & [-0.45, 0.03] &       &  \\
    EBOS  & Quartic MTE & Weak  & [-0.36, 0.04] & [-0.63, 0.37] &       &  \\
          &       &       &       &       &       &  \\
    WROX  & Linear MTE & Possibly Strong & [-0.32, -0.14] & [-0.30, -0.16] & \multirow{4}[0]{*}{28.15} & \multirow{4}[0]{*}{Strong} \\
    WROX  & Quadratic MTE & Possibly Strong & [-0.33, -0.12] & [-0.31, -0.13] &       &  \\
    WROX  & Cubic MTE & Weak  & [-0.33, -0.13] & [-0.40, 0.05] &       &  \\
    WROX  & Quartic MTE & Weak  & [-0.33, -0.12] & [-0.89, 0.44] &       &  \\
          &       &       &       &       & \textcolor[rgb]{ 0,  .114,  .784}{} &  \\
    BMC   & Linear MTE & Possibly Strong & [-0.24, 0.00] & [-0.16, -0.07] & \multirow{4}[0]{*}{31.47} & \multirow{4}[0]{*}{Strong} \\
    BMC   & Quadratic MTE & Weak  & [-0.28, -0.02] & [-0.21, 0.04] &       &  \\
    BMC   & Cubic MTE & Weak  & [-0.28, 0.00] & [-0.35, 0.12] &       &  \\
    BMC   & Quartic MTE & Weak  & [-0.37, -0.07] & [-0.30, 0.18] &       &  \\
          &       &       &       &       &       &  \\
    ROX   & Linear MTE & Possibly Strong & [-0.03, 0.23] & [0.06, 0.17] & \multirow{4}[0]{*}{34.68} & \multirow{4}[0]{*}{Strong} \\
    ROX   & Quadratic MTE & Possibly Strong & [-0.02, 0.28] & [0.08, 0.23] &       &  \\
    ROX   & Cubic MTE & Weak  & [-0.02, 0.34] & [0.05, 0.38] &       &  \\
    ROX   & Quartic MTE & Possibly Strong & [-0.10, 0.31] & [0.13, 0.24] &       &  \\
          &       &       &       &       &       &  \\
    DOR   & Linear MTE & Possibly Strong & [-0.42, -0.18] & [-0.37, -0.25] & \multirow{4}[1]{*}{43.60} & \multirow{4}[1]{*}{Strong} \\
    DOR   & Quadratic MTE & Possibly Strong & [-0.60, -0.26] & [-0.56, -0.36] &       &  \\
    DOR   & Cubic MTE & Weak  & [-0.60, -0.26] & [-0.56, -0.16] &       &  \\
    DOR   & Quartic MTE & Weak  & [-0.67, -0.29] & [-0.79, -0.08] &       &  \\
    \bottomrule
    \end{tabular}%
    \begin{tablenotes}
    \footnotesize Note: Valid Test is based on I.\cite{andrews.i:2018} first-stage pretesting procedure under 10\% size distortion. $F$-statistics are computed by regressing treatment on a set of saturated instruments, and the Rule of Thumb (RoT) suggests strong identification if the $F$-statistics exceed 10.
    \end{tablenotes}
    \end{threeparttable}
  \label{tab:IVstrength_AMPRTE}%
\end{table}%

\subsection{Testing exact under-identification}
Alternatively, researchers may want to test the full-rank property of the matrix $A$ in the moment condition, i.e., $H_0^{\text{rank}}: \rank(A) \leq 2M+1$. Rejections of such tests suggest that the propensity scores are sufficiently separate from each other so that the relevance condition holds under a \textit{pointwise asymptotic} framework. However, for a sequence of DGPs that leads to weak identification,  $H_0^{\text{rank}}$ may or may not be rejected since Assumption \ref{asp:identification}.7 holds but becomes nearly violated as sample size diverges. This implies a \textit{larger} critical value is needed to detect weak identification. As choice of valid critical values for testing weak identification is beyond the scope of this paper, I do not pursue this direction but instead comparing the existing test statistics and their corresponding critical values for testing $H_0^{\text{rank}}$. The model may be weakly identified if the test statistic falls below but is close to the relevant critical value.

The test for $H_0^{\text{rank}}$ considered here is based on the analytical bootstrap test from \cite{chen/fang:2019}. Their test focuses exactly on the null hypothesis $H_0^{\text{rank}}$ rather than testing the equality $\rank(A) = 2M+1$ as in earlier work by \citet{kleibergen/paap:2006}. They show that KP's approach can be size-incorrect to test for $H_0^{\text{rank}}$. Since IV relevance condition at population level is directly related to $H_0^{\text{rank}}$, I adopt \cite{chen/fang:2019}'s approach to study this problem. 

The test results are collected in Table \ref{tab:IVstrength_CF19}. From this table we see that the \textit{population} IV relevance condition holds for all courts and models we consider except for the SBO court based on quartic MTE models. As argued above, this result does not necessarily show evidence of strong identification. When looking at the ratios of test statistics and critical values, we observe that these quantities are close to one under cubic and quartic settings, indicating potential weak identification concerns in such scenarios. 

\begin{table}[htbp]
  \centering
  \caption{Analytical Bootstrap Test for $H_0^{\text{rank}}: \rank(A) \leq 2M+1$}
    \begin{tabular}{lcccc}
    \toprule
    Court & Linear & Quadratic & Cubic & Quartic \\
    \midrule
    \multicolumn{5}{c}{Ratio: Test Statistics/Critical Values} \\
    SBO   & 20.27 & 1.44  & 1.83 & 0.57 \\
    EBOS  & 17.47 & 3.01  & 2.19 & 1.11 \\
    WROX  & 23.29 & 11.41 & 3.68 & 2.46 \\
    BMC   & 16.47 & 7.54  & 3.66 & 1.20 \\
    ROX   & 16.73 & 6.08  & 2.70 & 1.01 \\
    DOR   & 27.58 & 7.72  & 2.90 & 1.32 \\
    \midrule
    \multicolumn{5}{c}{Test Statistics} \\
    SBO   & 469.57 & 1.02 & 9.02$\times 10^{-3}$ & 3.18$\times 10^{-5}$ \\
    EBOS  & 506.15 & 1.74 & 8.49$\times 10^{-3}$ & 4.80$\times 10^{-5}$ \\
    WROX  & 454.48 & 3.05 & 16.68$\times 10^{-3}$ & 12.18$\times 10^{-5}$ \\
    BMC   & 243.39 & 1.05 & 3.56$\times 10^{-3}$ & 0.86$\times 10^{-5}$ \\
    ROX   & 205.75 & 0.39 & 0.82$\times 10^{-3}$ & 0.12$\times 10^{-5}$ \\
    DOR   & 418.09 & 1.71 & 6.08$\times 10^{-3}$ & 1.50$\times 10^{-5}$ \\
    \midrule
    \multicolumn{5}{c}{Critical Values} \\
    SBO   & 23.17 & 0.71 & 4.92$\times 10^{-3}$ & 5.60$\times 10^{-5}$ \\
    EBOS  & 28.97 & 0.58 & 3.88$\times 10^{-3}$ & 4.34$\times 10^{-5}$ \\
    WROX  & 19.52 & 0.27 & 4.54$\times 10^{-3}$ & 4.94$\times 10^{-5}$ \\
    BMC   & 14.78 & 0.14 & 0.97$\times 10^{-3}$ & 0.72$\times 10^{-5}$ \\
    ROX   & 12.30 & 0.06 & 0.31$\times 10^{-3}$ & 0.12$\times 10^{-5}$ \\
    DOR   & 15.16 & 0.22 & 2.09$\times 10^{-3}$ & 1.14$\times 10^{-5}$ \\
    \bottomrule
    \end{tabular}%
  \label{tab:IVstrength_CF19}%
\end{table}%

\newpage

\section{Numerical Example of Additive Separability Bias}
\label{appendix:numerical_example}
In this appendix, we present a numerical example and show that  misspecification of additive separability structure can lead to potentially large bias of causal estimand even if covariates $W$ only exhibits binary variation.

As alluded in section \ref{sec:bias_add_sep}, the bias of ATE would depend both on the heterogeneity of the coefficient $\rho_1(W) - \rho_0(W)$ as well as how $W$ influence the selection into treatment. Consider a linear MTE model as follows:
\begin{align*}
    P & = \frac{\exp(bW + Z)}{1 + \exp(bW + Z)} \\
    Y_d & = \mu_d(W) + \rho_d(W)\left(U - \frac{1}{2}\right) + e_d \\
    D & = \indicator[P \leq U]
\end{align*}
where $\Prob(W = 1) = \Prob(W = 0) = 0.5$ and $F_Z(z) = (1+e^{-z})^{-1}$, following a standard logistic distribution. Let $U \indep (Z,W)$ and $Z\indep W$. For the parameter specification, I assume
\begin{itemize}
    \item $\mu_1(w) = 0.5$ and $\mu_0(w) = 0$ for $w = 0,1$.
    \item $\rho_1(0) = 1, \rho_1(1) = 1+\delta_\rho$.
    \item $\rho_0(0) = -1, \rho_0(1) = -1+\delta_\rho$.
    \item $(e_1, e_0) \sim \normal(0_{2\times 1}, I_2)$.
    \item $(b,\delta_\rho) \in [-5,5]^2$.
\end{itemize}
From the above specification, ATE equals 0.5 for all pairs of $(b,\delta_\rho)$. When $b=0$, the covariate $W$ does not shift the probability of being treated. When $\delta_\rho = 0$, there is no heterogeneous treatment effects of covariates varying with selection unobservable (i.e., the additive separability condition holds). Following the intuition provided by Theorem \ref{thm:bias_formula_specific}, we would expect bias of ATE estimand when $b\neq 0$ and $\delta_\rho \neq 0$.

The ATE estimand under additive separability and its deviation from the true ATE are shown in Figure \ref{fig:bias_plot}. This figure demonstrates that bias of ATE estimand is pronounced when both $b$ and $\delta_\rho$ are large, and there is no bias on ATE estimand if either $b$ or $\delta_\rho$ equals zero, which coincides the prediction by Theorem \ref{thm:bias_formula_specific}. Importantly, this bias has the potential to change the sign of estimand, therefore resulting in problematic policy recommendations.

\begin{figure}[htbp]
    \caption{Bias of ATE Estimand under Additive Separability}
    
    \centering
    
    \includegraphics[scale = 0.7]{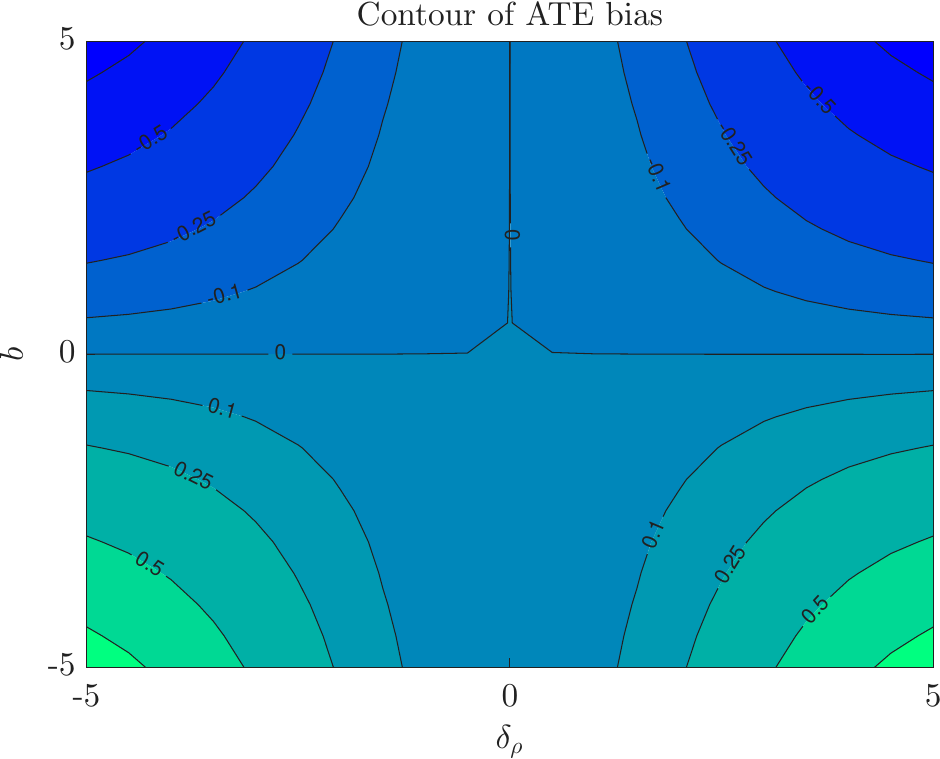} 
    
    \includegraphics[scale = 0.7]{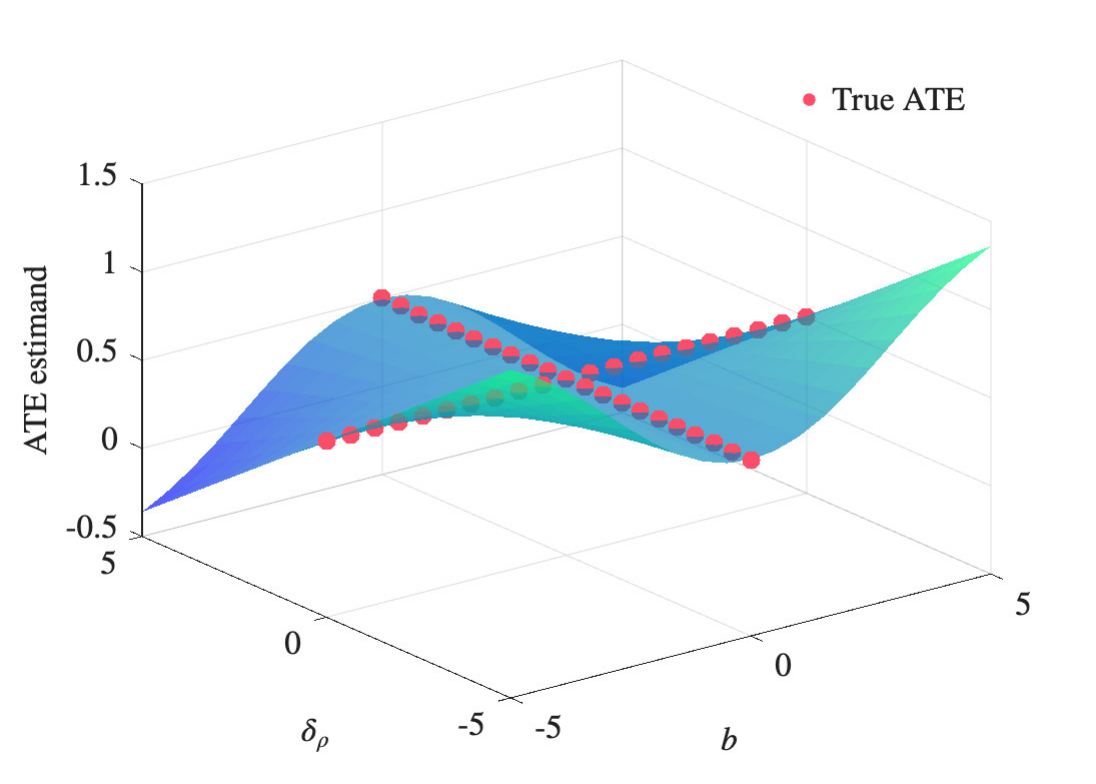}
    
    \label{fig:bias_plot}
    
    \begin{tablenotes}
    \footnotesize Top: Contour plot of ATE bias for $(b,\delta_\rho) \in [-5,5]^2$. Bottom: 3D plot of ATE estimand (under additive separability) over the same range, with scatter points showing where the estimand equals true ATE.
    \end{tablenotes}
    
\end{figure}

\newpage 

\section{Additional Literature Discussion}
\label{appendix:litreview}

In this appendix, I provide further discussion on other inference procedures in weakly identified models and review how covariates can be incorporated in these methods.

\subsection{Other approaches to inference on subvectors or functions of parameters}
\label{appendix:lit_subvector_inference}
\citet[Section 12 in supplementary appendix]{andrews/guggenberger:2019} and D. \citet[Section 2]{andrews.d:2017} give thorough literature reviews on inference for weakly identified models. Next I briefly discuss several other main approaches for inference on subvectors or functions of parameters and explain why they cannot be applied to MTE models considered in this paper.

\begin{enumerate}
\item Concentrating out strongly identified nuisance parameters \par
Sometimes researchers assume certain parameters are strongly identified under the null hypothesis. When this holds, identification-robust tests can be constructed by either concentrating out these strongly identified parameters or substituting them with consistent, asymptotically normal estimators. This approach has been adopted in many papers, including several influential work by \cite{stock/wright:2000}, \cite{kleibergen:2005}, \cite{andrews/mikusheva:2016b}, and \cite{andrews/guggenberger:2019}. The resulting confidence sets achieve asymptotic similarity. However, a key limitation is that this method cannot accommodate weakly identified nuisance parameters. This limitation is particularly relevant for the MTE model considered here, where elements of the parameter vector $\theta$ may be weakly identified when propensity scores exhibit limited variation (see footnote 4 on page 17).

\item Least-favorable critical values and identification-category-selection method \par
\citet{andrews/cheng:2012,andrews/cheng:2013,andrews/cheng:2014} and \citet{han/mccloskey:2019} propose uniformly valid inference results for inference on functions of weakly identified parameters based on conventional test statistics such as quasi-likelihood ratio and Wald statistics. These methods are developed for models for which the identification status is determined by whether a vector of strongly identified parameters is zero. When this vector equals zero, the sample objective function does not involve the weakly identified parameters \citep[Assumption A]{andrews/cheng:2012} or the Jacobian matrix of the objective function is approximately a singular matrix \citep[Assumption ID]{han/mccloskey:2019}. For the MTE models studied in this paper, the identification status is determined by the vector of propensity scores, with under-identification occurring if this propensity score vector approaches a set on which Assumption \ref{asp:identification}.7 fails, rather than a single point. As a result, failure of identification can happen along multiple directions in higher-order MTE models, rather than induced only by a vector approaching zero. 
While some studies have attempted to generalize the framework of \cite{andrews/cheng:2012}, they do not yet cover MTE models studied in this paper. For example, \cite{cheng:2015} considers nonlinear regression models with mixed identification strength where each structural parameter has a one-to-one correspondence with its identification strength parameter. \cite{cox:2025} develops identification-robust inference for a class of minimum-distance models with restricted parameter spaces, with a primary focus on low-dimensional factor models under non-negativity restrictions of variances. 

\item Profiling-based Anderson-Rubin test \par
\cite{gkmc:2012} and \cite{gkm:2019} propose using profiled Anderson-Rubin test to study inference on a subvector of endogenous coefficients in linear IV models. Their tests rely on the assumption that structural errors are homoskedastic, which is proved to be a key assumption for maintaining asymptotic validity by \cite{lee:2015}. In another subsequent work, \cite{gkm:2021} relax the homoskedasticity and develop an adaptive procedure by switching into the general subvector test in D. \cite{andrews.d:2017} if the structural residuals do not follow the approximate homoskedastic structure. Despite the simplicity and convenience of their tests, the assumption of homoskedasticity and linear IV structure do not apply to MTE models. In another related work by \cite{andrews/mikusheva:2016a}, the subvector test is constructed by profiling Anderson-Rubin-type statistics in a minimum distance model. However, as they note (page 1260), their method reduces to conservative projection inference in linear IV models. A similar limitation likely applies to MTE models, given the linear parametric structure of the moment function. 
\end{enumerate}

\subsubsection*{Asymptotic similarity of the subvector inference}

Asymptotic similarity guarantees that, in large samples, the confidence set coverage equals $1-\alpha$, or equivalently, the size of the test equals $\alpha$, regardless of identification strength (for a formal definition of similarity, see \cite{andrews/cheng/guggenberger:2020}). This property means that confidence sets are not conservative, i.e., unnecessarily wide with coverage exceeding the nominal level, in large samples. While \cite{andrews/guggenberger:2017} has established this property for many identification-robust inference methods on the full parameter vector, it rarely holds for subvector inference with weakly identified nuisance parameters. This failure occurs because weakly identified nuisance parameters usually cannot be consistently estimated under weak identification. Below, I review the main approaches in this literature and present evidence of their inability to achieve asymptotic similarity.

Methods following \cite{andrews/cheng:2012,andrews/cheng:2013,andrews/cheng:2014} determine critical values by considering the least-favorable case across all sequences of DGPs. By construction, this robust approach can yield critical values that exceed the true quantile of the test statistic's limiting distribution. Consequently, these tests are not asymptotically similar, as evidenced by the over-coverage probabilities in \citet[Figure 7]{andrews/cheng:2012}.

For linear IV models with homoskedastic errors, the methods of \cite{gkmc:2012} and \cite{gkm:2019} also lack asymptotic similarity. As shown in \citet[Figure 4]{gkm:2019}, null rejection probabilities fall below the significance level when the endogenous coefficient not under testing is weakly identified.

The methods of \cite{chaudhuri/zivot:2011}, D. \cite{andrews.d:2017}, and I. \cite{andrews.i:2018} also fail to achieve asymptotic similarity. These approaches minimize an identification-robust test statistic over a set of weakly identified nuisance parameters (either a confidence set or the full parameter domain), but derive critical values using the true parameter values. Under weak identification, the minimizing nuisance parameters may not converge to their true values, causing the profiled test statistics' distribution to be stochastically dominated by the test statistics evaluated at the true value. This leads to null rejection probabilities falling below the significance level, as demonstrated by the under-rejection of the null hypothesis in both \citet[Table 2]{chaudhuri/zivot:2011} and D. \citet[Table I]{andrews.d:2017} for AR/QLR1 tests.

It is worth noting that reparameterizing the model parameters can sometimes make simple plug-in tests feasible, since the new parameter vector may be separated into a strongly identified nuisance vector and a weakly identified scalar parameter whose value is uniquely determined under the null for subvector testing, see \cite{cox:2024} for examples of low-dimensional factor models. However, as I show in Appendix \ref{appendix:control_function}, applying such a reparameterization to the simple MTE model with a scalar endogenous coefficient ($M=1$) yields at least two weakly identified nuisance parameters, which are the slope coefficients of MTE functions. By contrast, the approach proposed in section \ref{sec:linearMTE_inference} involves only a single weakly identified parameter, namely, the causal parameter $c'\theta$ of interest.

\subsection{Covariates in weakly identified models}
\label{appendix:cov_weakIV}
The covariates are often ignored in the literature of weakly identified models, partly because the majority of studies has been focused on the weak identification issues in linear IV models. That is, 
\begin{align*}
	Y = X \beta + W \eta_Y + U \\
	X = Z \pi + W \eta_X + V.
\end{align*}
where $W$ denotes a set of covariates that possibly include a constant term, $X$ denotes a vector of endogenous variables, and $Z$ denotes a vector of instrument variables. 
In such case we can project out the linear effects of covariates $(W\eta_Y, W\eta_X)$ by the standard Frisch–Waugh–Lovell theorem so it reduces to the standard IV model without covariates:
\begin{align*}
	Y^{\perp W} = X^{\perp W}\beta + U^{\perp W} \\
	X^{\perp W} = Z^{\perp W}\pi + V^{\perp W}, 
\end{align*}
where $Y^{\perp W}$ denotes the residual of $Y$ from regressing on $W$. This enable us to discuss the weak identification-robust inference on $\beta$ by making assumptions on the joint distribution of $(Y^{\perp W}, X^{\perp W}, Z^{\perp W})$ instead of the distribution of $(Y, X, Z, W)$. This argument has been used in a sequence of literature.\footnote{For full vector inference, see \citet{stock/staiger:1997, kleibergen:2002, moreira:2003}, albeit \citet{moreira:2003} only considers the projection of the instrument $Z$ on exogenous covariates $W$. For subvector inference, see \citet{gkmc:2012, gkm:2019}.} With a set of high-dimensional covariates, \cite{ma:2023} considers identification-robust inference on the doubly-robust estimand of LATE.

In nonlinear models subject to weak identification, partialing out the effects of covariates becomes significantly more challenging. Specifically, inference on weakly identified target parameters becomes more complicated due to the nonstandard asymptotic behavior of estimators for covariate coefficients. To address identification-robust inference in models with covariates, the existing literature often assumes that the estimation of covariate effects is consistent and asymptotically normal (CAN) when the values of weakly identified coefficients are fixed under the null hypothesis. In other words, this means that covariate effects are assumed to be "strongly identified," while target parameters may be weakly identified. This property usually aligns with the assumption that covariates are "exogenous", in the sense that covariates are independent of structural unobservables but are not necessarily excluded from the outcome equation. 

Building on this framework, several studies have imposed the exogeneity of covariates to facilitate identification-robust inference in weakly identified models. For instance, in the consumption capital asset pricing model with constant relative risk aversion preferences, \cite{stock/wright:2000} treat the discount factor as strongly identified, while the utility parameter is potentially weakly identified. In the case of IV quantile regressions  \citep{chernozhukov/christian:2005}, \cite{andrews/mikusheva:2016b} assume that the estimation of covariate effects is CAN once the endogenous coefficients are known, allowing the use of a consistent null-imposed estimator of covariate coefficients. Similarly, in the endogenous Probit model,\footnote{In the same context,  \cite{andrews/cheng:2014} develop an identification-robust inference method by adjusting critical values for conventional tests (such as Wald and QLR). This approach also assumes that covariate effects can be consistently estimated under the null values of weakly identified parameters.} \cite{andrews/guggenberger:2019} explore robust inference for the scalar endogenous coefficient, assuming that other covariates are exogenous and their effects can be consistently estimated under the null hypothesis. In the context of MTE models, this assumption leads to an additive separable structure on the outcome equation $\Exp[Y_d \mid U, W]$ and a parametric specification for the first-stage regression. Relatedly, \cite{han/mccloskey:2019} analyze models with nearly singular Jacobians, which include the Normal MTE model \citep{bjorklund/moffitt:1987} as a special case.\footnote{However, the moment conditions formed by two-stage regressions \citep{han/mccloskey:2019} are not recommended, as they lead to a singular asymptotic variance matrix shown in Appendix \ref{appendix:control_function}.} While the Normal MTE model is not their primary focus, they also impose exogeneity and parametric assumptions to facilitate their analysis.

\newpage

\section{Two-stage Regression Approach}
\label{appendix:control_function}

In this appendix, I show that the moment function derived from the commonly used two-stage regression approach leads to singular variance under the failure of identification. Non-singularity of moment covariance matrix is crucial for conducting uniformly valid identification-robust inference on structural parameters. This assumption has been commonly imposed by various important strands of literature, for example, the conditional inference procedure \citep[Assumption 2]{andrews/mikusheva:2016b}, the least-favorable critical value approach \citep[Assumption GMM 4\textsuperscript{*}]{andrews/cheng:2014}, and the well-known Robust LM and CLR tests \citep[Eq. (3.3)]{andrews/guggenberger:2017}.

Consider a parametric MTE model proposed by \cite{kline/walters:2019}:
\begin{align*}
    \Exp[Y_d\mid U = u] &= \mu_d + \rho_d h(u) \\
    D &= \indicator[U \leq p(Z)],
\end{align*}
where $U$ is normalized to be uniformly distributed on the unit interval $[0,1]$, $Z$ is a discrete exogenous instrument independent of both $U$ and $Y_d$, and $h(u)$ is a strictly increasing continuous function with $\Exp[h(U)] = 0$. This specification incorporates  linear MTE model \citep{brinch/mogstad/wiswall:2017}, Heckman normal MTE model \citep{bjorklund/moffitt:1987}, and Logit selection model \citep{dubin/mcfadden:1984} as special examples. 

Under the MTE assumptions, we have 
\begin{equation}
\label{eq:CF_regression}
    \Exp[Y\mid Z, D=d] = \mu_d + \rho_d\lambda_d(p(Z))
\end{equation}
where $\lambda_0: (0,1)\to\R$ and $\lambda_1: (0,1)\to\R$ are the control functions defined as below:
\[
    \lambda_1(p) = \Exp[h(U)\mid U \leq p], ~~
    \lambda_0(p) = \Exp[h(U)\mid U > p].
\]
By the separate regression \eqref{eq:CF_regression}, we find that $\mu_d$ and $\rho_d$ are point identified once $p(Z)$ has nontrivial variation for both treated and control groups. Therefore, weak identification occurs when propensity scores $\{p(z): z = z_0,\ldots,z_K\}$ converge to a singleton under a sequence of DGPs.

Typically people estimate propensity score on the first stage and run the second-stage regression \eqref{eq:CF_regression} to obtain estimators of $\mu_d$ and $\rho_d$. Inference on those parameters can be achieved based on the following moment conditions:
\begin{align*}
    \Exp_{\theta^*}[g^{\text{TS}}(\theta)] = 
    0_{(K+5)\times 1} \quad \text{if~} \theta = \theta^*,
\end{align*}
where $\theta = (\mu_0, \mu_1, \rho_0, \rho_1, \{p(z): z = z_0, \ldots, z_K\})$ is an arbitrary vector of parameters, $\theta^*$ denotes the true value of parameters, and 
\begin{equation}
\label{eq:moment_two_stage}
    g^{\text{TS}}(\theta) \equiv 
    \begin{pmatrix}
        [Y - \mu_0 - \rho_0 \lambda_0(p(Z))] \times \indicator[D=0]\\
        [Y - \mu_0 - \rho_0 \lambda_0(p(Z))] \times \indicator[D=0] \times \lambda_0(p(Z)) \\
        [Y - \mu_1 - \rho_1 \lambda_1(p(Z))] \times \indicator[D=1] \\
        [Y - \mu_1 - \rho_1 \lambda_1(p(Z))] \times \indicator[D=1] \times \lambda_1(p(Z)) \\
        \indicator[Z=z_0] (p(z_0) - \indicator[D=1]) \\
        \vdots \\
        \indicator[Z=z_K] (p(z_K) - \indicator[D=1])
    \end{pmatrix}.
\end{equation}

\begin{proposition}
\label{prop:singular_Var}
Suppose the parameter $\theta$ satisfies $p(z_0) = p(z_1) = \ldots = p(z_K)$, then
\[
    \var_{\theta^*}(g^{\text{TS}}(\theta)) \equiv \Exp_{\theta^*}\left[\left\{g^{\text{TS}}(\theta) - \Exp_{\theta^*}(g^{\text{TS}}(\theta))\right\}\left\{g^{\text{TS}}(\theta) - \Exp_{\theta^*}(g^{\text{TS}}(\theta))\right\}'\right]
\]
is singular for all $\theta^*$ in the parameter space.
\end{proposition}
\begin{proof}[Proof of Proposition \ref{prop:singular_Var}]
    Since propensity scores all equal to $p(z_0)$, this implies the first moment condition can be written as
    \[
        g_1^{\text{TS}}(\theta) \equiv \indicator[D=1]\times [Y - \mu_0 - \rho_0 \lambda_0(p(Z))] =  \indicator[D=1]\times [Y - \mu_0 - \rho_0\lambda_0(p(z_0))]
    \]
    while the second moment condition gives
    \begin{align*}
        g_2^{\text{TS}}(\theta) 
        &\equiv  \indicator[D=1]\times [Y-\mu_0-\rho_0(p(Z))] \times \lambda_0(p(Z)) \\
        &= \indicator[D=1]\times [Y-\mu_0-\rho_0(p(z_0))] \times \lambda_0(p(z_0)) \\
        &= g_1^{\text{TS}}(\theta) \lambda_0(p(z_0)).
    \end{align*}
    The moment function $g^{\text{TS}}(\theta)$ is multicollinear for each $\theta^*$ in the parameter space. Therefore, its variance is singular.
\end{proof}

The singularity of moment variance for two-stage regressions also appears after reparametrizing the model along the lines of \cite{han/mccloskey:2019}\footnote{This paper derives the a useful reparametrization technique for models with nearly singular Jacobian, which includes the Heckman selection models and the parametric MTE models studied here (see their Example 2.1 and 2.2). The reparametrized model follows the structure posited by \cite{andrews/cheng:2012, andrews/cheng:2013, andrews/cheng:2014} and therefore allowing  inference on functions of model parameters. However, as argued in Appendix \ref{appendix:lit_subvector_inference}, their analysis cannot be extended to polynomial MTE models with higher-order polynomials with multiple selection coefficients.}. Applying their reparametrization technique to the moment condition $g^{\text{TS}}$ above gives the reparametrized moment:
\begin{equation}
\label{eq:moment_repara}
    g^{\text{RP}}(\tilde{\theta}) 
    =
    \begin{pmatrix}
    [Y - \alpha_0 - \rho_0\left\{\lambda_0(p(Z)) - \lambda_0(p(z_0))\right\}] \times \indicator[D=0] \\
    [Y - \alpha_0 - \rho_0\left\{\lambda_0(p(Z)) - \lambda_0(p(z_0))\right\}] \times \indicator[D=0] \times \lambda_0(p(Z)) \\
    [Y - \alpha_1 - \rho_1\left\{\lambda_1(p(Z)) - \lambda_1(p(z_0))\right\}] \times \indicator[D=1] \\
    [Y - \alpha_1 - \rho_1\left\{\lambda_1(p(Z)) - \lambda_1(p(z_0))\right\}] \times \indicator[D=1] \times \lambda_1(p(Z)) \\
    \indicator[Z=z_0] (p(z_0) - \indicator[D=1]) \\
    \vdots \\
    \indicator[Z=z_K] (p(z_K) - \indicator[D=1])
    \end{pmatrix}
\end{equation}
with the reparametrized parameter $\tilde{\theta} = (\alpha_0, \alpha_1, \rho_0, \rho_1, \{p(z): z = z_0, \ldots, z_K\})$ and $\alpha_d = \mu_d - \rho_d \lambda_d(p(z_0))$ for $d = 0,1$. It is clear that $\theta \mapsto \tilde{\theta}$ is a one-to-one mapping, and the reparametrized intercept coefficient $\alpha_d$ is strongly identified from this model for all possible values of propensity scores $\{p(z): z = z_0, z_1,\ldots, z_K\}$, whereas $(\rho_0,\rho_1)$ become weakly identified under limited variation of propensity scores. We can partial out the strongly-identified parameters (see Appendix \ref{appendix:lit_subvector_inference} for the discussion on this approach) to conduct inference on $(\rho_0, \rho_1)$ based on the reparameterized moment \eqref{eq:moment_repara}. However, the validity of such procedure requires the variance of moment condition $g^{\text{RP}}$ to be nonsingular, which does not hold as this moment condition is constructed by two-stage regressions:
\begin{corollary}
Suppose the parameter $\tilde{\theta}$ satisfies $p(z_0) = p(z_1) = \ldots = p(z_K)$, then 
\[
    \var_{\tilde{\theta}^*} (g^{\text{RP}}(\tilde{\theta})) \equiv 
    \Exp_{\tilde{\theta}^*}\left[\left\{g^{\text{RP}}(\tilde{\theta}) - \Exp_{\tilde{\theta}^*}(g^{\text{RP}}(\tilde{\theta}))\right\}\left\{g^{\text{RP}}(\tilde{\theta}) - \Exp_{\tilde{\theta}^*}(g^{\text{RP}}(\tilde{\theta}))\right\}'\right]
\]
is singular for all values of $\tilde{\theta}^*$ in the parameter space.
\end{corollary}
\begin{proof}
Similar to the proof of Proposition \ref{prop:singular_Var}, we note that the second moment condition is a factor $\lambda_0(p(z_0))$ of the first moment condition. Therefore, the singularity follows by the multicollineairty. 
\end{proof}

Although the current setup focuses on discrete instrument without covariates, it is noteworthy that such singularity continues to exist even if instrument is continuous and we estimate this model under additive separability if covariates are included. The source of singularity arises from the zero variation of control function variables $\lambda_d(p(Z))$ under the failure of identification, regardless of how $p(z)$ is estimated from the data. In order to achieve robustness against weak instruments, we need to base on other valid moment conditions derived from the model rather than using the ones derived by two-step regression approach. Recently, \cite{andrews/guggenberger:2019} extend the robustness towards the singularity by selecting and rotating moment conditions that correspond to positive eigenvalues of covariance matrix. However, a simpler solution to address singularity in this setting is to replace the factor $\lambda_d(p(Z))$ multiplied by the residual of the second stage regressions in \eqref{eq:moment_two_stage} and \eqref{eq:moment_repara} with some other functions of instrument $Z$ that has positive variation under identification failure. This is employed in the moment conditions \eqref{eq:linsystem} constructed in this paper so that singularity problem does not occur.

\putbib
\end{bibunit}

\end{document}